\documentclass[11pt]{article}
\usepackage[margin=0.75in]{geometry}
\usepackage{graphicx}
\usepackage{amsmath,amssymb,amsfonts,amsthm,mathrsfs}
\usepackage{color}
\usepackage{times}
\usepackage[small,bf]{caption}
\usepackage{makecell}
\newcommand{\mylabel}[2]{#2\def\@currentlabel{#2}\label{#1}}
\makeatother
\usepackage{hyperref}
\usepackage{cases}
\usepackage{xspace}
\usepackage{adjustbox}
\usepackage{multirow}

\usepackage{enumitem}
\setlist{nosep}

\usepackage[normalem]{ulem}

\newcommand{\nc}{\newcommand}
\nc{\DMO}{\DeclareMathOperator}
\nc\m[2]{m_{#1}(#2)}
\nc{\BR}{\mathbb{R}}
\nc{\BN}{\mathbb{N}}
\nc{\BZ}{\mathbb{Z}}
\nc{\ep}{\varepsilon}
\renewcommand{\epsilon}{\varepsilon}
\nc{\ra}{\rightarrow}

\nc{\Rprot}{R}
\nc{\Sprot}{S}
\nc{\Aprot}{A}
\nc{\Pdist}{D}
\nc{\Qdist}{F}
\nc{\pdist}{d}
\nc{\qdist}{f}
\nc{\DP}{differentially private\xspace}

\nc{\SD}{\mathscr{D}}
\nc{\la}{\lambda}
\DMO{\KL}{KL}
\DMO{\Unif}{Unif}
\DMO{\SQ}{SQ}
\nc{\nn}{\varnothing}

\renewcommand{\P}{\mathbb{P}}
\nc{\pp}{p}
\nc{\PP}{P}
\nc{\QQ}{Q}
\nc{\DD}{D}

\DMO{\RAPPOR}{{RAPPOR}}
\nc{\RAP}{\RAPPOR}
\DMO{\RR}{RR}
\nc{\MD}{\mathcal{D}}
\nc{\MO}{\mathcal{O}}
\nc{\MM}{\mathcal{M}}
\nc{\MZ}{\mathcal{Z}}
\nc{\MU}{\mathcal{U}}
\nc{\MP}{\mathcal{P}}
\nc{\poly}{\mathrm{poly}}
\DMO{\treesum}{TreeSum}
\DMO{\lapsum}{LapSum}
\DMO{\checksum}{CheckSum}
\nc{\MDts}{\MD_{\treesum}}
\nc{\MDls}{\MD_{\lapsum}}
\nc{\MDcs}{\MD_{\checksum}}
\nc{\MC}{\mathcal{C}}
\nc{\MT}{\mathcal{T}}
\nc{\MS}{\mathcal{S}}
\nc{\MX}{\mathcal{X}}
\nc{\MY}{\mathcal{Y}}
\nc{\MA}{\mathcal{A}}
\nc{\MB}{\mathcal{B}}
\nc{\MJ}{\mathcal{J}}
\nc{\MF}{\mathcal{F}}
\nc{\MQ}{\mathcal{Q}}
\nc{\p}{\mathbb{P}}
\nc{\E}{\mathbb{E}}
\nc{\tablesize}{s}
\DMO{\Hist}{hist}
\nc{\hist}{\mathrm{hist}}

\nc{\ba}{\mathbf{a}}
\nc{\bx}{\mathbf{x}}
\nc{\by}{\mathbf{y}}
\nc{\bz}{\mathbf{z}}

\DMO{\sr}{sr}
\DMO{\Med}{Med}
\DMO{\Ber}{Ber}
\DMO{\Bin}{Bin}
\DMO{\Had}{Had}
\renewcommand{\matrix}{\mathrm{matrix}}
\nc{\ME}{\mathcal{E}}
\DMO{\View}{View}
\nc{\B}{B}
\nc{\M}{M}
\nc{\ha}{\kappa}
\nc{\hk}{k}
\DMO{\pre}{pre}
\nc{\MH}{\mathcal{H}}
\DMO{\FO}{FO}
\DMO{\CM}{CM}
\DMO{\hb}{\beta}
\nc{\MW}{\mathcal{W}}

\nc{\BB}{\{0,1\}}
\nc{\bW}{\mathbf{W}}
\nc{\eell}{\ell}
\nc{\EELL}{L}
\nc{\q}{q}

\makeatletter
\newtheorem*{rep@theorem}{\rep@title}
\newcommand{\newreptheorem}[2]{%
\newenvironment{rep#1}[1]{%
 \def\rep@title{#2 \ref{##1}}%
 \begin{rep@theorem}}%
 {\end{rep@theorem}}}
\makeatother

\newtheorem{theorem}{Theorem}[section]
\newtheorem{corollary}[theorem]{Corollary}
\newtheorem{lemma}[theorem]{Lemma}
\newtheorem{informal theorem}[theorem]{Informal Theorem}
\newtheorem{fact}[theorem]{Fact}
\newtheorem{claim}[theorem]{Claim}
\newreptheorem{lemma}{Lemma}

\theoremstyle{definition}
\newtheorem{defn}{Definition}[section]
\newtheorem{remark}{Remark}[section]

\usepackage{xcolor}

\usepackage[boxruled,vlined,nofillcomment,linesnumbered]{algorithm2e}
\usepackage{graphicx}
\usepackage{bbm}
\usepackage{mathabx}
\usepackage{subfig}
\nc{\One}{\mathbbm{1}}
\SetKwProg{Fn}{}{\string:}{}

\usepackage{thm-restate}

\title{On the Power of Multiple Anonymous Messages\footnote{A $1$-page abstract based on this work will be presented at the Symposium on Foundations of Responsible Computing (FORC) 2020.}}

\author{
  Badih Ghazi \hspace*{0.5cm}
  Noah Golowich\thanks{MIT EECS. Supported at MIT by a Fannie \& John Hertz Foundation Fellowship, an MIT Akamai Fellowship, and an NSF Graduate Fellowship.  This work was done while at Google Research.}
  \hspace*{0.5cm}
   Ravi Kumar \hspace*{0.5cm}
   Rasmus Pagh\thanks{Visiting from BARC and IT University of Copenhagen.} \hspace*{0.5cm}
   Ameya Velingker \\
   Google Research \\
   Mountain View, CA \\
   \texttt{badihghazi@gmail.com, nzg@mit.edu, ravi.k53@gmail.com,} \\
   \texttt{pagh@itu.dk, ameyav@google.com}
}

\date{}

\begin{document}
\maketitle

\if 0

\fi

% {\bf \it In the case our paper is accepted our intent is to publish this as a 1-page extended abstract at FORC.}
\begin{abstract}
    
    An exciting new development in differential privacy is the \emph{shuffled} model, in which an anonymous channel enables non-interactive, differentially private protocols with error much smaller than what is possible in the local model, while relying on weaker trust assumptions than in the central model. 
    In this paper, we study basic counting problems in the shuffled model and establish separations between the error that can be achieved in the single-message shuffled model and in the shuffled model with multiple messages per user.  
    
    \medskip
    \noindent
    For the %fundamental 
    problem of \emph{frequency estimation} for $n$ users and a domain of size $B$, we obtain:
        
    \begin{itemize}
    \item A nearly tight lower bound of $\tilde{\Omega}( \min(\sqrt[4]{n}, \sqrt{\B}))$ on the error in the single-message shuffled model. This implies that the protocols obtained from the amplification via shuffling work of Erlingsson et al.~(SODA 2019) and Balle et al.~(Crypto 2019) are essentially optimal for single-message protocols. A key ingredient in the proof is a lower bound on the error of locally-private frequency estimation in the low-privacy (aka high $\epsilon$) regime. For this we develop new techniques to extend the results of Duchi et al.~(FOCS 2013; JASA 2018) and Bassily \& Smith~(STOC 2015), whose techniques were restricted to the high-privacy case.
    
%     \item \noah{Remove this maybe? A nearly tight lower bound of $\Omega\left(\frac{\log{B}}{\log\log{B}}\right)$ on the sample complexity with constant relative error in the single-message shuffled model. This improves on the lower bound of $\Omega(\log^{1/17} B)$ obtained by Cheu et al.~(Eurocrypt 2019).}
    \item Protocols in the \emph{multi-message} shuffled model with $\poly(\log{B}, \log{n})$ bits of communication per user and $\poly\log{B}$ error, which provide an exponential improvement on the error compared to what is possible with single-message algorithms. This implies protocols with similar error and communication guarantees for several well-studied problems such as heavy hitters, $d$-dimensional range counting, M-estimation of the median and quantiles, and more generally sparse non-adaptive statistical query algorithms.
    \end{itemize}

    \medskip
    \noindent
    For the related 
    %\modified{and well-studied} \
    \emph{selection} problem on a domain of size $\B$, we prove:
  \begin{itemize}
    \item A nearly tight lower bound of $\Omega(B)$ on the number of users in the single-message shuffled model. This significantly improves on the $\Omega(B^{1/17})$ lower bound obtained by Cheu et al.~(Eurocrypt 2019), and when combined with their $\tilde{O}(\sqrt{B})$-error multi-message protocol, implies the first separation between single-message and multi-message protocols for this problem.
      \end{itemize}

\end{abstract}

\newpage

\newpage
\tableofcontents

\thispagestyle{empty}
\setcounter{page}{0}

\newpage

\section{Introduction}
\label{sec:intro}

With increased public awareness and the introduction of stricter regulation of how personally identifiable data may be stored and used, user privacy has become an issue of paramount importance in a wide range of practical applications.
While many formal notions of privacy have been proposed (see, e.g.,~\cite{li2007t}), \emph{differential privacy (DP)}~\cite{dwork2006calibrating,dwork2006our} has emerged as the gold standard due to its broad applicability and nice features such as composition and post-processing (see, e.g.,~\cite{dwork2014algorithmic, vadhan2017complexity} for a comprehensive overview).
A primary goal of DP is to enable processing of users' data in a way that (i) does not reveal substantial information about the data of any single user, and (ii) allows the accurate computation of functions of the users' inputs.
The theory of DP studies what trade-offs between privacy and accuracy are feasible for desired families of functions.

Most work on DP has been in the \emph{central} (a.k.a. \emph{curator}) setup, where numerous private algorithms with small error have been devised (see, e.g.,~\cite{blum2008learning,dwork2009complexity,DworkRothBook}).
The premise of the central model is that a curator can access the raw user data before releasing a differentially private output.
In distributed applications, this requires users to transfer their raw data to the curator --- a strong limitation in cases where users would expect the entity running the curator (e.g., a government agency or a technology company) to gain little information about their data.

To overcome this limitation, recent work has studied the \emph{local} model of DP~\cite{kasiviswanathan2008what} (also~\cite{warner1965randomized}), where each individual message sent by a user is required to be private. Indeed, several large-scale deployments of DP in practice, at companies such as Apple~\cite{greenberg2016apple,dp2017learning}, Google~\cite{erlingsson2014rappor,CNET2014Google}, and Microsoft~\cite{ding2017collecting}, have used local DP. 
While estimates in the local model require weaker trust assumptions than in the central model, they inevitably suffer from significant error. For many types of queries, the estimation error is provably larger than the error incurred in the central model by a factor growing with the square root of the number of users.

\paragraph{Shuffled Privacy Model.}

The aforementioned trade-offs have motivated the study of the \emph{shuffled} model of privacy as a middle ground between the central and local models. 
While a similar setup was first studied in cryptography in the work of Ishai et al.~\cite{ishai2006cryptography} on cryptography from anonymity, the shuffled model was first proposed for privacy-preserving protocols by Bittau et al.~\cite{bittau17} in their Encode-Shuffle-Analyze architecture.
In the shuffled setting, each user sends one or more messages to the analyzer using an \emph{anonymous} channel that does not reveal where each message comes from.
This kind of anonymization is a common procedure in data collection and is easy to explain to regulatory agencies and users.
The anonymous channel is equivalent to all user messages being randomly shuffled (i.e., permuted) before being operated on by the analyzer, leading to the model illustrated in Figure~\ref{fig:esa}; see Section~\ref{sec:model} for a formal description of the shuffled model.
In this work, we treat the shuffler as a black box, but note that various efficient cryptographic implementations of the shuffler have been considered, including onion routing, mixnets, third-party servers, and secure hardware (see, e.g.,~\cite{ishai2006cryptography,bittau17}). A comprehensive overview of recent work on anonymous communication can be found on Free Haven's Selected Papers in Anonymity website%
\footnote{\url{https://www.freehaven.net/anonbib/}}.
%\cite{anon_comm}} \badih{not sure if there's a way to make this citation more standard?}).

\begin{figure}[b]
\centering
\includegraphics[width=0.7\textwidth]{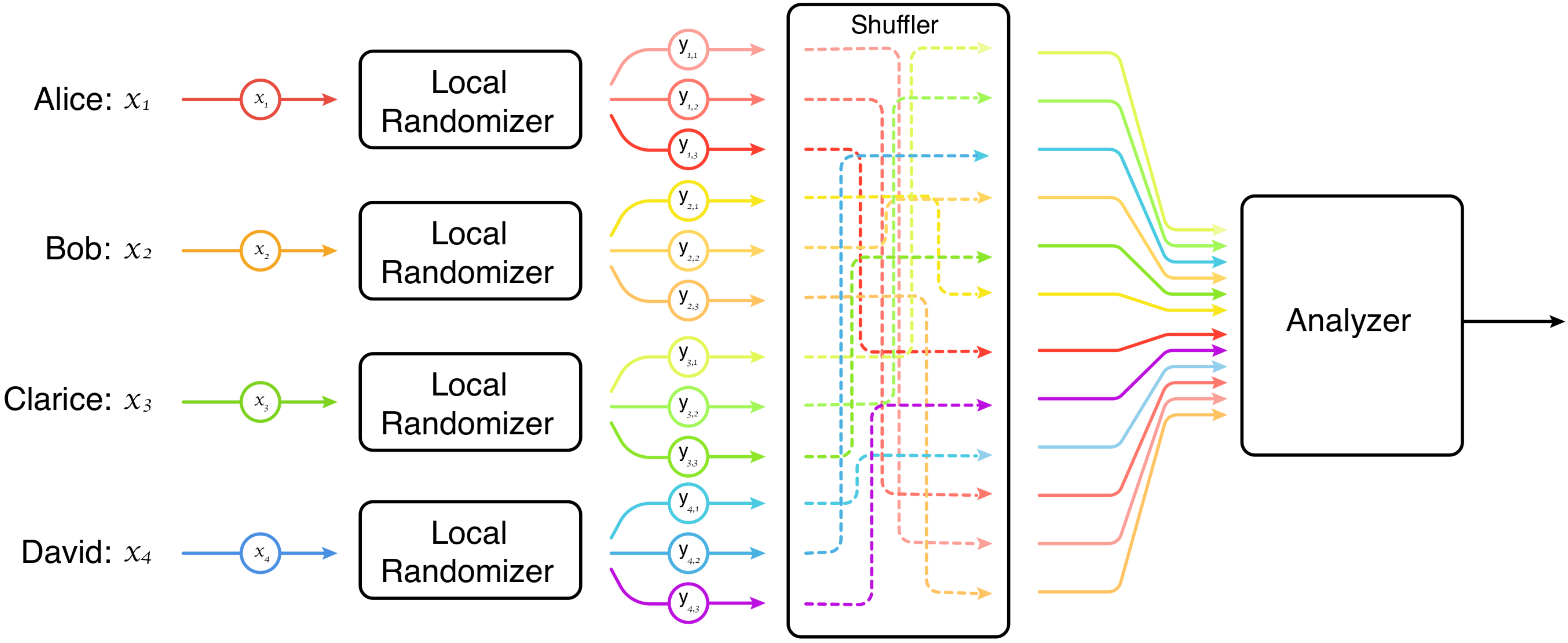}
\caption{Computation in the shuffled model consists of local randomization of inputs in the first stage, followed by a shuffle of all outputs of the local randomizers, after which the shuffled output is passed on to an analyzer.}
\label{fig:esa}
\end{figure}

The DP properties of the shuffled model were first analytically studied, independently, in the works of Erlingsson et al.~\cite{erlingsson2019amplification} and Cheu et al.~\cite{cheu_distributed_2018}. 
Protocols within the shuffled model are non-interactive and fall into two categories: 
\emph{single-message} protocols, in which each user sends one message (as in the local model), and \emph{multi-message} protocols, in which a user can send more than one message.
In both variants, the messages sent by all users are shuffled before being passed to the analyzer. The goal is to design private protocols in the shuffled model with as small error and total communication as possible.  An example of the power of the shuffled model was established by Erlingsson et al.~\cite{erlingsson2019amplification} and extended by Balle et al.~\cite{balle_privacy_2019}, who showed that every local DP algorithm directly yields a single-message protocol in the shuffled model with significantly better privacy.  

\subsection{Results}
In this work, we study several basic problems related to \emph{counting} in the shuffled model of DP.
In these problems, each of $n$ users holds an element from a domain of size $B$.
We consider the problems of frequency estimation, variable selection, heavy hitters, median, and range counting and study whether it is possible to obtain $(\varepsilon,\delta)$-DP%
\footnote{Formally stated in Definition~\ref{def:dp}.}
in the shuffled model with accuracy close to what is possible in the central model, while keeping communication low.

The \emph{frequency estimation} problem (a.k.a. \emph{histograms} or \emph{frequency oracles}) is at the core of all the problems we study.
In the simplest version, each of $n$ users gets an element of a domain $[B] := \{1,\dots,B\}$ and the goal is to estimate the number of users holding element $j$, for any query element $j \in [B]$. 
Frequency estimation has been extensively studied in DP where in the central model, the smallest possible error is $\Theta(\min(\log(1/\delta)/\epsilon, \log(B)/\epsilon, n))$ (see, e.g.,~\cite[Section 7.1]{vadhan2017complexity}). By contrast, in the local model of DP, the smallest possible error is known to be $\Theta(\min(\sqrt{n \log(\B)}/\epsilon, n)$ under the assumption that $\delta < 1/n$~\cite{bassily2015local}. 

In the high-level exposition of our results given below, we let $n$ and $B$ be any positive integers, $\epsilon > 0$ be any constant, and $\delta > 0$ be inverse polynomial in $n$. This assumption on $\epsilon$ and $\delta$ covers a regime of parameters that is relevant in practice. We will also make use of tilde notation (e.g., $\tilde O$, $\tilde{\Theta}$) to indicate the possible suppression of multiplicative factors that are polynomial in $\log B$ and $\log n$.

\paragraph{Single-Message Bounds for Frequency Estimation.}

For the frequency estimation problem, we show the following results in the shuffled model where each user sends a single message.
\begin{theorem}[Informal version of Theorems~\ref{thm:freq_lb} \& \ref{thm:freq_ub}]\label{th:histograms_single_message}
The optimal error of private frequency estimation in the single-message shuffled model is $\tilde{\Theta}( \min(\sqrt[4]{n}, \sqrt{\B}))$. 
%For large domain size, $B > 2^{\sqrt{n}}$, we get particularly tight bounds: 
%the optimal error is at most $O(\min(n,\log{B}))$ and at least $\Omega\left(\min\left(n,\frac{\log{B}}{\log\log{B}}\right)\right)$.
\end{theorem}
The main contribution of Theorem~\ref{th:histograms_single_message} is the lower bound.  To prove this result, we obtain improved bounds on the error needed for frequency estimation in local DP in the weak privacy regime where $\epsilon$ is around $\ln{n}$.  The upper bound in Theorem~\ref{th:histograms_single_message}  follows by combining the recent
result of Balle et al.~\cite{balle_privacy_2019} (building on the earlier result of Erlingsson et al.~\cite{erlingsson2019amplification}) with RAPPOR \cite{erlingsson2014rappor} and $B$-ary randomized response \cite{warner1965randomized} (see Section~\ref{sec:overview_lb} and Appendix~\ref{apx:freq_ub} for more details).

Theorem~\ref{th:histograms_single_message} implies that in order for a single-message \DP protocol to get error $o(n)$ one needs to have $n = \omega\left( \frac{\log \B}{\log \log \B} \right)$ users; see Corollary~\ref{cor:freq_lb_const}.
This improves on a result of Cheu et al.~\cite[Corollary 32]{cheu_distributed_2018}, which gives a lower bound of $n = \omega(\log^{1/17} \B)$ for this task.

\paragraph{Single-Message Bounds for Selection.}

It turns out that the techniques that we develop to prove the lower bound in Theorem~\ref{th:histograms_single_message} can be used to get a nearly tight $\Omega(B)$ lower bound on the number of users necessary to solve the \emph{selection} problem. 
In the selection problem\footnote{Sometimes also referred to as \emph{variable selection}.}, each user $i \in [n]$ is given an arbitrary subset of $[B]$, represented by the indicator vector $x_i \in \{0,1\}^B$, and the goal is for the analyzer to output an index $j^* \in [B]$ such that
\begin{equation}\label{eq:variable_selection_def}
\sum_{i \in [n]} x_{i, j^*} \geq \max_{j \in [B]} \sum_{i \in [n]} x_{i,j} - \frac{n}{10}.
\end{equation}
In other words, the analyzer's output should be the index of a domain element that is held by an approximately maximal number of users. The choice of the constant $10$ in \eqref{eq:variable_selection_def} is arbitrary; any constant larger than $1$ may be used.  

The selection problem has been studied in several previous works on differential privacy, and it has many applications to machine learning, hypothesis testing and approximation algorithms (see~\cite{duchi2013local, steinke2017tight,ullman2018tight} and the references therein).
Our work improves an $\Omega(B^{1/17})$ lower bound in the single-message shuffled model due to Cheu et al.~\cite{cheu_distributed_2018}.
For $\epsilon=1$, the exponential mechanism~\cite{mcsherry2007mechanism} implies an $(\epsilon, 0)$-DP algorithm for selection with $n = O(\log{B})$ users in the central model, whereas in the local model, it is known that any $(\epsilon, 0)$-DP algorithm for selection requires $n = \Omega(B \log{B})$ users~\cite{ullman2018tight}. 
Variants of the selection problem appear in several natural statistical tasks such as feature selection and hypothesis testing (see, e.g.,~\cite{steinke2017tight} and the references therein).

\begin{theorem} [Informal version of Theorem~\ref{thm:selection_formal}]
  \label{th:selection_single_message}
For any single-message differentially private protocol in the shuffled model that solves the selection problem given in Equation~(\ref{eq:variable_selection_def}), the number $n$ of users should be~$\Omega(\B)$. 
\end{theorem}
The lower bound in Theorem~\ref{th:selection_single_message} nearly matches the $O(B\log B)$ upper bound on the required number of users that holds even in the local model (and hence in the single-message shuffled model) and that uses the $B$-randomized response~\cite{warner1965randomized, ullman2018tight}.
Cheu et al. \cite{cheu_distributed_2018} have previously obtained a multi-message protocol for selection with $O(\sqrt{B})$ users, and combined with this result Theorem~\ref{th:selection_single_message} yields the first separation between single-message and multi-message protocols for selection.

\paragraph{Multi-Message Protocols for Frequency Estimation.}

We next present (non-interactive) multi-message protocols in the shuffled model of DP for frequency estimation with only \emph{polylogarithmic} error and communication. 
%This improves exponentially on $O(\sqrt{B})$ error obtained by the single-message protocol of Balle et al. \cite{balle_privacy_2019} as well as on the $O(\sqrt[4]{n})$ error obtained by shuffling the $B$-randomized response (see Sections~\ref{sec:overview_lb} and~\ref{sec:overview_ub} below for more details). 
This is in strong contrast with what is possible for any protocol in the single-message shuffled setup where Theorem~\ref{th:histograms_single_message} implies that the error has to grow polynomially with $\min(n,B)$, even with unbounded communication.  In addition to error and communication, a parameter of interest is the query time, which is the time to estimate the frequency of any element $j \in [B]$ from the data structure constructed by the analyzer.
\begin{theorem}[Informal version of Theorems~\ref{thm:hist_full} \&~\ref{thm:hist_full_pub}]\label{thm:hist_intro}
There is a private-coin (resp., public-coin) multi-message protocol in the shuffled model for frequency estimation with error $\tilde{O}(1)$, total communication of $\tilde{O}(1)$ bits per user, and query time $\tilde{O}(n)$ (resp., $\tilde{O}(1)$).
\end{theorem}

Combining Theorems~\ref{th:histograms_single_message} and~\ref{thm:hist_intro} yields the first separation between single-message and multi-message protocols for frequency estimation. 
Moreover, Theorem~\ref{thm:hist_intro} can be used to obtain multi-message protocols with small error and small communication for several other widely studied problems (e.g., heavy hitters, range counting, and median and quantiles estimation), discussed in Section~\ref{sec:applications}.
 Finally, Theorem~\ref{thm:hist_intro} implies the following consequence for statistical query (SQ) algorithms with respect to a distribution $\MD$ on $\MX$ (see Appendix~\ref{app:counting-queries} for the basic definitions). We say that a non-adaptive SQ algorithm $\MA$ making at most $\B$ queries $q : \MX \ra \{ 0,1\}$ is {\it $k$-sparse} if for each $x \in \MX$, the Hamming weight of the output of the queries is at most $k$.
 %it is always the case that at most $k$ of the queries take the value $1$ on $x$. 
 Then, under the assumption that users' data is drawn i.i.d. from $\MD$, the algorithm $\MA$ can be efficiently simulated in the shuffled model as follows: % can be phrased in terms of {\it sparse counting queries}; see Appendix~\ref{app:counting-queries} for more details.
  \begin{corollary}[Informal version of Corollary \ref{cor:cq_sparse}]
    \label{cor:sq}
  For any non-adaptive $k$-sparse SQ algorithm $\MA$ with $\B$ queries and $\beta > 0$, there is a (private-coin) shuffled model protocol satisfying $(\ep, \delta)$-DP whose output has total variation distance at most $\beta$ from that of $\MA$, such that the number of users is $n \leq \tilde O \left( \frac{k}{\ep \tau} + \frac{1}{\tau^2} \right)$, and the per-user communication is $\tilde O\left(\frac{k^2}{\ep^2}\right)$, where $\tilde O(\cdot)$ hides logarithmic factors in $\B, n, 1/\delta, 1/\ep$, and $1/\beta$. 
\end{corollary}
Corollary \ref{cor:sq} improves upon the simulation of non-adaptive SQ algorithms in the {\it local model} \cite{kasiviswanathan2008what}, for which the number of users must grow as $\frac{k}{\ep^2 \tau^2}$ as opposed to $\frac{1}{\tau^2} + \frac{k}{\ep \tau}$ in the shuffled model. 
We emphasize that the main novelty of Corollary \ref{cor:sq} is in the regime that $k^2/\ep^2 \ll \B$; in particular, though prior work on low-communication private summation in the shuffled model \cite{cheu_distributed_2018,ghazi2019private,balle_private_2020} implies an algorithm for simulating $\MA$ with roughly the same bound on the number of users $n$ as in Corollary \ref{cor:sq} and communication $\Omega(\B)$, it was unknown whether the communication could be reduced to have logarithmic dependence on $\B$, as in Corollary \ref{cor:sq}. % dependence of communication on $\B$ could be reduced to a logarithmic one. 

\begin{table}[t]
    \centering
    \footnotesize
%\begin{adjustbox}{center}
\begin{tabular}{|l|c|c|c|c|c|c|}
       \hline
          & \multicolumn{2}{|l|}{\bf \thead{Local}} & {\bf \thead{Local + shuffle}} & {\bf \thead{Shuffled,\\ single-message}} & {\bf \thead{Shuffled,\\ multi-message}} & {\bf \thead{Central}}\\
       \hline
         \hline
         Expected max. error & $\tilde{O}(\sqrt{n})$ & $\tilde{\Omega}(\sqrt{n})$ & $\tilde{O}(\min(\sqrt[4]{n}, \sqrt \B))$ & $\tilde{\Omega}( \min(\sqrt[4]{n}, \sqrt{\B}))$ & $\tilde{\Theta}(1)$ & $\tilde{\Theta}(1)$\\
         \hline
         \makecell[l]{Communication\\ per user} & $\Theta(1)$ & any & \makecell{$O(\B)$ (err $\sqrt[4]{n}$)\\ $\tilde{O}(1)$ (err $\sqrt{\B}$)} & any & $\tilde{\Theta}(1)$ & n.a. \\ % $\tilde{\Theta}(1)$\\         
       \hline
      References & \cite{bassily2017practical} & \cite{bassily2015local} &~\cite{warner1965randomized,erlingsson2014rappor,balle_privacy_2019} & Theorems~\ref{thm:freq_ub} \& \ref{thm:freq_lb} & Theorem~\ref{thm:hist_full} &~\cite{mcsherry2007mechanism,steinke2017tight}\\
    \hline
    \end{tabular}
%    \end{adjustbox}
     \caption{Upper and lower bounds on expected maximum error (over all $B$ queries, where the sum of all frequencies is $n$) for frequency estimation in different models of DP. The bounds are stated for fixed, positive privacy parameters $\varepsilon$ and $\delta$, and $\tilde{\Theta}/\tilde{O}/\tilde{\Omega}$ asymptotic notation suppresses factors that are polylogarithmic in $B$ and $n$. The communication per user is in terms of the total number of bits sent. In all upper bounds, the protocol is symmetric with respect to the users, and no public randomness is needed. References are to the first results we are aware of that imply the stated bounds. 
     }
    \label{table:frequency_estimation_results}
  \end{table}

\subsection{Overview of Single-Message Lower Bounds}\label{sec:overview_lb}

We start by giving an overview of the   lower bound of $\tilde \Omega( \min \{ n^{1/4}, \sqrt{\B} \})$ in Theorem~\ref{th:histograms_single_message} on the error of any single-message frequency estimation protocol. We first focus on the case where $n \leq \B^2$ and thus $\min \{n^{1/4}, \sqrt{\B}\} = n^{1/4}$. 
% Henceforth we assume that $n \leq \B^2$.
The main component of the proof in this case is a lower bound of $\tilde \Omega(n^{1/4})$ for frequency estimation for $(\ep_L, \delta_L)$-\emph{local} DP protocols\footnote{Note that we use the subscripts in $\ep_L$ and $\delta_L$ to distinguish the privacy parameters of the \emph{local} model from the $\epsilon$ and $\delta$ parameters (without a subscript) of the shuffled model.} when $\ep_L  =\ln(n) + O(1)$.  While lower bounds for local DP frequency estimation were previously obtained in the seminal works of Bassily and Smith~\cite{bassily2015local} and Duchi, Jordan and Wainwright~\cite{duchi_minimax}, two critical reasons make them less useful for our purposes: (i) their dependence on $\ep_L$ is sub-optimal when $\ep_L = \omega(1)$ (i.e., low error regime) and (ii) they only apply to the case where $\delta_L = 0$ (i.e., pure privacy).\footnote{As we discuss in Remark \ref{rmk:black-box}, generic reductions \cite{cheu_distributed_2018,bun2018heavy} showing that one can efficiently simulate an approximately differentially private protocol (i.e., with $\delta_L > 0$) with a pure differentially private protocol (i.e., with $\delta_L = 0$) are insufficient to obtain tight lower bounds.}  We prove new error bounds in the low error and approximate privacy regime in order to obtain our essentially tight lower bound in Theorem \ref{th:histograms_single_message} for single-message shuffled protocols.  We discuss these and outline the proof next.

Let $R$ be an $(\ep_L, \delta_L)$-locally \DP randomizer. 
The general approach~\cite{bassily2015local, duchi_minimax} is to show that if $V$ is a random variable drawn uniformly at random from $[\B]$ and if $X$ is a random variable that is equal to $V$ with probability  parameter $\alpha \in (0,1)$, and is drawn uniformly at random from $[\B]$ otherwise, then the mutual information between $V$ and the local randomizer output $R(X)$ satisfies
\begin{equation}
  \label{eq:inf_ub_intro}
  I(V; R(X)) \leq  \frac{\log \B}{4n}.
\end{equation}
\noindent
Once~(\ref{eq:inf_ub_intro}) is established, the chain rule of mutual information implies that $I(V; R(X_1), \ldots, R(X_n)) \leq \frac{\log{\B}}{4}$, where $X_1, \ldots, X_n$ are independent and identically distributed given $V$. Fano's inequality~\cite{coverthomas} then implies that the probability that any analyzer receiving  $R(X_1), \ldots, R(X_n)$ correctly guesses $V$ is at most $1/4$; on the other hand, an $\Omega(\alpha n)$-accurate analyzer must be able to determine $V$ with high probability since its frequency in the dataset $X_1, \ldots, X_n$ is roughly $\alpha n$, greater than the frequency of all other $v \in [\B]$. This approach thus yields a lower bound of $\Omega(\alpha n)$ on frequency estimation. 
%\badih{We next describe how previous work upper-bounded the mutual information before outlining the novel ingredients in our approach.}
% a lower bound of $\tilde \Omega(n^{1/4})$ %\badih{parametrize?} on frequency estimation.

To prove the desired $\tilde \Omega(n^{1/4})$ lower bound using this approach, it turns out we need a bound of the form
\begin{equation}
  \label{eq:ep-mutinf}
I(V; R(X)) \leq \tilde O(\alpha^4 n e^{\ep_L}),
\end{equation}
where both $\delta_L > 0$ and $\ep_L = \omega(1)$.  (We will in fact choose 
$\alpha = \tilde \Theta( n^{-3/4})$ and $ \ep_L = \ln(n) + O(1)$; as we will discuss later,~(\ref{eq:ep-mutinf}) is essentially tight in this regime.)

% \badih{We know that the bound in Equation~(\ref{eq:ep-mutinf}) is tight up to logarithmic factors? at least in some regime? if so, we should mention it..}\noah{I only know it's tight for the parameters $\ep = \ln(n), \alpha = n^{-3/4}$, i.e., single-message shuffled parameters. I mentioned that in the upper bound section below. I'm not sure about in general -- one could imagine replacing the RHS with various powers of $\alpha,n,e^{\ep_L}$ to get lots of possible inequalities like the above one, and I haven't put much thought into what would be tight for more general settings of the parameters.}

\paragraph{Limitations of Previous Approaches}

We first state the existing upper bounds on $I(V; R(X))$, which only use the privacy of the local randomizer.  Bassily and Smith \cite[Claim 5.4]{bassily2015local} showed an upper bound of $I(V; R(X)) \leq O(\ep_L^2 \alpha^2)$ with $\ep_L = O(1)$ and $\delta_L = o(1/(n \log n))$, which thus satisfies~(\ref{eq:inf_ub_intro}) with $\alpha = \Theta \left( \sqrt{\frac{\log \B}{\ep_L^2 n}} \right)$. 
For $\delta_L = 0$, Duchi et al. \cite{duchi_minimax} generalized this result to the case $\ep_L \geq 1$, proving that\footnote{This bound is not stated explicitly in \cite{duchi_minimax}, though \cite[Lemma 7]{duchi_minimax} proves a similar result whose proof can readily be modified appropriately.}
$I(V; R(X)) \leq O(\alpha^2 e^{2\ep_L})$. Both of these bounds are weaker than 
%$\alpha^4 n e^{\ep_L} < \alpha^2 e^{2\ep_L}$, and thus 
%(\ref{eq:ep-mutinf}) is stronger than 
(\ref{eq:ep-mutinf}) for the above setting of $\alpha$
and $\ep_L$.  
%if $e^{\ep_L}/(\alpha^2 n) > 1$, which will hold for our choices of the parameters, i.e., $\alpha = \tilde \Theta( n^{-3/4})$ and $ \ep_L = \ln(n) + O(1)$ (and thus $e^{\ep_L}/(\alpha^2 n) = \tilde \Theta( n^{3/2})$). 

However, proving the mutual information bound in (\ref{eq:ep-mutinf}) turns out to be impossible if we only use the privacy of the local randomizers!  In fact, the bound can be shown to be {\it false} if all we assume about $R$ is that it is $(\ep_L, \delta_L)$-locally \DP for some $\ep_L \approx \ln n$ and $\delta_L = n^{-O(1)}$. For instance, it is violated if one takes $R$ to be $R_{\RR}$, the local randomizer of the $\B$-randomized response \cite{warner1965randomized}. Consider for example the regime where $B \leq n \leq \B^2$, and the setting where $R_{\RR}(v)$ is equal to $v$ with probability $1-\B/n$, and is uniformly random over $[\B]$ with the remaining probability of $\B/n$. In this case,  the local randomizer $R_{\RR}(\cdot)$ is $(\ln(n) + O(1), 0)$-differentially private.  A simple calculation shows that $I(V; R_{\RR}(X)) = \tilde\Theta(\alpha)$. Whenever $\alpha \ll 1/\sqrt{n}$, which is the regime we have to consider in order to obtain any non-trivial lower bound\footnote{i.e., any stronger lower bound than what holds even in the \emph{local} model} in the single-message shuffled model, it holds that $\alpha \gg \alpha^4 n \exp(\ln(n))$, thus contradicting~(\ref{eq:ep-mutinf}) (see Remark~\ref{rmk:inf-counterexample}).  The insight derived from this counterexample is actually crucial, as we describe in our new technique next.

%  We remark that the aforementioned previous results based on Fano's inequality~\cite{bassily2015local,duchi_minimax} prove a \modified{weaker} statement, where $R$ is an $(\ep_L, \delta)$-locally \DP randomizer with either $\delta_L = 0$~\cite{duchi_minimax} or $\ep_L = O(1)$~\cite{bassily2015local}. However, a key difficulty in using that technique in our case is that their mutual information upper bound only relies on the \emph{privacy} of the local randomizer, and it turns out that  \modified{Thus, prior approaches which rely entirely upon privacy to upper bound $I(V;R(X))$ are insufficient to show (\ref{eq:inf_ub_intro}).}

\paragraph{Mutual Information Bound from Privacy and Accuracy}

Departing from previous work, we manage to prove the stronger bound (\ref{eq:ep-mutinf}) as follows. Inspecting the counterexample based on the $B$-randomized response outlined above, we first observe that any analyzer in this case must have error at least $\Omega(\sqrt{\B})$, which is larger than $\alpha n$, the error that would be ruled out by the subsequent application of Fano's inequality! This led us to appeal to accuracy, in addition to privacy, when proving the mutual information upper bound. We thus leverage the additional available property that the local randomizer $R$ can be combined with an analyzer $A$ in such a way that the mapping $(x_1, \ldots, x_n) \mapsto A(R(x_1), \ldots, R(x_n))$ computes the frequencies of elements of every dataset $(x_1, \ldots, x_n)$ accurately, i.e., to within an error of $O(\alpha n)$.
At a high level, our approach for proving the bound in~(\ref{eq:ep-mutinf}) then proceeds by:
\begin{enumerate}
    \item[(i)] Proving a structural property satisfied by the randomizer corresponding to any accurate frequency estimation protocol. Namely, we show in Lemma \ref{lem:lb_tvd} that if there is an accurate analyzer, the total variation distance between the output of the local randomizer on any given input, and its output on a uniform input, is close to 1. % of a local randomizer on a uniform input $X$, is close to 1 (i.e., $\Delta(R(v), R(X)) \approx 1$).
    \item[(ii)] Using the $(\ep_L, \delta_L)$-DP property of the randomizer along with the structural property in (i) in order to upper-bound the mutual information $I(V;R(X))$. 
\end{enumerate}
% \badih{Mention somewhere something like: We hope that the structural property in (i) will be of independent interest, given how basic frequency estimation is as a problem, and how connected it is to other algorithmic tasks..}\noah{I think we should be careful here since it is the same as the structural property of CSS -- that's why I was thinking of citing GGK+20, to say that something similar (admiteddly the degenerate case $\B = 2$, so not really our contribution) already has found other applications.} \badih{Is our bound a generalization of CCS? or are they incomparable? Also is our proof technique similar to that of CCS, or is one easier than the other? If the bound of CCS was not enough for our purposes and we had to work harder to prove our bound, one option might be to articulate why this was the case..} \noah{There are some ways in that CSS is more general than us, so we're not quite a generalization. Our proof technique is different from CSS in the same way that GGK+ is different from it -- we both don't condition on ``bad transcripts''. But honestly, I think the best thing for a reader to do for the proof of the structural property in (i) is to read GGK+ since it's phrased better there than in the proof of Lemma 3.15. I think the best way to bring out our contribution is to emphasize the {\it application of (i)} -- took a stab at this below.}
We believe that the application of the structural property in (i) to proving bounds of the form (\ref{eq:ep-mutinf}) is of independent interest. As we further discuss below, this property is, in particular, used (together with privacy of $R$) to argue that for most inputs $v \in [\B]$, the local randomizer output $R(v)$ is unlikely to equal a message that is much less likely occur when the input is uniformly random than when it is $v$. Note that it is somewhat counter-intuitive that accuracy is used in the proof of this fact, as one way to achieve very accurate protocols is to ensure that $R(v)$ is equal to a message which is unlikely when the input is any $u \neq v$. We now outline the proofs of (i) and (ii) in more detail.
  
The gist of the proof of (i) is an anti-concentration statement. Let $v$ be a fixed element of $[B]$ and let $X$ be a random variable uniformly distributed on $[B]$. Assume that the total variation distance  $\Delta(R(v), R(X))$ is not close to $1$, and that a small fraction of the users have input $v$ while the rest have uniformly random inputs. Let $\MZ$ denote the range of the local randomizer $R$. First, we consider the special case where $\MZ$ is $\{ 0,1\}$. Then the distribution of the shuffled outputs of the users with $v$ as their input is in bijection with a binomial random variable with parameter $p := \p[R(v) = 1]$, and the same is true for the distribution of the shuffled outputs of the users with uniform random inputs $X$ (with parameter $q := \p[R(X) = 1]$). Then, we use the anti-concentration properties of binomial random variables in order to argue that if $|p-q| = \Delta(R(v), R(X))$ is too small, then with nontrivial probability the shuffled outputs of the users with input $v$ will be indistinguishable from the shuffled outputs of the users with uniform random inputs. This is then used to contradict the supposed accuracy of the analyzer. To deal with the general case where the range $\MZ$ is any finite set, we repeatedly apply the data processing inequality for total variation distance in order to reduce to the binary case (Lemma \ref{lem:tvd_hist}). The full proof appears in Lemma \ref{lem:lb_tvd}.

Equipped with the property in (i), we now outline the proof of the mutual information bound in (ii).
% \modified{This assumption of an accurate analyzer allows us to conclude in Lemma \ref{lem:lb_tvd} that the total variation distance between the output of a local randomizer on any given input $v$, and the output of a local randomizer on a uniform input $X$, is close to 1 (i.e., $\Delta(R(v), R(X)) \approx 1$). This statement is quite intuitive: if, for some $v$, $\Delta(R(v), R(X))$ is too small, then if a subset of many users have input $v$, one can see by binomial anti-concentration that with nontrivial probability their shuffled outputs will ``look like'' the shuffled outputs $R(X)$ of other users with uniformly random inputs. In particular, an analyzer will not be able to accurately determine how many users hold $v$, contradicting accuracy.}
Denote by
\begin{itemize}
    \item $\MT_v$ % $\MT_v := \left\{ z \left| \frac{\p[R(v) = z]}{\p[R(X) = z]} \gg 1\right.\right\}$
the set of messages {\it much more likely} to occur when the input is $v$ than when it is uniform,
\item $\MY_v$ % $\MY_v := \left\{ z \left| \frac{\p[R(v) = z]}{\p[R(X) = z]} < 1 \right.\right\}$
the set of messages {\it less likely} to occur when the input is $v$ than when it is uniform.
\end{itemize}
Note that the union $\MT_v \cup \MY_v$ is \emph{not} the entire range $\MZ$ of messages; in particular, it does not include messages that are {\it a bit more likely} to occur when the input is $v$ than when it is uniform.\footnote{For clarity of exposition in this overview, we refrain from quantifying the likelihoods in each of these cases; for more details on this, we refer the reader to Section~\ref{sec:interm_lb}.}
On a high level, it turns out that the mutual information $I(V; R(X))$ will be large, i.e., $R(X)$ will reveal a significant amount of information about $V$, if either of the following events occurs:% \badih{Should we switch the order of (a) and (b) since we're currently outlining their proofs in the opposite order? Better to keep the shorter proof first..}\noah{did so}
% \badih{Just to double-check, $\MT_v$ and $\MY_v$ aren't complements here?}\noah{no, because $\MT_v$ requires {\it much} more likely...so they are disjoint but not everything}\badih{I see, it might be a bit difficult for the reader to grasp this distinction in this overview.. do you think we can cheat a bit here and present a simplified version in terms of only $\MT_v$ and $\overline{\MT}_v$ (and say at the end of the overview or in a footnote that the actual proof is more subtle/complicated..?)? If not, it might be good if we can somehow tell the reader why (a) and (b) are not symmetric (in terms of ``more likely'' versus ``much more likely''). Also, regarding the asymmetry, is there a short/elegant explanation why (a) requires DP but not (b)?}\noah{I modified below discussion accordingly. It now explains the asymmetry and why (a) requires DP. I don't know how to explain why (b) does not require DP other than to explain the proof and note that nowhere do we use DP...}
\begin{enumerate}
\item[(a)] There are too many inputs $v \in [\B]$ such that the mass $\p[R(X) \in \MY_v]$ is small. Intuitively, for such $v$, the local randomizer $R$ fails to ``hide'' the fact that a uniform input $X$ is $v$ given that $X$ indeed equals $v$ and $R(X) \in \MY_v$.
\item[(b)] There are too many inputs $v \in [\B]$ such that the mass $\p[R(v) \in \MT_v]$ is large. Such inputs make it too likely that $X = v$ given that $R(X) \in \MT_v$, which makes it more likely in turn that $V=v$. 
\end{enumerate}
We first note that the total variation distance $\Delta(R(v), R(X))$ is upper-bounded by $\p[R(X) \in \MY_v]$. On the other hand, the accuracy of the protocol along with property (i) imply that $\Delta(R(v), R(X))$ is close to $1$. By putting these together, we can conclude that event (a) does not occur (see Lemma \ref{lem:lb_tvd} for more details). % \noah{ probably delete this...We thus use the existence of an accurate analyzer to deduce a lower bound on the total variation distance between the output of the local randomizer for a \emph{fixed} input, and its output for a \emph{random} input; this is done in Lemma~\ref{lem:lb_tvd}. This lower bound is then used to upper-bound the mutual information and establish the inequality in~(\ref{eq:inf_ub_intro}). We note that in Lemma~\ref{lem:lb_tvd}, we need to establish a bound that is asymptotically tight, and to do so it turns out we have to consider a particular non-i.i.d distribution on inputs. This in turn necessitates a lower bound on the total variation distance of histograms of certain mixture distributions (see Lemma \ref{lem:tvd_hist} and Section~\ref{sec:single_msg_bds} for more details). }
% Instead, to prove the inequality in (\ref{eq:inf_ub_intro}), we need to additionally assume that the local randomizer $R$ can be combined with some analyzer $A$ in such a way that the mapping $(x_1, \ldots, x_n) \mapsto A(R(x_1), \ldots, R(x_n))$ computes the frequencies of elements of every dataset $(x_1, \ldots, x_n)$ \emph{accurately} (i.e., to within an un-normalized error of $\tilde O(n^{1/4})$). \modified{This assumption of an accurate analyzer allows us to conclude in Lemma \ref{lem:lb_tvd} that the total variation distance between the output of a local randomizer on any given input $v$, and the output of a local randomizer on a uniform input $X$, is close to 1 (i.e., $\Delta(R(v), R(X)) \approx 1$). This statement is quite intuitive: if, for some $v$, $\Delta(R(v), R(X))$ is too small, then if a subset of many users have input $v$, one can see by binomial anti-concentration that with nontrivial probability their shuffled outputs will ``look like'' the shuffled outputs $R(X)$ of other users with uniformly random inputs. In particular, an analyzer will not be able to accurately determine how many users hold $v$, contradicting accuracy.}

To prove that event (b) does not occur, we use the $(\ep_L, \delta_L)$-DP guarantee of the local randomizer $R$. Namely, we will use the inequality $\p[R(v) \in \MS] \leq e^{\ep_L} \cdot \p[R(X) \in \MS] + \delta$ for various subsets $\MS$ of $\MZ$. Unfortunately, setting $\MS = \MT_v$ does not lead to a good enough upper bound on $\p[R(v) \in \MT_v]$; indeed, for the local randomizer $R = R_{\RR}$ corresponding to the $B$-ary randomized response, we will have $\MT_v = \{ v \}$ for $n \gg \B$, and so $\p[R(v) \in \MT_v] = 1-\B/n \approx 1$ for any $v$. Thus, to establish (b), we need to additionally use the accuracy of the analyzer $A$ (i.e., property (i) above), together with a careful double-counting argument to enumerate the probabilities that $R(v)$ belongs to subsets of $\MT_v$ of different granularity (with respect to the likelihood of occurrence under input $v$ versus a uniform input). For the details, we refer the reader to Section \ref{sec:interm_lb} and Lemma \ref{lem:intermed_mutinf}.

Having established the above lower bound for locally differentially private estimation in the low-privacy regime, the final step is to apply a lemma of Cheu et al.~\cite{cheu_distributed_2018} (restated as Lemma~\ref{lem:pratio} below), stating that any lower bound for $(\ep + \ln(n), \delta)$-locally differentially private protocols implies a lower bound for $(\ep, \delta)$-differentially private protocols in the single-message shuffled model (i.e., we take $\ep_L = \ep + \ln(n)$). Moreover, for $\ep_L = \ln(n) + O(1)$ and $\alpha = \tilde \Theta(n^{-3/4})$, we observe that~(\ref{eq:ep-mutinf}) implies~(\ref{eq:inf_ub_intro}), and thus a lower bound of $\tilde \Omega(\alpha n) = \tilde \Omega(n^{1/4})$ for frequency estimation in the single-message shuffled model follows. Finally, we point out that while the above outline focused on the case where $n \le \B^2$, it turns out that this is essentially without loss of generality as the other case where $n > \B^2$ 
%(so that $\min \{n^{1/4}, \sqrt{\B}\} = \sqrt{\B}$)
can be reduced to the former (see Lemma~\ref{lem:invn}).

\paragraph{Tightness of Lower Bounds} The lower bounds sketched above are nearly tight. The upper bound of Theorem~\ref{th:histograms_single_message} follows from combining existing results showing that the single-message shuffled model provides privacy amplification of locally differentially private protocols~\cite{erlingsson2019amplification,balle_privacy_2019}, with known locally differentially private protocols for frequency estimation~\cite{warner1965randomized,erlingsson2014rappor,duchi_minimax, balle_privacy_2019}. In particular, as recently shown by Balle et al.~\cite{balle_privacy_2019}, a pure $(\ep_L, 0)$-differentially private local randomizer yields a protocol in the shuffled model that is $\left(O\left( e^{\ep_L} \sqrt{\frac{\log(1/\delta)}{n}}\right), \delta\right)$-differentially private and that has the same level of accuracy.\footnote{Note that we cannot use the earlier amplification by shuffling result of~\cite{erlingsson2019amplification}, since it is only stated for $\ep_L = O(1)$ whereas we need to amplify a much less private local protocol, having an $\ep_L$ close to $\ln{n}$.} Then:
\begin{itemize}
    \item When combined with RAPPOR~\cite{erlingsson2014rappor,duchi_minimax}, we get an upper bound of $\tilde O(n^{1/4})$ on the error.
    \item When combined with the $\B$-randomized response~\cite{warner1965randomized,acharya2019hadamard}, we get an error upper bound of $\tilde O(\sqrt{\B})$.
\end{itemize}
The full details appear in Appendix~\ref{apx:freq_ub}. Put together, these imply that the minimum in our lower bound in Theorem~\ref{th:histograms_single_message} is tight (up to logarithmic factors). It also follows that the mutual information bound in Equation~(\ref{eq:ep-mutinf}) is tight (up to logarithmic factors) for $\ep_L = \ln(n) + O(1)$ and $\alpha = n^{-3/4}$ (which is the parameter settings corresponding to the single-message shuffled model); indeed, a stronger bound in Equation~(\ref{eq:ep-mutinf}) would lead to larger lower bounds in the single-message shuffled model thereby contradicting the upper bounds discussed in this paragraph.

\paragraph{Lower Bound for Selection: Sharp Bound on Level-1 Weight of Probability Ratio Functions} We now outline the proof of the nearly tight lower bound on the number of users required to solve the \emph{selection} problem in the single-message shuffled model (Theorem~\ref{th:selection_single_message}).
The main component of the proof in this case is a lower bound of $\Omega(B)$ users for selection for $(\ep_L, \delta_L)$-\emph{local} DP protocols when $\ep_L  =\ln(n) + O(1)$.

In the case of local $(\ep_L, 0)$-DP (i.e., pure) protocols, Ullman~\cite{ullman2018tight} proved a lower bound $n = \Omega\left(\frac{\B \log \B}{(\exp(\ep_L) - 1)^2}\right)$. There are two different reasons why this lower bound is not sufficient for our purposes:
\begin{enumerate}
    \item It does not rule out DP protocols with $\delta_L > 0$ (i.e., approximate protocols), which are necessary to consider for our application to the shuffled model.
    \item For the low privacy setting of $\ep_L  =\ln(n) + O(1)$, the bound simplifies to $ n = \tilde \Omega(\B / n^2)$, i.e.,  $n = \tilde \Omega(\B^{1/3})$, weaker than what we desire.
\end{enumerate}

% At this point, one would hope to extend the result of Ullman~\cite{ullman2018tight}, which shows, for $\ep_L = \ep + \ln n$, a lower bound of $n = \Omega\left(\frac{\B \log \B}{(\exp(\ep_L) - 1)^2}\right) = \Omega\left(\frac{\B \log \B}{n^2}\right)$ for $(\ep_L, 0)$-locally \DP protocols, to the case of $(\ep_L, \delta)$-locally \DP protocols.  While this by itself is non-trivial, even if we could do so, we would only get a lower bound of $ n = \tilde \Omega(\B / n^2)$, i.e.,  $n = \tilde \Omega(\B^{1/3})$, which is substantially weaker than the desired $n = \tilde \Omega(\B)$ lower bound.
To prove our near-optimal lower bound, we remedy both of the aforementioned limitations by allowing positive values of $\delta_L$ and achieving a better dependence on $\epsilon_L$.  As in the proof of frequency estimation, we reduce proving Theorem~\ref{th:selection_single_message} to the task of showing the following mutual information upper bound:
% \badih{Follow the same general structure as in the above proof for frequency estimation? i.e., parametrize Equation~(\ref{eq:sel-intro-mi-ub}), then say what parameters were achieved by Ullman, then state the improved bound that we aim for? then we can keep the blue paragraph below that's specific to our proof?}\noah{did so -- ended up somewhat longer than I originally intended, but it's probably ok}
\begin{equation}
  \label{eq:sel-intro-mi-ub}
I((L,J); R(X_{L,J})) \leq \tilde O \left( \frac{1}{\B} \right) +  O(\delta_L (\B + n)),
\end{equation}
where $L$ is a uniform random bit, $J$ is a uniform random coordinate in $[B]$, and $X_{L,J}$ is uniform over the subcube $\{ x \in \BB^\B : x_J = L \}$. Indeed, once (\ref{eq:sel-intro-mi-ub}) holds and $\delta_L < o(1/(\B n))$, the chain rule implies that the mutual information between all users' messages and the pair $(L,J)$ is at most $O\left(\frac{n\ln(\B)}{\B}\right)$. It follows by Fano's inequality that if $n = o(\B)$, no analyzer can determine the pair $(L,J)$ with high probability (which any protocol for selection must be able to do).

For any message $z$ in the range of $R$, define the Boolean function $f_z(x) := \frac{\p[R(x) = z]}{\p[R(X_{L,J}) = z]}$ where $x \in \BB^\B$. Let $\bW^1[f]$ denote the level-$1$ Fourier weight of a Boolean function $f$.  To prove inequalities of the form (\ref{eq:sel-intro-mi-ub}), the prior work of Ullman~\cite{ullman2018tight} shows that $I((L,J); R(X_{L,J}))$ is determined by  $\bW^1[f_z]$, up to normalization constants.  In the case where $\delta_L = 0$ and $\ep_L = \ln(n) + O(1)$, $f_z \in [0,e^{\ep_L}]$, and by Parseval's identity $\bW^1[f_z] \leq O(e^{2\ep_L})$ for any message $z$, leading to
\begin{equation}
\label{eq:ullman-mutinf}
I((L,J); R(X_{L,J})) \leq O\left(\frac{e^{2\ep_L}}{\B}\right).
\end{equation}
Unfortunately, for our choice of $\ep_L = \ln(n) + O(1)$, (\ref{eq:ullman-mutinf}) is weaker than (\ref{eq:sel-intro-mi-ub}).

To show (\ref{eq:sel-intro-mi-ub}), we depart from the previous approach in the following ways:
%, even when we restrict to the (insufficient) case where $\delta_L = 0$, we depart from the previous approach as folows:
\begin{enumerate}
    \item[(a)] We show that the functions $f_z$ take values in $[0, O(e^{\ep_L})]$ for {\it most} inputs $x$; this uses the $(\ep_L, \delta_L)$-local DP of the local randomizer $R$.
    \item[(b)] Using the {\it Level-$1$ inequality} from the analysis of Boolean functions \cite{odonnell2014analysis} (see Theorem \ref{thm:level1} below), we upper bound $\bW^1[g_z]$ by $O(\ep_L)$, where $g_z$ is the truncation of $f_z$ defined by $g_z(x) = f_z(x)$ if $f_z(x) \leq O(n)$, and $g_z(x) = 0$ otherwise.
    \item[(c)] We bound $I((L,J); R(X_{L,J}))$ by $\bW^1[g_z]$, using the fact $f_z$ is sufficiently close to its truncation $g_z$. % \badih{say something about how we go from $\bW^1[g_z]$ to the mutual information upper bound?}\noah{did so}
\end{enumerate}

The above line of reasoning, formalized in Section  \ref{sec:selection_lb}, allows us to show
$$
I((L,J); R(X_{L,J})) \leq O \left( \frac{\ep_L}{\B} + \delta \cdot (\B + e^{\ep_L}) \right),
$$
which is sufficient to establish that (\ref{eq:sel-intro-mi-ub}) holds.
% \badih{Maybe elaborate on what the {\it level-1 inequality} is, and how we apply in our context?} \noah{see above}

Having proved a lower bound on the error of any $(\ep + \ln n, \delta)$-local DP protocol for selection with $\ep = O(1)$, the final step in the proof is to apply a lemma of \cite{cheu_distributed_2018} to deduce the desired lower bound in the single-message shuffled model. % For more details, we refer the reader to Section~\ref{sec:selection_lb}.

\subsection{Overview of Multi-Message Protocols}\label{sec:overview_ub}

An important consequence of our lower bound in Theorem~\ref{th:histograms_single_message} is that one cannot achieve an error of $\tilde O(1)$ using \emph{single-message} protocols. This in particular rules out any approach that uses the following natural two-step recipe for getting a private protocol in the shuffled model with accuracy better than in the local model:
\begin{enumerate}
    \item Run any known locally differentially private protocol with a setting of parameters that enables high-accuracy estimation at the analyzer, but exhibits low privacy locally.
    \item Randomly shuffle the messages obtained when each user runs step $1$ on their input, and use the privacy amplification by shuffling bounds~\cite{erlingsson2019amplification, balle_privacy_2019} to improve the privacy guarantees.
\end{enumerate}
Thus, shuffled versions of  the $B$-randomized response~\cite{warner1965randomized,acharya2019hadamard}, RAPPOR~\cite{erlingsson2014rappor,duchi_minimax,acharya2019hadamard}, the Bassily--Smith protocol~\cite{bassily2015local}, TreeHist and Bitstogram~\cite{bassily2017practical}, and the Hadamard response protocol~\cite{acharya2019hadamard,acharya2019communication}, will still incur an  error of $\Omega(\min(\sqrt[4]{n}, \sqrt{B}))$.

Moreover, although the single-message protocol of Cheu et al.~\cite{cheu_distributed_2018} for binary aggregation (as well as the multi-message protocols given in~\cite{ghazi2019scalable, DBLP:journals/corr/abs-1906-09116, ghazi2019private, balle_privacy_2019constantIKOS} for the more general task of real-valued aggregation) can be applied to the one-hot encodings of each user's input to obtain a multi-message protocol for frequency estimation with error $\tilde{O}(1)$, the communication per user would be $\Omega(B)$ bits, which is clearly undesirable.

Recall that the main idea behind (shuffled) randomized response is for each user to send their input with some probability, and random noise with the remaining probability. Similarly, the main idea behind (shuffled) Hadamard response is for each user to send a uniformly random index from the support of the Hadamard codeword corresponding to their input with some probability, and a random index from the entire universe with the remaining probability. In both protocols, the user is sending a message that either depends on their input or is noise; this restriction turns out to be a significant limitation.  Our main insight is that multiple messages allows users to simultaneously send both types of messages, leading to a sweet spot with exponentially smaller error or communication.

\paragraph{Our protocols.}
We design a multi-message version of the private-coin Hadamard response of Acharya et al.~\cite{acharya2019hadamard,acharya2019communication} where each user sends a small \emph{subset} of indices sampled uniformly at random from the support of the Hadamard codeword corresponding to their input, and in addition sends a small subset of indices sampled uniformly at random from the entire universe $[B]$. 
To get accurate results it is crucial that a subset of indices is sampled, as opposed to just a single index (as in the local model protocol of~\cite{acharya2019hadamard,acharya2019communication}). We show that in the regime where the number of indices sampled from inside the support of the Hadamard codeword and the number of noise indices sent by each user are both logarithmic, the resulting multi-message algorithm is private in the shuffled model, and it has polylogarithmic error and communication per user (see Theorem~\ref{thm:hist_full},  Lemmas~\ref{thm:h_had_priv},~\ref{thm:h_had_acc}, and~\ref{thm:Hadamard_efficiency} for more details).

A limitation of our private-coin algorithm outlined above is that the time for the analyzer to answer a single query is $\tilde{O}(n)$. 
This might be a drawback in applications where the analyzer is CPU-limited or where it is supposed to produce real-time answers.  In the presence of public randomness, we design an algorithm that remedies this limitation, having error, communication per user, and query time all equal to 
$\tilde{O}(1)$. 
Furthermore, the frequency estimates of this algorithm have one-sided error, and never underestimate the frequency of an element.
This algorithm is based on a multi-message version of randomized response combined in a delicate manner with the Count Min data structure~\cite{cormode2005improved} (for more details, see Section~\ref{subsec:count-min}).  Previous work~\cite{bassily2015local,bassily2017practical} on DP have used Count Sketch \cite{charikar2002finding}, which is a close variant of Count Min, to go from frequency estimation to heavy hitters.  In contrast, our use of Count Min 
% \footnote{An advantage of Count Min over Count Sketch is that it never underestimates the desired counts; this slightly simplifies the analysis.}
has the purpose of reducing the amount of communication per user.

\subsection{Applications}\label{sec:applications}

\paragraph{Heavy Hitters.}
Another algorithmic task that is closely related to frequency estimation is computing the \emph{heavy hitters} in a dataset distributed across $n$ users, where the goal of the analyzer is to (approximately) retrieve the identities and counts of all elements that appear at least $\tau$ times, for a given threshold $\tau$. It is well-known that in the central DP model, it is possible to compute $\tau$-heavy hitters for any $\tau = \tilde{\Theta}(1)$ whereas in the local DP model, it is possible to compute $\tau$-heavy hitters if and only if $\tau = \tilde{\Theta}(\sqrt{n})$. 
By combining with known reductions (e.g., from Bassily et al.~\cite{bassily2017practical}), our multi-message protocols for frequency estimation yield multi-message protocols for computing the $\tau$-heavy hitters with $\tau = \tilde{\Theta}(1)$ and total communication of $\tilde{\Theta}(1)$ bits per user (for more details, see Appendix~\ref{app:hh_reduction}).

\paragraph{Range Counting.}
In range counting, each of the $n$ users is associated with a point in $[B]^d$ and the goal of the analyzer is to answer arbitrary queries of the form: given a rectangular box in $[B]^d$, how many of the points lie in it?\footnote{We formally define range queries as a special case of counting queries in Section \ref{sec:rq}.} This is a  basic algorithmic primitive that captures an important family of database queries and is useful in geographic applications. This problem has been well-studied in the central model of DP, where Chan et al.~\cite{ChanSS11} obtained an upper bound of $(\log B)^{O(d)}$ on the error (see Section~\ref{sec:related} for more related work). It has also been studied in the local DP model~\cite{cormode2018answering}; in this case, the error has to be at least $\Omega(\sqrt{n})$ even for $d=1$.

We obtain private protocols for range counting in the multi-message shuffled model with exponentially smaller error than what is possible in the local model (for a wide range of parameters). Specifically, we give a private-coin multi-message protocol with $(\log{B})^{O(d)}$ messages per user each of length $O(\log{n})$ bits, error $(\log{B})^{O(d)}$, and query time $\tilde{O}(n \log^d B)$. Moreover, we obtain a public-coin protocol with similar communication and error but with a much smaller query time of $\tilde{O}(\log^d B)$ (see Section~\ref{sec:rq} for more details).

We now briefly outline the main ideas behind our multi-message protocols for range counting. We first argue that even for $d=2$, the total number of queries is $\Theta(B^2)$ and the number of possible queries to which a user positively contributes is also $\Theta(B^2)$. Thus, direct applications of DP algorithms for aggregation or for frequency estimation would result in polynomial error and polynomial communication per user. Instead, we combine our multi-message protocol for frequency estimation (Theorem~\ref{thm:hist_intro}) with a communication-efficient implementation, in the multi-message shuffled model, of the space-partitioning data structure used in the central model protocol of Chan et al.~\cite{ChanSS11}.
The idea is to use a collection $\MB$ of $O(B\log^d B)$ $d$-dimensional rectangles in $[B]^d$ (so-called \emph{dyadic intervals}) with the property that an arbitrary rectangle can be formed as the disjoint union of $O(\log^d B)$ rectangles from $\MB$.
Furthermore, each point in $[B]^d$ is contained in $O(\log^d B)$ rectangles from $\MB$.
This means that it suffices to release a private count of the number of points inside each rectangle in $\MB$ --- a frequency estimation task where each user input contributes to $O(\log^d B)$ buckets.
To turn this into a protocol with small maximum communication in the shuffled model, we develop an approach analogous to the matrix mechanism~\cite{li2010optimizing,li2012adaptive}.
We argue that the transformation of the aforementioned central model algorithm for range counting into a private protocol in the multi-message shuffled model with small communication and error is non-trivial and relies on the specific protocol structure. In fact, the state-of-the-art range counting algorithm of Dwork et al.~\cite{dwork2015pure} in the central model does not seem to transfer to the shuffled model.

\paragraph{M-Estimation of Median.}

A very basic statistic of any dataset of real numbers is its \emph{median}. For simplicity, suppose our dataset consists of real numbers lying in $[0,1]$. It is well-known that there is no DP algorithm for estimating the \emph{value} of the median of such a dataset with error $o(1)$ (i.e., outputting a real number whose absolute distance to the true median is $o(1)$) \cite[Section $3$]{vadhan2017complexity}. This is because the median of a dataset can be highly sensitive to a single data point when there are not many individual data points near the median. Thus in the context of DP, one has to settle for weaker notions of median estimation. One such notion is \emph{M-estimation}, which amounts to finding a value $\tilde{x}$ that approximately minimizes $\sum_i |x_i - \tilde{x}|$ (recall that the median is the minimizer of this objective). This notion has been studied in previous work on DP including by~\cite{lei2011differentially,duchi_minimax} (for more on related work, see Section~\ref{sec:related} below).  Our private range counting protocol described above yields a multi-message protocol with communication $\tilde{O}(1)$ per user and that $M$-estimates the median up to error $\tilde{O}(1)$, i.e., outputs a value $y \in [0,1]$ such that $\sum_i |x_i - y| \le \min_{\tilde{x}} \sum_i |x_i - \tilde{x}| + \tilde{O}(1)$ (see Theorem~\ref{thm:M_est_median} in Appendix~\ref{sec:basic_stats}). Beyond $M$-estimation of the median, our work implies private multi-message protocols for estimating \emph{quantiles} with $\tilde{O}(1)$ error and $\tilde{O}(1)$ bits of communication per user (see Appendix~\ref{sec:basic_stats} for more details).

\subsection{Related Work}\label{sec:related}

\paragraph{Shuffled Privacy Model.}
Following the proposal of the Encode-Shuffle-Analyze architecture by Bittau et al. \cite{bittau17}, several recent works have sought to formalize the trade-offs in the shuffled model with respect to standard local and central DP~\cite{erlingsson2019amplification,balle_privacy_2019} as well as devise private schemes in this model for tasks such as secure aggregation~\cite{cheu_distributed_2018, balle_privacy_2019, ghazi2019scalable, DBLP:journals/corr/abs-1906-09116, ghazi2019private, balle_privacy_2019constantIKOS}. In particular, for the task of \emph{real} aggregation, Balle et al. \cite{balle_privacy_2019} showed that in the single-message shuffled model, the optimal error is $\Theta(n^{1/6})$ (which is better than the error in the local model which is known to be $\Theta(n^{1/2})$).\footnote{Although the single-message real summation protocol of Balle et al. \cite{balle_privacy_2019} uses the $B$-ary randomized response, when combined with their lower bound on single-message protocols, it does not imply any lower bound on single-message frequency estimation protocols. The reason is that their upper bound doe not use the $\ell_{\infty}$ error bound for the $B$-ary randomized response as a black box.} By contrast, recent follow-up work gave multi-message protocols for the same task with error and communication of $\tilde{O}(1)$ \cite{ghazi2019scalable, DBLP:journals/corr/abs-1906-09116, ghazi2019private, balle_privacy_2019constantIKOS}\footnote{A basic primitive in these protocols is a ``split-and-mix'' procedure that goes back to the work of Ishai et al.~\cite{ishai2006cryptography}.}.
Our work is largely motivated by the aforementioned body of works demonstrating the power of the shuffled model, namely, its ability to enable private protocols with lower error than in the local model while placing less trust in a central server or curator.

Wang et al.~\cite{wang2019practical} recently designed an extension of the shuffled model and analyzed its trust properties and privacy-utility tradeoffs. 
They studied the basic task of frequency estimation, and benchmarked several algorithms, including one based on single-message shuffling.
However, they did not consider improvements through multi-message protocols, such as the ones we propose in this work.
Very recently, Erlingsson et al.~\cite{erlingsson2020encode} studied multi-message (``report fragmenting'') protocols for frequency estimation in a practical shuffled model setup. Though they make use of a sketching technique, like we do, their methods cannot be parameterized to have communication and error polylogarithmic in $n$ and $B$ (which our Theorem~\ref{thm:hist_intro} achieves). This is a result of using an estimator (based on computing a mean) that does not yield high-probability guarantees.
% \rasmus{I agree with Ameya's comment to the previous description, and tried to modify to make it accurate.}
%\badih{Is this accurate? Also need to fill in the missing reference..} \ameya{Are you sure the Erlingsson et al. '20 protocol is actually similar to ours? My understanding is that they focus on using ``attribute fragmenting'' based on what was in Cheu et al. Very different from our Hadamard response based approach.}

\paragraph{Private Frequency Estimation, Heavy Hitters, and Median.}
Frequency estimation and its extensions (considered below) has been extensively studied in concrete computational models including data structures, sketching, streaming, and communication complexity, (e.g.,~\cite{misra1982finding,charikar2002finding,estan2003new,cormode2005improved,cormode2005s,cormode2008finding,munro1980selection,manku1998approximate,greenwald2001space,gilbert2002fast,yi2013optimal,karnin2016optimal}).
Heavy hitters and frequency estimation have also been studied extensively in the standard models of DP, e.g.,~\cite{warner1965randomized,hsu_hh,bassily2015local,bassily2017practical,wang_freqest,bun2018heavy,acharya2019communication}.
The other problems we consider in the shuffled model, namely, range counting, M-estimation of the median, and quantiles, have been well-studied in the literature on data structures and sketching~\cite{CormodeYi20} as well as in the context of DP in the central and local models. Dwork and Lei~\cite{dwork2009differential} initiated work on establishing a connection between DP and robust statistics, and gave private estimators for several problems including the median, using the paradigm of propose-test-release.
% as well as interquartile range, $\alpha$-trimmed mean, and regression.
Subsequently, Lei~\cite{lei2011differentially} provided an approach in the central DP model for privately releasing a wide class of M-estimators (including the median) that are statistically consistent. While such M-estimators can also be obtained indirectly from non-interactive release of the density function~\cite{wasserman}, the aforementioned approach exhibits an improved rate of convergence. Furthermore, motivated by risk bounds under privacy constraints, Duchi et al.~\cite{duchi_minimax} provided private versions of information-theoretic bounds for minimax risk of M-estimation of the median.

Frequency estimation can be viewed as the problem of distribution estimation in the $\ell_\infty$ norm where the distribution to be estimated is the empirical distribution of a dataset $(x_1, \ldots, x_n)$. Some works~\cite{ye2017optimal,kairouz2016discrete} have established tight lower bounds for locally \DP distribution estimation in the weak privacy setting with loss instead given by either $\ell_1$ or $\ell_2^2$. However, their techniques proceed by using Assouad's method~\cite{duchi_minimax} and are quite different from the approach we use for the $\ell_\infty$ norm in the proof of Theorem~\ref{th:histograms_single_message} (specifically, in the proof of Theorem~\ref{thm:local_lb}).

We also note that an anti-concentration lemma qualitatively similar to our Lemma~\ref{lem:lb_tvd} was used by Chan et al.~\cite[Lemma 3]{chan_optimal_2012} to prove lower bounds on private aggregation, but they operated in a multi-party setting with communication limited by a sparse communication graph. After the initial release of this paper, Ghazi et al.~\cite{ghazi_pure_2020} proved a similar anti-concentration lemma to establish a lower bound on private summation for protocols with short messages. The lemmas in both of these papers do not  apply to the more general case of frequency estimation with an arbitrary number $\B$ of buckets, as is the case throughout this paper.

\paragraph{Range Counting.} Range counting queries have also been an important subject of study in several areas including database systems and algorithms (see \cite{cormode2011sketch} and the references therein). Early works on differentially private frequency estimation~, e.g., \cite{Dwork2006DP, hardt_release}, apply naturally to range counting, though the approach of summing up frequencies yields large errors for queries with large ranges. 

For $d = 1$, Dwork et al.~\cite{DworkNPR10} obtained an upper bound of $O\left(\frac{\log^2 B}{\epsilon}\right)$ and a lower bound of $\Omega(\log B)$ for obtaining $(\epsilon, 0)$-DP.  Chan et al.~\cite{ChanSS11} extended the analysis to $d$-dimensional range counting queries in the central model, for which they obtained an upper bound of roughly $(\log B)^{O(d)}$. Meanwhile, a lower bound of Muthukrishnan and Nikolov~\cite{MuthukrishnanN12} showed that for $n \approx B$, the error is lower bounded by $\Omega\left((\log n)^{d - O(1)}\right)$. Since then, the best-known upper bound on the error for general $d$-dimensional range counting has been $(\log{B}+\log(n)^{O(d)})/\varepsilon$~\cite{dwork2015pure}, obtained using ideas from~\cite{DworkNPR10, ChanSS11} along with a k-d tree-like data structure. We note that for the special case of $d=1$, it is known how to get a much better dependence on $B$ in the central model, namely, exponential in $\log^* B$~\cite{beimelNissimStemmer2013, bun2015differentially}.

Xiao et al.~\cite{xiao2010differential} showed how to obtain private range count queries by using Haar wavelets, while Hay et al.~\cite{hay_hist} formalized the method of maintaining a hierarchical representation of data; the aforementioned two works were compared and refined by Qardaji et al.~\cite{qardaji}. Cormode et al.~\cite{cormode2018answering} showed how to translate many of the previous ideas to the local model of DP. We also note that the matrix mechanism of Li et al.~\cite{li2010optimizing,li2012adaptive} also applies to the problem of range counting queries. An alternate line of work for tackling multi-dimensional range counting that relied on developing private versions of k-d trees and quadtrees was presented by Cormode et al.~\cite{cormode_spatialdecomp}.

\paragraph{Secure Multi-Party Computation.}
If we allow user interaction in the computation of the queries, then there is a rich theory, within cryptography, of \emph{secure multi-party computation} (SMPC) that allows $f(x_1,\dots,x_n)$ to be computed without revealing anything about $x_i$ except what can be inferred from $f(x_1,\dots,x_n)$ itself (see, e.g., the book of Cramer et al.~\cite{cramer2015secure}).
Kilian et al.~\cite{kilian2008fast} studied SMPC protocols for heavy hitters, obtaining near-linear communication complexity with a multi-round protocol.
%, but we are not aware of any work on optimizing SMPCs for range queries.
In contrast, all results in this paper are about \emph{non-interactive} (single-round) protocols in the shuffled-model (in the multi-message setting, all messages are generated at once).
Though generic SMPC protocols can be turned into differentially private protocols (see, e.g., Section 10.2 in~\cite{vadhan2017complexity} and the references therein), they almost always use multiple rounds, and often have large overheads compared to the cost of computing $f(x_1,\dots,x_n)$ in a non-private setting.

\subsection{Organization}

We start with some notation and background in Section~\ref{sec:prelim}. In Section~\ref{sec:single_msg_bds}, we prove our lower bounds for single-message protocols in the shuffled model; corresponding upper bounds can be found in Appendix~\ref{apx:freq_ub}. In Section~\ref{sec:freq_oracle_heavy_hitters}, we present and analyze our multi-message protocols for frequency estimation (with missing proofs in Appendix~\ref{sec:CM_app}). In Section~\ref{sec:rq}, we give our multi-message protocols for range counting. We conclude with some interesting open questions in Section~\ref{sec:conclusion}. The proof of Corollary~\ref{cor:sq} is given in Appendix~\ref{app:counting-queries}. The reduction from frequency estimation to heavy hitters appears in Appendix~\ref{app:hh_reduction}. The reduction from range counting to M-estimation of the median and quantiles is given in Appendix~\ref{sec:basic_stats}.

\section{Preliminaries}\label{sec:prelim}

\paragraph{Notation.} For any positive integer $\B$, let $[\B] = \{ 1, 2, \ldots, \B \}$. For any set $\MY$, we denote by $\MY^*$ the set consisting of sequences of elements of $\MY$, i.e., $\MY^* = \bigcup_{n \geq 0} \MY^n$. % We will often use the variables $j,x$ to index over $[\B]$. 
%
% Given a (discretized) dataset $\bx \in [d]^n$, for $j \in [d]$, denote the {\it frequency} of $j$ with respect to $\bx$ by $f_\bx(j) := | \{ i \in [n] : x_i = j\}|$. Given an (undiscretized) dataset $\ba \in [0,1]^n$, for $d \in \BN$, let $\bx \in [d]^n$ be defined by $x_i = \lceil a_i \cdot d \rceil$, and define $f_{\ba,d}(j) := f_\bx(j)$ for $j \in [d]$.
%A {\it multiset} is an (unordered) set in which each element may be repeated arbitrarily many times. 
Suppose $\MS$ is a multiset whose elements are drawn from a set $\MX$.  With a slight abuse of notation, we will write $\MS \subset \MX$ and for $x \in \MX$, we write $\m{\MS}{x}$ to denote the \emph{multiplicity} of $x$ in $\MS$.
For an element $x \in \MX$ and a non-negative integer $k$, let $k \times \{ x \}$ denote the multiset with $k$ copies of $x$ (e.g., $3 \times \{ x \} = \{ x,x,x\}$). For a positive real number $a$, we use $\log(a)$ to denote the logarithm base 2 of $a$, and $\ln(a)$ to denote the natural logarithm of $a$.  Let $\Bin(n, p)$ denote the binomial distribution with parameters $n > 0$ and $p \in (0, 1)$.  

\subsection{Differential Privacy}
We now introduce the basics of differential privacy that we will need.  Fix a finite set $\MX$, the space of reports of users. A {\it dataset} is an element of $\MX^*$, namely a tuple consisting of elements of $\MX$. Let $\hist(X) \in \BN^{|\MX|}$ be the histogram of $X$: for any $x \in \MX$, the $x$th component of $\hist(X)$ is the number of occurrences of $x$ in the dataset $X$. We will consider datasets $X, X'$ to be \emph{equivalent} if they have the same histogram (i.e., the ordering of the elements $x_1, \ldots, x_n$ does not matter). For a multiset $\MS$ whose elements are in $\MX$, we will also write $\hist(\MS)$ to denote the histogram of $\MS$ (so that the $x$th component is the number of copies of $x$ in $\MS$).

Let $n \in \BN$, and consider a dataset $X = (x_1, \ldots, x_n) \in \MX^n$.  For an element $x \in \MX$, let $f_X(x) = \frac{\hist(X)_x}{n}$ be the {\it frequency} of $x$ in $X$, namely the fraction of elements of $X$ which are equal to $x$.
Two datasets $X, X'$ are said to be {\it neighboring} if they differ in a single element, meaning that we can write (up to equivalence) $X = (x_1, \ldots, x_{n-1}, x_n)$ and $X' = (x_1, \ldots, x_{n-1}, x_n')$. In this case, we write $X \sim X'$. Let $\MZ$ be a set; we now define the differential privacy of a randomized function $P : \MX^n \ra \MZ$:
\begin{defn}[Differential privacy~\cite{dwork2006calibrating,dwork2006our}]
  \label{def:dp}
A randomized algorithm $P : \MX^n \ra \MZ$ is {\it $(\ep, \delta)$-differentially private} if for every pair of neighboring datasets $X \sim X'$ and for every set $\MS \subset \MZ$, we have
$$
\p[P(X) \in \MS] \leq e^\ep \cdot \p[P(X') \in \MS] + \delta,
$$
where the probabilities are taken over the randomness in $P$.  Here, $\ep \geq 0, \delta \in [0,1]$.
\end{defn}
We will use the following compositional property of differential privacy.
\begin{lemma}[Post-processing, e.g.,~\cite{DworkRothBook}]
    \label{lem:post_process}
    If $P$ is $(\ep, \delta)$-differentially private, then for every randomized function $A$, the composed function $A \circ P$ is $(\ep, \delta)$-differentially private.
\end{lemma}

\subsection{Shuffled Model}\label{sec:model}
We briefly review the \emph{shuffled model} of differential privacy~\cite{bittau17, erlingsson2019amplification, cheu_distributed_2018}. 
The input to the model is a dataset $(x_1, \ldots, x_n) \in \MX^n$, where item $x_i \in \MX$ is held by user $i$. A protocol in the shuffled model is the composition of three algorithms:
\begin{itemize}[nosep]
\item The {\it local randomizer} $R: \MX \ra \MY^*$ takes as input the data of one user, $x_i \in \MX$, and outputs a sequence $(y_{i,1}, \ldots, y_{i,m_i})$ of {\it messages}; here $m_i$ is a positive integer.
\item The {\it shuffler} $S: \MY^* \ra \MY^*$ takes as input a sequence of elements of $\MY$, say $(y_1, \ldots, y_m)$, and outputs a random permutation, i.e., the sequence $(y_{\pi(1)}, \ldots, y_{\pi(m)})$, where $\pi \in S_m$ is a uniformly random permutation on $[m]$. The input to the shuffler will be the concatenation of the outputs of the local randomizers.
\item The {\it analyzer} $A: \MY^* \ra \MZ$ takes as input a sequence of elements of $\MY$ (which will be taken to be the output of the shuffler) and outputs an answer in $\MZ$ which is taken to be the output of the protocol~$P$.
\end{itemize}
We will write $P = (R, S, A)$ to denote the protocol whose components are given by $R$, $S$, and $A$.  The main distinction between the shuffled and local model is the introduction of the shuffler $S$ between the local randomizer and the analyzer. Similar to the local model, in the shuffled model the analyzer is untrusted; hence privacy must be guaranteed with respect to the input to the analyzer, i.e., the output of the shuffler. Formally, we have: % We next define what it means for $P$ to be differentially private in the shuffled model.
\begin{defn}[Differential privacy in the shuffled model,~\cite{erlingsson2019amplification, cheu_distributed_2018}]
\label{def:dp_shuffled}
 A protocol $P = (R, S, A)$ is {\it $(\ep, \delta)$-differentially private} if, for any dataset $X = (x_1, \ldots, x_n)$, the algorithm 
$$
(x_1, \ldots, x_n) \mapsto S(R(x_1), \ldots, R(x_n))
$$
is $(\ep, \delta)$-differentially private. 
\end{defn}
Notice that the output of $S(R(x_1), \ldots, R(x_n))$ can be simulated by an algorithm that takes as input the {\it multiset} consisting of the union of the elements of $R(x_1), \ldots, R(x_n)$ (which we denote as $\bigcup_i R(x_i)$, with a slight abuse of notation) and outputs a uniformly random permutation of them. Thus, by Lemma~\ref{lem:post_process}, it can be assumed without loss of generality for privacy analyses that the shuffler simply outputs the multiset $\bigcup_i R(x_i)$. 
For the purpose of analyzing accuracy of the protocol $P = (R, S, A)$, we define its \emph{output} on the dataset $X = (x_1, \ldots, x_n)$ to be $P(X) := A(S(R(x_1), \ldots, R(x_n)))$. We also remark that the case of {\it local differential privacy}, formalized in Definition~\ref{def:dp_local}, is a special case of the shuffled model where the shuffler $S$ is replaced by the identity function.
\begin{defn}[Local differential privacy~\cite{kasiviswanathan2008what}]
  \label{def:dp_local}
A protocol $P = (R,A)$ is {\it $(\ep, \delta)$-differentially private in the local model} (or {\it $(\ep, \delta)$-locally differentially private}) if the function $x \mapsto R(x)$ is $(\ep, \delta)$-\DP in the sense of Definition~\ref{def:dp}. We say that the {\it output} of the protocol $P$ on an input dataset $X = (x_1, \ldots, x_n)$ is $P(X) := A(R(x_1), \ldots, R(x_n))$. 
\end{defn}

\section{Single-Message Lower and Upper Bounds}\label{sec:single_msg_bds}
In this section, we prove Theorem~\ref{th:histograms_single_message}, which determines (up to polylogarithmic factors) the accuracy of frequency estimation in the single-message shuffled model. Using similar techniques, we also prove Theorem~\ref{th:selection_single_message}, which establishes a tight (up to polylogarithmic factors) lower bound on the number of users required to solve the selection problem in the single-message shuffled model. Our theorems give tight versions (see Corollary~\ref{cor:freq_lb_const}) of Corollaries~30 and~32 %6.6 and 6.8 <- wrong numbers for Eurocrypt version
of~\cite{cheu_distributed_2018}, which were each off from the respective optimal bounds by a polynomial of degree~17. We will use the following definition throughout this section:
% We will often use the following definition.
\begin{defn}[$(\alpha, \beta)$-accuracy]
  \label{def:accuracy}
  Let $\MZ$ be a finite set, let $\B \in \BN$, and let $e_v \in \{0,1\}^\B$ be the binary indicator vector with $(e_v)_j = 1$ if and only if $j = v$.
  We say that a (randomized) protocol $P : [\B]^n \ra [0,1]^\B$ for frequency estimation is {\it $(\alpha, \beta)$-accurate} if for each dataset $X = (x_1, \ldots, x_n) \in [\B]^n$, we have that $$\p_P\left[\max_{j \in [\B]} \left| P(X)_j - \frac 1n \sum_{i=1}^n (e_{x_i})_j \right| \leq \alpha\right] \geq 1-\beta.$$

  Often we will either have $P = (R,A)$ for a local randomizer $R$ and an analyzer $A$ (corresponding to the local model) or $P = (R,S,A)$ (corresponding to the shuffled model). In such a case, we will slightly abuse notation and refer to the local randomizer $R : [\B] \ra \MZ$ as {\it $(\alpha, \beta)$-accurate} if there exists an analyzer $A : \MZ^n \ra [0,1]^\B$ such that the corresponding local or shuffled-model protocol is $(\alpha, \beta)$-accurate.
 %  is an {\it $(\alpha, \beta)$-accurate} local randomizer if the following holds: there is an analyzer $\Aprot : \MZ^n \ra [0,1]^\B$ such that for any $x_1, \ldots, x_n \in [\B]$, for each $j \in [\B]$, we have that $\left| \frac 1n \sum_{i=1}^n x_{i,j} - \Aprot(\Rprot(x_1), \ldots, \Rprot(x_n))_j \right| \leq \alpha$ with probability at least $\beta$. In such a case we will often abuse notation and say that entire local-model protocol $P = (\Rprot, \Aprot)$ is $(\alpha, \beta)$-accurate as well.
\end{defn}

% Theorem~\ref{thm:freq_lb} shows that Theorem~\ref{thm:freq_ub} is tight, up to polylogarithmic factors.
Theorem~\ref{thm:freq_lb} establishes lower bounds on the (additive) error of frequency estimation in the single-message differentially-private shuffled model. 
\begin{theorem}[Lower bound for single-message \DP\ frequency estimation]
  \label{thm:freq_lb}
  There is a sufficiently small constant $c > 0$ such that the following holds: Suppose $n, \B \in \BN$ with $n \geq 1/c$, and $0 < \delta < c/n$. Any $(\ep, \delta)$-differentially private $n$-user single-message shuffled model protocol that is $(\alpha, 1/4)$-accurate satisfies: % that is given as input a dataset $X = (x_1, \ldots, x_n) \in [\B]^n$ and produces as output frequency estimates $(\hat x_1, \ldots, \hat x_\B) \in [0,1]^\B$. Then there is a dataset $X = (x_1, \ldots, x_n) \in [\B]^n$ so that:% there is a distribution $D_{n,\B}$ over $(\{0,1\}^\B)^n$ so that if $(x_1, \ldots, x_n) \sim D_{n,\B}$, we have:
   \begin{numcases}{
      \alpha % \E_{R} \left[\max_{j \in [\B]} \left| \hat x_j - \frac 1n \sum_{i=1}^n (e_{x_i})_j \right|\right]
       \geq}
    \Omega \left(\frac{\log \B}{ n \log\log\B} \right) ~~ & for ~~ $\frac{\log \B}{c\log \log \B}  \leq n \leq (\log^2 \B)(\log \log \B)$,  % &
     \quad \quad\ \ \text{(``Small-sample'')}
    \label{eq:const_lb} \\
    \Omega \left(\frac{1}{n^{3/4} \sqrt[4]{\log n}}\right) & for ~~ $ (\log^2 \B)(\log \log \B) \leq n \leq \frac{\B^2}{\log \B} $,   \quad\quad\quad\quad  \text{(``Intermediate-sample'')}
    \label{eq:interm_lb} \\
    \Omega \left(\frac{\sqrt{\B}}{n \sqrt{\log \B}}\right) & for ~~ $ n > \frac{\B^2}{\log \B}$.   \quad\quad\quad\quad\quad\quad\quad\quad\quad\quad\quad\quad\quad\  \text{(``Large-sample'')}
    \label{eq:small_lb}
    \end{numcases}
  \end{theorem}
Note that the lower bound on the additive error $\alpha$ is divided into 3 cases, which we call the {\it small-sample regime} (\ref{eq:const_lb}), the {\it intermediate-sample regime} (\ref{eq:interm_lb}), and the {\it large-sample regime} (\ref{eq:small_lb}). While the division into separate regimes makes our bounds more technical to state, we point out that this seems necessary in light of the very different protocols that achieve near-optimality in the various regimes (as discussed in Section~\ref{sec:overview_lb} and Appendix~\ref{app:counting-queries}). Moreover, the bound for the low-sample regime of Theorem \ref{thm:freq_lb} is established in Lemma \ref{lem:const_lb}, while the bounds for the intermediate-sample and large-sample regimes of Theorem \ref{thm:freq_lb} are established in Corollary \ref{cor:intermediate_lb} and Lemma \ref{lem:freq_lb_small_cor}, respectively. We note that the proof of the intermediate-sample regime (\ref{eq:interm_lb}) is the most technically involved and constitutes the bulk of the proof of Theorem \ref{thm:freq_lb}.

Furthermore, we observe that the lower bounds \eqref{eq:const_lb}, \eqref{eq:interm_lb}, and \eqref{eq:small_lb} also hold, up to constant factors, for the \emph{expected error} $\E_{R} \left[\max_{j \in [\B]} \left| P(X)_j - \frac 1n \sum_{i=1}^n (e_{x_i})_j \right|\right]$ of $P$ on a dataset $X$. This follows as an immediate consequence of Theorem~\ref{thm:freq_lb} and Markov's inequality.

% For a protocol $P$ in the above theorem, we refer to the quantity $  \E_{R} \left[\max_{j \in [\B]} \left| \hat x_j - \frac 1n \sum_{i=1}^n (e_{x_i})_j \right|\right]$ as the {\it expected error of $P$ (for the dataset $X$)}. We remark that $x_1, \ldots, x_n$ are not necessarily independent under $D_{n,\B}$.

In the course of proving Theorem \ref{thm:freq_lb} in the small-sample regime (i.e., (\ref{eq:const_lb})), we shall see that the constants can be chosen in such a way so as to establish Corollary \ref{cor:freq_lb_const} below (in particular, Corollary \ref{cor:freq_lb_const} follows from Lemma \ref{lem:real_const}):
  \begin{corollary}[Lower bound for constant-error frequency estimation]
    \label{cor:freq_lb_const}
Let $c$ be the constant of Theorem~\ref{thm:freq_lb}. If $P$ is a $(1,\delta)$-\DP protocol for frequency estimation in the shuffled model with $\delta < c/n$ which is $(1/10, 1/10)$-accurate, then $n \geq \Omega\left( \frac{\log \B}{\log \log \B}\right)$.
\end{corollary}
Corollary~\ref{cor:freq_lb_const} improves upon Corollary 32 of~\cite{cheu_distributed_2018}, both in the lower bound on the error (which was $\Omega(\log^{1/17}\B)$ in~\cite{cheu_distributed_2018}) and on the dependence on $\delta$ (which was $\delta < O(n^{-8})$ in~\cite{cheu_distributed_2018}).

% \subsection{Lower bound for locally DP frequency estimation for the low privacy setting}
The primary component of the proof of Theorem~\ref{thm:freq_lb} is a lower bound on the additive error of $(\ep_L, \delta_L)$-{\it locally} \DP\ protocols $P = (R,A)$, when both $\ep_L \gg 1$ (the {\it low-privacy} setting) and $\delta_L > 0$ simultaneously hold (see Lemma \ref{lem:pratio}). In particular, we prove the following:
\begin{theorem}[Lower bound for locally differentially private frequency estimation]
  \label{thm:local_lb}
  There is a sufficiently small constant $c > 0$ such that the following holds.
  Suppose $n, \B \in \BN$ with $n \geq 1/c$, and that $\ep_L, \delta_L > 0$ with $\delta_L < c\min \{1/(n \log n), \exp(-\ep_L) \}$. Any $(\ep_L, \delta_L)$-locally differentially private protocol that is $(\alpha, 1/4)$-accurate satisfies:
 {\small \begin{numcases}{\alpha \geq}
    \Omega \left( \frac{\ln \B}{n \ep_L} \right) & for ~~ $n \geq \frac{\ln \B}{c\ep_L}$, \qquad\qquad\qquad\qquad\qquad\qquad\qquad \text{(``Small-sample'')} \label{eq:local_const}\\
    \tilde \Omega\left( \frac{1}{\sqrt n \cdot \exp(\ep_L/4)} \right) & \parbox{2.5in}{for ~~ $n \geq (\ln \B) \exp(\ep_L/2)$ \\ and $\frac 23 \cdot \ln(n) \leq \ep_L + \ln(1+\ep_L) + \frac 1c \leq 2 \ln(\B)$,} \label{eq:local_interm1} \qquad  \text{(``Intermediate-sample'')} \\
    \tilde \Omega \left( \frac{1}{n^{2/3}} \right) & for ~~ $ \ln^{3/2}(\B) \leq n \leq \B^3$ and $\ep_L \leq \frac 23 \cdot \ln(n) $,\label{eq:local_interm2} \ \ \quad \text{(``Intermediate-sample'')} \\
    \tilde \Omega \left( \frac{\sqrt{\B}}{n} \right) & for ~~ $ n \geq \B^2$ and $\ep_L \leq 2 \ln (\B)$, \qquad \qquad \qquad\quad \text{(``Large-sample'')} \label{eq:local_small1}\\
    \tilde \Omega \left( \frac{\B}{n} \right) & for ~~ $n \geq \B^3$ and $\ep_L \leq 2 \ln (\B)$. \qquad \qquad \qquad\quad \text{(``Large-sample'')}\label{eq:local_small2}
    \end{numcases}}
  \end{theorem}
  Again, the lower bound is divided into cases---the bound for low-sample regime of Theorem~\ref{thm:local_lb} (namely, \eqref{eq:local_const}) is established in Lemma~\ref{lem:const_lb}, while the bounds for the intermediate-sample (namely, \eqref{eq:local_interm1} and \eqref{eq:local_interm2}) and large-sample (namely, \eqref{eq:local_small1} and \eqref{eq:local_small2}) regimes are established in Lemma~\ref{lem:intermediate_lb} and Lemma~\ref{lem:local_small}, respectively.

It turns out that Theorem~\ref{thm:freq_lb} is tight in each of the three regimes (small-sample, intermediate-sample, and large-sample), up to polylogarithmic factors in $\B$ and $n$, as shown by Theorem~\ref{thm:freq_ub}:
\begin{theorem}[Upper bound for single-message shuffled DP frequency estimation]
\label{thm:freq_ub}
Fix $\B,n \in \BN$, $\delta = n^{-O(1)}$, and $\ep \leq 1$ that satisfies $\ep = \omega(\ln^{2}(n)/\min\{\sqrt{\B},\sqrt{n}\})$. 
For $n \in \BN$, there is a shuffled model protocol $P = (\Rprot, \Sprot, \Aprot)$ so that for any $X = (x_1, \ldots, x_n) \in [\B]^n$, the frequency estimates $P(X) \in [0,1]^\B$ produced by $P$ satisfy % $P(X)$ produces frequency estimates $\hat x := (\hat x_1, \ldots, \hat x_\B) \in [0,1]^\B$, then we have:
\begin{numcases}{
\E \left[ \max_{j \in [\B]} \left| P(X)_j - \frac 1n\sum_{i=1}^n (e_{x_i})_j \right| \right] \leq}
O \left( \frac{\log \B}{n} \right) \quad & for ~~ $ n \leq \frac{\ep^2\log^2 \B}{\log^{3} \log \B}$, \label{eq:const_ub}\\
O \left( \frac{\ln^{3/4}(n) \sqrt{\log \B}}{n^{3/4}\sqrt \ep}\right)\quad &for ~~ $ \frac{\ep^2\log^2 \B}{\log^3 \log \B} \leq n \leq \B^2$,\label{eq:interm_ub} \\
O \left( \frac{\sqrt{\B \ln(n) \ln (\B)}}{n\ep}\right) \quad &for ~~ $ n > \B^2 $.\label{eq:small_ub}
\end{numcases}
\end{theorem}
The proof of Theorem~\ref{thm:freq_ub} follows by combining existing protocols for locally differentially private frequency estimation with the privacy amplification result of \cite{balle_privacy_2019}. For completeness, we provide the proof in Appendix~\ref{apx:freq_ub}.

The remainder of this section is organized as follows. In Section~\ref{sec:fe_lb} we collect some tools that will be used in the proofs of our error lower bounds. In Section~\ref{sec:const_lb} we establish Theorem~\ref{thm:freq_lb} in the small-sample regime (i.e., (\ref{eq:const_lb})). In Sections~\ref{sec:interm_lb} and \ref{sec:lb_tvd} we establish Theorem~\ref{thm:freq_lb} in the intermediate and large-sample regimes (i.e., (\ref{eq:interm_lb}) and (\ref{eq:small_lb}). Finally, in Section~\ref{sec:selection_lb} we show how similar techniques used to prove Theorem~\ref{thm:freq_lb} lead to a tight lower bound on the selection problem (Theorem~\ref{thm:selection_formal}).

% \subsection{Aside: inadequacy of ``black-box'' reduction}
\begin{remark}
\label{rmk:black-box}
Before proceeding with the proof of Theorem \ref{thm:local_lb} (and thus Theorem \ref{thm:freq_lb}), we briefly explain why the approach of~\cite{cheu_distributed_2018}, which establishes a weak variant of Theorem \ref{thm:local_lb}, cannot obtain the tight bounds that we are able to achieve here. Recall that this approach used:
\begin{enumerate}
    \item[(i)] in a black-box manner, known lower bounds of Bassily and Smith \cite{bassily2015local} and Duchi et al.~\cite{duchi_minimax} on the error of ``pure'' $(\ep_L, 0)$-locally \DP frequency estimation protocols, together with
    \item[(ii)] a result of Bun et al.~\cite{bun2018heavy} stating that by modifying an $(\ep_L, \delta_L)$-locally \DP protocol, one can produce an $(8\ep_L, 0)$-locally \DP protocol without significant loss in accuracy.
\end{enumerate}
% Note that we use subscripts in $\ep_L, \delta_L$ to distinguish the privacy parameters of the \emph{local} model from the $\epsilon, \delta$ parameters without subscripts of the shuffled model. 

It seems to be quite challenging to get tight bounds in the single-message shuffled model using this two-step technique. This is because when $\ep_L \approx \ln n$, the error lower bounds for $(\ep_L, 0)$-differentially private frequency estimation in the local model decay as $\exp(-a\ep_L)$ for some constant $a$. Suppose that for some constant $C \geq 1$, one could show that by modifying any $(\ep_L, \delta_L)$-locally \DP\ protocol one could obtain a $(C \ep_L, 0)$-locally \DP\ protocol without a large loss in accuracy (for instance, Bun et al. \cite{bun2018heavy} achieves $C = 8$.) Then the resulting error lower bound for shuffled-model protocols would decay as $\exp(-aC \ln n) = n^{-aC}$. This bound will necessarily be off by a polynomial in $n$ unless we can determine the optimal constant $C$. The proof for $C = 8 $ \cite{bun2018heavy,cheu_distributed_2018} is already quite involved, and in order for this approach to guarantee tight bounds in the single-message setup, we would need to achieve $C = 1$, i.e., turn any $(\ep_L, \delta_L)$-locally \DP\ protocol into one with $\delta_L = 0$ and essentially no increase in $\ep_L$ whatsoever.
\end{remark}

\subsection{Preliminaries for Lower Bounds}
\label{sec:fe_lb}

% We prove the accuracy lower bounds for each of the regimes (\ref{eq:const_lb}), (\ref{eq:interm_lb}), (\ref{eq:small_lb}) of Theorem~\ref{thm:freq_lb} separately in Sections~\ref{sec:const_lb}, \ref{sec:interm_lb}, and~\ref{sec:small_lb}. % First we collect some definitions 
In this section we collect some useful definitions and lemmas. Throughout this section, we will use the following notational convention:
\begin{defn}[Notation $\pp_{x,\MS}$]
  \label{def:pp}
  For a fixed local randomizer $\Rprot : \MX \ra \MZ$ (which will be clear from the context), and for $x \in \MX, \MS \subset \MZ, z \in \MZ$, we will write $\pp_{x,\MS} := \P_R[\Rprot(x) \in \MS]$ and $\pp_{x,z} := \P_R[\Rprot(x) = z]$, where the probability is over the randomness of $\Rprot$.

  Moreover, we will additionally write $\PP_x$ to denote the distribution on $\MZ$ given by $R(x)$. In particular, the density of $\PP_x$ at $z\in \MZ$ is $\pp_{x,z}$.
  \end{defn}

We say that a local randomizer $R : \MX \ra \MZ$ is $(\ep,\delta)$-\DP in the $n$-user shuffled model if the composed protocol $(x_1, \ldots, x_n) \mapsto S(R(x_1), \ldots, R(x_n))$ is $(\ep, \delta)$-\DP. Lemma~\ref{lem:pratio} establishes that a protocol $R$ that is $(\ep, \delta)$-\DP in the shuffled model is in fact $(\ep + \ln n, \delta)$-\DP in the {\it local} model of differential privacy, which means that the function $x \mapsto R(x)$ is itself $(\ep + \ln n, \delta)$-\DP.
 \begin{lemma}[Theorem 6.2,~\cite{cheu_distributed_2018}]
    \label{lem:pratio}
    % Suppose $n \geq 2$ and $\delta \leq 1/(2n)$.
Suppose $\MX, \MZ$ are finite sets. If $\Rprot : \MX \ra \MZ$ is $(\ep, \delta)$-\DP in the $n$-user single-message shuffled model, then $R$ is $(\ep + \ln n,\delta)$-locally differentially private. 

(That is, for all $x,y \in \MX$, and for all $\MS \subset \MZ$, we have
    $$
\pp_{y,\MS} \leq \pp_{x,\MS} \cdot e^\ep n + \delta.
$$
Recall $\pp_{y,\MS} = \p[R(y) \in \MS], \pp_{x,\MS} = \p[R(x) \in \MS]$ per Definition~\ref{def:pp}.)
\end{lemma}

As discussed in Section~\ref{sec:intro}, to prove Theorem~\ref{thm:freq_lb} (as well as Theorem~\ref{thm:selection_formal}), 
% To derive tight bounds, we instead opt, roughly speaking, to follow
we use similar ideas to those in the the results of~\cite{duchi_minimax,bassily2015local} to {\it directly} derive a lower bound on the error of locally private frequency estimation in the low and approximate privacy setting (i.e., for $(\ep_L, \delta_L)$-locally \DP protocols with $\ep_L \approx \ln n$ and $\delta_L > 0$). By Lemma~\ref{lem:pratio}, doing so suffices to derive a lower bound for frequency estimation in the single-message shuffled model. Our lower bounds for local-model protocols, on their own, may be of independent interest. The locally private frequency estimation lower bounds of~\cite{duchi_minimax,bassily2015local}, as well as our proof, rely on Fano's inequality, which we recall as Lemma~\ref{lem:fano} below.
% We will also make use of Fano's inequality, which has been widely used to prove lower bounds on the minimax rates of tasks such as mean estimation, linear regression~\cite{duchi2014information}, and, more relevant to our application, minimax rates for frequency estimation~\cite{duchi_minimax,bassily2015local} and for selection~\cite{ullman2018tight} in the setting of {\it local} differential privacy.

For random variables $X,Y$ distributed on a finite set $\MX$, let $I(X;Y)$ denote the mutual information between $X,Y$. We refer the reader to \cite{coverthomas} for more background on basic information theory. 
\begin{lemma}[Fano's inequality]
  \label{lem:fano}
  Suppose $Z, Z'$ are jointly distributed random variables on a finite set $\MZ$. Then
  $$
\p[Z = Z']\leq \frac{I(Z;Z') + 1}{\log |\MZ|}.
  $$
\end{lemma}

Additionally, it will be useful to phrase some of our arguments in terms of the hockey stick divergence between distributions:
\begin{defn}[Hockey stick divergence]
  \label{def:hockey}
Suppose $\Pdist,\Qdist$ are probability distributions on a space $\MX$ that are absolutely continuous with respect to some measure $G$ on $\MX$; let the densities of $\Pdist,\Qdist$ with respect to $G$ be given by $\pdist, \qdist$.
  For any $\rho \geq 1$, the {\it hockey stick divergence of order $\rho$} between $\Pdist,\Qdist$ is defined as:
  $$
\SD_\rho(\Pdist || \Qdist) := \int_\MX \left[ \pdist(x) - \rho \cdot \qdist(x) \right]_+ dG(x),
$$
where $[a]_+ = \max\{a, 0\}$ for $a \in \BR$. 
\end{defn}
The {\it total variation distance} $\Delta(D, F)$ between two distributions $D,F$ on a set $\MX$ is defined as 
$$\sup_{\MS \subseteq \MX} | D(\MS) - F(\MS)| \enspace .$$ 
Note that for $\rho = 1$ the hockey stick divergence of order $\rho$ is the total variation distance, i.e., $\SD_1(\Pdist || \Qdist) = \SD_1(\Qdist || \Pdist) = \Delta(\Pdist,\Qdist)$. 
The following fact is well-known:
\begin{fact}[Characterization of hockey stick divergence]
  \label{fac:hockey}
  Using the notation of Definition~\ref{def:hockey}, we have:
  $$
\SD_\rho(\Pdist || \Qdist) = \sup_{\MS \in \MX} \left( \Pdist(\MS) - \rho \cdot \Qdist(\MS) \right).
  $$
\end{fact}

\if 0
By Fact~\ref{fac:hockey} and Lemma~\ref{lem:pratio} it follows that: % For a real number $x$, let $[x]_+ := \max\{x, 0\}$. Corollary~\ref{cor:n_hs} of Lemma~\ref{lem:pratio} will be useful in the following section.
\begin{corollary}
  \label{cor:n_hs}
  If $\Rprot : \MX \ra \MZ$ is $(\ep, \delta)$-\DP in the $n$-user shuffled model, then for all $x,y \in \MX$ and for all $\MS \subset \MZ$, we have
  \begin{equation}
    \label{eq:n_hs}
\SD_{e^\ep n} (\PP_{y} || \PP_{x}) = \sum_{z \in \MZ} \left[\pp_{y,z} - e^\ep n \cdot \pp_{x,z}\right]_+ \leq \delta.
  \end{equation}
\end{corollary}
\fi

For a boolean function $f : \BB^\B \ra \BR$, the {\it Fourier transform} of $f$ is given by the function $\hat f(S) := \E_{x \sim \Unif(\BB^\B)} \left[ f(x) \cdot (-1)^{\sum_{j=1}^\B x_j \cdot \One[j \in S]} \right]$, where $S \subseteq [\B]$ is any subset. The {\it Fourier weight at degree 1} of such a function is defined by $\bW^1[f] := \sum_{j \in [\B]} \hat f(\{ j\})^2$. We refer the reader to \cite{odonnell2014analysis} for further background on the Fourier analysis of boolean functions.

\subsection{Small-Sample Regime}
\label{sec:const_lb}
In this section we establish Theorem~\ref{thm:freq_lb} in the case that $n \leq \log^2 \B$ (i.e., we prove (\ref{eq:const_lb})). As we noted following Lemma~\ref{lem:pratio}, we will prove a slightly more general statement, allowing $R$ to be any $(\ep + \ln n, \delta)$-{\it locally} \DP randomizer for some $\ep > 0$. Similar results are known~\cite{duchi_minimax,bassily2015local}; however, the work of~\cite{bassily2015local} only applies to the case that $R$ is $(\ep_L, \delta_L)$-locally \DP with $\ep_L = O(1)$, and~\cite{duchi_minimax} only consider $(\ep_L, 0)$-locally \DP protocols. Moreover, their dependence on the privacy parameter $\ep_L$ is not tight: in particular, for the ``small-sample regime'' of $n \leq O(\log^2 \B)$ that we consider in this section, the bounds of~\cite{duchi_minimax} decay as $e^{-2\ep_L}$, whereas we will be able to derive bounds scaling as $1/\ep_L$. We will then apply this bound with $\ep_L = \ep + \ln n$ being the privacy parameter of the locally \DP protocol furnished by Lemma~\ref{lem:pratio}. % , and for $\ep' \geq \ln n$, their approach yields lower bounds that scale as $\Omega(1/n^{3/2})$, which is trivial in light of Lemma~\ref{lem:invn}. Nevertheless, 

The proof of the error lower bound relies on the following Lemma~\ref{lem:mutinf}, which bounds the mutual information between a uniformly random index $V \in [\B]$, and $R(V)$. It improves upon analogous results in \cite{duchi_minimax,bassily2015local}, for which the dependence on $\ep_L$ is $(e^{\ep_L} - 1)^2$, when $\ep_L$ is large.
\begin{lemma}[Mutual information upper bound for small-sample regime]
  \label{lem:mutinf}
  Fix $n \in \BN$. Let $\Rprot$ be an $(\ep_L, \delta_L)$-\DP local randomizer (in the sense of Definition~\ref{def:dp}). Let $V \sim [\B]$ be chosen uniformly at random. Then
  $$
I(V; \Rprot(V)) \leq 2\delta_L \cdot \log \B + 1 + \ep_L \log e.
  $$
\end{lemma}
\begin{proof}
  For $v \in \MX$, following Definition~\ref{def:pp}, let $\PP_v$ denote the distribution of $\Rprot(v)$, i.e., the distribution of the output of the local randomizer when it gets an input of $v$. Let $\bar \PP := \frac{1}{\B} \sum_{v \in [\B]} \PP_v$. Notice that the distribution of $R(V)$, where $V \sim [\B]$ is uniform, is $\PP_v$. It follows that% Letting $U_\B$ denote the uniform distribution on $[\B]$ and $U_\B \otimes \bar \PP$ denote the product distribution, we have:
  \begin{equation}
    \label{eq:inf_def}
I(V; \Rprot(V)) = \frac{1}{\B} \sum_{v \in [\B]} \sum_{z \in \MZ} \p[R(v) = z] \cdot \log \left( \frac{\p_R[R(v) = z]}{\p_{V \sim [\B], R}[R(V) = z]} \right)=\frac 1\B \sum_{v \in [\B]} \KL(\PP_v || \bar \PP).
\end{equation}
  We now upper bound $\KL(\PP_v || \bar \PP)$ for each $v \in [\B]$. 
 We first claim that for any $v_0 \in [\B]$,
\begin{equation}
  \label{eq:kl_prob}
\pp_{Z \sim \Rprot(v_0)} \left[ \log \left( \frac{\pp_{v_0,Z}}{\frac{1}{\B} \sum_{j \in [\B]} \pp_{j,Z}} \right) > 1 + \ep_L\log e  \right] \leq 2\delta_L.
\end{equation}
To see that (\ref{eq:kl_prob}) holds, let $\MS := \left\{ z \in \MZ : \frac{\pp_{v_0,z}}{\frac{1}{\B} \sum_{j \in [\B]} \pp_{j,Z}} > 2e^{\ep_L} \right\}$. If (\ref{eq:kl_prob}) does not hold, then $\pp_{v_0,\MS} > 2\delta_L$ and $\pp_{v_0,\MS} > (2e^{\ep_L}) \cdot \frac{1}{\B} \sum_{j \in [\B]} \pp_{j,\MS}$. On the other hand, we have from $(\ep_L,\delta_L)$-differential privacy of $R$ that % from Lemma~\ref{lem:pratio} that
$$
\pp_{v_0,\MS} \leq \left( \frac{1}{\B} \sum_{j \in [\B]} \pp_{j,\MS} \right) \cdot e^{\ep_L} +\delta_L.
$$
The above equation is a contradiction in light of the fact that for positive real numbers $a,b$, $a + b \leq \max \{ 2a, 2b \}$.

\noindent Notice that for any $z \in \MZ, v \in [\B]$, it is the case that $\log\left(\frac{\pp_{v,z}}{\frac 1\B \sum_{j \in [\B]} \pp_{j,z}}\right) \leq \log \B$. It follows that (\ref{eq:kl_prob}) implies that
$$
\KL(\PP_{v_0} || \bar \PP) \leq 2\delta_L \cdot \log \B + 1 + \ep_L \log e.
$$
The statement of Lemma~\ref{lem:mutinf} follows from the above equation and (\ref{eq:inf_def}).
\end{proof}

Lemma \ref{lem:real_const}, together with Lemma~\ref{lem:pratio}, establishes Corollary \ref{cor:freq_lb_const}: % in the case that $n \approx \frac{\log \B}{\log \log \B}$ (
in particular, by Lemma \ref{lem:pratio}, any single-message shuffled-model $(\ep, \delta)$-\DP\ protocol $P$ yields a local-model $(\ep + \ln n, \delta)$-\DP\ protocol with the same accuracy. Thus we may set $\ep_L = \ln n + \ep$ in Lemma \ref{lem:real_const}, so that $n \geq \Omega(\log(\B)/\ep_L) = \Omega(\log(\B)/\log n)$ becomes $n \geq \Omega(\log(\B)/\log \log(\B))$). % In particular, Lemma \ref{lem:real_const} immediately implies Corollary \ref{cor:freq_lb_const}.
Lemma \ref{lem:real_const} is also used in the proof of (\ref{eq:const_lb}) of Theorem \ref{thm:freq_lb}. 
The proof is by a standard application of Fano's inequality \cite{duchi_minimax,bassily2015local}.
\begin{lemma}[Sample-complexity lower bound for constant-error frequency estimation]
  \label{lem:real_const}
  Suppose $\delta_L < 1/(4n), 0 < \ep_L < \log(\B)/20$, and $P = (R,A)$ is a local-model protocol that satisfies $(\ep_L, \delta_L)$-local differential privacy and $(1/3, 1/2)$-accuracy. Then $n > \frac{\log \B}{20 \ep_L}$. 
%   $R : [\B] \ra \MZ$ is an $(\ep_L, \delta_L)$-\DP local randomizer, and $A : \MZ \ra [0,1]^\B$ is an analyzer such that the following holds. Given a dataset $X = (x_1, \ldots, x_n) \in [\B]^n$, suppose the locally differentially private algorithm $X \mapsto A(R(x_1), \ldots, R(x_n)) := (\hat x_1, \ldots, \hat x_\B) \in [0,1]^\B$ produces frequency estimates $(\hat x_1, \ldots, \hat x_\B)$. Then there are $(x_1, \ldots, x_n) \in [\B]^n$ so that
%   \begin{equation}
%     \label{eq:14}
%     \E_R \left[ \max_{j \in [\B]} \left| \hat x_j - \frac 1n \sum_{i=1}^n (e_{x_i})_j \right| \right] > 1/4 
%   \end{equation}
%   if $n \leq \frac{\log \B}{20 \ep_L}$.
\end{lemma}
\begin{proof}
  Suppose for the purpose of contradiction that % the quantity in (\ref{eq:14}) is always bounded above by $1/4$ yet
  $n \leq \frac{\log \B}{20 \ep_L}$. % Note that $\frac{n \log n}{\log \B} \leq \frac{1}{20}$. 
  
  Let $D$ be the distribution on $(\{0,1\}^\B)^n$ that is uniform over all tuples $(e_v, e_v, \ldots, e_v)$, for $v \in [\B]$ (recall that $e_v$ is the unit vector for component $v$, i.e., the vector with a $1$ in the $v$th component).
  
  % Let $\alpha := \E_{X=(x_1, \ldots, x_n) \sim D, R} \left[ \max_{j \in [\B]} \left| \hat x_j - \frac 1n \sum_{i=1}^n (e_{x_i})_j \right| \right]$, where $(\hat x_1, \ldots, \hat x_\B) = A(R(x_1), \ldots, R(x_n))$.
  Consider any sample $X = (x_1, \ldots, x_n)$ in the support of $D$, so that $x_1 = \cdots = x_n = e_v$ for some $v \in [\B]$. If the error $ \max_{j \in [\B]} \left| P(X)_j - \frac 1n \sum_{i=1}^n (e_{x_i})_j \right| $ is strictly less than $1/2$, then the function
  $$
f(\hat x_1, \ldots, \hat x_j) := \arg \max_{j \in [\B]} \hat x_j
$$
will satisfy $f(\hat x_1, \ldots, \hat x_j) = v$. It follows that % by Fano's inequality that
\begin{align}
  &\p_{X \sim D, R} \left[ \max_{j \in [\B]} \left| \hat x_j - \frac 1n \sum_{i=1}^n (e_{x_i})_j \right|< 1/2 \right] \nonumber\\
               & \leq \p_{X \sim D, R} \left[ f(A(R(x_1), \ldots, R(x_n))) = V \right]\nonumber\\
  \label{eq:fano}
               & \leq \frac{I(f(A(R(x_1), \ldots, R(x_n))); V) + 1}{\log \B}\\
               & \leq \frac{I ((R(x_1), \ldots, R(x_n)); V) + 1}{\log \B}\nonumber\\
  \label{eq:inf_manip}
               & \leq \frac{n \cdot I(R(V); V) + 1}{\log \B}\\
  \label{eq:irvv}
               & \leq \frac{n \cdot (2\delta_L \log \B + 1 + \ep_L \log e) + 1}{\log \B} \\
  \label{eq:asm1}
               & \leq 2\delta_L n +  \frac{5 n \ep_L}{\log \B} \\
  \label{eq:asm2}
  & < 1/2
\end{align}
where (\ref{eq:fano}) follows by Fano's inequality and the random variable $V$ is so that $x_1 = \cdots = x_n = e_V$ and $V$ is uniform over $[\B]$). Moreover, (\ref{eq:irvv}) follows from Lemma~\ref{lem:mutinf} (and the fact that $V$ is uniform over $[\B]$), (\ref{eq:inf_manip}) follows from the chain rule for mutual information, and (\ref{eq:asm1}), (\ref{eq:asm2}) follow from our assumptions on $n, \B, \delta_L, \ep_L$. We now arrive at the desired contradiction to the $(1/3, 1/2)$-accuracy of $P$. % As long as $\alpha < 1$, we arrive at the desired contradiction.
\end{proof}

The next lemma is an adaptation to the local model of a standard result \cite[Fact 2.3]{steinke2013between}, stating that the optimal error of a differentially private frequency estimation protocol decays at most inverse linearly in the number of users~$n$. % We allow the users' data $x_i$ to belong to $\{0,1\}^\B$ (i.e., allowing each user to have any subset of the $\B$ possible items) in contrast to the previous lemmas (where each user only held 1 element) since we will also use Lemma~\ref{lem:invn} in our lower bounds on selection (Section~\ref{sec:selection_lb})
\begin{lemma}[Inverse-linear dependence of error on $n$]
  \label{lem:invn}
%   Suppose $\MZ$ is a finite set and $R : [\B] \ra \MZ$.
  Suppose $P = (R,A)$ is an $(\ep_L, \delta_L)$-locally differentially private algorithm (Definition~\ref{def:dp_local}) for $n$-user frequency estimation on $[\B]$ that satisfies $(\alpha, \beta)$-accuracy. 

%   outputting frequency estimates $P(X) \in [0,1]^\B$ so that for any dataset $X = (x_1, \ldots, x_n) \in (\{0,1\}^\B)^n$,
%   \begin{equation}
%     \label{eq:p_acc}
% \E_P \left[ \max_{j \in [\B]} \left| P(X)_j - \frac 1n \sum_{i=1}^n ({x_i})_j \right| \right] \leq \alpha.
% \end{equation}
  Let $n \geq n' \geq \lfloor 2\alpha n/c\rfloor$ for any $c \leq 2\alpha n$. Then there is an $\left(\ep_L, \delta_L\right)$-\DP protocol $P' = (R', A')$ for $n'$-user frequency estimation on $[\B]$ that satisfies $(c, \beta)$-accuracy.
%   so that for any dataset $X' = (x_1', \ldots, x_n') \in (\{0,1\}^\B)^n$,
% \begin{equation}
%   \label{eq:pprime_acc}
% \E_{P'} \left[ \max_{j \in [\B]} \left| P'(X')_j - \frac{1}{n'} \sum_{i=1}^{n'} ({x_i'})_j \right| \right] \leq c.
% \end{equation}
% The expectation in (\ref{eq:p_acc}) and (\ref{eq:pprime_acc}) is over the randomness in $R(x_i), R'(x_i')$, respectively.
%   Moreover, exactly the same statement holds in the single-message shuffled-model: i.e., if $P = (R,S,A)$ and $P' = (R',S,A')$.
\end{lemma}
% Recall that for $x \in [\B]$, $e_x \in \{0,1\}^\B$ is the unit vector with, for $j \in [\B]$, $(e_x)_j = 1$ iff $x=j$.
\begin{proof}
  The algorithm $P'$ is given as follows: we have $R' = R$. The analyzer $A'$, on input $(z_1, \ldots, z_{n'}) \in \MZ^{n'}$, generates $n-n'$ i.i.d.~copies of $R(e_1)$, which we denote by $z_{n'+1}, \ldots, z_n$. (Recall $e_1 = (1, 0, \ldots, 0)$.) Then $A'$ computes the vector $v := A(z_1, \ldots, z_n) \in [0,1]^\B$, and outputs $v'$, where $v'_j = \frac{n}{n'} \cdot v_j$ for $j > 1$, and $v'_1 = \frac{n}{n'} \cdot \left(v_1 - \frac{n-n'}{n}\right)$.

  % To see that (\ref{eq:pprime_acc}) holds, note that (\ref{eq:p_acc}) gives that
  To see that $P'$ satisfies $(c, \beta)$ accuracy, let's fix any input dataset $X = (x_1, \ldots, x_{n'})$. Let $x_j = 1$ for $j > n'$, and set $X' = (x_1, \ldots, x_{n'}, x_{n'+1}, \ldots, x_n)$. The $(\alpha, \beta)$-accuracy of $P$ gives that with probability at least $1-\beta$, $\max_{j \in [\B]} \left| P(X')_j - \frac 1n \sum_{i=1}^n (x_i)_j \right| \leq \alpha$. In such an event, we have that
  $$
 \max_{j \in [\B]} \left| \frac{n'v'_j}{n} - \frac{1}{n} \sum_{i=1}^{n'} ({x_i})_j \right| \leq \alpha.
$$
Multiplying the above by $n/n'$ and noting that $n/n'\leq c/\alpha$ gives that $P'$ satisfies $(c, \beta)$-accuracy.
% To see that $R'$ satisfies $(\ep_L, (n'/n) \cdot \delta_L)$-differential privacy, we use Fact~\ref{fac:hockey}. That $R : [\B] \ra \MZ$ is $(\ep_L, \delta_L)$-\DP gives, for any $u,v \in [\B]$,
% $$
% \sum_{z \in \MZ} \left[\p[R(v) = z] - \exp(\ep_L) \cdot \p[R(u) = z]\right]_+
% $$
\end{proof}
A technique similar to the one used in Lemma~\ref{lem:invn} can be used to show that the dependence of the error on $\ep$ must be $\Omega(1/\ep)$ in the central model \cite[Fact 2.3]{steinke2013between}. However, doing so requires each user's input to be duplicated a total of $\Theta(1/\ep)$ times, and it is not clear how to implement such a transformation in the local model. (In the {\it multi-message shuffled model}, though, such a transformation can be done and one would recover the $\Omega(1/\ep)$ lower bound.)

Finally we may establish (\ref{eq:const_lb}); for ease of the reader we state it as a separate lemma:
\begin{lemma}[Proof of Theorems~\ref{thm:freq_lb} and \ref{thm:local_lb} in small-sample regime; i.e., (\ref{eq:const_lb}) \& (\ref{eq:local_const})]
  \label{lem:const_lb}
  There is a sufficiently small positive constant $c$ so that the following holds. Suppose $n, \B \in \BN$ and $\ep_L, \delta_L \geq 0$ with $n \geq \log \B/(\ep_L c)$, $0<  \delta_L < c/n$, and $0 \leq \ep_L \leq \log \B$. Then there is no protocol for $n$-user frequency estimation on $[\B]$ that satisfies $(\ep_L, \delta_L)$-local differential privacy and $\left( \frac{c \log \B}{n \ep_L}, 1/2\right)$-accuracy. %$R : [\B] \ra \MZ$ is an $(\ep + \ln n,\delta)$-\DP local randomizer. Then for any analyzer $A : \MZ^n \ra [0,1]^\B$,
 %  Then there is a dataset $X = (x_1, \ldots, x_n) \in [\B]^n$ so that the estimates $P(X) \in [0,1]^\B$ of the protocol $P = (R,A)$ satisfy %$X = (x_1, \ldots, x_n) \in [\B]^n$ and produces as output frequency estimates $(\hat x_1, \ldots, \hat x_\B) \in [0,1]^\B$. Then there is a dataset $X = (x_1, \ldots, x_n) \in [\B]^n$ so that:% there is a
%   $$
%       \E_{R} \left[\max_{j \in [\B]} \left| P(X)_j - \frac 1n \sum_{i=1}^n (e_{x_i})_j \right|\right] \geq
%       \frac{c\log \B}{ n \ep_L}.
%       $$
  
  For $n \geq \frac{\log \B}{c \log \log \B}$, $\delta < c/n, \ep \leq \log n$, there is no protocol for $n$-user frequency estimation on $[\B]$ that satisfies $(\ep, \delta)$-differential privacy in the single-message shuffled model and $\left( \frac{c \log \B}{n \log \log \B}, 1/2\right)$-accuracy.
%   $P = (R,S,A)$ that is $(\ep, \delta)$-\DP in the $n$-user shuffled model, we have that for some dataset $X$,
%         $$
%        \E_{R} \left[\max_{j \in [\B]} \left| P(X)_j - \frac 1n \sum_{i=1}^n (e_{x_i})_j \right|\right] \geq
%       \frac{c\log \B}{ n \log\log\B}.
%       $$
    \end{lemma}
    \begin{proof}
      Suppose that the statement of the lemma did not hold for some protocol $P$. By Lemma~\ref{lem:invn} with $\alpha = \frac{c\log \B}{ n \ep_L}$ and $n' = \left\lfloor\frac{8c \log \B}{ \ep_L}\right\rfloor$ there is an $(\ep_L, \delta_L)$-locally \DP protocol $P'=(R', A')$ for $n'$-user frequency estimation that is $(1/4, 1/2)$-accurate. As long as $c < 1/160$ we have a contradiction by Lemma~\ref{lem:real_const}.
%       
%       $$
% \E_{P'} \left[ \max_{j \in [\B]} \left| P'(X')_j - \frac{1}{n'} \sum_{i=1}^{n'} (e_{x_i'})_j \right| \right] \leq 1/4
% $$
% for any dataset $X'= (x_1', \ldots, x_{n'}') \in [\B]^{n'}$. As long as $c < 1/160$, we have a contradiction by Lemma~\ref{lem:real_const}.

The second statement of the lemma follows by applying Lemma~\ref{lem:pratio} and taking $\ep_L = \ep + \ln n$.
      \end{proof}

\subsection{Intermediate-Sample and Large-Sample Regimes}
\label{sec:interm_lb}
In this section we prove Theorem~\ref{thm:freq_lb} in the intermediate and large-sample regimes (i.e., (\ref{eq:interm_lb}) and (\ref{eq:small_lb})), which is the most technical part of the proof of Theorem~\ref{thm:freq_lb}. As we did in the small-sample regime, we in fact prove a more general statement giving a lower bound on the accuracy of all {\it locally differentially private} protocols in the low and approximate-privacy setting: % To ease the exposition, we will introduce some additional notation that we will use throughout this section.
  \begin{lemma}[Proof of Theorem \ref{thm:local_lb} in the intermediate-sample regime; i.e., (\ref{eq:local_interm1}) \& (\ref{eq:local_interm2})]% Frequency estimation lower bound for local DP in intermediate-sample regime]
    \label{lem:intermediate_lb}
    There is a sufficiently small positive constant $c$ such that the following holds. Suppose $\ln\B > 1/c^{56}$,
    \begin{equation}
      \label{eq:logb_sandwich}
\frac{1}{c} + \max \left \{ \frac{\ln n }{3},\frac{\ep_L + \ln(1+\ep_L)}{2}\right\} \leq \ln \B \leq \min \left\{ n^{2/3}, \frac{n}{\sqrt{\exp(\ep_L)(1+\ep_L)}}\right\},
\end{equation}
and
\begin{equation}
  \label{eq:deltal_ub}
  \delta_L \leq \min \left\{e^{-\ep_L}, \frac{1}{2\sqrt[4]{n^2 (\ln^2 \B) \exp(\ep_L) (1+\ep_L)}}, \frac{1}{2(\ln^{1/3} \B) n^{2/3}} \right\}.
  \end{equation}
    % $1/c \leq \ln^2 \B \leq n \leq \B^2$, $\delta_L \leq c/n$, and $\ep_L \leq $.
  Then there is no protocol for $n$-user frequency estimation that is $(\ep_L, \delta_L)$-locally \DP and\\ $\left(\alpha, 1/4\right)$-accurate for
  $$
\alpha = c \cdot \min \left\{ \frac{1}{\sqrt[4]{n^2\exp(\ep_L)(1+\ep_L)}}, \frac{ \ln^{1/7} \B}{n^{2/3}} \right\}.
  $$
\end{lemma}
\begin{remark} The term $\ln^{1/7}\B$ in the definition of $\alpha$ above can be replaced by $\ln^{\zeta}\B$ for any constant $\zeta < 1/6$. Moreover, the requirement that $\ln \B$ is greater than the (very large) constant $1/c^{56}$ can easily be reduced to $1/c^6$ (with no change in $c$) by replacing this term $\ln^{1/7}\B$ with $1$.
\end{remark}

  By Lemma~\ref{lem:pratio} the following is a corollary of Lemma~\ref{lem:intermediate_lb}, establishing Theorem \ref{thm:freq_lb} in the intermediate-sample regime:
  \begin{corollary}[Proof of Theorem \ref{thm:freq_lb} in intermediate-sample regime; i.e., (\ref{eq:interm_lb})]
    \label{cor:intermediate_lb}
    For a sufficiently small positive constant $c < 1$, if $\log \B \geq 1/c$,
    \begin{equation}
      \label{eq:withc}\frac 1c \cdot \log^2 \B \cdot \log \log \B \leq n \leq \frac{c\B^2}{\log \B},
      \end{equation}
    $\delta \leq c/n$ and $\ep \leq 1$, there is no protocol for $n$-user frequency estimation in the single-message shuffled model that is $(\ep, \delta)$-\DP and $\left( \frac{c}{n^{3/4} \sqrt[4]{\log n}}, 1/4 \right)$-accurate.
%     If $\alpha \leq \frac{c}{n^{3/4} \sqrt{\ln n}}$ and $\beta \geq 1/3$, then there is no $(\alpha, \beta)$-accurate local randomizer which is $(\ep, \delta)$-DP in the shuffled model with $\delta \leq \min \{ 1/n, cn^{-3/4}/\sqrt{\ln \B} \}$. 
  \end{corollary}
  \begin{remark}
    Notice that the bounds on $n$ in the inequality (\ref{eq:interm_lb}) do not involve $c$, unlike those in (\ref{eq:withc}). To ensure that (\ref{eq:interm_lb}) holds for $\frac{\log^2 \B}{c \log \log \B} \geq n \geq (\log^2 \B)(\log \log \B)$, note that (\ref{eq:const_lb}) holds for {\it all} $n \geq \frac{\log \B}{c \log \log \B}$ (Lemma \ref{lem:const_lb}) and increase the constant in the $\Omega(\cdot)$ in (\ref{eq:interm_lb}) by a factor of at most $1/c$. To ensure that (\ref{eq:interm_lb}) holds for $\frac{c\B^2}{\log \B} \leq n \leq \frac{\B^2}{\log \B}$, we use the reduction to locally \DP protocols given by Lemma~\ref{lem:pratio}, and then note that a locally \DP protocol which is $(\alpha, \beta)$-accurate for $n'$ users implies a locally \DP protocol with the same privacy parameters and which is $(\alpha n'/n, \beta)$-accurate for $n \leq n'$ users by simulating the presence of $n'-n$ fake users who hold a fixed and known item. (This latter reduction also requires increasing the constant in the $\Omega(\cdot)$ in (\ref{eq:interm_lb}) by a factor of at most $1/c$.)

    Finally we note that similar reductions hold for the proof of Theorem \ref{thm:local_lb} using Lemma \ref{lem:intermediate_lb} as well.
  \end{remark}
\begin{proof}[Proof of Corollary~\ref{cor:intermediate_lb}]
  By Lemma~\ref{lem:pratio}, it suffices to show that there is no protocol for $n$-user frequency estimation in the local model that is $(1 + \ln n, \delta)$-\DP and $\left( \frac{c}{n^{3/4} \sqrt[4]{\log n}}, 1/4 \right)$-accurate. We now apply Lemma~\ref{lem:intermediate_lb} with $\ep_L = 1+\ln n$ and $\delta_L = \delta$. The left-hand side of (\ref{eq:logb_sandwich}) holds (though perhaps with a different constant $c$ than the one used here) since $n \leq \frac{c\B^2}{\log \B}$ and $c < 1$, and the right-hand side of (\ref{eq:logb_sandwich}) holds since $cn \geq \log^2 \B \cdot \log \log \B$ (as long as $c$ is sufficiently small). Moreover, (\ref{eq:deltal_ub}) holds as long as 
  $$
\delta \leq \frac{1}{2\cdot (en)^{3/4} \cdot (1 + \ln(en))^{1/4} \cdot (\log^{1/2} \B) },
$$
which is guaranteed by $\delta \leq c/n$ and $cn \geq \log^2 \B \log \log \B$ for sufficiently small $c$. As $\frac{1}{\sqrt[4]{n^2 \exp(\ep_L) (1+\ep_L)}}< n^{-2/3}$ for our choice of $\ep_L$, Lemma~\ref{lem:intermediate_lb} now yields the desired result.
\end{proof}

The hard distribution used to prove Theorem~\ref{thm:freq_lb} in the small-sample regime set each user's data $X_i \in [\B]$ to be equal to some fixed $V \in [\B]$ (Lemma~\ref{lem:real_const}). At a high level, to prove Lemma~\ref{lem:intermediate_lb}, we must adapt this argument to allow us to gain a more fine-grained control over the accuracy of protocols. We do so using the same distribution as in previous works \cite{duchi_minimax,bassily2015local}: in particular, each user's $X_i$ is now only equal to $V$ with some small probability (which is roughly the target accuracy $\alpha$) and otherwise is uniformly random. Formally, we make the following definition: 
% Fix a local randomizer $\Rprot : \MX \ra \MZ$ that is $(\ep, \delta)$-DP in the $n$-user shuffled model.
% Let $V \sim [\B]$ be uniformly distributed.
For each $v \in [\B]$ and $\gamma \in (0,1)$, define a distribution of $X \in [\B]$, denoted by $X \sim \DD_{v,\gamma}$, as
  $$
  X = \begin{cases}
    v \quad & \mbox{w.p.~$\gamma$} \\
    \Unif([\B]) \quad & \mbox{w.p.~$1-\gamma$},
  \end{cases}
  $$
  where $\Unif([\B])$ denotes the uniform distribution on $[\B]$. Let $\bar \DD_\gamma$ denote the joint distribution of $(V,X)$, where $V \sim \Unif([\B])$ and $X \sim \DD_{V,\gamma}$. (Note that the marginal distribution of $X$ under $\bar \DD_\gamma$ is the mixture distribution $\DD_{V,\gamma} = \frac{1}{\B} \sum_{v \in [\B]} \DD_{v,\gamma}$, which is just the uniform distribution on $[\B]$.) Analogously to Lemma~\ref{lem:mutinf}, we wish to derive an upper bound on $I(V; R(X))$ when $(V,X) \sim \bar \DD_\gamma$. It is known  that if $R$ is $(\ep_L, \delta_L)$-\DP and $\ep_L = O(1)$, then $I(V; R(X)) \leq O(\gamma^2 \ep_L^2 + \tilde O(\delta_L / (\ep_L \gamma)))$ \cite{bassily2015local}, and that for {\it any} $\ep_L \geq 0$, if $R$ is $(\ep_L, 0)$-\DP, then $I(V; R(X)) \leq O(\gamma^2 (e^{\ep_L} - 1)^2)$ \cite{duchi_minimax}.

  \begin{remark}
    \label{rmk:inf-counterexample}
    Suppose we attempt to prove Lemma~\ref{lem:intermediate_lb} following this strategy, at least when $\frac{1}{\sqrt[4]{n^2 \exp(\ep_L)(1+ \ep_L)}} < n^{-2/3}$, which is the regime we encounter for single-message shuffled-model protocols. To do so, it is natural to try to improve the upper bound of Duchi et al.~\cite{duchi_minimax} of $I(V; R(X)) \leq O(\gamma^2 \exp(2\ep_L))$ to $I(V; R(X)) \leq \tilde O(\gamma^2 \exp(\ep_L/2))$, which turns out to be sufficient to establish Lemma \ref{lem:intermediate_lb}. However, this is actually {\it false}, as can be seen by the local randomizer $R_{\RR} : [\B] \ra [\B]$ of $\B$-randomized response \cite{warner1965randomized} (see also Appendix \ref{apx:freq_ub}). In particular, suppose we take $\ep_L = (\ln n)  + O(1)$, $n > 10\B$, and $\gamma \ll \exp(-\ep_L/2) = \Theta(\sqrt{1/n})$; it is in fact necessary to treat these settings of the parameters to prove (\ref{eq:interm_lb}). For these parameters it is easy to check that $I(V; R_{\RR}(X)) = \Theta(\gamma \log \B) \gg \gamma^2 \cdot \exp(\ep_L/2)$. Thus it may seem that one cannot derive tight bounds by upper bounding $I(V;R(X))$ when $(V,X) \sim \bar \DD_\gamma$.
    \end{remark}

  It is, however, possible to salvage the technique outlined in Remark \ref{rmk:inf-counterexample}: the crucial observation is that the best-possible additive error of any single-message shuffled-model protocol where each user uses $R_{\RR}$ is $\tilde \Theta(\sqrt{\B}/n)$. When $\ep = O(1) + \ln n$ (as we will have when applying Lemma \ref{lem:pratio}), it is the case that $\sqrt{\B}/n > \frac{1}{\sqrt[4]{n^2 \exp(\ep_L) (1+\ep_L)}}$ when $n < \tilde O(\B^2)$. Therefore, there is still hope to prove a lower bound of $\tilde \Omega \left(\frac{1}{\sqrt[4]{n^2 \exp(\ep_L) (1+\ep_L)}}\right)$ on the additive error when $\ep = O(1) + \ln n$ and $n \leq \tilde O(\B^2)$ if we {\it additionally assume} that the additive error of any local-model protocol using $R$ is bounded above by $\frac{1}{\sqrt[4]{n^2 \exp(\ep_L) (1+\ep_L)}}$. This is indeed what we manage to do in Lemma \ref{lem:intermed_mutinf} below:
%   It is, however, possible to salvage  $R_{\RR}$ is not $\left( o(\sqrt{\B}/n), 1/4)$-accurate (i.e., the best analyzer $A$ is such that the protocol $P = (R,A)$ is not $(o(\sqrt{\B}/n, 1/4))$-accurate; Definition~\ref{def:accuracy}). It is therefore natural to wonder whether we can show the desired upper bound on $I(V;R(X))$ if we additionally constrain $R$ to be $(\alpha, \beta)$ accurate for $\alpha$ sufficiently small. Indeed, this is exactly what we show:
\begin{lemma}[Mutual information upper bound for intermediate-sample regime]
  \label{lem:intermed_mutinf}
  There is a sufficiently large positive constant $C$ such that the following holds. Suppose $n,\alpha, \beta, \gamma, \delta, \ep \geq 0$, $(V,X) \sim \bar \DD_\gamma$ and $\Rprot : [\B] \ra \MZ$ is an $(\alpha, 1/4)$-accurate local randomizer with $C \max\{1/n, 1/\sqrt{n\B} \} \leq \alpha \leq \gamma$, and $C\alpha^2 n \leq 1$ which is $(\ep_L, \delta_L)$-differentially private for $n$-user frequency estimation in the local model with $\delta_L \leq \min \left\{ \frac{\gamma}{\log \B}, e^{-\ep_L}\right\}$. Then
  \begin{equation}
    \label{eq:i3terms}
I(V; \Rprot(X)) \leq C \cdot \left( \gamma^2\alpha^2 n e^{\ep_L} \cdot (1 + \ep_L)+ \gamma\alpha^2 n + \gamma^2\right).
\end{equation}
%       $$
% I(V; \Rprot(X)) \leq C \cdot \left(\gamma^2 \alpha^2 n^2 \log n + \gamma \alpha^2 n + \gamma^2\right).
% $$
\end{lemma}
Typically the term $\gamma^2 \alpha^2 n e^{\ep_L} (1+\ep_L)$ is the dominating one on the right-hand side of (\ref{eq:i3terms}). In particular, in the application of Lemma \ref{lem:intermed_mutinf} to establish (\ref{eq:interm_lb}), we will have $\gamma = \tilde \Theta(n^{-3/4}), \alpha = \tilde \Theta(n^{-3/4})$ and $\ep_L = \ln(n) + O(1)$, so that $\gamma^2 \alpha^2 n e^{\ep_L} (1+\ep_L) = \tilde \Theta(1/n)$, whereas $\gamma \alpha^2 n = \tilde \Theta(n^{-5/4})$ and $\gamma^2 = \tilde \Theta(n^{-3/2})$.
    \begin{remark}
The statement of Lemma~\ref{lem:intermed_mutinf} still holds if $R$ is only assumed to be $(\alpha, \beta)$-accurate for any constant $\beta < 1/2$. 
    \end{remark}
% \begin{remark}
% A choice of $\delta = O(n^{-1})$ will always satisfy the requirements of Lemma~\ref{lem:intermediate_lb}.
% \end{remark}
% Note that Theorem~\ref{thm:intermediate_lb} rules out even those protocols which are accurate with very tiny (i.e., exponentially small in $n$) probability.
    We postpone the proof of Lemma~\ref{lem:intermed_mutinf} for now and assuming it, prove Lemma~\ref{lem:intermediate_lb}.
\begin{proof}[Proof of Lemma~\ref{lem:intermediate_lb}]
  Let $a =200$ and $c < 1$ be a sufficiently small positive constant, to be specified later. Let $\alpha = \min \left\{\frac{c}{\sqrt[4]{n^2\exp(\ep_L)(1+\ep_L)}}, \frac{c \ln^{1/7}\B}{n^{2/3}} \right\}$ be the desired error lower bound. Set $\gamma := \min \left\{\alpha \cdot a \sqrt{\ln \B}/c, 1/3 \right\}$. We make the following observations about $\gamma$:
  \begin{enumerate}
  \item $\gamma < 1/2$ is clear from definition of $\gamma$. % This follows since $c$ can be made sufficiently small and $\sqrt{\ln \B}/{\sqrt n} < 1$ by assumption.
  \item $\gamma^2 \cdot n\B \geq a^2\ln \B$. This is clear if $\gamma \geq 1/3$ by choosing $c$ small enough (recall $\ln \B > 1/c^3$). Otherwise, note that $\gamma^2 \cdot n\B \geq \min \left\{ \frac{a^2 \B\ln \B}{\sqrt{\exp(\ep_L) (1+\ep_L)}}, \frac{a^2 \B \ln \B}{n^{1/3}}\right\} \geq a^2 \ln \B$ since $\max \left\{e^{\ep_L} (1+\ep_L), n^{2/3} \right\} \leq \B^2$.
  \item $\gamma n \geq a \ln \B$. Again this is clear if $\gamma \geq 1/3$. Otherwise, $\gamma n \geq \min \left\{\frac{a\sqrt{n\ln \B}}{\sqrt[4]{\exp(\ep_L)(1+\ep_L)}}, a \sqrt{\ln \B} n^{1/3} \right\} \geq a \ln \B$ since $\min \left\{ n^{2/3}, \frac{n}{\sqrt{\exp(\ep_L)(1+\ep_L)}}\right\} \geq \ln \B $.
  \end{enumerate}
  
  Suppose that $P = (\Rprot, \Aprot)$ is a single-message shuffled model protocol which is $(\alpha, 1/4)$-accurate and $(\ep_L, \delta_L)$-\DP where $\ep_L, \delta_L$ satisfy (\ref{eq:logb_sandwich}) and (\ref{eq:deltal_ub}). 
%  $$\delta_L \leq \min \left\{e^{-\ep_L}, \frac{c}{\sqrt[4]{n^2 (\ln^2 \B) \exp(\ep_L) (1+\ep_L)}} \right\}.$$
  Now suppose $V \sim [\B]$ uniformly and $X_1, \ldots, X_n \sim \DD_{V,\gamma}$ are independent (conditioned on $V$). 
  
  Fix an arbitrary $v \in [\B]$, and let us momentarily condition on the event that $V = v$. Consider the conditional distribution of $X_1, \ldots, X_n \sim \DD_{v,\gamma}$. For any $u \neq v$, we have, by the Chernoff bound, in the case that $\gamma/3 \leq 1/\B$,
  \begin{equation}
    \label{eq:union1}
\P\left[ \frac 1n \sum_{i=1}^n (e_{X_i})_u \geq 1/\B + \gamma/3 \right] \leq \exp\left( - \frac{(\gamma \B)^2 \cdot n/\B}{27} \right) = \exp(-\gamma^2 \B n/27) \leq \exp(-a^2 \ln \B/27).
\end{equation}
In the case that $\gamma/3 > 1/\B$, again by the Chernoff bound, we have
\begin{equation}
  \label{eq:union2}
\P\left[ \frac 1n \sum_{i=1}^n (e_{X_i})_u \geq 1/\B + \gamma/3 \right] \leq \exp \left( -\frac{(\gamma \B) \cdot n/\B}{9} \right) = \exp(-\gamma n /9) = \exp(-a \sqrt{\ln \B} n^{1/4}/9) \leq \exp(-a \ln \B/9).
\end{equation}
Next, note that by definition of the distribution $\DD_{V,\gamma}$, we have that for each $1 \leq i \leq n$, $\E[(e_{X_i})_v] = \frac{1-\gamma}{\B} + \gamma$. It then follows by the Chernoff bound that % Also by the Chernoff bound,
\begin{align}
  \P\left[ \frac 1n \sum_{i=1}^n (e_{X_i})_v \leq (1-\gamma)/\B + 2\gamma/3 \right]
  & \leq \exp(-a^2 (1-\gamma) \ln \B/72) + \exp(-\gamma n / 72) \nonumber\\
  \label{eq:union3}
  & \leq \exp(-a^2 \ln \B / 144) + \exp(-a\ln \B / 72).
\end{align}
Since $1/\B + \gamma/3 < (1-\gamma)/\B + 2\gamma/3$ and $P = (R,A)$ is $(\gamma/3,1/3)$-accurate (as $\gamma/3 \geq \alpha$), it follows by a union bound over all $u \in [\B]$ in (\ref{eq:union1}), (\ref{eq:union2}) and (\ref{eq:union3}) that with probability at least
\begin{equation}
  \label{eq:bmexpa}
1 - 1/3 - \exp(-a/9) - 2 \exp(-a\ln \B/72),
\end{equation}
we have that
$$
\arg \max_{u \in [\B]} P((X_1, \ldots, X_n)) = \arg \max_{u \in [\B]} \Aprot(\Rprot(X_1), \ldots, \Rprot(X_n))_u = v.
$$
Moreover, by our choice of $a = 200$, we ensure that the probability in (\ref{eq:bmexpa}) is strictly greater than $1/4$. For such $a$, using the fact that $v \in [\B]$ is arbitrary, we have shown that
\begin{equation}
  \label{eq:plt14}
\P\left[ \arg \max_{u \in [\B]} \{ \Aprot(\Rprot(X_1), \ldots, \Rprot(X_n))_u\} = V \right] > 1/4.
\end{equation}

Now we will apply Lemma~\ref{lem:intermed_mutinf} to derive an upper bound on the probability in the above equation. First we check that the conditions of Lemma~\ref{lem:intermed_mutinf} are met. By (\ref{eq:logb_sandwich}) and $\ln \B \geq 1/c^3$ we have that
\begin{equation}
  \label{eq:alphan}
(\alpha n)^2 \geq \min\left\{ c^2 n^{2/3}, \frac{c^2 n}{\sqrt{\exp(\ep_L) (1+\ep_L)}}\right\} \geq c^2 \ln \B \geq 1/c,
  \end{equation}
  so by choosing $c$ small enough, we can guarantee that $\alpha \geq C/n$, where $C$ is the constant of Lemma~\ref{lem:intermed_mutinf}. Similarly, by (\ref{eq:logb_sandwich}), we have that
  \begin{equation}
    \label{eq:alphabn}
    \alpha^2 n\B \geq \min \left\{ c^2 \B/n^{1/3}, \frac{c^2\B}{\sqrt{\exp(\ep_L)(1+\ep_L)}}\right\} \geq c^2 \exp(1/c),
  \end{equation}
  and again by choosing $c$ small enough, we can guarantee that $\alpha \geq C/{\sqrt{n\B}}$, where $C$ is the constant of Lemma~\ref{lem:intermed_mutinf}.

  The choice of $\gamma$ ensures that $\gamma \geq \alpha$, and $C\alpha^2 n =\min \left\{\frac{Cc^2}{\sqrt{\exp(\ep_L)(1+\ep_L)}}, \frac{Cc^2\ln^{2/7} \B}{n^{1/3}}\right\} \leq Cc^2$, which can be made less than 1 by choosing $c$ sufficiently small. Finally, $\delta_L \leq \min \left\{ e^{-\ep_L}, \frac{\gamma}{2\ln \B} \right\} \leq \min \left\{ e^{-\ep_L}, \frac{\gamma}{\log \B} \right\}$ by (\ref{eq:deltal_ub}). Therefore, by Fano's inequality and Lemma ~\ref{lem:intermed_mutinf}, for any function $f : [0,1]^\B \ra [\B]$,
\begin{align}
  \P[f(\Aprot(\Rprot(X_1), \ldots, \Rprot(X_n))) = V] & \leq \frac{I(V; (\Rprot(X_1), \ldots, \Rprot(X_n))) + 1}{\ln \B} \nonumber\\
                                                      & \leq \frac{ \sum_{i=1}^n I(V; \Rprot(X_i)) + 1}{\ln \B}\nonumber\\
                                                      & \leq \frac{1 + n \cdot \left(C \cdot \left(\gamma^2 \alpha^2 ne^{\ep_L} \cdot (1 + \ep_L) + \gamma \alpha^2 n + \gamma^2\right)\right)}{\ln \B} \nonumber\\
  \label{eq:1/c}
   & \leq \frac{1 + n C \cdot \left( \gamma^2 \alpha^2 ne^{\ep_L} \cdot (1 + \ep_L) + \frac{2a}{c^2} \cdot \gamma \alpha^2 n \right)}{\ln \B} \\
    \label{eq:use_mutinf}
  & \leq \frac{1 + Cc \ln \B}{\ln \B}. %\frac{C' \cdot \left(\gamma^2 \alpha^2 n^2 e^{\ep_L} \cdot (1+ \ep_L)  + 1 \right)}{\ln \B},
 %                                                      & \leq \frac{C \cdot a^2 \log \B \alpha^2 n^{3/2}\log n+1}{\log \B},
\end{align}
% where $C'$ is some sufficiently large constant (depending on $C$).
Inequality~(\ref{eq:1/c}) follows since
$$
\frac{\gamma^2}{\gamma \alpha^2 n} \leq \frac{\alpha a \sqrt{\ln \B}/c}{\alpha^2 n} \leq \frac{a}{c} \cdot \frac{\sqrt{\ln \B}}{\alpha n} \leq \frac{a}{c^2},
% \alpha^2 n e^{\ep_L} (1+\ep_L) \geq c^2  \cdot \min \left\{ \sqrt{\exp(\ep_L) (1+\ep_L)}, \frac{(\ln^{1/3} \B) \exp(\ep_L) (1+\ep_L)}{n^{1/3}} \right\}
$$
where we have used inequality~(\ref{eq:alphan}). Inequality~(\ref{eq:use_mutinf}) follows since, by choice of $\alpha$,
\begin{align}
  \max \left\{\gamma^2 \alpha^2 n \exp(\ep_L) (1+\ep_L), \frac{a}{c^2} \cdot \gamma \alpha^2n \right\} & \leq \max \left\{ \frac{c^2a^2 (\ln \B) n \exp(\ep_L)(1+\ep_L) }{n^2 \exp(\ep_L)(1+\ep_L)}, \frac{a^2 (\ln^{13/14} \B) \cdot n}{cn^2} \right\} \nonumber\\
  % \frac{ac\sqrt{\ln \B}\cdot n}{(n^2 \exp(\ep_L)(1+\ep_L))^{3/4}} \right\}\\
  \label{eq:lnBlarge}
  & \leq  \frac{c \ln \B}{n}, %\max \left\{ \frac{c (\ln \B)}{n} , \frac{c (\ln \B)}{n} \right\}, %\frac{ac\sqrt{\ln \B}}{\sqrt{n} \cdot (\exp(\ep_L)(1+\ep_L))^{3/4}}\right\}.
\end{align}
where (\ref{eq:lnBlarge}) follows from $\ln \B > 1/c^{56} > (a^2 / c^2)^{14}$ and a choice of $c < 1/a^2$. 

As long as $c < 1/(10C)$, the expression in (\ref{eq:use_mutinf}) is bounded above by 1/4, which contradicts (\ref{eq:plt14}).
% To arrive at the last inequality, we have used that $\gamma^2 = \frac{a^2 \ln \B}{n^{3/2}} \leq 1$ for sufficiently large $n$ and similarly that $\gamma \alpha^2 n = \frac{a \sqrt{\ln \B} \alpha^2}{n^{3/2}} \leq 1$ for sufficiently large $n$.
% We thus get a contradiction since by choosing the constant $c$ in the lemma statement to be sufficiently small, we can guarantee $\alpha \leq \frac{1}{an^{3/4} \sqrt{C\ln n}}$, which implies that the expression (\ref{eq:use_mutinf}) is bounded above by $1/4$. 
\end{proof}

%     \begin{lemma}
%     Suppose $n \leq \todo{?}$,  $V \sim \Unif([\B])$ and $X_1, \ldots, X_n \sim \MD_\B(V)$. Then
%       $$
% I(V; (\Rprot(X_1), \ldots, \Rprot(X_n)))  \leq \todo{?}.
%       $$
%     \end{lemma}
We now prove Lemma~\ref{lem:intermed_mutinf}. The proof uses the assumption that the local randomizer $R$ is $(\alpha, 1/4)$-accurate (Definition~\ref{def:accuracy}) to derive, for each $v \in [\B]$, a lower bound on the total variation distance between the distributions of $R(v)$ and $R(V)$, where $V \sim \Unif([\B])$. Intuitively, it makes sense that if, for some $v \in [\B]$, the distribution of $R(v)$ is close to $R(V)$, then no analyzer can reliably compute how many users hold the item $v$. However, showing rigorously that this holds for {\it any} analyzer $A$ is nontrivial, and we state this result as a separate lemma:
  \begin{lemma}[Lower bound on total variation distance between $R(v)$ \& $R(V)$]
    \label{lem:lb_tvd}
    Suppose $\Rprot : [\B] \ra \MZ$ is an $(\alpha/6, \beta)$-accurate local randomizer such that % $n \geq \frac{24 \log(4/(1-2\beta))}{\alpha}$.
    \begin{equation}
      \label{eq:constrain_alpha}
    \max \left\{ \frac{3\B \log (4/(1-2\beta))}{n}, \sqrt{\frac{3\B \log(4/(1-2\beta))}{n}} \right\} \leq \frac{\alpha\B}{4}.
    \end{equation}
    Let the distribution of $R(v)$ be denoted $\PP_v$ (Definition~\ref{def:pp}) and the distribution of $R(V)$, where $V \sim \Unif([\B])$ be denoted $\QQ$. Then there is some $C  = \Theta\left( \frac{1}{(1-2\beta)^2} \right)$, such that for each $v \in [\B]$, $\Delta(\PP_{v}, \QQ) \geq 1 - C\alpha^2 n$. %, where we are using the above notation for $\PP_v, \QQ$. 
\end{lemma}
A result similar to Lemma \ref{lem:lb_tvd} was established in \cite{chan_optimal_2012}; however, their result only establishes a lower bound on $\Delta(\PP_v, \PP_u)$ for $u \neq v$, which does not lead to tight bounds on $\Delta(\PP_v, \QQ)$, as we need. % (and as is given by Lemma \ref{lem:lb_tvd}). 
The proof of Lemma~\ref{lem:lb_tvd} is provided in Section~\ref{sec:lb_tvd}.

    \begin{proof}[Proof of Lemma~\ref{lem:intermed_mutinf}]
      % Fix $\mu \in \todo{\approx \sqrt n}$.
      Fix a local randomizer $R : [\B] \ra \MZ$ satisfying the requirements of the lemma statement. Recall (per Definition~\ref{def:pp}) that for $v \in [\B]$, we use $\PP_v$ to denote the distribution of $\Rprot(v)$, and that $\pp_{v,\cdot}$ denotes the density of $\PP_v$, so that for each $z \in \MZ$, $\pp_{v,z} = \p_{R}[\Rprot(v) = z]$. 
      For $v \in [\B]$, additionally let $\PP_{v,\gamma}$ denote the distribution of $\Rprot(X)$ when $X \sim \DD_{v,\gamma}$, and let $\QQ$ denote the distribution of $\Rprot(X)$ when $(X,V) \sim \bar \DD_\gamma$. Note that $\QQ$ is the distribution of $R(X)$ when $X \sim \Unif([\B])$, so indeed does not depend on $\gamma$. 
      First note (see (\ref{eq:inf_def})) that
      $$
I(V; \Rprot(X)) = \frac{1}{\B} \sum_{v \in [\B]} \KL(\PP_{v,\gamma} || \QQ).
$$
For each $\MS \subset \MZ$ and $z \in \MZ$, write $q_\MS := \P[Z \in \MS]$ and $q_z := \p[Z=z]$, where $Z \sim \QQ$. Notice that for each $z \in \MZ$, we have $\pp_{v,z} \leq \B \cdot q_z$ since $q_{v,z} = \frac{1}{\B} \sum_{v \in [\B]} \pp_{v,z}$. 

% Next, for each $h \in H: = \{ \B, \B/2, \B/4, \ldots \}$ and $v \in [\B]$, set
Next, for each $h \geq 0$ and $v \in [\B]$, set 
$$
\MU_{v,h} := \left\{ z \in \MZ : \frac{\pp_{v,z}}{q_z} = h \right\}.
% \MU_{v,h} := \left\{ z \in \MZ : \frac{\pp_{v,z}}{q_z} \in \{ h/2, h \} \right\}.
$$
Let $\rho_v$ be the (Borel) probability measure on $\BR$ given by, for $\MS \subset \BR$,
$$
\rho_v(\MS) := \sum_{h \geq 0} \pp_{v,\MU_{v,h}} \delta_h(\MS),
$$
where the sum is well-defined since only finitely many $h$ are such that $\MU_{v,h}$ is nonempty. (Here $\delta_h$ is the measure such that $\delta_h(\MS) = 1$ if $h \in \MS$, and otherwise $\delta_h(\MS) = 0$.) Since $\sum_{z \in \MZ} \pp_{v,z} = \sum_{h \geq 0} \pp_{v,\MU_{v,h}} = 1$, $\rho_v$ is indeed a probability measure.

Next notice that
\begin{align}
  \KL(\PP_{v,\gamma} || \QQ) &= \sum_{z \in \MZ} \left( \gamma \pp_{v,z} + (1-\gamma) q_z\right) \cdot \ \log\left( 1 + \gamma \left( \frac{\pp_{v,z}}{q_z} - 1 \right) \right) \nonumber\\
                              &  \leq \sum_{z : \pp_{v,z} \geq q_z} \pp_{v,z}\left( \frac{1}{\pp_{v,z}/q_z} + \gamma \left( 1 - \frac{1}{\pp_{v,z}/q_z} \right) \right) \cdot \log\left(1 + \gamma  \cdot \frac{\pp_{v,z}}{q_z}\right) \nonumber\\
  &+\sum_{z : \pp_{v,z} < q_z} \left( \gamma \pp_{v,z} + (1-\gamma) q_z\right) \cdot \ \log \left( 1 + \gamma \left( \frac{\pp_{v,z}}{q_z} - 1 \right) \right)\nonumber\\ 
                              & = \int_{h =1}^{\B} \left( \frac{1}{h} + \gamma \left(1 - \frac{1}{h} \right)\right)  \log(1 + \gamma h) d\rho_v(h) \nonumber \\
                                % \sum_{h \in H, h > 1} \p_{v, U_{v,h}} \cdot \left( \frac 2h + \gamma \left( 1 - \frac 2h \right) \right) \cdot \log(1 + \gamma h) \nonumber\\
  \label{eq:kl_ub_all}
  &+\sum_{z \in \cup_{h < 1} \MU_{v,h}} \left( \gamma \pp_{v,z} + (1-\gamma) q_z\right) \cdot \log\left( 1 + \gamma \left( \frac{\pp_{v,z}}{q_z} - 1 \right) \right).
\end{align}

We begin by working towards an upper bound on the first term in the above expression (\ref{eq:kl_ub_all}), corresponding to values $h \geq 1$. Our first goal is to show that for $h \gg 1$, $\pp_{v,\MU_{v,h}}$ is small for most $v$. To do so, define, for any $\lambda \geq 1$,
$$
\MT_{v,\lambda} := \left\{ z \in \MZ : \frac{\pp_{v,z}}{q_z} \geq  \lambda \right\} = \bigcup_{h \geq \lambda} \MU_{v,h}.
  $$
  % (We have, for instance, that $\MT_{v,1} = \bigcup_{h \geq 2 \sqrt{n}} \MU_{v,h}$.)
  We next make the following claim:
  \begin{claim}
    \label{clm:bcn}
    For any $v \in [\B]$, for each $z \in \MT_{v,\lambda}$, there are at most $\B/\la$ values of $v' \in [\B]$ such that $z \in \MT_{v',\la}$.
  \end{claim}
  Claim~\ref{clm:bcn} is a simple consequence of Markov's inequality on the distribution of the random variable $\pp_{v,z}$, where $v \sim [\B]$ uniformly. 
  \if 0
  \begin{proof} To see this claim, suppose for the purpose of contradiction that $u_1, \ldots, u_k \in [\B]$ are distinct such that there is some $z \in \MT_{u_s,c}$ for each $s \in [k]$, and that $k \geq \B/(c\mu)$. For each $s \in [k]$ we have by construction that
  $$
\pp_{u_s,z} > \frac{c\mu}{\B} \sum_{j=1}^\B \pp_{j,z}.
$$
Thus
$$
\frac{1}{k} \sum_{s=1}^k \pp_{u_s,z} > \frac{c\sqrt{n}}{\B} \sum_{j=1}^{\B} \pp_{j,z},
$$
which is a contradiction in light of the fact that $k \geq \B/(c\sqrt{n})$.
\end{proof}
\fi

%   By Lemma~\ref{lem:pratio}, for all $u,v \in [\B]$ with $u \neq v$, we have that
%   \begin{equation}
%     \label{eq:use_ration}
%   \pp_{u, T_{u,c} \backslash T_{v,c}} \leq 2e^\ep n \pp_{v, T_{u,c} \backslash T_{v,c}} + 2\delta.
  % \pp_{u, T_{u,c} } \leq 2e^\ep n \pp_{v, T_{u,c} } + 2\delta.
% \end{equation}
% For each $u \in [\B]$, set $\tilde \pp_{u,T_{u,c}} := \pp_{u,T_{u,c}} - 2\delta$, and for each $z \in T_{u,c}$, set $\tilde \pp_{u,z} := \pp_{u,z} - 2\delta \cdot \frac{\pp_{u,z}}{\pp_{u,T_{u,c}}}$.

Next, consider any $z \in \MZ$ such that there is some $u \in [\B]$ with $z \in \MT_{u,\la}$; let the set of such $z$ be denoted by $\MW_\la$, i.e., $\MW_\la = \bigcup_u \MT_{u,\la}$. By Claim~\ref{clm:bcn}, for each $z \in \MW_\la$, there are at most $\B/\la$ values $u \in [\B]$ such that $z \in \MT_{u,\la}$. Let the set of such values be denoted by $\MS_z \subset [\B]$, and construct an ordering of those $u \in \MS_z$ so that $\pp_{u,z}$ are in decreasing order with respect to this ordering. Now, for a fixed $\la$ and for each $u \in [\B]$, and $1 \leq k \leq \B/\lambda$, construct a subset $\tilde \MT_{u,\la}^{(k)} \subset \MT_{u,\la}$ such that each $z \in \MW_\la$ appears in at most one set $\tilde \MT_{u,\la}^{(k)}$ (over all $u \in [\B]$), and such a $u$ is the $k$th element of $\MS_z$ with respect to the ordering above (if it exists). % Formally, we may inductively set
% $$
% \tilde T_{u,\la}^{(k)} := \left\{ z \in T_{u,\la} | \pp_{u,z} \geq \pp_{v,z} \ \forall \ v \in S_z \backslash \bigcup_{j=1}^{k-1} \tilde T_{u,\la}^{(j)} \mbox{ and } u < v \ \forall \ v \in S_z \ \mbox{ s.t. } \pp_{u,z} = \pp_{v,z}
%   \right\}.
%   $$
It is an immediate consequence of this construction that for each fixed $k$, the sets $\tilde \MT_{u,\la}^{(k)}$, $u \in [\B]$, are pairwise disjoint. Moreover, for each fixed $u$, the sets $\tilde \MT_{u,\la}^{(k)}$, $1 \le k \leq \B/\la$ are pairwise disjoint, and their union is $\MT_{u,\la}$. It follows that % Next, note that
$$
\sum_{k=1}^{\B/\la} \sum_{u \in [\B]} \pp_{u, \tilde \MT_{u,\la}^{(k)}} =  \sum_{z \in \MZ} \sum_{k=1}^{\B/\la}\sum_{u : z \in \tilde \MT_{u,\la}^{(k)}} \pp_{u,z} = \sum_{z \in \MZ} \sum_{u : z \in \MT_{u,\la}} \pp_{u,z} = \sum_{u \in [\B]} \pp_{u,\MT_{u,\la}}.
$$
% DONT THINK THIS IS NEEDED:
% We also have by construction that % It is also the case that
% $$
% \sum_{u \in [\B]} \pp_{u,\tilde \MT_{u,\la}^{(1)}} \geq \sum_{u \in [\B]} \pp_{u,\tilde \MT_{u,\la}^{(2)}} \geq \cdots \geq \sum_{u \in [\B]} \pp_{u,\tilde \MT_{u,\la}^{(\B/\la)}}.
% $$
%   \begin{align}
%     & \sum_{u \in [\B]} \pp_{u, \tilde T_{u,\la}^{(k)}} \nonumber\\
%     & = \sum_{z \in \MZ} \sum_{u : z \in \tilde T_{u,\la}} \pp_{u,z} \nonumber\\
%     & \geq \frac{\la}{\B}\sum_{z \in \MZ} \sum_{u : z \in T_{u,\la}} \pp_{u,z} \nonumber\\
%     \label{eq:tildet}
%     & = \frac{\la}{\B} \sum_{u \in [\B]} \pp_{u, T_{u,\la}}.
%   \end{align}
From $(\ep_L, \delta_L)$-differential privacy of $R$ we have that \begin{equation}
  \label{eq:epldp}
  \max_{\MS \subset \MZ} \{\pp_{u,\MS} - e^{\ep_L} \cdot \pp_{v,\MS}\} \leq \delta_L.
\end{equation}
By Fact~\ref{fac:hockey} and the fact that the sets $\tilde \MT_{u,\la}^{(1)}, \ldots, \tilde \MT_{u,\la}^{(\B/\la)}$ are pairwise disjoint for any $u \in [\B]$, we have that for all $u,v \in [\B]$,
  \begin{equation}
    \label{eq:puv_hs}
\sum_{k=1}^{\B/\la} \sum_{z \in \tilde \MT_{u,\la}^{(k)}} \left[ \pp_{u,z} - e^{\ep_L} \cdot \pp_{v,z} \right]_+ \leq \delta_L.
  \end{equation}
Averaging (\ref{eq:epldp}) over $v \in [\B]$ gives that $\max_{\MS \subset \MZ} \{\pp_{u,\MS} - e^{\ep_L} q_\MS\} \leq \delta_L$. Fact~\ref{fac:hockey} then gives
%   Averaging over all $v \in [\B]$ and using convexity of the function $x \mapsto [x]_+$ gives
  \begin{equation}
    \label{eq:pq_hs}
\sum_{k=1}^{\B/\la}  \sum_{z \in \tilde \MT_{u,\la}^{(k)}} \left[ \pp_{u,z} - e^{\ep_L} q_z \right]_+ \leq \delta_L.
  \end{equation}
  By (\ref{eq:puv_hs}) and (\ref{eq:pq_hs}), we have that, for all $u,v \in [\B]$,
  \begin{equation}
    \label{eq:fixuv}
e^{\ep_L} \cdot \sum_{k=1}^{\B/\la}\sum_{z \in \tilde \MT_{u,\la}^{(k)}} \min \{ \pp_{v,z}, q_z \} \geq \pp_{u, \MT_{u,\la}} - 2\delta_L.
\end{equation}

\noindent
For each $v \in [\B]$ and $1 \leq k \leq \B/\la$,
  $$
1- \Delta(\PP_v, \QQ) \geq \sum_{z \in \MZ} \min \{ \pp_{v,z}, q_z \} \geq \sum_{u \in [\B]} \sum_{z \in \tilde \MT_{u,\la}^{(k)}}\min\{ \pp_{v,z} ,q_z \}.
  $$

\noindent
Averaging over $k$ and using (\ref{eq:fixuv}), it follows that for any $v \in [\B]$,
\begin{align}
1 - \Delta(\PP_v, \QQ)  & \geq
                                       \frac{\la}{\B} \sum_{k=1}^{\B/\la} \sum_{u \in [\B]} \sum_{z \in \tilde \MT_{u,\la}^{(k)}} \min\{\pp_{v,z}, q_z \} \nonumber\\
                                     & \geq \frac{\la}{\B} \sum_{u \in [\B]} \frac{\pp_{u,\MT_{u,\la}} - 2\delta_L}{e^{\ep_L}}\nonumber\\
%   & \geq 
%     \left( \frac{1}{2e^\ep n} \sum_{u \in [\B]} \p_{u, \tilde T_{u,\la}} \right) - \frac{4\delta\B}{2e^\ep n} \nonumber\\
  & \geq \left( \frac{\la}{e^{\ep_L} \B} \sum_{u \in [\B]} \pp_{u,\MT_{u,\la}} \right) - \frac{2\delta_L\la}{e^{\ep_L}}.\nonumber
\end{align}
By Lemma~\ref{lem:lb_tvd} and the $(\alpha, 1/4)$-accuracy of $\Rprot$, together with the fact that $$\alpha/4 \geq \max \left\{\frac{3 \log (4/(1-2 \cdot 1/4))}{n}, \sqrt{\frac{3 \log(4/(1-2\cdot 1/4))}{n\B}} \right\}$$ as long as the constant $C$ is chosen large enough (recall the assumption $\alpha \geq C\max\{1/n, 1/\sqrt{n\B}\}$), we have that, perhaps by making $C$ even larger, $1 - \Delta(\PP_v, \QQ) \leq C \alpha^2 n$. In particular, it follows that
\begin{equation}
  \label{eq:psum_ub}
\frac{1}{\B} \sum_{u \in [\B]}\pp_{u,\MT_{u,\la}} \leq \frac{e^{\ep_L}}{\la} \cdot \left( C\alpha^2 n + \frac{2\delta_L \la}{e^{\ep_L}} \right) = \frac{C \alpha^2 n e^{\ep_L}}{\la} + 2\delta_L. % \leq \frac{\alpha^2 n^{2} (2Ce^\ep + 4)}{\la},
\end{equation}
% where the inequality uses $\delta \leq \alpha^2 n^2/\B$. 
Using the inequality $\log(1+\gamma h) \leq \gamma h$, we can now upper bound the first term in (\ref{eq:kl_ub_all}), when averaged over $v \in [\B]$, as follows: %  for any $ \geq 1$,
\begin{align}
  & \frac{1}{\B} \sum_{v \in [\B]}\int_{h = 1}^{\B} \left( \frac{1}{h} + \gamma \left(1 - \frac{1}{h} \right)\right)  \log(1 + \gamma h) d\rho_v(h) \nonumber\\
%   & \leq \frac{1}{\B} \sum_{v \in [\B]} \int_{h=1}^{\mu} \left( \gamma + \gamma^2 h \right) d\rho_v(h) + \frac{1}{\B} \sum_{v \in [\B]}\int_{h = \mu}^{\B} \left( \frac{1}{h} + \gamma \left(1 - \frac 1h \right)\right) \log(1+ \gamma h) d\rho_v(h) \nonumber\\
  & \leq \gamma + \frac{1}{\B} \sum_{v \in [\B]} \int_{h=1}^{2\exp(\ep_L)} \gamma  \log(1 + \gamma h) d\rho_v(h)  + \frac{1}{\B} \sum_{v \in [\B]} \int_{h=2\exp(\ep_L)}^{\B} \gamma \log(1+\gamma h) d\rho_v(h)\nonumber\\
  \label{eq:exp_h}
  & \leq \gamma + \gamma^2 \int_{h=1}^{2\exp(\ep_L)} h \cdot \left(\frac{1}{\B} \sum_{v \in [\B]} d\rho_v(h)\right) + \gamma \log \B\int_{h=2\exp(\ep_L)}^{\B} \left( \frac{1}{\B} \sum_{v \in [\B]} d\rho_v(h) \right).
\end{align}
(In the integrals above, for an integral of the form $\int_{h=x}^y$ we integrate over the {\it closed} interval $[x,y]$, so that point masses at $x$ and $y$ are included.) Let $\rho$ be the Borel measure on $\BR$ defined by $\rho = \frac 1\B \sum_{v \in [\B]} \rho_v$.

Recall that $(\ep_L, \delta_L)$-differential privacy of $R$ gives that for any $\MS \subset \MZ$, $\pp_{v,\MS} \leq e^{\ep_L} q_\MS + \delta_L$. Thus, setting $\MS = \{ z \in \MZ : \pp_{v,z} \geq 2e^{\ep_L} q_z \}$ gives $\pp_{v,\MS} \leq e^{\ep_L} q_\MS + \delta_L \leq e^{\ep_L} \cdot \pp_{v,\MS}/(2e^{\ep_L}) + \delta_L$, so that $\pp_{v,\MS} \leq 2\delta_L$. It follows by averaging this inequality over all $v \in [\B]$ that 
\begin{equation}
  \label{eq:forgotten_delta}
\gamma \log \B \int_{h=2\exp(\ep_L)}^\B d\rho(h) \leq \gamma \log \B \cdot 2\delta_L.
\end{equation}

\noindent
Next, note that (\ref{eq:psum_ub}) gives us that for any $\la \geq 1$,
$$
\int_{h=\la}^{2\exp(\ep_L)} d\rho(h) \leq \int_{h \geq \la} d\rho(h) =  \frac{1}{\B} \sum_{u \in [\B]} \pp_{u,\MT_{u,\la}} \leq \frac{C\alpha^2 n \cdot e^{\ep_L}}{\la} + 2\delta_L.
$$
It follows that
\begin{align}
  & \int_{h = 1}^{2\exp(\ep_L)} h d\rho(h) \nonumber\\
  & \leq \sum_{h' \in \{ 1, 2  , \ldots, 2^{\lceil \log 2\exp(\ep_L) \rceil}\}} \int_{h'/2}^{h'} hd\rho(h) \nonumber\\
  & \leq \sum_{h' \in \{ 1, 2, \ldots, 2^{\lceil \log 2\exp(\ep_L) \rceil}\}} 2h' \cdot \left(\frac{2C\alpha^2 n e^{\ep_L}}{h'}+2\delta_L\right) \nonumber\\
  \label{eq:exp_h_ub}
  & \leq (4 + \ep_L \cdot \log(e)) \cdot \left(\alpha^2 n \cdot 4Ce^{\ep_L} \right) + 32 \delta_L e^{\ep_L}.
\end{align}

Next we upper bound the terms $h < 1$ in the sum of (\ref{eq:kl_ub_all}). Again using Lemma~\ref{lem:lb_tvd} and the $(\alpha, 1/4)$-accuracy of $\Rprot$, we see that there is a constant $C$ such that for each $v \in [\B]$, it holds that $1-\Delta(\PP_{v}, \QQ) \leq C\alpha^2n$. % there is a set $Y_v$ such that $q_{Y_v} - \pp_{v,Y_v} \geq C \gamma^2 n$. In particular, it follows that $\pp_{v,Y_v} \leq C\gamma^2 n$ and $q_{Y_v} \geq 1- C \gamma^2 n$.
Hence we have
$$
\sum_{z : \pp_{v,z} < q_z} \pp_{v,z} \leq C\alpha^2 n
$$
and
$$
\sum_{z : \pp_{v,z} < q_z} q_z \geq 1 - C\alpha^2 n.
$$

We next need the following claim:
\begin{claim}
  \label{clm:y_separate}
  Let $\tau \in (0,1)$. Suppose $\MY$ is a finite set and for each $z \in \MY$, $p_z, q_z \in [0,1]$ are defined such that $\sum_{z \in \MY} p_z \leq \tau$ and $1-\tau \leq \sum_{z \in \MY} q_z \leq 1$. Suppose also that $p_z \leq q_z$ for all $z \in \MY$. % that for some $\Lambda \geq 2$, $p_z / q_z \leq \Lambda$ for all $z \in \MY$. % Suppose also that for some $m \in \BR_{+}$ and $\delta \in (0,1)$,
%   $$
% \sum_{z \in \MY} [p_z - m \cdot q_z]_+ \leq \delta.
%   $$
  Then for any $\gamma \in (0,1/2)$,
  $$
\sum_{z \in \MY} (\gamma p_z + (1-\gamma) q_z) \cdot \log \left( 1 + \gamma \left( \frac{p_z}{q_z} - 1\right)\right) \leq -\gamma +  2\gamma \tau + \gamma^2 (1 + \tau).
  $$
\end{claim}
\begin{proof}
  Using the fact that $\log(1+x) \leq x$ for all $x \geq -1$, we have
  \begin{align}
    & \sum_{z \in \MY} (\gamma p_z + (1-\gamma) q_z) \cdot \log \left( 1 + \gamma \left( \frac{p_z}{q_z} - 1\right)\right) \nonumber\\
    & \leq \sum_{z \in \MY } (\gamma p_z + (1-\gamma) q_z) \gamma \cdot \left( \frac{p_z}{q_z} - 1 \right)  \nonumber\\ %\sum_{z \in \MY : p_z > q_z} (\gamma p_z + (1-\gamma) q_z) \gamma \cdot \log \left(1 + \gamma \left(\frac{p_z}{q_z} - 1 \right)\right) \nonumber\\
    & \leq -(1-\tau)(1-\gamma) \gamma + \sum_{z \in \MY } (\gamma p_z + (1-\gamma) q_z) \gamma p_z / q_z \nonumber\\ % \sum_{z \in \MY : p_z > q_z} (\gamma p_z + (1-\gamma) q_z) \gamma \cdot \log \left(1 + \gamma \left(\frac{p_z}{q_z} \right)\right) \nonumber\\
    & \leq -(1-\tau-\gamma) \gamma + \gamma \tau + \sum_{z \in \MY} \gamma^2 p_z^2 / q_z\nonumber\\%  +  \sum_{z \in \MY : p_z > q_z} (\gamma p_z + (1-\gamma) q_z) \gamma \cdot \log \left(1 + \gamma \left(\frac{p_z}{q_z} \right)\right) \nonumber\\% \sum_{z \in \MY} \gamma^2 p_z^2/q_z \nonumber\\
    & \leq -\gamma + 2\gamma \tau + \gamma^2 + \gamma^2 \tau\nonumber.
  \end{align}
\end{proof}
Using Claim~\ref{clm:y_separate} with $\MY = \{ z \in \MZ : \pp_{v,z} < q_z\}$, $\tau = C\alpha^2 n \leq 1$ gives us that we may upper bound the second term in (\ref{eq:kl_ub_all}) as follows:
\begin{equation}
  \label{eq:hlt1_ub}
\sum_{z \in \cup_{h < 1} \MU_{v,h}} \left( \gamma \pp_{v,z} + (1-\gamma) q_z\right) \cdot \log\left( 1 + \gamma \left( \frac{\pp_{v,z}}{q_z} - 1 \right) \right) \leq -\gamma + 2C\gamma\alpha^2 n + 2\gamma^2.
\end{equation}
Combining (\ref{eq:kl_ub_all}), (\ref{eq:exp_h}), (\ref{eq:forgotten_delta}), (\ref{eq:exp_h_ub}), and (\ref{eq:hlt1_ub}), we obtain
\begin{align}
  \frac{1}{\B} \sum_{v \in [\B]} \KL(\PP_{v,\gamma} || \QQ) & \leq \left( \gamma + \gamma^2 (4 + \ep_L \cdot \log(e)) \cdot \alpha^2 n \cdot 4C e^{\ep_L} + \gamma^2 \cdot 32 e^{\ep_L} \delta_L + 2\gamma \delta_L \log \B \right) + \left( -\gamma + 2C\gamma\alpha^2 n + 2\gamma^2 \right) \nonumber\\
  & \leq \gamma^2\alpha^2 n e^{\ep_L} \cdot 4C \cdot (4 + \ep_L \log(e))+ 2C\gamma\alpha^2 n + 36 \gamma^2\nonumber.
\end{align}
The second inequality above uses the facts that $\delta_L \leq \gamma / \log \B$ and $\delta \leq 1/e^{\ep_L}$.

\end{proof}

Finally we prove Theorems \ref{thm:local_lb} and \ref{thm:freq_lb} in the large-sample regime. The proof is a simple application of Lemma \ref{lem:invn}.
\begin{lemma}[Proof of Theorem \ref{thm:local_lb} in large-sample regime; i.e., (\ref{eq:local_small1}) \& (\ref{eq:local_small2})]
  \label{lem:local_small}
  There is a sufficiently small positive constant $c$ so that the following holds. Suppose $n, \B \in \BN$ with $n \geq 1/c$ and $\ep_L, \delta_L \geq 0$ with $0<  \delta_L < c/(n \log n)$, and $0 \leq \ep_L \leq 2 \ln\B - \ln \ln \B - 1/c$. Then there is no protocol for $n$-user frequency estimation on $[\B]$ that satisfies $(\ep_L, \delta_L)$-local differential privacy and $\left( \alpha, 1/2\right)$-accuracy where
  $$
  \alpha \geq \begin{cases}
    \frac{ c\sqrt{\B}}{n\log^{1/4} \B} &\text{for} ~~ n \geq \B^2, \\
    \frac{c \B\ln^{1/7} \B}{n} &\text{for} ~~ n \geq \B^3.
  \end{cases}
  $$
\end{lemma}
\begin{proof}
  We first treat the case that $n \geq \B^2$. Set $n_0 = \B^2$ and $\ep_L = 2 \ln \B - \ln \ln \B - 1/c$. Lemma \ref{lem:intermediate_lb} establishes that there is no $(\ep_L, \delta_L)$-locally \DP\ protocol that satisfies $\left( \frac{c}{\sqrt{n_0\B\sqrt{\log \B}}}, 1/4 \right)$-accuracy, for a sufficiently small constant $c$. But $1/\sqrt{n_0\B\sqrt{\log \B}} = \frac{\sqrt{\B}}{n_0\log^{1/4} \B}$ as $n_0 = \B^2$.

  But by Lemma \ref{lem:invn}, any $(\ep_L, \delta_L)$-locally \DP\ protocol with $ n$ users, $\delta < c/( n \log  n)$ and $\ep_L = 2 \ln \B - \ln \ln \B - 1/c$ which is $\left( \frac{c\sqrt{\B}}{4 n \log^{1/4} \B}, 1/4 \right)$-accurate yields an $(\ep_L, \delta_L)$-locally \DP\ protocol with $n_0$ users which is $\left( \frac{c \sqrt{\B}}{n_0\log^{1/4}\B}, 1/4\right)$-accurate.

  The proof for the case $n \geq \B^3$ is virtually identical, with $n_0 = \B^2$ replaced by $n_1 = \B^3$, for which the lower bound of Lemma \ref{lem:intermediate_lb} states that there is no $(\ep_L, \delta_L)$-locally \DP\ protocol that satisfies $\left( \frac{\ln^{1/7} \B}{n^{2/3}}, 1/4\right)$-accuracy. 
\end{proof}

\begin{lemma}[Proof of Theorem \ref{thm:freq_lb} in large-sample regime; i.e., (\ref{eq:small_lb})]
  \label{lem:freq_lb_small_cor}
  For a sufficiently small positive constant $c$ so that the following holds. If $\log \B > 1/c$, $0 \leq \delta < c/n$, and $0 \leq \ep \leq 1$, then there is no protocol for $n$-user frequency estimation in the single-message shuffled-model that is $(\ep, \delta)$-\DP\ and $\left( \frac{c\sqrt{\B}}{n \log \B}, 1/4\right)$-accurate.
\end{lemma}
\begin{proof}
  Here we cannot use Lemma \ref{lem:pratio} in tandem with Lemma \ref{lem:local_small} since Lemma \ref{lem:local_small} requires a privacy parameter $\ep_L \leq 2 \ln \B$ and the one produced by Lemma \ref{lem:pratio} grows as $\ln n$, which can be arbitrarily large. However, note that it is evident that Lemma \ref{lem:invn} still applies if the protocol $P = (R,A)$ in the lemma statement is replaced by a single-message shuffled-model protocol $P = (R,S,A)$ (and the the protocol $P'$ guaranteed by the lemma is $P' = (R', S, A')$).

  In particular, letting $c$ be the constant of Corollary \ref{cor:intermediate_lb}, for $n_0 = \frac{c\B^2}{\log \B}$, Corollary \ref{cor:intermediate_lb} guarantees that for $\delta \leq 1/n_0, \ep \leq 1$, there is no $(\ep, \delta)$-\DP\ protocol in the single-message shuffled model that is $\left( \frac{c}{n_0^{3/4} \log^{1/4} n_0}, 1/4\right)$-accurate. But $\frac{c}{n_0^{3/4} \log^{1/4} n_0} = \frac{c^{5/4} \sqrt{\B}}{n_0 (\log^{1/4} n_0)(\log^{1/4} \B)}$ by our choice of $n_0$.

  By the modification of Lemma \ref{lem:invn} mentioned in the previous paragraph, there is a sufficiently small constant $c' > 0$ so that for any $n \geq \frac{c\B^2}{\log \B}$, there is no $(\ep, \delta)$-\DP\ protocol for frequency estimation in the $n$-user single-message shuffled model that is $\left( \frac{c'\sqrt{\B}}{n \sqrt{\log \B}}, 1/4 \right)$-accurate.
\end{proof}

\subsection{Proof of Lemma~\ref{lem:lb_tvd}}
\label{sec:lb_tvd}
  In this section we prove Lemma~\ref{lem:lb_tvd}. We first establish a more general statement in Lemma~\ref{lem:tvd_hist} below, which shows that if two distributions $D,F$ have total variation distance bounded away from 1, then two distributions which can be obtained as the histograms of mixtures of i.i.d.~samples from $D,F$ have small total variation distance (much smaller than $\Delta(D,F)$). 
  
We first recall the notation that we use to denote histograms of distributions. Given a tuple of random variables $(Y_1, \ldots, Y_n)$, $\Hist(Y_1, \ldots, Y_n)$ denotes the distribution of the histogram of $(Y_1, \ldots, Y_n)$, i.e., of the function that maps each $z \in \MZ$ to $| \{ i : Y_i = z \}|$. We will denote histograms as functions $h : \MZ \ra \BN$. If $(z_1, \ldots, z_n) \in \MZ^n$ is such that all the $z_i$ are distinct, then its histogram $h = \Hist(z_1, \ldots, z_n)$ is a function $h : \MZ \ra \{0,1\}$. 
\begin{lemma}[Total variation distance between histograms of mixture distributions]
  \label{lem:tvd_hist}
  Suppose $\Pdist,\Qdist$ are distributions on a finite set $\MZ$. Suppose that for $\gamma \leq 1/\sqrt n$ such that $(1- \gamma)n/2$ is an integer, if $Z_1, \ldots, Z_n \sim \Pdist$ and $W_1, \ldots, W_n \sim \Qdist$ are iid, then
%   $\ppi : [n] \ra [n]$ is a random permutation, we have
  \begin{align}
    \label{eq:delta_hist}
    &     \Delta(\Hist(Z_{1}, \ldots, Z_{(1-\gamma)n/2}, W_{(1-\gamma)n/2 + 1}, \ldots, W_{n})), \\
    & \Hist(Z_{1}, \ldots, Z_{(1+\gamma)n/2}, W_{(1+\gamma)n/2 + 1}, \ldots, W_{n})) \geq c.\nonumber
  \end{align}
  Then $\Delta(\Pdist,\Qdist) \geq 1-c'\gamma^2n$ for $c' = \Theta(1/c^2)$.
\end{lemma}
Notice that in the statement of Lemma \ref{lem:tvd_hist} 
\begin{proof}[Proof of Lemma \ref{lem:tvd_hist}]
  We first introduce some notation. Given a set $\MZ$, let $\MH_\MZ$ denote the set of all histograms on $\MZ$; notice that elements $h \in \MH_\MZ$ can be thought of as functions $h : \MZ \ra \BZ_{\geq 0}$. Given distributions $\Pdist,\Qdist$ on a set $\MZ$, as well as positive integers $\nu, n$ with $\nu \leq n$, let $R_{\Pdist,\Qdist}^{\nu, n}$ denote the distribution of the random variable
  $$
\hist(Z_1, \ldots, Z_\nu, W_{\nu + 1}, \ldots, W_n),
$$
where $Z_1, \ldots, Z_\nu \sim \Pdist$ i.i.d., and $W_{\nu + 1}, \ldots, W_n \sim \Qdist$ i.i.d. Thus (\ref{eq:delta_hist}) may equivalently be written as
\begin{equation}
  \label{eq:delta_hist_r}
\Delta\left(R_{\Pdist,\Qdist}^{(1-\gamma)n/2, n}, R_{\Pdist,\Qdist}^{(1+\gamma)n/2, n}\right) \geq c.
\end{equation}
  
  A key tool in the proof of Lemma~\ref{lem:tvd_hist} is the data processing inequality (for total variation distance), stated below for convenience:
  \begin{lemma}[Data processing inequality]
    \label{lem:dpi}
Suppose $\MZ, \MZ'$ are sets, $\Pdist_0, \Pdist_1$ are distributions on $\MZ$ and $f : \MZ \ra \MZ'$ is a randomized function. Suppose $Z_0 \sim \Pdist_0, Z_1 \sim \Pdist_1$. Then $\Delta(f(Z_0), f(Z_1)) \leq \Delta(Z_0, Z_1)$.
\end{lemma}
Using Lemma~\ref{lem:dpi} twice, we will reduce to the case in which $|\MZ| = 3$, in two stages. For the first stage, consider the set $\MZ \sqcup \MZ := \{ (z, b) : z \in \MZ, b \in \{0,1\} \}$, as well as the function $f : \MZ \sqcup \MZ \ra \MZ$, defined by $f((z,b)) = z$. Moreover define distributions $\Pdist', \Qdist'$ on $\MZ \sqcup \MZ$, as follows: for all $z \in \MZ$, we have
\begin{align*}
  \Pdist'((z,0)) &= \Pdist(z) - \min \{ \Pdist(z), \Qdist(z) \}.\\
  \Qdist'((z,0)) &= \Qdist(z) - \min \{ \Pdist(z), \Qdist(z) \}.\\
  \Pdist'((z,1)) &= \min \{ \Pdist(z), \Qdist(z) \}.\\
  \Qdist'((z,1)) &= \min \{ \Pdist(z), \Qdist(z) \}.
\end{align*}
It is immediate from the definition of $\Pdist', \Qdist'$ that $\Delta(\Pdist',\Qdist') = \Delta(\Pdist,\Qdist)$. Moreover, it is clear that if $Z_0, Z_1$ are distributed according to $\Pdist', \Qdist'$, respectively, then $f(Z_0), f(Z_1)$ are distributed according to $\Pdist,\Qdist$, respectively. To describe the effect of applying $f$ to histograms on $\MZ \sqcup \MZ$, we make the following definition:
\begin{defn}[Push-forward histogram]
Consider sets $\MY, \MY'$ together with a (possibly randomized) function $f : \MY \ra \MY'$. For a histogram $h \in \MH_{\MY}$, the {\it push-forward histogram} $f_* h \in \MH_{\MY'}$ (i.e., $f_* h : \MY' \ra \BZ_{\geq 0}$) is defined as follows. If $h$ is expressed as $h = \hist(y_1, \ldots, y_n)$ for $y_1, \ldots, y_n \in \MY$, then $f_*h$ is the (possibly random) histogram given by $\hist(f(y_1), \ldots, f(y_n))$.
\end{defn}
It follows that if $H \in \MH_{\MZ \sqcup \MZ}$ is a random variable distributed according to $R_{\Pdist',\Qdist'}^{\nu, n}$, then the push-forward histogram $f_* H : \MZ \ra \BZ_{\geq 0}$ (which in this case is given by $f_*H(z) = H((z,0)) + H((z,1))$), is distributed according to $R_{\Pdist,\Qdist}^{\nu, n}$. Since $f_*H$ is a (deterministic) function of $H$, it follows from Lemma~\ref{lem:dpi} that
\begin{equation}
  \label{eq:f_reduction}
\Delta\left(R_{\Pdist,\Qdist}^{(1-\gamma)n/2, n}, R_{\Pdist,\Qdist}^{(1+\gamma)n/2, n}\right) \leq \Delta \left( R_{\Pdist',\Qdist'}^{(1-\gamma)n/2, n}, R_{\Pdist',\Qdist'}^{(1+\gamma)n/2, n} \right).
\end{equation}

Next, we define a randomized function $g : \{-1, 0, 1 \} \ra \MZ \sqcup \MZ$, as follows. First make the following definitions:
$$
\rho_0 = \sum_{z \in \MZ} \Pdist'((z,1)) = \sum_{z \in \MZ} \Qdist'((z,1)),  \quad \rho_1 = \sum_{z \in \MZ} \Qdist'((z,0)) = \sum_{z \in \MZ} \Pdist'((z,0)) = 1-\rho_0.
$$
Let $g(0)$ be the distribution over $ \{ (z,1) : z \in \MZ \}$ that assigns to the point $(z,1)$ a mass of $\frac{\Pdist'((z,1))}{\rho_0} = \frac{\Qdist'((z,1))}{\rho_0}$. Let $g(-1)$ be the distribution over $\{ (z,0) : z \in \MZ \}$ that assigns to the point $(z,0)$ a mass of $\frac{\Pdist'((z,0))}{\rho_{1}}$. Finally let $g(1)$ be the distribution over $\{ (z,0) : z \in \MZ \}$ that assigns to the point $(z, 0)$ a mass of $\frac{\Qdist'((z,0))}{\rho_1}$.

Finally define distributions $\Pdist'', \Qdist''$ on $\{ -1, 0 , 1\}$ as follows:
\begin{align*}
  \Pdist''(-1) = \rho_{1}, & \quad \Pdist''(0) = \rho_0, \quad \Pdist''(1) = 0 \\
  \Qdist''(-1) = 0, & \quad \Qdist''(0) = \rho_0, \quad \Qdist''(1) = \rho_1.
\end{align*}
From the \emph{definitions} of $\Pdist'', \Qdist''$ we see that $\Delta(\Pdist'', \Qdist'') = \rho_1 = \Delta(\Pdist,\Qdist)$. Moreover, if $Z_0', Z_1'$ are distributed according to $\Pdist'', \Qdist''$, respectively, then $g(Z_0'), g(Z_1')$ are distributed according to $\Pdist', \Qdist'$. % Next, for a histogram $H' \in \MH_{\{-1,0,1\}}$, define the {\it randomized push-forward histogram} $g_* H : \MZ \sqcup \MZ \ra \BZ_{\geq 0}$ (i.e., $g_*H \in \MH_{\MZ \sqcup \MH}$), as follows:
Next, let $H' \in \MH_{\{-1,0,1\}}$ be the random histogram distributed according to $R_{\Pdist'',\Qdist''}^{\nu, n}$. Then the push-forward histogram $g_*H' \in \MH_{\MZ \sqcup \MZ}$ is distributed according to $R_{\Pdist',\Qdist'}^{\nu, n}$. It follows from Lemma~\ref{lem:dpi} that
\begin{equation}
  \label{eq:g_reduction}
 \Delta \left( R_{\Pdist',\Qdist'}^{(1-\gamma)n/2, n}, R_{\Pdist',\Qdist'}^{(1+\gamma)n/2, n} \right) \leq \Delta \left( R_{\Pdist'', \Qdist''}^{(1-\gamma)n/2, n}, R_{\Pdist'', \Qdist''}^{(1+\gamma)n/2, n} \right).
\end{equation}
Thus, since $\Delta(\Pdist'', \Qdist'') = \Delta(\Pdist', \Qdist') = \Delta(\Pdist,\Qdist)$, and by (\ref{eq:delta_hist_r}), (\ref{eq:f_reduction}), and (\ref{eq:g_reduction}) it suffices to show that there is some constant $c'$ such that, assuming $\Delta \left( R_{\Pdist'', \Qdist''}^{(1-\gamma)n/2, n}, R_{\Pdist'', \Qdist''}^{(1+\gamma)n/2, n} \right) \geq c$, it follows that $1 - \rho_0 =\Delta(\Pdist'', \Qdist'') \geq 1 - c' \gamma^2 n$.

Let random variables $H^{(1-\gamma)n/2}, H^{(1+\gamma)n/2}$ be distributed according to $R_{\Pdist'', \Qdist''}^{(1-\gamma)n/2, n}$ and $R_{\Pdist'', \Qdist''}^{(1+\gamma)n/2, n}$, respectively. Since $H^{(1\pm \gamma)n/2}$ are each histograms of $n$ elements, they are completely determined by their values on $-1$ and $1$. Next note that the distribution of the tuple $(H^{(1-\gamma)n/2}(-1), H^{(1-\gamma)n/2}(1))$ is the distribution of $(Bin((1-\gamma)n/2, 1-\rho_0), Bin((1+\gamma)n/2, 1-\rho_0))$, where the two binomial random variables are independent. Similarly, the distribution of the tuple $(H^{(1+\gamma)n/2}(-1), H^{(1+\gamma)n/2}(1))$ is the distribution of $(Bin((1+\gamma)n/2, 1-\rho_0), Bin((1-\gamma)n/2, 1-\rho_0))$. To upper bound
$$
\Delta ( (H^{(1-\gamma)n/2}(-1), H^{(1-\gamma)n/2}(1)), (H^{(1+\gamma)n/2}(-1), H^{(1+\gamma)n/2}(1)))
$$
it therefore suffices to upper bound $\Delta\left( Bin((1-\gamma)n/2, 1-\rho_0), Bin((1+\gamma)n/2, 1-\rho_0)\right)$. To do this, we use \cite[Theorem 2, Eq. (15)]{roos}, which implies that as long as $\gamma \geq 1/n$ and $\frac{3 \gamma^2 n(1-\rho_0)}{\rho_0} < 1/4$, we have that
\begin{equation}
  \label{eq:delta_use}
\Delta\left( Bin((1-\gamma)n/2, 1-\rho_0), Bin((1+\gamma)n/2, 1-\rho_0)\right) \leq 2\sqrt{e} \cdot \sqrt{\frac{3\gamma^2 n(1-\rho_0)}{\rho_0}}
\end{equation}
If $\frac{3\gamma^2 n(1-\rho_0)}{\rho_0} \geq 1/4$ then we get $\rho_0 \leq 12 \gamma^2 n$. Otherwise, (\ref{eq:delta_use}) together with $\Delta \left( R_{\Pdist'', \Qdist''}^{(1-\gamma)n/2, n}, R_{\Pdist'', \Qdist''}^{(1+\gamma)n/2, n} \right) \geq c$ gives us that $\frac c2 \leq 2\sqrt{e} \cdot \sqrt{3\gamma^2 n/\rho_0}$, meaning that $\rho_0 \leq O(\gamma^2 n / c^2)$. In particular, for some constant $c' = \Theta(1/c^2)$, we have shown that $1-\rho_0 \geq 1-c'\gamma^2 n$. 
\end{proof}

Now we are ready to prove Lemma~\ref{lem:lb_tvd} using Lemma~\ref{lem:tvd_hist}.
\begin{proof}[Proof of Lemma~\ref{lem:lb_tvd}]
By increasing $\alpha$ by at most a constant factor we may ensure that $(1 \pm \alpha)n/2$ are integers. We define two distributions of inputs, $D_1$ and $D_2$. For $D_1$, exactly $(1-\alpha)n/2$ of the $x_i$ are set to $v$ and the remaining $(1+\alpha)n/2$ of the $x_i$ are drawn uniformly from $[\B]$. For $D_2$, exactly $(1+\alpha)n/2$ of the $x_i$ are set to $v$ and the remaining $(1-\alpha)n/2$ of the $x_i$ are drawn uniformly from $[\B]$. Under both $D_1$ and $D_2$, the subset of which users are chosen to have their $x_i$ fixed to $v$ is chosen uniformly at random.

  By the Chernoff bound, for $\lambda \geq 0$, as long as $\alpha \leq 1$, we have that
  \begin{eqnarray}
    \label{eq:chernoff_hist}
    \P_{(x_1, \ldots, x_n) \sim D_1} \left[ \sum_{i=1}^n (e_{x_i})_v \geq (1-\alpha)n/2 + (1+\lambda) (1+\alpha)n/(2\B) \right] &\leq&  \exp(-n\min\{\lambda,\lambda^2\} / (3\B)).\\
    \label{eq:chernoff_hist-2}
    \P_{(x_1, \ldots, x_n) \sim D_2} \left[ \sum_{i=1}^n (e_{x_i})_v \leq (1+\alpha) n/2 + (1-\lambda)(1-\alpha) \cdot n/(2\B) \right] & \leq & \exp(-n\min\{\lambda,\lambda^2\} / (3\B)).
  \end{eqnarray}
%   For a histogram $h : [\B] \ra \BZ_{\geq 0}$ with $\sum_{v \in [\B]} h(v) = n$, define $G(h) \in [\B]^n$ to be the random variable obtained by permuting uniformly at random the multiset consisting of $h(v)$ copies of $v$ for each $v \in [\B]$. Note that for $b \in \{1,2\}$, the distribution of $G(H_b)$ is exactly the distribution of $(x_1, \ldots, x_n) \sim D_b$. 
%, we obtain that % outputs a $\frac 13 \cdot (\alpha - \lambda/\B)$-additive approximation to

  Choose $\lambda$ so that $\exp(-n\min\{\lambda,\lambda^2\}/(3\B)) = (1-2\beta)/4$, i.e., $\min\{\lambda,\lambda^2\}/\B = 3/n \cdot \log(4/(1-2\beta))$. Explicitly, we have:
  $$
\lambda =
\max \left\{
  \frac{3 \B \log(4/(1-2\beta))}{n}, %\quad : \quad \B > \frac{n}{12 \log(4/(1-2\beta))} \\
  \sqrt{\frac{3 \B \log(4/(1-2\beta))}{n}} % \quad : \quad \B \leq \frac{n}{12 \log(4/(1-2\beta))}.
  \right\}.% \end{cases}
  $$

  We now consider two possible cases for $\lambda$:
  
{\bf Case 1.}  In the case that $\lambda = 3\B \log(4/(1-2\beta))/n$, we have $\lambda \geq 1$. Moreover, since $2\lambda /\B = 6 \log(4/(1-2\beta))/n \leq \alpha/2$ (by (\ref{eq:constrain_alpha})), the fact that $R$ is $(\alpha/6, \beta)$-accurate implies that $R$ is $(\frac 13 \cdot (\alpha - 2\lambda /\B), \beta)$-accurate. In turn it follows that $R$ is $(\frac 13 \cdot (\alpha - (\lambda + \alpha)/\B), \beta)$-accurate.

{\bf Case 2.} In the case that $\lambda = \sqrt{3\B \log(4/(1-2\beta))/n}$, we have $\lambda \leq 1$. Since $\lambda / \B = \sqrt{3 \log(4/(1-2\beta))/(n\B)} \leq \alpha/4 \leq \alpha/2 - \alpha/\B$ (by (\ref{eq:constrain_alpha})), the fact that $R$ is $(\alpha/6, \beta)$-accurate implies that $R$ is $(\frac 13 \cdot (\alpha - (\lambda + \alpha)/\B), \beta)$-accurate.

$\newline$
In both cases, it follows that, by Definition~\ref{def:accuracy}, there exists some analyzer $A : \MZ^n \ra [0,1]$ so that, for any dataset $X = (x_1, \ldots, x_n)$, with probability $1-\beta$ over the local randomizers, $\left|A(R(x_1), \ldots, R(x_n)) - \frac 1n \sum_{i=1}^n (e_{x_i})_v \right| \leq  \frac 13 \cdot (\alpha - (\la+\alpha)/\B)$. (Here we are only using accuracy on the $v$th coordinate.) Define $f : [0,1] \ra \{0,1\}$ by $f(x) = 1$ if $x \geq \frac 12 + \frac{(1+ \la \alpha)n}{2\B}$ and $f(x) = 0$ otherwise. Using (\ref{eq:chernoff_hist}) and (\ref{eq:chernoff_hist-2}) it follows that
\begin{equation}
  \label{eq:d12_tv}
\E_{(x_1, \ldots, x_n) \sim D_2} \left[ f(A(R(x_1), \ldots, R(x_n))) \right] - \E_{(x_1, \ldots, x_n) \sim D_1} \left[ f(A(R(x_1), \ldots, R(x_n))) \right] \geq 1-2\beta - 2\exp(-n\lambda / (12\B)).
\end{equation}
For $b \in \{1,2\}$, let $H_b : \MZ \ra \BZ_{\geq 0}$ denote the random variable that is the histogram of the $R(x_i)$ when $(x_1, \ldots, x_n)$ are drawn according to $D_b$, and let $\tilde D_b$ be the distribution of $H_b$.  For a histogram $h : \MZ \ra \BZ_{\geq 0}$ with $\sum_{z \in \MZ} h(z) = n$, define $G(h) \in \MZ^n$ to be the random variable obtained by permuting uniformly at random the multiset consisting of $h(z)$ copies of $z$ for each $z \in \MZ$. Note that for $b \in \{1,2\}$, the distribution of $G(H_b)$ is exactly the distribution of $(R(x_1), \ldots, R(x_n))$ when $(x_1, \ldots, x_n) \sim D_b$. It follows from (\ref{eq:d12_tv}) and our choice of $\la$ that
$$
\E_{H_2 \sim \tilde D_2} \left[ f(A(G(H_2))) \right] - \E_{H_1\sim \tilde D_1} \left[ f(A(G(H_1))) \right] \geq \frac{1-2\beta }{2}.
$$
  Hence $\Delta(\tilde D_1, \tilde D_2) \geq 1-2\beta - 2 \exp(-n\lambda/(12\B))$. % This follows because by (\ref{eq:chernoff_hist}), $\Aprot$ must output different answers on $H_1$ and $H_2$ with probability at least $\beta - 2\exp(-n\lambda/(12\B))$. \todo{verify this}

  By definition, $\tilde D_1$ is the distribution of $\hist(Z_1, \ldots, Z_{n(1-\alpha)/2}, W_{n(1-\alpha)/2+1}, \ldots, W_n)$ and $\tilde D_2$ is the distribution of $\hist(Z_1, \ldots, Z_{n(1+\alpha)/2}, W_{n(1+\alpha)/2+1}, \ldots, W_n$ when $Z_1, \ldots, Z_n \sim \PP_v$ iid, and $W_1, \ldots, W_n \sim \QQ$ iid. 
By Lemma~\ref{lem:tvd_hist} with $\Pdist = \PP_v, \Qdist = \QQ$, we must have that $\Delta(\PP_v, \QQ) \geq 1 - c \alpha^2 n$ for some $c = \Theta\left( \frac{1}{(1-2\beta)^2} \right)$. 
\end{proof}

\subsection{Lower Bounds for Single-Message Selection}
\label{sec:selection_lb}
In this section we prove Theorem~\ref{th:selection_single_message}, stated formally below:
\begin{theorem}[Nearly tight lower bound for single-message shuffled model selection]
  \label{thm:selection_formal}
Suppose $n, \B \in \BN$ and $\ep \leq O(1)$ and $\delta \leq o(1/(n\B))$. Any single-message shuffled model protocol that is $(\ep, \delta)$-\DP and solves the selection problem with $n$ users and with probability at least $4/5$ must have $n \geq \Omega \left(\B \right)$. 
\end{theorem}
As we did for frequency estimation in the single-message shuffled model, we will prove Theorem~\ref{thm:selection_formal} by appealing to Lemma~\ref{lem:pratio} and proving an analogous statement for $(\ep_L, \delta_L)$-\DP\ {\it local-model} protocols where the privacy parameter $\ep_L$ is approximately $\ln n$. The proof is similar in structure to that of \cite{ullman2018tight}, established a lower bound on $n$ which is tight in the case that $\ep_L = O(1)$ and $\delta_L = 0$.
% The technique to prove Theorem~\ref{thm:selection_formal} is somewhat similar to that in Section~\ref{sec:const_lb} used to prove Corollary~\ref{cor:freq_lb_const}.

We begin by defining a distribution under which we shall show selection to be hard. % in the shuffled model.
\begin{defn}[Distributions $\DD_{ \eell, j}$]
  Fix $\B \in \BN$. For $\ell \in \{0,1\}, j \in [\B]$, let the distribution $\DD_{\ell, j}$ be the uniform distribution on the subcube $\{ x \in \BB^\B : x_j = \ell \}$.

  Let $\bar \DD$ denote the joint distribution of $(L, J, X)$, where $(L, J) \sim \Unif(\BB \times [\B])$ and conditioned on $L,J$, $X \sim \DD_{L, J}$. 
  
% Fix $\B \in \BN$. Let $U_\B$ denote the uniform distribution over $\BB^\B$. For each $\eell \in \{ -1, 1\}$, $j \in [\B]$, $\alpha \in [0,1]$, define the distribution $\DD_{\alpha, \eell, j}$ by
 %  $$
% \DD_{\alpha, \eell, j} := \alpha \cdot (U_\B | x_j = \eell) + (1-\alpha) \cdot U_\B.
% $$
% Also let $\DD_{\eell,j} = \DD_{1, \eell, j}$. %, and $\bar \DD = \frac{1}{2^{\B+1}} \sum_{b \in \{-1,1\}, j \in [\B]} \DD_{b,j}$.
% If $\EELL, J$ are random variables distributed according to some distribution $F$ on $\{ -1, 1\} \times [\B]$, then define the following mixture distribution:
% $$
% \DD_{\EELL, J} := \sum_{\eell \in \{-1,1\}, j \in [\B]} \P_{F}[\EELL=\eell, J=j] \cdot \DD_{\eell, j}.
% $$
\end{defn}
% Notice, for instance, that if $\EELL, J$ are uniform and independent over $\{ 0,1 \}$ and $[\B]$, respectively, then $\DD_{\EELL, J}$ is the uniform distribution on $\BB^\B$. %  = U_\B$ for any $\alpha \in [0,1]$.

\begin{defn}[Distributions $\PP_{\eell, j}$]
  \label{def:selection_PP}
  Next suppose that for some finite set $\MZ$, $\Rprot : \BB^\B \ra \MZ$ is a fixed local randomizer. Let $\PP_{\eell, j}$ be the distribution of $\Rprot(X)$ when $X \sim \DD_{\eell, j}$. % Finally let $\bar M = \frac{1}{2^{\B+1}} \sum_{b \in \{-1,1\}, j \in [\B]} M_{\alpha, b, j}$ be the distribution of $\Rprot(X)$ when $X \sim U_\B$.
%   Let $\PP_{\eell,j} = \PP_{1,\eell,j}$. If $\EELL, J$ are distributed according to some distribution $\MD$ on $\{ -1, 1\} \times [\B]$, then define the following mixture distribution:
% $$
% \PP_{\alpha, \EELL, J} := \sum_{\eell \in \{-1,1\}, j \in [\B]} \P_{\MD}[\EELL=\eell, J=j] \cdot \PP_{\alpha, \eell, j}.
% $$
% When $\EELL \in \{-1,1\}, J \in [\B]$ are independent variables uniformly distributed over their respective domains, we will often use $\bar \PP$ as a shorthand to denote $\PP_{\EELL,J}$.
  Let $Q$ be the distribution $\frac{1}{2\B} \sum_{j \in [\B], \ell \in \{0,1\}} \PP_{\ell, j}$. Note that $Q$ is the distribution of $R(X)$ when $X \sim \Unif(\BB^\B)$.
\end{defn}

% Theorem~\ref{thm:selection_1m_lb} is a straightforward consequence of
To prove Theorem~\ref{thm:selection_formal}, we first establish Lemma~\ref{lem:selection_amax} below, which applies to any protocol that is differentially private in the {\it local model} of differential privacy (with a large privacy parameter $\ep_L$):
\begin{lemma}
  \label{lem:selection_amax}
For a sufficiently small positive constant $c$, the following holds. Suppose $R : \BB^\B \ra \MZ$ is an $(\ep_L, \delta_L)$-(locally) differentially private protocol with $\delta_L \leq \frac{c}{n(\B + \exp(\ep_L))}$. Moreover suppose that $\Aprot : \MZ^n \ra \{ 0, 1 \} \times [\B]$ is a function so that, if $\EELL \sim \{ 0, 1\}, J \sim [\B]$ are uniform and independent, then
\begin{equation}
  \label{eq:1/3contradict}
\P_{\EELL, J, X_1, \ldots, X_n \sim (\DD_{\EELL,J} | \EELL, J)} \left[ \Aprot(\Rprot(X_1), \ldots, \Rprot(X_n)) = (\EELL,J) \right] \geq \frac 13.
\end{equation}
Then
$$
n \geq \frac{c \B \log \B}{1+ \ep_L}.
$$
\end{lemma}

Theorem~\ref{thm:selection_formal} is a straightforward consequence of Lemma~\ref{lem:selection_amax} and Lemma~\ref{lem:pratio} (see \cite{ullman2018tight}). We provide the proof for completeness.
\begin{proof}[Proof of Theorem~\ref{thm:selection_formal}]
Let $c_0 \in (0,1)$ be a sufficiently small positive constant to be specified later.  Suppose for the purpose of contradiction that $P_S = (R,S,A)$ is an $(\ep, \delta)$-\DP\ single-message shuffled model protocol that solves the selection problem with $n < c_0 \B$ users and probability at least $4/5$. By Lemma \ref{lem:pratio}, $P_L := (R,A)$ is an $(\ep + \ln n, \delta)$-locally \DP protocol that solves the selection problem with $n$ users and probability at least $4/5$.

It follows by a Chernoff bound and a union bound that if $J \sim \Unif([\B])$ and $X_1, \ldots, X_n \sim \DD_{1,J} | J$, then
\begin{equation}
  \label{eq:jj}
\p_{J, X_1, \ldots, X_n \sim  \DD_{1,J} | J}\left[A(R(X_1), \ldots, R(X_n)) = J\right] \geq 3/4
\end{equation}
as long as $n \geq \Omega(\sqrt{\log \B})$. (In particular, we can guarantee that with probability at least $1-1/20$, for all $j' \neq J$, $\sum_{i=1}^n (X_i)_{j'} < \sum_{i=1}^n (X_i)_J - n/10 = 9n/10$.)

It follows from (\ref{eq:jj}) that if $L \sim \Unif(\BB)$ is independent of $J$,
\begin{equation}
  \p_{J, X_1, \ldots, X_n \sim  \DD_{L,J} | J}\left[A(R(X_1), \ldots, R(X_n)) = (J,L)\right] \geq 3/8\nonumber.
\end{equation}
The above equation is a contradiction to Lemma \ref{lem:selection_amax} in light of the fact that $R$ is $(\ep + \ln n, \delta)$-\DP, $n <c_0 \B <  \frac{c \B\log \B}{1 + \ep + \ln n}$, and
$$
\delta < \frac{c}{n (\B + \exp(\ep + \ln(n)))}.
$$
(The above bound on $\delta$ can be seen by noting that $\delta < c_0 / (nB)$ by assumption and $\exp(\ep + \ln(n)) = O(n) \leq O(\B)$.)

\end{proof}

The bulk of the proof of Theorem~\ref{thm:selection_formal} is to establish an upper bound on $I((L,J); R(X))$, when $(L,J,X) \sim \bar \DD$. % since the structure of the input space $\MX = \{0,1\}^\B$ for selection is more complex than the structure of the input space $\MX = [\B]$ under consideration for frequency estimation. In particular, in the case of selection, for each possible answer $j \in [\B]$, there are $2^{\B - 1}$ possible input values $x \in \{0,1\}^\B$ which contribute to $j$. This structure will ultimately lead us to consider the level-1 Fourier coefficients $\hat f(\{ j\})$ for certain functions $f : \{0,1\}^n \ra \BR$, and a key ingredient in our proof will be the Level-1 inequality from the analysis of boolean functions \cite{odonnell2014analysis}.
Lemma~\ref{lem:izbj} below provides this upper bound.
\begin{lemma}
  \label{lem:izbj}
 Suppose $\ep_L \geq 0$, $\delta_L \in (0,1)$, and $R$ is $(\ep_L, \delta_L)$-\DP. Then
    we have that
    $$
\E_{\eell \sim \EELL, j \sim J} \left[ \KL(\PP_{\eell,j} || \QQ)\right] \leq O \left( \frac{1 + \ep_L}{\B} + \delta \cdot(\B + e^{\ep_L}) \right),
$$
where $\PP_{\ell, j}, \QQ$ are as defined in Definition~\ref{def:selection_PP}.
\end{lemma}
The proof of Lemma~\ref{lem:selection_amax} from Lemma~\ref{lem:izbj} is entirely standard \cite{ullman2018tight}. We provide a proof for completeness.
\begin{proof}[Proof of Lemma~\ref{lem:selection_amax}]
  Suppose $L,J$ are drawn uniformly from $\BB \times [\B]$, and then $X_1, \ldots, X_n \sim \DD_{L,J}$ are drawn i.i.d. Let $Z_1 = R(X_1), \ldots, Z_n = R(x_n)$ denote the resulting random variables after passing $X_1, \ldots, X_n$ through the local randomizer $R$. (In particular, $Z_1, \ldots, Z_n$ are drawn i.i.d.~according to $\PP_{L,J}$.) By Fano's inequality, for any deterministic function $f : \MZ^n \ra \BB \times [\B]$, we have that
  \begin{align}
&     \p_{L, J, Z_1, \ldots, Z_n} \left[ f(Z_1, \ldots, Z_n) = (L,J)\right] \nonumber\\
    & \leq \frac{1 + I((Z_1, \ldots, Z_n); (L,J))}{\log 2\B}\nonumber \\
    & \leq \frac{1 + n \cdot I(Z_1; (L, J))}{\log 2\B} \nonumber\\
    & = \frac{1 + n \cdot \KL((Z_1, L, J) || Z_1 \otimes (L, J))}{\log 2\B} \nonumber\\
    & =\frac{1 + n \cdot \E_{(\ell, j) \sim \Unif(\BB \times [\B])}[ \KL(\PP_{\ell, j} || \QQ)]}{\log 2\B} \nonumber\\
    \label{eq:use_channel}
    & \leq \frac{1+ C \cdot \left( \frac{n(1+\ep_L)}{\B} + \delta n \cdot(\B + e^{\ep_L})\right)}{\log 2\B},
  \end{align}
  where (\ref{eq:use_channel}) uses Lemma~\ref{lem:izbj} and $C$ is a sufficiently large constant. If $n < \frac{c \B \log \B}{1 + \ep_L}$, then using the assumption on $\delta$ we may bound (\ref{eq:use_channel}) above by
  $$
\frac{1 + C \cdot c \log \B + c}{\log \B},
$$
which is strictly less than $1/3$ for a sufficiently small constant $c$, thus contradicting (\ref{eq:1/3contradict}).
\end{proof}

Finally we prove Lemma~\ref{lem:izbj}. %
\begin{proof}[Proof of Lemma~\ref{lem:izbj}]
Recall the notation of Definition~\ref{def:pp}: For $x \in \BB^d$ and $z \in \MZ$, we have $\pp_{x,z} = \P_\Rprot[\Rprot(x) = z]$, and for $\MS \subset \MZ$, $\pp_{x,\MS} = \P_R[R(x) \in \MS]$. Also set $\q_z = \frac{1}{2^\B} \sum_{x \in \BB^\B} \pp_{x,z} = \P_{X \sim U_\B, \Rprot} [\Rprot(X) = z] = \pp_{Z \sim \QQ} [Z = z]$ and $q_\MS = \sum_{z \in \MS} q_z$ for $\MS \subset \MZ$. Notice that
    \begin{align}
      &  \E_{\ell \sim \Unif(\{0,1\}), j \sim \Unif([\B])} \left[
      \KL(\PP_{\eell,j} || \QQ) \right]\nonumber\\
      &=  \E_{\eell,j} \left[
        \sum_{z \in \MZ} \P_{Z \sim \PP_{\eell,j}} [ Z = z] \cdot \log \left( \frac{\P_{Z \sim \PP_{\eell,j}}[Z=z]}{\P_{Z \sim \QQ}[Z=z]}\right)  \right]
      \nonumber\\
      &= \frac{1}{2\B} \sum_{\ell \in \{0,1\}, j \in [\B]} \sum_{z \in \MZ} \left( \frac{1}{2^{\B-1}} \sum_{x \in \BB^\B : x_j = \eell} \pp_{x,z}\right) \cdot \log \left( \frac{\frac{1}{2^{\B-1}} \sum_{x \in \BB^\B : x_j = \eell} \pp_{x,z}}{q_z}\right)\nonumber\\
      &= \sum_{z \in \MZ} \frac{1}{2^\B} \sum_{y \in \BB^\B} \pp_{y,z} \cdot \frac{1}{\B} \sum_{j \in [\B]} \log \left( \frac{\frac{1}{2^{\B-1}} \sum_{x : x_j = y_j} \pp_{x,z}}{q_z} \right)\nonumber\\
      \label{eq:kl_ub}
      &=  \sum_{z \in \MZ} q_z \cdot \frac{1}{\B} \sum_{j \in [\B]}\left(\frac{1}{2^\B} \sum_{y \in \BB^\B} \frac{\pp_{y,z}}{q_z} \cdot \log \left( \frac{1}{2^{\B-1}} \sum_{x : x_j = y_j} \frac{\pp_{x,z}}{q_z}\right) \right).
    \end{align}
        %         For each $j \in [\B], \eell \in \{ \pm 1 \}$, let $U_{\B,\eell,j}$ be the uniform distribution over the subcube $\{ x \in \BB^\B : x_j = \eell \}$.
    For each $z \in \MZ$, define a function $f_z : \BB^\B \ra \BR_{\geq 0}$ by $f_z(x) := \frac{\pp_{x,z}}{q_z}$. Thus, for any $j \in [\B]$, $\E_{x \sim \Unif(\BB^\B)}[f_z(x)] = \frac{\E_{x \sim \DD_{0,j}}[f_z(x)] + \E_{x \sim \DD_{1,j}}[f_z(x)]}{2} = 1$. We may now upper bound (\ref{eq:kl_ub}) by
    \begin{equation}
      \label{eq:fz}
  \E_{\eell \sim \Unif(\{ 0,1 \}), j \sim \Unif([\B])} \left[
      \KL(\PP_{\eell,j} || \QQ) \right] \leq \sum_{z \in \MZ} q_z \cdot \frac{1}{2\B} \sum_{j, \eell} \E_{y \sim \DD_{\eell,j}} [f_z(y)] \cdot \log \left( \E_{y \sim \DD_{\eell,j}} [f_z(y)] \right).
  \end{equation}
  For each $z \in \MZ$, $x \in \BB^\B$, set
  $$
  g_z(x) = \begin{cases}
    f_z(x) &: f_z(x) \leq 2e^{\ep_L} \\
    0 &: f_z(x) > 2e^{\ep_L},
  \end{cases}
  $$
  and $h_z(x) := f_z(x) - g_z(x)$.
\if 0
  For non-negative real numbers $a, b$ we have, by convexity of the function $x \mapsto x \log x$, that
  $$
  (a+b) \log(a+b) \leq % 2\max\{a,b\} \log(2\max\{a,b\}) \leq
  \frac 12 \cdot \left(2a \log(2a) + 2b \log (2b)\right).
$$
Applying this inequality with $a = \E_{y \sim \DD_{\ell, j}}[g_z(y)], b = \E_{y \sim \DD_{\ell, j}}[h_z(y)]$ for each $\ell \in \BB, j \in [\B]$, we obtain from (\ref{eq:fz}) that
\begin{align}
  & \E_{\eell \sim \Unif(\{ 0,1 \}), j \sim \Unif([\B])} \left[
    \KL(\PP_{\eell,j} || \QQ) \right] \nonumber\\
 %  & \leq \sum_{z \in \MZ} q_z \cdot \frac{1}{\B} \sum_{j, \ell}  \left(\E_{y \sim \DD_{\eell,j}} [g_z(y)] \cdot \log \left( \E_{y \sim \DD_{\eell,j}} [2g_z(y)] \right) + \E_{y \sim \DD_{\eell,j}} [h_z(y)] \cdot \log \left( \E_{y \sim \DD_{\eell,j}} [2h_z(y)] \right)\right)\nonumber\\
  \label{eq:fz2}
    & \leq \sum_{z \in \MZ} q_z \cdot \frac{1}{\B} \sum_{j, \ell}  \left(\E_{y \sim \DD_{\eell,j}} [2g_z(y)] \cdot \log \left( \E_{y \sim \DD_{\eell,j}} [2g_z(y)] \right) + \E_{y \sim \DD_{\eell,j}} [2h_z(y)] \cdot \log \left( \E_{y \sim \DD_{\eell,j}} [2h_z(y)] \right)\right).
\end{align}
\fi
  
We next note the following basic fact.
\if 0
\begin{fact}
  \label{fac:pinsker}
  Suppose $a,b \geq 0, \frac{a+b}{2} \leq 1$. Then
  $$
a \log a + b \log b \leq \frac 12 \cdot (a-b)^2.
  $$
% Fix some $c \in [0,1]$.  The value of the solution to the optimization problem:
%   \begin{align*}
%     & \max \ \ a \log a + b \log b  \\
%     \mbox{s.t. } & a, b, \geq 0, \ \ \frac{a+b}{2} = 1, \ \ \ a - b = c
%  \end{align*}
%   is given by $O(c^2)$. 
\end{fact}
\fi
\begin{fact}
  \label{fac:gen_pinsker}
  Suppose $g_0, g_1, h_0, h_1 \geq 0$ are real numbers such that $\frac{g_0 + g_1 + h_0 + h_1}{2} = 1$. Then
  $$
g_0 \log(g_0 + h_0) + g_1 \log(g_1 + h_1) \leq \frac 12 (g_1 - g_0)^2 + \frac 12 (g_1 - g_0)(h_1 - h_0).
  $$
\end{fact}
\begin{proof}[Proof of Fact~\ref{fac:gen_pinsker}]
  Let $c \in [0,1]$ be such that $g_0 + h_0 = 1-c$ and $g_1 + h_1 = 1+c$. Then using the fact that $\log(1+x) \leq x$ for all $x \geq -1$,
  \begin{align*}
    & g_0 \log(g_0 + h_0) + g_1 \log(g_1 + h_1) \\
    & = g_0 \log(1-c) + g_1 \log(1+c) \\
    & \leq -g_0 c + g_1 c \\
    & = \frac{(g_1 - g_0) +(h_1 - h_0)}{2} \cdot (g_1 - g_0),
  \end{align*}
  which leads to the desired claim.
\end{proof}

Recall that for a boolean function $f : \BB^\B \ra \BR$ we have $\hat f( \{ j \}) = \frac{1}{2} \left( \E_{x \sim \DD_{j,0}} [f(x)] - \E_{x \sim \DD_{j,1}}[f(x)] \right)$ for each $j \in [\B]$. Using Fact~\ref{fac:gen_pinsker} in (\ref{eq:fz}) with $g_0 = \E_{x \sim\DD_{0, j}}[g_z(x)]$, $g_1 = \E_{x \sim \DD_{1,j}}[g_z(x)]$, $h_0 = \E_{x \sim \DD_{0,j}}[h_z(x)]$, and $h_1 = \E_{x \sim \DD_{1,j}}[h_z(x)]$ for each $z \in \MZ$, $j \in [\B]$, we obtain
% (\ref{eq:fz}), and Fact~\ref{fac:pinsker}, we get that
\begin{align}
  & \E_{\eell \sim\Unif(\BB), j \sim \Unif([\B])} \left[
    \KL(\PP_{\eell,j} || \QQ) \right]\nonumber\\
  & \leq \sum_{z \in \MZ} q_z \cdot \frac{1}{\B} \sum_{j \in [\B]} \hat g_z(\{ j \})^2+ \hat g_z( \{ j \}) \hat h_z(\{ j\})+ \sum_{z \in \MZ} q_z \cdot \frac{1}{2\B} \sum_{j \in [\B], \ell \in \BB} \E_{x \sim \DD_{\ell, j}}[h_z(x)] \cdot \log \left( \E_{x \sim \DD_{\ell, j}}[f_z(x)] \right) \nonumber\\
  \label{eq:logub}
  & \leq \sum_{z \in \MZ} \frac{q_z}{\B} \bW^1[g_z] + \sum_{z \in \MZ} \frac{q_z}{\B} \sum_{j \in [\B]} \hat g_z(\{ j \}) \hat h_z( \{j \}) + \sum_{z \in \MZ} \frac{q_z}{2} \sum_{j,\ell} \E_{x \sim \DD_{\ell, j}}[h_z(x)]\\
  \label{eq:3terms}
  & = \sum_{z \in \MZ} \frac{q_z}{\B} \bW^1[g_z] + \sum_{z \in \MZ} \frac{q_z}{\B} \sum_{j \in [\B]} \hat g_z(\{ j \}) \hat h_z( \{j \}) + \sum_{z \in \MZ} \frac{\B q_z}{2} \cdot \E_{x \sim \Unif(\BB^\B)}[h_z(x)].
%   & \leq \sum_{z \in \MZ} q_z \cdot \frac{2}{\B} \sum_{j \in [\B]} \hat g_z(\{ j\})^2 + \hat h_z(\{ j\})^2 \nonumber\\
%   & = \sum_{z \in \MZ} q_z \cdot \frac{2}{\B} \cdot \left(\bW^1[g_z]  + \bW^1[h_z] \right).\nonumber
\end{align}
where the (\ref{eq:logub}) uses the fact that $f_z(x) = \frac{\pp_{x,z}}{q_z} \leq 2^\B$ for any $x \in \BB^\B, z \in \MZ$.

Next, notice that for an arbitrary non-negative-valued boolean function $f : \BB^\B \ra \BR_{\geq 0}$, for and $j \in [\B]$ we have $\hat f ( \{ j\}) \leq \E_{x \sim \Unif(\BB^\B)}[f(x)]$. Also using the fact that $g_z(x) \leq 2e^{\ep_L}$ for each $z \in \MZ, x \in \BB^\B$, we see that
\begin{equation}
  \label{eq:innprod}
\sum_{z \in \MZ} \frac{q_z}{ \B} \sum_{j \in [\B]} \hat g_z(\{ j\}) \hat h_z(\{ j\}) \leq \sum_{z \in \MZ} 2e^{\ep_L} q_z \cdot \E_{x \sim \Unif(\BB^\B)}[h_z(x)].
\end{equation}
Next we derive an upper bound on $\sum_{z \in \MZ} q_z\cdot \E_{x \sim \Unif(\BB^\B)}[h_z(x)]$. Here we will use the $(\ep_L, \delta_L)$-differential privacy of $R$; intuitively, the differential privacy of $R$ constrains the ratio $\pp_{x,z} / q_z$ to be small for most $x \in \BB^\B, z \in \MZ$, except with probability $\delta_L$. % indeed, we must have, for any $x,x \in \BB^\B$, and any $z \in $\pp_{x,z}

For each $x \in \BB^\B$, set $\MT_x := \left\{ z \in \MZ : \frac{\pp_{x,z}}{q_z} > 2e^{\ep_L} \right\}$. Note that $h_z(x) > 0$ if and only if $z \in \MT_x$. Then
\begin{align}
  & \sum_{z \in \MZ} q_z \cdot \E_{x \sim \Unif(\BB^\B)} [h_z(x)]\nonumber \\
  & = \sum_{z \in \MZ} q_z \cdot \E_{x \sim \Unif(\BB^\B)} \left[ \One[z \in \MT_x] \cdot h_z(x) \right]\nonumber \\
  & = \E_{x \sim \Unif(\BB^\B)} \left[ \sum_{z \in \MT_x} q_z \cdot h_z(x) \right] \nonumber\\
  & = \E_{x \sim \Unif(\BB^\B)} \left[ \sum_{z \in \MT_x} \pp_{x,z} \right]\nonumber\\
  \label{eq:mtx}
  & = \E_{x \sim \Unif(\BB^\B)} \left[ \pp_{x, \MT_x} \right],
\end{align}
where we have used that for $z \in \MT_x$, $h_z(x) = f_z(x) = \pp_{x,z} / q_z$. But since $R$ is $(\ep_L, \delta_L)$-\DP, we have that $\pp_{x,\MT_x} \leq e^{\ep_L} \cdot \pp_{y, \MT_x} + \delta_L$ for any $x,y \in \BB^\B$. Averaging over all $y$, we obtain $\pp_{x,\MT_x} \leq e^{\ep_L} \cdot q_{\MT_x} + \delta_L \leq e^{\ep_L} \cdot \frac{\pp_{x,\MT_x}}{2e^{\ep_L}} + \delta_L$, so $\pp_{x, \MT_x} \leq 2\delta_L$. Since this holds for all $x$, it follows by (\ref{eq:3terms}), (\ref{eq:innprod}) and (\ref{eq:mtx}) that
\begin{align}
 \E_{\eell \sim\Unif(\BB), j \sim \Unif([\B])} \left[
  \KL(\PP_{\eell,j} || \QQ) \right] & \leq \left(\sum_{z \in \MZ} \frac{q_z}{\B} \bW^1[g_z]\right) + \left( 2e^{\ep_L} + \B/2 \right) \cdot \sum_{z \in \MZ} q_z \cdot \E_{x \sim \Unif(\BB^\B)} [h_z(x)] \nonumber\\
    \label{eq:pre_level1}
  & \leq \left(\sum_{z \in \MZ} \frac{q_z}{\B} \bW^1[g_z]\right) + \left( 2e^{\ep_L} + \B/2\right) \cdot 2\delta_L.
\end{align}
By definition of $g_z$ we have that $\hat g_z(\varnothing) =  \E_{x \sim \Unif(\BB^\B)}[g_z(x)] \leq \E_{x \sim \Unif(\BB^\B)}[f_z(x)] = 1$. For each $z \in \MZ$, define a function $g_z' : \BB^\B \ra \BR_{\geq 0}$, by $g_z'(x) = g_z(x) + \left( 1 - \hat g_z(\varnothing) \right)$. Certainly $\bW^1[g_z] = \bW^1[g_z']$, $0 \leq g_z'(x) \leq 1 + 2e^{\ep_L}$ for all $x$, and $\E_{x \sim \Unif(\BB^\B)}[g_z'(x)] = 1$. Now we apply the level-1 inequality, stated below for convenience.
\begin{theorem}[Level-1 Inequality, \cite{odonnell2014analysis}, Section 5.4]
  \label{thm:level1}
  Suppose $f : \BB^\B \ra \BR_{\geq 0}$ is a non-negative-valued boolean function with $0 \leq f(x) \leq L$ for all $x \in \BB^\B$. Suppose also that $\E_{x \sim \Unif(\BB^\B)}[f(x)] = 1$. Then $\bW^1[f] \leq 6 \ln(L)$. 
\end{theorem}
Using Theorem~\ref{thm:level1}, for each $z \in \MZ$, with $f = g_z', L = 1 + 2e^{\ep_L}$, we get that $\bW^1[g_z'] \leq 6 \ln(1 + 2e^{\ep_L}) \leq 6 \ln(3e^{\ep_L})$. From (\ref{eq:pre_level1}) it follows that
$$
 \E_{\eell \sim\Unif(\BB), j \sim \Unif([\B])} \left[
  \KL(\PP_{\eell,j} || \QQ) \right]  \leq \frac{6 \ln(3e^{\ep_L})}{\B} + (2e^{\ep_L}+ \B/2) \cdot 2\delta_L,
$$
as desired.
\end{proof}

\section{Multi-Message Protocols for Frequency Estimation}%Oracle and Heavy Hitters}
\label{sec:freq_oracle_heavy_hitters}
\label{sec:hh}

In this section, we present new algorithms for private frequency estimation 
%(see Theorem~\ref{thm:CM_main}) and range counting (see Theorems~\ref{thm:cm_rq_final} and~\ref{thm:had_rq_final}) 
in the shuffled model that significantly improve on what can be achieved in the local model of differential privacy.
By our previous lower bounds, such protocols must necessarily use multiple messages.
Our results are summarized in Table~\ref{tab:results}, which focuses on communication requirements of users, the size of the additive error on query answers, and the time required to answer a query (after creating a data structure based on the shuffled dataset).
To our best knowledge, the only previously known upper bounds for these problems in the shuffled model (going beyond local differential privacy) followed via a reduction to private aggregation~\cite{cheu_distributed_2018}.
Using the currently best protocol for private aggregation in the shuffled model~\cite{DBLP:journals/corr/abs-1906-09116,ghazi2019scalable} yields the result stated in the first row of Table~\ref{tab:results}.
Since the time to answer a frequency query differs by a factor of $\tilde{\Theta}(n)$ between our public and private coin protocols, we include query time bounds in the table. We start by stating the formal guarantees on our \emph{private-coin} multi-message protocol.
\begin{theorem}[Frequency estimation via private-coin multi-message shuffling]
  \label{thm:hist_full}
Let $n$ and $B$ be positive integers and $\epsilon > 0$ and $\delta \in (0,1)$ be real numbers. Then there exists a private-coin $(\epsilon, \delta)$-differentially private algorithm in the shuffled model for frequency estimation on $n$ users and domain size $B$ with error $O\left(\log B + \frac{\sqrt{\log(\B)\log(1/(\ep\delta))}}{\ep}\right)$ and with $O\left(\frac{\log (1/\ep\delta)}{\ep^2}\right)$ messages per user, where each message consists of $O(\log{n}\log{\B})$ bits. Moreover, any frequency query can be answered in time $O\left(\frac{ n  \log\left(\frac{1}{\ep\delta}\right)  \log{n} \log \B}{\ep^2}\right)$.
\end{theorem}
Theorem~\ref{thm:hist_full} is proved in Section~\ref{subsec:Had}. We next state the formal guarantees of our \emph{public-coin} multi-message protocols, whose main advantage compared to the private-coin protocol is that it has polylogarithmic query time.

\begin{theorem}[Frequency estimation via public-coin multi-message shuffling]
  \label{thm:hist_full_pub}
Let $n$ and $B$ be positive integers and $\epsilon > 0$ and $\delta \in (0,1)$ be real numbers. Then there exists a public-coin $(\epsilon, \delta)$-differentially private algorithm in the shuffled model for frequency estimation on $n$ users and domain size $B$ with error $O\left(\frac{\log^{3/2}(\B)\sqrt{\log (\log \B/\delta)}}{\ep}\right)$ and with $O\left( \frac{\log^3(\B) \log(\log(\B)/\delta)}{\ep^2} \right)$ messages per user, where each message consists of $O(\log{n} + \log\log{B})$ bits. Moreover, any frequency query can be answered in time $O(\log{B})$.
\end{theorem}
The public-coin protocol in Theorem~\ref{thm:hist_full_pub} and its analysis are presented in Section~\ref{subsec:count-min}.

Prior work~\cite{bassily2015local,bassily2017practical,bun2018heavy} has focused on the case of computing heavy hitters when each user holds only a single element.
% , i.e., the case $k = 1$ of the $(\ha, \hb, \hk)$-frequency oracle.
While we focus primarily on this case, we will also consider the 
% and we do not try to optimize our algorithms for large $k$.
application of frequency estimation to the task of computing range counting queries (given in Section~\ref{sec:rq} below), where we apply our protocols for a frequency oracle as a black box and need to deal with the case in which a user can hold $k > 1$ inputs. Thus, in the rest of this section, we state our results for more general values of $k$ (although we do not attempt to optimize our algorithms for large values of $k$, as $k$ will be at most $\poly\log(n)$ in our application to range counting queries).
Moreover, our results for $k \geq 1$ can be interpreted as establishing bounds for privately computing sparse families of counting queries (see Appendix~\ref{app:counting-queries}). 

For clarity, we point out that the privacy of our protocols holds for every setting of the public random string.
In other words, public randomness is assumed to be known by the analyzer, and affects error but not privacy.

% We also note that consequences of the results (presented in this section) for \emph{counting queries} are discussed in Appendix~\ref{app:counting-queries}.

%In the local model frequency estimation has error XXX~\cite{}

\begin{table}[t]
    \centering
    \bgroup
    \def\arraystretch{1.5}
    \footnotesize
    \centerline{
    \begin{tabular}{|c|c|c|c|c|}
        \hline
        {\bf Problem} & {\bf \thead{Messages\\ per user}} & {\bf \thead{Message size\\ in bits}} & {\bf \thead{Error}} & {\bf Query time}\\
        \hline
        \hline
        \thead{Frequency estimation\\ (private randomness) \\
        \cite{cheu_distributed_2018,DBLP:journals/corr/abs-1906-09116,ghazi2019scalable}} & 
       $\B$ & $\log \B$ & \begin{tabular}{@{}c@{}} $\sqrt{\log(\B)\log\tfrac{1}{\delta}}/\varepsilon$\\ (expected error)\end{tabular} & 1 \\
        \hline
        \thead{Frequency estimation\\ (private randomness) \\ Section~\ref{subsec:Had}} & $\frac{\log (1/\ep\delta)}{\ep^2}$ & $\log{n}\log{\B}$ & \begin{tabular}{@{}c@{}} $\log B + \frac{\sqrt{\log(\B)\log(1/(\ep\delta))}}{\ep}$
        %\\ (whp., two-sided error)
        \end{tabular} & $\frac{ n  \log \left(\tfrac{1}{\ep\delta}\right)  \log{n} \log \B}{\ep^2}$\\
        \hline
         \multirow{2}{*}{\thead{Frequency estimation\\ (public randomness) \\ Section~\ref{subsec:count-min}}} &
         $\frac{\log^3(\B) \log(\log(\B)/\delta)}{\ep^2}$ & $\log{n} + \log\log{B}$ & 
         $\frac{\log^{3/2}(\B)\sqrt{\log (\log \B/\delta)}}{\ep}$ & $\log \B$\\
         & $B^\eta$ & $\log B$ & $\sqrt{\log(\B)\log\tfrac{1}{\delta}}/\varepsilon$ & $1$\\
% Old version, low communication only:
%         \thead{Frequency estimation\\ (public randomness) \\ Section~\ref{subsec:count-min}} &
%         $\frac{\log^3(\B) \log(\log(\B)/\delta)}{\ep^2}$ & $\log{n} + \log\log{B}$ & 
%         $\frac{\log^{3/2}(\B)\sqrt{\log (\log \B/\delta)}}{\ep}$ & $\log \B$\\
         \hline
    \end{tabular}}
    \egroup
    \caption{Overview of bounds on frequency estimation in the shuffled model with multiple messages.  Each user is assumed to hold $k=1$ value from $[B]$, and $\eta > 0$ is a constant. The query time stated is the additional time to answer a query, assuming a preprocessing of the output of the shuffler that takes time linear in its length. Note that frequencies and counts are not normalized, i.e., they are integers in~$\{0,\dots,n\}$. For simplicity of presentation in this table, constant factors are suppressed, the bounds are stated for error probability $\beta = B^{-O(1)}$, and the following are assumed: $n$ is bounded above by $\B$, and $\delta < 1/\log \B$.
    }
    \label{tab:results}
\end{table}

\subsection{Private-Coin Protocol}\label{subsec:Had}

% The protocol based on the Count Min sketch of the previous section achieves nearly optimal (i.e., up to polylogarithmic factors) time and accuracy for computing a $(\ha, \hb, \hk)$-frequency oracle for $k = 1$; see Table~\ref{tab:results}.
% However, the protocol relied on the availability of public randomness, i.e., the ability of the server to send a common string of random bits to all users.
In this section, we give a private-coin protocol (i.e., one where the only source of randomness is the private coin source at each party) for frequency estimation with polylogarithmic error and polylogarithmic bits of communication per user. In the case of local DP, private-coin protocols were recently obtained by Acharya et al. in~\cite{acharya2019hadamard,acharya2019communication}. These works made use of the Hadamard response for the local randomizers instead of previous techniques developed in the local model and which relied on public randomness. The Hadamard response was also used in~\cite{cormode2018answering,cormode2018marginal,nguyen2016collecting} for similar applications, namely, private frequency estimation.

\paragraph{Overview.}
For any power of two $\B \in \BN$, let $H_\B \in \{-1, 1\}^{\B \times \B}$ denote the $\B \times \B$ Hadamard matrix and for $j \in [\B-1]$, set $\MH_{\B,j} := \{ j' \in [\B] ~\mid~ H_{j+1,j'} = 1\}$\footnote{Since the first row of the Hadamard matrix $H_\B$ is all 1's, we cannot use the first row in our frequency estimation protocols. This is the reason for the subscript of $j+1$ in the definition of $\MH_{\B,j}$.}. By orthogonality of the rows of $H_\B$, we have that $|\MH_{\B,j}| = \B/2$ for any $j \in [\B-1]$ and for all $j \neq j'$, it is the case that $| \MH_{\B,j} \cap \MH_{\B,j'}| = \B/4$. For any $\tau \in \BN$, we denote the $\tau$-wise Cartesian product of $\MH_{\B,j}$ by $\MH_{\B,j}^\tau \subset [\B]^\tau$. In the \emph{Hadamard response}~\cite{acharya2019hadamard}, a user whose data consists of an index $j \in [\B]$ sends to the server a random index $j' \in [\B]$ that is, with probability $\frac{e^\ep}{1 + e^\ep}$, chosen uniformly at random from the ``Hadamard codeword'' $\MH_{\B,j}$ and, with probability $\frac{1}{1 + e^\ep}$, chosen uniformly from $[\B] \setminus \MH_{\B,j}$.
% {\it Hadamard response} can be used to compute heavy hitters in the local model of differential privacy with communication bounded by $O(\log \B)$. In Hadamard response,.

In the shuffled model, much less randomization is needed to protect a user's privacy than in the local model of differential privacy, where the Hadamard response was previously applied. In particular, we can allow the users to send more information about their data to the server, along with some ``blanket noise''\footnote{This uses the expression of Balle et al.~\cite{balle_privacy_2019}} which helps to hide the true value of {\it any} one individual's input. Our adaptation of the Hadamard response to the multi-message shuffled model for computing frequency estimates (in the case where each user holds up to $k$ elements) proceeds as follows (see Algorithm~\ref{alg:hh_had} for the detailed pseudo-code). Suppose the $n$ users possess data $\MS_1, \ldots, \MS_n \subset [\B]$ such that $|\MS_i| \leq k$ --- equivalently, they possess $x_1, \ldots, x_n \in \{0,1\}^{\B}$, such that for each $i \in [n]$, $\| x_i \|_1 \leq k$ (the nonzero indices of $x_i$ are the elements of $\MS_i$).  Given $x_i$, the local randomizer $R^{\Had}$ {\it augments} its input by adding $k - \| x_i \|_1$ arbitrary  elements from the set $\{ \B + 1, \ldots, 2\B-1 \}$ (recall that $k < \B$). (Later, the analyzer will simply ignore the augmented input in $\{ \B + 1, \ldots, 2\B-1 \}$ from the individual randomizers. The purpose of the augmentation is to guarantee that all sets $\MS_i$ will have cardinality exactly $k$, which facilitates the privacy analysis.) Let the augmented input be denoted $\tilde x_i$, so that $\tilde x_i \in \{0,1\}^{2\B-1}$ and $\| \tilde x_i \|_1 = k$. For each index $j$ at which $(\tilde x_i)_j \neq 0$, the local randomizer chooses $\tau$ indices $a_{j,1}, \ldots, a_{j,\tau}$ in $\MH_{2\B,j}$ uniformly and independently, and sends each tuple $(a_{j,1}, \ldots, a_{j,\tau})$ to the shuffler. It also generates $\rho$ tuples $(\tilde a_{g,1}, \ldots, \tilde a_{g,\tau})$ where each of $\tilde a_{g,1}, \ldots, \tilde a_{g,\tau}$ is uniform over $[2\B]$, and sends these to the shuffler as well; these latter tuples constitute ``blanket noise'' added to guarantee differential privacy. 

Given the output of the shuffler, the analyzer $A^{\Had}$ determines estimates $\hat x_j$ for the frequencies of each $j \in [\B]$ by counting the number of messages $(a_1, \ldots, a_\tau) \in [2\B]^\tau$ which belong to $\MH_{2\B,j}^\tau$. The rationale is that each user $i$ such that $j \in \MS_i$ will have sent such a message in $\MH_{2\B,j}^{\tau}$. As the analyzer could have picked up some of the blanket noise in this count, as well as tuples sent by users holding some $j' \neq j$, since $\MH_{2\B,j}^\tau \cap \MH_{2\B,j'}^\tau \neq \emptyset$, it then corrects this count (Algorithm~\ref{alg:hh_had}, Line~\ref{ln:debias}) to obtain an unbiased estimate $\hat x_j$ of the frequency of $j$. 

\begin{algorithm}[!ht]
\Fn{$R^{\Had}(n, \B, \tau, \rho, k)$}{
\KwIn{Set $\MS \subset [\B]$ specifying $i$'s input set;\\ 
Parameters $n, \B, \tau, \rho,k \in \BN$}
\KwOut{A multiset $\MT \subset \{0,1\}^{\log 2\B \cdot \tau}$}
\For{$j = \B + 1, \B + 2, \ldots, 2\B-1$}{\tcp{Augmentation step}
\If{$|\MS| < k$}{
$\MS \gets \MS \cup \{ j \}$\label{ln:augment}
}
}
\For{$j \in \MS$}{
Choose $a_{j,1}, \ldots, a_{j,\tau} \in \MH_{2\B,j}$ uniformly and independently at random
}\label{ln:ajs}
% Let $\tilde \rho \gets \Bin(\B^\tau, \rho / \B^\tau)$.\\
\For{$g = 1, 2, \ldots, \rho$}{
Choose $\tilde a_{g,1}, \ldots, \tilde a_{g,\tau} \in [2\B]$ uniformly and independently at random
}\label{ln:tildek}
\Return{$\MT :=\bigcup_{j \in\MS}\{ (a_{j,1}, \ldots,  a_{j,\tau})\} \cup \bigcup_{1 \leq g \leq \tilde\rho} \{  (\tilde a_{g,1} , \ldots,  \tilde a_{g,\tau})\}$} \tcp{Each element of $\MT$ is viewed as an element of $(\{0,1\}^{\log 2\B})^{\tau}$, by associating each element of $[2B]$ with its binary representation.}
}
\Fn{$A^{\Had}(n,\B,\tau,\rho, k)$}{
\KwIn{Multiset $\{ y_1, \ldots, y_m \}$ consisting of outputs of local randomizers, $y_i \in (\{0,1\}^{\log 2\B})^\tau$;\\
Parameters $n, \B, \tau, \rho, k \in \BN$}
\KwOut{A vector $\hat x \in \BR^{\B}$ containing estimates of the frequency of each $j \in [\B]$}
\For{$j \in [\B]$}{
Let $\hat x_j \gets 0$
}
\For{$j \in [\B]$}{
\For{$i \in [m]$}{\label{ln:mforloop}
Write $y_i \in (\{0,1\}^{\log 2\B})^\tau$ as $y_i := (a_{i,1},\ldots,  a_{i,\tau})$, with $a_{i,1}, \ldots, a_{i,\tau} \in \{0,1\}^{\log 2\B}$ \\
\If{$\{ a_{i,1}, \ldots, a_{i,\tau} \} \subset \MH_{2\B,j}$}{\label{ln:if_had}
$\hat x_j \gets \hat x_j + 1$
}
}
}
\For{$j \in [\B]$}{
$\hat x_j \gets \frac{1}{1 - 2^{-\tau}} \cdot \left(\hat x_j - (\rho + k) n2^{-\tau}\right)$\quad \tcp{De-biasing step}\label{ln:debias}
}
\Return{$\hat x$}
}
\caption{Local randomizer and analyzer for frequency estimation via Hadamard response}
\label{alg:hh_had}
\end{algorithm}

\paragraph{Analysis.}
The next theorem summarizes the privacy, accuracy and efficiency properties of Algorithm~\ref{alg:hh_had} for general values of $k$.
\begin{theorem}\label{thm:Had_k_main}
There is a sufficiently large positive absolute constant $\zeta$ such that the following holds. Suppose $n, \B, k \in \BN$ with $k < \B$, and $0 \leq \ep, \delta,\beta \leq 1$. Consider the shuffled-model protocol $P^{\Had} = (R^{\Had}, S, A^{\Had})$ with $\tau = \log n$ and $\rho = \frac{36k^2}{\epsilon^2} \left(\ln \frac{e k}{\epsilon \delta}\right)$.
Then $P^{\Had}$ is a $(\ep, \delta)$-differentially private protocol (Definition~\ref{def:dp_shuffled}) with $O\left( \frac{k^2 \log(k/\epsilon\delta)}{\epsilon^2}\right)$ messages per user, each consisting of $O(\log n \log B)$ bits, such that for inputs $x_1, \ldots, x_n \in \{0,1\}^\B$ satisfying $\| x_i\|_1 \leq k$, the estimates $\hat x_j$ produced by the output of $P^{\Had}(n, B, \tau, \rho, k)$ satisfy
\begin{equation}
  \label{eq:qhad_acc}
\p \left[\forall j \in [\B] ~:~ \left|\hat x_j - \sum_{i=1}^n x_{i,j}\right| \leq O \left( \log(B/\beta) + \frac{k \sqrt{\log(B/\beta) \log(k/\epsilon\delta)}}{\epsilon}\right)\right] \geq 1 - \beta.
\end{equation}
Moreover, any frequency query can be answered in time $O\left( n \log n \log B \left(\frac{k^2 \log(k/\epsilon\delta)}{\epsilon^2} \right)\right)$.
\end{theorem}
Before we prove Theorem~\ref{thm:Had_k_main}, instantiating Theorem~\ref{thm:Had_k_main} with $k=1$ directly implies Theorem~\ref{thm:hist_full}.

Theorem~\ref{thm:Had_k_main} is a direct consequence of Lemmas~\ref{thm:h_had_priv},~\ref{thm:h_had_acc}, and~\ref{thm:Hadamard_efficiency}, which establish the privacy, accuracy, and efficiency guarantees, respectively, of protocol $P^{\Had}$. The remainder of this section presents and proves the aforementioned lemmas. We begin with Lemma~\ref{thm:h_had_priv}, which establishes DP guarantees of $P^{\Had}$.

\begin{restatable}[Privacy of $P^{\Had}$]{lemma}{Hadprivacy}
\label{thm:h_had_priv}
Fix $n, \B \in \BN$ with $\B$ a power of 2. Let $\tau = \log n$, $\ep \leq 1$, and $\rho = \frac{36 \ln 1/\delta}{\ep^2}$. Then the algorithm $S \circ R^{\Had}(n, \B, \tau, \rho, k)$ is $(k\ep, \delta \exp(k\ep) / \ep)$-differentially private.
\end{restatable}
% \begin{proof}
%Consider neighboring datasets $X = (x_1, \ldots, x_n), X' = (x_1, \ldots, x_{n-1}, x_n')$, where each $x_i, x_i' \in \{0,1\}^\B$ with $\| x_i\|_1 \leq k$ and $\| x_i' \|_1 \leq k$ for all $i \in[n]$. 
% \end{proof}

\begin{proof}
For convenience let $P := S \circ R^{\Had}(n, \B, \tau, \rho, k)$ be the protocol whose $(\ep, \delta)$-differential privacy we wish to establish. With slight abuse of notation, we will assume that $P$ operates on the augmented inputs $(\tilde x_1, \ldots, \tilde x_n)$ (see Algorithm~\ref{alg:hh_had}, Line~\ref{ln:augment}). In particular, for inputs $(x_1, \ldots, x_n)$ that lead to augmented inputs $(\tilde x_1, \ldots, \tilde x_n)$, we will let $P(\tilde x_1, \ldots, \tilde x_n)$ be the output of $P$ when given as inputs $x_1, \ldots, x_n$. Let $\MY$ be the set of multisets consisting of elements of $\{ 0,1\}^{\log 2\B \times \tau}$; notice that the output of $P$ lies in $\MY$. 

By symmetry, it suffices to show that for any augmented inputs of the form $\tilde X = (\tilde x_1, \ldots, \tilde x_{n-1}, \tilde x_n)$ and $\tilde X' = (\tilde x_1, \ldots, \tilde x_{n-1}, \tilde x_n')$, and for any subset $\MU \subset \MY$, we have that
\begin{equation}
\label{eq:had_private_set}
\p[P(\tilde x_1, \ldots, \tilde x_n) \in \MU] \leq e^\ep \cdot \p[P(\tilde x_1, \ldots, \tilde x_{n-1}, \tilde x_n') \in \MU] + \delta.
\end{equation}
We first establish (\ref{eq:had_private_set}) for the special case that $\tilde x_n, \tilde x_n'$ differ by 1 on two indices, say $j, j'$, while having the same $\ell_1$ norm: in particular, we have $|(\tilde x_n)_j - ( \tilde x_n')_{j}| = 1$ and $|(\tilde x_n)_{j'} - (\tilde x_n')_{j'}| = 1$. % We will first consider the case that $(x_n')_j - (x_n)_j = 1$. % % while $(x_n)_{j'} - (x_n')_{j'} = -1$. 
By symmetry, without loss of generality we may assume that $j = 1, j' = 2$ and that $(\tilde x_n)_j - (\tilde x_n')_j = 1$ while $(\tilde x_n')_{j'} - (\tilde x_n)_{j'} = 1$. 
To establish (\ref{eq:had_private_set}) in this case, we will in fact prove a stronger statement: % first, for $i \in [n]$, write $k_i = \| x_i\|_1$. 
for inputs $(\tilde x_1, \ldots, \tilde x_n)$, define the {\it view} of an adversary, denoted by $\View_P(\tilde x_1, \ldots, \tilde x_n)$, as the tuple consisting of the following components:
\begin{itemize}[nosep]
% \item The tuple $(x_1, \ldots, x_{n-1})$ of inputs of the first $n-1$ users.
\item For each $i \in [n-1]$, the set $\hat \MS_i := \bigcup_{j  : (\tilde x_i)_j = 1} \{ (a_{j,1}, \ldots, a_{j,\tau}) \}$ of tuples output by user $i$ corresponding to her true input $\tilde x_i$.
\item The set $\hat \MS_n := \bigcup_{j : j \not \in \{1,2\}, (\tilde x_n)_j = 1} \{ (a_{j,1}, \ldots, a_{j,\tau}) \}$ of tuples output by user $n$ corresponding to her true (augmented) input $\tilde x_n$, except (if applicable) the string that would be output if $(\tilde x_n)_1 = 1$ or $(\tilde x_n)_2  = 1$. 
\item The multiset $\{ y_1, \ldots, y_m\}$ consisting of the outputs of the $n$ users of the protocol $P$.
\end{itemize}
It then suffices to show the following:
\begin{equation}
\label{eq:had_views}
\p_{V \sim \View_P(\tilde x_1, \ldots, \tilde x_n)} \left[ \frac{\p[\View_P(\tilde x_1, \ldots, \tilde x_{n-1}, \tilde x_n) = V]}{\p[\View_P(\tilde x_1, \ldots, \tilde x_{n-1}, \tilde x_n') = V]} \geq e^\ep \right] \leq \delta.
\end{equation}
% Indeed, if (\ref{eq:had_views}) holds, then for any set $\MS$ as in (\ref{eq:had_private_set}), we may consider the cartesian product of $\MS$ with all possible values of $\hat \MS_1, \ldots, \hat \MS_n$, \todo{finish}.
(See \cite[Theorem 3.1]{balle_privacy_2019} for a similar argument.)

Notice that each of the elements $y_1, \ldots, y_m$ in the output of the protocol $P$ consists of a tuple $(a_1, \ldots, a_\tau)$, where each $a_1, \ldots, a_\tau \in [2\B]$.  Now we will define a joint distribution (denoted by $\MD$) of random variables $(W_{a_1, \ldots, a_\tau})_{a_1, \ldots, a_\tau \in [2\B]}, Q, Q'$, where, for each $(a_1, \ldots, a_\tau) \in [2\B]^\tau$, $W_{a_1, \ldots, a_\tau} \in \BZ_{\geq 0}$, and $Q,Q' \in [2\B]^\tau$, as follows. 
For each tuple $(a_1, \ldots, a_\tau) \in [2\B]^\tau$, we let $W_{a_1, \ldots, a_\tau}$ be jointly distributed from a multinomial distribution over $[2\B]^\tau$ with $\rho n $ trials. %  independent random variables distributed as $\Bin(\rho n, B^{-\tau})$.
For each $(a_1, \ldots, a_\tau) \in [2\B]^\tau$, let $\hat W_{a_1, \ldots, a_\tau}$ be the random variable representing the number of tuples $(\tilde a_{g,1}, \ldots, \tilde a_{g,\tau})$ generated on Line~\ref{ln:tildek} of Algorithm~\ref{alg:hh_had} satisfying $(\tilde a_{g,1}, \ldots, \tilde a_{g,\tau}) = (a_1, \ldots, a_\tau)$. Notice that the joint distribution of all $W_{a_1, \ldots, a_\tau}$ is the same as the joint distribution of $\hat W_{a_1, \ldots, a_\tau}$, for $(a_1, \ldots, a_\tau) \in [2\B]^\tau$. % the number of tuples $(\tilde a_{g,1}, \ldots, \tilde a_{g,\tau})$ generated on Line~\ref{ln:tildek} of Algorithm~\ref{alg:hh_had} equal to $(a_1, \ldots, a_\tau)$
Intuitively, $W_{a_1, \ldots, a_\tau}$ represents the blanket noise added by the outputs $(\tilde a_{g,1}, \ldots , \tilde a_{g,\tau})$ in Line~\ref{ln:tildek} of Algorithm~\ref{alg:hh_had}. Also let $Q, Q' \in [2\B]^\tau$ be random variables that are distributed uniformly over $\MH_{2\B,1}^\tau, \MH_{2\B,2}^\tau$, respectively.
% Notice that the procedure to generate the tuples $(\tilde a_{g,1}, \ldots, \tilde a_{g,\tau})$ in Line~\ref{ln:tildek} of Algorithm~\ref{alg:hh_had} % TODO: rename Y_{a_1, \ldots, a_\tau} to w_{a_1, \ldots, a_\tau}, and Z, Z' to q, q'.
% Notice that the procedure to generate the tuples $(\tilde a_{g,1}, \ldots, \tilde a_{g,\tau})$ in Line~\ref{ln:tildek} of the local randomizer of Algorithm~\ref{alg:hh_had} can be described equivalently as follows: for each tuple $(a_1, \ldots, a_\tau) \in [\B]^\tau$, it is included in $\MT$ with probability $\rho/\B^\tau$, independently of all other tuples $(a_1, \ldots, a_\tau)$.
Then since the tuples $(\tilde a_{g,1}, \ldots, \tilde a_{g,\tau})$ are distributed independently of the tuples $(a_{j,1}, \ldots, a_{j,\tau})$ ($j \in \MS_i$), (\ref{eq:had_views}) is equivalent to
{\scriptsize\begin{equation}
\label{eq:big_fraction}
\p_{w_{a_1, \ldots, a_\tau}, q, q' \sim \MD} \left[ \frac
{\p_{ W_{a_1, \ldots, a_\tau}, Q, Q' \sim \MD}\left[ \forall (a_1, \ldots, a_\tau) \in [2\B]^\tau : W_{a_1, \ldots, a_\tau} + \One[Q = (a_1, \ldots, a_\tau)]  = w_{a_1, \ldots, a_\tau} + \One[q = (a_1, \ldots, a_\tau)]  \right]}
{\p_{ W_{a_1, \ldots, a_\tau}, Q, Q' \sim \MD}\left[
\forall (a_1, \ldots, a_\tau) \in [2\B]^\tau : W_{a_1, \ldots, a_\tau} + \One[Q' = (a_1, \ldots, a_\tau)] = w_{a_1, \ldots, a_\tau} + \One[q = (a_1, \ldots, a_\tau)]  \right]} \geq e^\ep
\right] \leq \delta.
\end{equation}}
% For $0 \leq w \leq \rho n$, let
% $$
% p(w) = {\rho n \choose w} \cdot (B^{-\tau})^w \cdot (1 - B^{-\tau})^{\rho n - w},
% $$
% be the probability that a random variable distributed as $\Bin(\rho n, B^{-\tau})$ is equal to $w$. As a notational convention, for $w < 0$ or $w > \rho n$, set $p(w) = 0$. 
Set $\tilde w_{a_1, \ldots, a_\tau}:= w_{a_1, \ldots, a_\tau} + \One[q = (a_1, \ldots, a_\tau)]
$. By the definition of $\MD$ we have % By independence of the $W_{a_1, \ldots, a_\tau}$ under $\MD$, we have
\begin{align}
& \p_{ W_{a_1, \ldots, a_\tau}, Q, Q' \sim \MD}\left[
\forall (a_1, \ldots, a_\tau) \in [2\B]^\tau : W_{a_1, \ldots, a_\tau} + \One[Q = (a_1, \ldots, a_\tau)] = \tilde w_{a_1, \ldots, a_\tau}  \right]\nonumber\\
&= \E_{Q \sim \MD} \left[ (2\B)^{-\tau \rho n} \cdot \binom{ (2\B)^\tau}{ \{ \tilde w_{a_1, \ldots, a_\tau} - \One[Q = (a_1, \ldots, a_\tau)] \}_{(a_1, \ldots, a_\tau) \in [2\B]^\tau}}\right]\nonumber\\
& = \left( \frac{2}{2\B} \right)^\tau \cdot (2\B)^{-\tau \rho n} \binom{(2\B)^\tau }{ \{ \tilde w_{a_1, \ldots, a_\tau} \}_{(a_1, \ldots, a_\tau) \in [2\B]^\tau}}
\cdot  \sum_{a_1', \ldots, a_\tau' \in \MH_{2\B, 1}} \tilde w_{a_1', \ldots, a_\tau'}.\nonumber 
\end{align}
In the above equation, the notation such as $\binom{(2\B)^\tau}{ \{ \tilde w_{a_1, \ldots, a_\tau} \}_{(a_1, \ldots, a_\tau) \in [2\B]^\tau}}$ refers to the multinomial coefficient, equal to $\frac{((2\B)^\tau)!}{\prod_{a_1, \ldots, a_\tau \in [2\B]} \tilde w_{a_1, \ldots, a_\tau}!}$. 
% \begin{align}
% & \p_{ W_{a_1, \ldots, a_\tau} \sim \MD}\left[ \forall (a_1, \ldots, a_\tau) \in \MH_0^\tau \cup \MH_1^\tau: W_{a_1, \ldots, a_\tau} + \One[Q = (a_1, \ldots, a_\tau)]  = \tilde w_{a_1, \ldots, a_\tau}  \right] \nonumber\\
% & = \prod_{(a_1, \ldots, a_\tau) \in [\B]^\tau} p(w_{a_1, \ldots, a_\tau})\nonumber
% \end{align}
Similarly, for the denominator of the expression in (\ref{eq:big_fraction}),
\begin{align}
& \p_{ W_{a_1, \ldots, a_\tau}, Q, Q' \sim \MD}\left[
\forall (a_1, \ldots, a_\tau) \in [2\B]^\tau : W_{a_1, \ldots, a_\tau} + \One[Q' = (a_1, \ldots, a_\tau)] = \tilde w_{a_1, \ldots, a_\tau}  \right]\nonumber\\
&= \E_{Q \sim \MD} \left[ (2\B)^{-\tau \rho n} \cdot \binom{ (2\B)^\tau}{ \{ \tilde w_{a_1, \ldots, a_\tau} - \One[Q' = (a_1, \ldots, a_\tau)] \}_{(a_1, \ldots, a_\tau) \in [2\B]^\tau}}\right]\nonumber\\
& = \left( \frac{2}{2\B} \right)^\tau \cdot (2\B)^{-\tau \rho n} \binom{(2\B)^\tau }{ \{ \tilde w_{a_1, \ldots, a_\tau} \}_{(a_1, \ldots, a_\tau) \in [2\B]^\tau}}
\cdot  \sum_{a_1', \ldots, a_\tau' \in \MH_{2\B,2}} \tilde w_{a_1', \ldots, a_\tau'}.\nonumber 
\end{align}
Thus, (\ref{eq:big_fraction}) is equivalent to
\begin{equation}
\label{eq:ratio_sums}
\p_{w_{a_1, \ldots, a_\tau}, q, q'\sim \MD} \left[ \frac{\sum_{a_1', \ldots, a_\tau' \in \MH_{2\B,1}} \tilde w_{a_1', \ldots, a_\tau'}}{\sum_{a_1', \ldots, a_\tau' \in \MH_{2\B,2}} \tilde w_{a_1, \ldots, a_\tau'}} \geq e^\ep \right] \leq\delta.
\end{equation}
Notice that $\sum_{a_1', \ldots, a_\tau' \in \MH_{2\B,1}} \tilde w_{a_1', \ldots, a_\tau'}$ is distributed as $1 + \Bin(\rho n, 2^{-\tau})$, since $q \in \MH_{2\B,1}^\tau$ with probability 1 (by definition of $Q \sim \MD$), and each of the $\rho n $ trials in determining the counts $w_{a_1, \ldots, a_\tau}$ belongs to $\MH_{2\B,2}$ with probability $2^{-\tau}$.  Similarly, $\sum_{a_1', \ldots, a_\tau' \in \MH_{2\B,1}} \tilde w_{a_1', \ldots, a_\tau'}$ is distributed as $\Bin(\rho n + 1, 2^{-\tau})$; notice in particular that $q$, which is distributed uniformly over $\MH_{2\B,1}^\tau$, is in $\MH_{2\B,2}^\tau$ with probability $2^{-\tau}$. By the multiplicative Chernoff bound, we have that, for $\eta \leq 1$, it is the case
\begin{equation}
\label{eq:rhoeta}
\p_{W \sim \Bin(\rho n, 1/n)} \left[|W - \rho| > \rho \eta \right] \leq \exp\left( \frac{-\eta^2 \rho}{3} \right).
\end{equation}
As long as we take $\rho =  \frac{36 \ln(1/\delta)}{\ep^2}$, inequality (\ref{eq:rhoeta}) will be satisfied with $\eta = \ep/6$, which in turn implies inequality (\ref{eq:ratio_sums}) since 
$$
\frac{(1 + \ep/6) \rho + 1}{\rho (1 - \ep/6)} \leq \frac{e^{\ep/6} \rho + 1}{e^{-\ep/3}\rho} \leq \frac{e^{4\ep/6}\rho}{e^{-\ep/3}\rho} \leq e^\ep,
$$
where the second inequality above uses $(\rho + 1)/\rho \leq e^{\ep / 2}$ for our choice of $\rho$. 

We have thus established inequality~(\ref{eq:had_private_set}) for the case that $\tilde x_n,\tilde x_n'$ differ by 1 on two indices. For the general case, consider any neighboring datasets $X = (\tilde x_1, \ldots, \tilde x_n)$ and $X' = (\tilde x_1, \ldots, \tilde x_{n-1}, \tilde x_n')$; we can find a sequence of at most $k-1$ intermediate datasets $(\tilde x_1, \ldots, \tilde x_{n-1}, \tilde x_n^{(\mu)})$, $1 \leq \mu \leq k-1$ such that $\tilde x_n^{(\mu)}$ and $\tilde x_n^{(\mu-1)}$ differ by 1 on two indices. Applying inequality~(\ref{eq:had_private_set}) to each of the $k$ neighboring pairs in this sequence, we see that for any $\MU \subset \MY$,
$$
\p[P(X) \in \MU] \leq e^{k\ep} \cdot \p[P(X') \in \MU] + \delta \cdot (1 + e^\ep + \cdots + e^{(k-1)\ep}) \leq e^{k\ep} \cdot \p[P(X') \in \MU] + \delta \cdot \frac{2e^{k\ep}}{\ep},
$$
where we have used $\ep \leq 1$ in the final inequality above.
\end{proof}

We next prove the accuracy of Algorithm~\ref{alg:hh_had}.
\begin{restatable}[Accuracy of $P^{\Had}$]{lemma}{Hadaccuracy}
\label{thm:h_had_acc}
Fix $n, \B \in \BN$ with $\B$ a power of 2. Then with $\tau, \rho$ as in Lemma~\ref{thm:h_had_priv}, the estimate $\hat x$ produced in $A^{\Had}$ in the course of the shuffled-model protocol $P^{\Had} = (R^{\Had}, S, A^{\Had})$ with input $x_1, \ldots, x_n \in \{0,1\}^\B$ satisfies
$$
\p \left[ \left\| \hat x - \sum_{i=1}^n x_i \right\|_\infty \leq  \sqrt{3 \ln(2\B/\beta) \cdot \max \{ {3 \ln(2\B/\beta)}, \rho + k \}}\right] \geq 1- \beta.
$$
\end{restatable}

\begin{proof}
Fix any $j \in [\B]$. Let $\zeta_j = \sum_{i=1}^n (x_i)_j$. We will upper bound the probability that $\xi_j := \hat x_j - \zeta_j$ is large. Notice that the distribution of $\hat x_j$ is given by 
$$
\frac{1}{1 - 2^{-\tau}} \cdot \left(\zeta_j +  \Bin(\rho n + kn - \zeta_j, 2^{-\tau}) - (\rho n + kn) 2^{-\tau}\right).
$$
This is because each of the $\rho n $ tuples $(\tilde a_{g,1}, \ldots, \tilde a_{g,\rho})$ chosen uniformly from $[2\B]^\tau$ on Line~\ref{ln:tildek} of Algorithm~\ref{alg:hh_had} has probability $2^{-\tau}$ of belonging to $\MH_{2\B,j}^\tau$, and each of the $kn - \zeta_j$ tuples $(a_{j',1}, \ldots, a_{j',\tau})$ (for $j' \in\MS_i$, $j' \neq j$, $i \in [n]$) chosen uniformly from $\MH_{2\B,j'}^\tau$ in Line~\ref{ln:ajs} of Algorithm~\ref{alg:hh_had} also has probability $2^{-\tau}$ of belonging to $\MH_{2\B,j}^\tau$. (Moreover, each of the $\zeta_j$ tuples $(a_{j,1}, \ldots, a_{j,\tau})$ chosen in Line~\ref{ln:ajs} of the local randomizer always belongs to $\MH_{2\B,j}^\tau$.)

Therefore, the distribution of $\xi_j$ is given by
$$
\frac{1}{1 - 2^{-\tau}} \cdot \left( \Bin(\rho n + kn - \zeta_j, 2^{-\tau}) - (\rho n + kn - \zeta_j) 2^{-\tau} \right).
$$
Using that $\tau = \log n$, we may rewrite the above as
$$
\frac{n}{n-1} \cdot \left( \Bin((\rho + k - \zeta_j /n) \cdot n, 1/n) - (\rho + k - \zeta_j / n) \right).
$$
For any reals $c > 0$ and $0 \leq \eta \leq 1$, by the Chernoff bound, we have
$$
\p_{\xi \sim \Bin(cn, 1/n)} \left[ |z - c| > \eta c \right] \leq 2\exp\left( \frac{-\eta^2 c}{3} \right).
$$
Moreover, for $\eta > 1$, we have
$$
\p_{\xi \sim \Bin(cn, 1/n)}  \left[ |z - c| > \eta c \right] \leq 2\exp\left( \frac{-\eta c}{3} \right).
$$

We have $2 \exp(-\eta^2 c / 3) \leq \beta/\B$ as long as $\eta \geq \sqrt{3 \ln(2\B/\beta)/ c}$. Set $c = \rho + k - \zeta_j / n$, so that $\rho \leq c \leq \rho + k$. First suppose that $\sqrt{3 \ln(2\B/\beta)/(\rho + k - \zeta_j / n)} \leq 1$. Then we see that
\begin{equation}
\label{eq:had_acc_1}
\p[|\hat x_j - \zeta_j| > \sqrt{3 \ln(2\B/\beta) \cdot (\rho + k)}] \leq \beta/\B.
\end{equation}
In the other case, namely $\rho + k - \zeta_j/n = c < 3 \ln(2\B/\beta)$, set $\eta = 3 \ln(2\B/\beta)/c$, and we see that
\begin{equation}
\label{eq:had_acc_2}
\p[|\hat x_j - \zeta_j | > 3 \ln(2\B/\beta)] \leq \beta/\B.
\end{equation}
The combination of (\ref{eq:had_acc_1}) and (\ref{eq:had_acc_2}) with a union bound over all $j \in [\B]$ completes the proof of Lemma~\ref{thm:h_had_acc}.
\end{proof}

\if 0
Setting $k = 1$ in Lemma~\ref{thm:h_had_acc} gives the following accuracy bound for $P^{\Had}$:
\begin{corollary}
\label{cor:h_had_acc_ke1}
Fix $n, \B$ with $\B$ a power of 2. Then with $\tau, \rho$ as in Lemma~\ref{thm:h_had_priv} and $k = 1$, the estimate $\hat x$ produced in $A^{\Had}$ in the course of the shuffled model protocol $P^{\Had} = (R^{\Had}, S, A^{\Had})$ with input $x_1, \ldots, x_n \in \{0,1\}^\B$ satisfies
$$
\p \left[ \left\| \hat x - \sum_{i=1}^n x_i \right\|_\infty \leq O \left( \sqrt{\log(\B/\beta) \cdot \max \left\{ \log(\B / \beta), \log(1/\delta) / \ep^2 \right\}} \right) \right] \geq 1 - \beta.
$$
\end{corollary}
\fi

Next we summarize the communication and computation settings of the shuffled model protocol $P^{\Had}$. % with the settings of $\tau, \rho$ in Theorem~\ref{thm:h_had_priv}:

\begin{restatable}[Efficiency of $P^{\Had}$]{lemma}{Hadefficiency}
\label{thm:Hadamard_efficiency}
Let $n, \B, \tau, \rho, k \in \BN$. Then the protocol $P^{\Had} = (R^{\Had}(n,\B,\tau,\rho, k), S, \\ A^{\Had}(n,\B,\tau,\rho,k))$  satisfies the following:
\begin{enumerate}
\item On input $(x_1, \ldots, x_n) \in \{0,1\}^\B$, the output of the local randomizers $R^{\Had}(n,\B,\tau,\rho,k)$ consists of $n(k+\rho)$ messages of length $\tau \log 2\B$ bits.
\item The runtime of the analyzer $A^{\Had}(n,\B,\tau,\rho,k)$ on input $\{ y_1, \ldots, y_m\}$ is at most $O(\B m \tau)$ and its output has space $O(\B\log(n(k+\rho)))$ bits. Moreover, if $\tau = \log n$ (i.e., as in Lemma~\ref{thm:h_had_priv}), and if its input $\{ y_1, \ldots, y_m \}$ is the output of the local randomizers on input $x_1, \ldots, x_n$ (so that $m = n(\rho + k)$), there is a modification of the implementation of $A^{\Had}$ in Algorithm~\ref{alg:hh_had} that, for $\beta \in [0,1]$, completes in time $O((\rho + k) n\log^3 \B + \B\rho \log \B/\beta)$ with probability $1-\beta$.  %  be made to be , and the storage cost of the data structure it outputs is $O(\B)$.
% \item The runtime of any frequency query on the data structure output by $A^{\Had}$ is $O(1)$.
\item There is a separate modification of $A^{\Had}(n,\B,\tau,\rho,k)$ that on input $\{ y_1, \ldots, y_m \}$ produces an output data structure $(\FO, \MA)$ with space $O(m\tau\log \B)$ bits, such that a single query $\MA(\FO, j)$ of some $j \in [\B]$ takes time $O(m\tau \log \B)$.
\end{enumerate}
\end{restatable}
\begin{proof}[Proof]
The first item is immediate from the definition of $R^{\Had}$ in Algorithm~\ref{alg:hh_had}. For the second item, note first that $A^{\Had}$ as written in Algorithm~\ref{alg:hh_had} takes time $O(\B m \tau\log \B)$: for each message $y_i = (a_{i,1}, \ldots, a_{i,\tau})$, it loops through each $j \in [\B]$ to check if each $a_{i,g} \in \MH_{2\B,j}$ for $1 \leq g \leq \tau$ (determination of whether $a_{i,g} \in \MH_{2\B,j}$ takes time $O(\log \B)$).

Now suppose that the messages $y_1, \ldots, y_m$ are the union of the multisets output by each of $n$ shufflers on input $x_1, \ldots, x_n$. Notice that for each $j' \in [2\B-1]$, the number of messages $(a_{j,1}, \ldots, a_{j,\tau}) \in \MH_{2\B,j}$ (Line~\ref{ln:ajs} of Algorithm~\ref{alg:hh_had}) such that $j \neq j'$ and also $(a_{j,1}, \ldots, a_{j,\tau}) \subset \MH_{2\B,j'}$ is distributed as $Bin(n', 1/n)$ for some $n' \leq n$ (recall $\tau = \log n$). Moreover, for each $j' \in [2\B-1]$, the number of messages of the form $(\tilde a_{g,1}, \ldots, \tilde a_{g,\tau})$ (Line~\ref{ln:tildek} of Algorithm~\ref{alg:hh_had})  satisfying $(\tilde a_{g,1}, \ldots, \tilde a_{g,\tau}) \subset \MH_{2\B, j'}$ is distributed as $Bin(\rho n, 1/n)$. Therefore, by the multiplicative Chernoff bound and a union bound, for any $0 \leq \beta \leq 1$, the sum, over all $j' \in [\B]$, of the number of messages $(a_1, \ldots, a_\tau)$ that belong to $\MH_{2\B,j'}$ is bounded above by
\begin{equation}
\label{eq:bound_solutions}
kn + \frac{4 \B (n + \rho n) (1 + \ln(\B/\beta))}{n} = kn + 4\B(1+\rho)(1 + \ln(\B/\beta)),
\end{equation}
with probability $1 - \beta$. Next, consider any individual message $y_i = (a_{i,1}, \ldots, a_{i,\tau})$ (as on Line~\ref{ln:if_had}). Notice that the set of $j$ such that $\{ a_{i,1}, \ldots, a_{i,\tau} \} \subset \MH_{2\B,j}$ can be described as follows: write $j = (j_1, \ldots, j_{\log 2\B}) \in \{0,1\}^{2\B}$ to denote the binary representation of $j$, and arrange the $\log 2\B$-bit binary representations of each of $a_{i,1}, \ldots, a_{i,\tau}$ to be the rows of a $\tau \times \log 2\B$ matrix $A \in \{0,1\}^{\tau \times \log 2\B}$. Then $\{ a_{i,1}, \ldots, a_{i,\tau} \} \subset \MH_{2\B,j}$ if and only if $Aj = 0$, where arithmetic is performed over $\mathbb{F}_2$. This follows since for binary representations $i = (i_1, \ldots, i_{\log 2\B}) \in \{0,1\}^{\log 2\B}$ and $j = (j_1, \ldots, j_{\log 2\B}) \in \{0,1\}^{\log 2\B}$, the $(i,j)$-element of $H_\B$ is $(-1)^{\sum_{t=1}^{\log 2\B} i_tj_t}$. Using Gaussian elimination one can enumerate the set of $j \in \{0,1\}^{\log 2\B}$ in the kernel of $A$ in time proportional to the sum of $O(\log^3 \B)$ and the number of $j$ in the kernel. Since the sum, over all messages $y_i$, of the number of such $j$ is bounded above by $(\ref{eq:bound_solutions})$ (with probability $1-\beta$), the total running time of this modification of $A^{\Had}$ becomes $O((\rho + k)n \log^3 \B + kn + \B\rho \log(\B/\beta))$.

For the last item of the theorem, the analyzer simply outputs the collection of all tuples $y_i = (a_{i,1}, \ldots, a_{i,\tau})$; to query the frequency of some $j \in [\B]$, we simply run the for loop on Line~\ref{ln:mforloop} of Algorithm~\ref{alg:hh_had}, together with the debiasing step of Line~\ref{ln:debias}.
% The last item of the theorem is immediate since $A^{\Had}$ simply returns the vector of all such frequency queries $\hat x_j$, $j \in [\B]$. 
\end{proof}

\subsection{Public-Coin Protocol with Small Query Time}\label{subsec:count-min}

In this subsection, we give a public-coin protocol for frequency estimation in the shuffled model with error $\poly\log{B}$ and communication per user $\poly(\log{B}, \log{n})$ bits. As discussed in Section~\ref{sec:overview_ub}, our protocol is based on combining the Count Min data structure~\cite{cormode2005improved} with a multi-message version of randomized response~\cite{warner1965randomized}. We start by giving a more detailed overview of the protocol.

\paragraph{Overview.}
On a high level, the presence of public randomness (which is assumed to be known to the analyzer) allows the parties to jointly sample random seeds for hash functions which they can use to compute and communicate (input-dependent) updates to a probabilistic data structure. The data structure that we will use is Count Min which we recall next. Assume that each of $n$ users holds an input from $[B]$ where $n \ll B$. We hash the universe $[B]$ into $s$ buckets where $s = O(n)$.%
~\footnote{It is possible to introduce a trade-off here: By increasing $s$ we can improve privacy and accuracy at the cost of requiring more communication. For simplicity we present our results for the case $s = O(n)$, minimizing communication, and discuss larger $s$ at the end of section~\ref{subsec:count-min}}
Then for each user, we increment the bucket to which its input hashes. This ensures that for every element of $[B]$, its hash bucket contains an overestimate of the number of users having that element as input. 
However, these bucket values are not enough to unambiguously recover the number of users holding any specific element of $[B]$---this is because on average, $B/s$ different elements hash to the same bucket. To overcome this, the Count Min data structure repeats the above idea $\tau = O(\log{B})$ times using independent hash functions.
Doing so ensures that for each element $j \in [B]$, it is the case that (i) no other element $j' \in [B]$ hashes to the same buckets as $j$ for all $\tau$ repetitions, and (ii) for at least one repetition, no element of $[B]$ that is held by a user (except possibly $j$ itself) hashes to the same bucket as $j$. 
To make the Count Min data structure differentially private, we use a multi-message version of randomized response~\cite{warner1965randomized}.
% \emph{privacy blanket} as was done by Balle et al.~\cite{balle_privacy_2019}.
Specifically, we ensure that sufficient independent noise is added to each bucket of each repetition of the Count Min data structure. This is done by letting each user independently, using its private randomness, increment every bucket with a very small probability. The noise stability of Count Min is used to ensure that the frequency estimates remain accurate after the multi-message noise addition. We further use the property that updates to this data structure can be performed using a logarithmic number of ``increments'' to entries in the sketch for two purposes: (i) to bound the privacy loss for a single user, and (ii) to obtain a communication-efficient implementation of the protocol.
The full description appears in Algorithm~\ref{alg:hh_CM}.

\begin{algorithm}[!ht]
\Fn{$R^{\CM}(n, \B, \tau, \gamma, \tablesize)$}{
\KwIn{Subset $\MS \subset [\B]$ specifying the user's input set\\ 
{\bf Parameters:} $n, \B, \tau, \tablesize \in \BN$ and $\gamma \in [0,1]$\\ 
\textbf{Public Randomness}: A random hash family $\{h_t:[B] \to [\tablesize], ~ \forall t \in [\tau]\}$}

\KwOut{A multiset $\MT \subset [\tau] \times [\tablesize]$}
\For{$j \in \MS$}{
\For{$t \in [\tau]$}{
Add the pair $(t, h_t(j))$ to $\MT$.
}
}
\For{$t \in [\tau]$}{
\For{$\ell \in [\tablesize]$}{
Sample $b_{t,\ell}$ from $\Ber(\gamma)$.\\
\If{$b_{t,\ell} = 1$}{
Add the pair $(t, \ell)$ to $\MT$.
}
}
}
\Return{$\MT$.}
}

\Fn{$A^{\CM}(n, B, \tau, s)$}{
\KwIn{Multiset $\{ y_1, \ldots, y_m \}$ containing outputs of local randomizers\\ 
{\bf Parameters:} $n, B, \tau, s \in \BN$\\ 
\textbf{Public Randomness}: A random hash family $\{h_t:[B] \to [\tablesize], ~ \forall t \in [\tau]\}$}
\KwOut{A noisy Count Min data structure $C: [\tau] \times [\tablesize] \to \BN$}
\For{$t \in [\tau]$}{
\For{$\ell \in [\tablesize]$}{
$C[t, \ell] = 0$.\\ 
}
}
\For{$j \in [m]$}{
$C[y_j] \gets C[y_j] + 1$.
}
\Return{$C$}
}

\Fn{$Q^{\CM}(n, B, \tau, s)$}{
\KwIn{Element $j \in [B]$\\ 
{\bf Parameters:} $n, B, \tau, s \in \BN$\\ 
\textbf{Public Randomness}: A random hash family $\{h_t:[B] \to [\tablesize], ~ \forall t \in [\tau]\}$}
\KwOut{A non-negative real number which is an estimate of the frequency of element $j$}
\Return{$\hat x_j := \max \left\{\min\{ C[t, h_t[j]] - \gamma n: ~ t \in [\tau]\}, 0 \right\}$}}
\caption{Local randomizer, analyzer and query for frequency estimation via Count Min.}
\label{alg:hh_CM}
\end{algorithm}

\paragraph{Analysis.}
We next show the accuracy, efficiency, and privacy guarantees of Algorithm~\ref{alg:hh_CM} which are summarized in the following theorem.
% By combining Lemmas~\ref{le:h_CM_acc} and~\ref{le:h_CM_priv}, we immediately obtain the following:
\begin{theorem}\label{thm:CM_main}
  There is a sufficiently large positive absolute constant $\zeta$ such that the following holds. Suppose $n, \B, k \in \BN$, and $0 \leq \ep, \delta,\beta \leq 1$. Consider the shuffled-model protocol $P^{\CM} = (R^{\CM}, S, A^{\CM})$ with $\tau = \log (2\B/\beta)$, $s = 2kn$, and
$$
\gamma =\frac 1n \cdot  \zeta \cdot \max \left\{ \log n, \frac{\log^2(\B/\beta) k^2 \log(\log( \B/\beta) k /\delta)}{\ep^2} \right\}.
$$
Then $P^{\CM}$ is $(\ep, \delta)$-differentially private (Definition~\ref{def:dp_shuffled}), each user sends $O(\gamma kn \log(\B/\beta))$ messages consisting of $O(\log n + \log \log \B/\beta)$ bits each with probability $1-\beta$, and for inputs $x_1, \ldots, x_n \in \{0,1\}^\B$ ($\| x_i\|_1 \leq k$), the estimates $\hat x_j$ produced by $Q^{\CM}$ satisfy:
\begin{equation}
   \label{eq:qcm_acc}
   \p \left[\forall j \in [\B] ~:~ \left|\hat x_j - \sum_{i=1}^n x_{i,j}\right| \leq O \left( \sqrt{\log\left( \frac{\B n }{\beta} \right) \cdot \left(\log\left( \frac{\B n }{\beta} \right) +  \frac{k^2 \log^2(\B/\beta) \log((\log \B/\beta) k /\delta)}{\ep^2}\right)} \right) \right]
     \geq 1 - \beta.
   \end{equation}
   Moreover, any frequency query can be answered in time $O(\log \B/\beta)$.
% (See (\ref{eq:qcm_complicated}) for an exact expression.)
\end{theorem}
Notice that by decreasing $\beta$ by at most a constant factor (and thus increasing the error bounds by at most a constant factor), we may ensure that $\tau = \log(2 \B/\beta)$ in the theorem statement is an integer. Note also that the additive error in (\ref{eq:qcm_acc}) is $\tilde O(k/\ep)$, where the $\tilde O(\cdot)$ hides factors logarithmic in $\B, n, k, 1/\delta, 1/\beta$. % We remark that another way to handle the case $k > 1$ 

Theorem~\ref{thm:CM_main} with $k=1$ directly implies Theorem~\ref{thm:hist_full_pub}. The next lemma is used to prove the accuracy of Algorithm~\ref{alg:hh_CM}.

\begin{restatable}[Accuracy of $P^{\CM}$]{lemma}{CMacc}
\label{le:h_CM_acc}
Let $n$, $\B$, and $\tau$ be positive integers, and $\gamma \in [0,1]$, $\xi \in [0,\sqrt{\gamma n}]$ be real parameters. Then the estimate $\hat x_j$ produced by $Q^{\CM}$ on input $j \in [B]$ and as an outcome of the shuffled-model protocol $P^{\CM} = (R^{\CM}, S, A^{\CM})$ with input $x_1, \ldots, x_n \in \{0,1\}^\B$ ($\| x_i\|_1 \leq k$) satisfies $\hat x_j \geq \sum_{i=1}^n x_{i,j}$ and
$$
\p \left[\left|\hat x_j - \sum_{i=1}^n x_{i,j}\right| \leq \xi \sqrt{\gamma n}\right] \geq 1- (kn/s)^{\tau} - 2^{\log(2s\tau) - \xi^2/3}.
$$
% In particular, for any $0 \leq \beta \leq 1$, for $s = 2n$, $\tau = \log n$, and $\gamma n = \zeta \log(\B n/\beta)$ for a sufficiently large constant $\zeta$, we have
% \begin{equation}
% \label{eq:cm_unionbound}
% \p \left[ \forall j \in [\B] : \hat x_j - \sum_{i=1}^n x_{i,j} \leq O(\log \B n / \beta) \right] \geq 1 - \beta.
% \end{equation}
\end{restatable}

\begin{proof}[Proof]
We consider the entries $\{ C[t, h_t[j]] ~\mid~ t \in [\tau]\}$ of the noisy Count Min data structure. We first consider the error due to the other inputs that are held by the users. Then we consider the error due to the noise blanket. We bound each of these two errors with high probability and then apply a union bound.

First, note that for any element $j \in [B]$, the probability that for every repetition index $t \in [\tau]$, some element $j' \in [B]$ held by one of the users (except possibly $j$ itself) satisfies $h_t(j') = h_t(j)$, is at most $(kn/s)^{\tau}$. % 
As in the original analysis of Count Min~\cite{cormode2005improved}, this holds even if the hash functions $h_t$ are sampled from a family of pairwise independent hash functions.

It remains to show that with probability at least $1-2^{1 + \log(s\tau) - \xi^2/3}$, the absolute value of the deviation of the blanket noise in each of these entries from its expectation $\gamma \cdot n$ is at most $O\left(\sqrt{ \gamma n} \right)$. By a union bound over all $s \tau$ pairs of bucket indices and repetition indices, it is enough to show that for each $t \in [\tau]$ and each $\ell \in [s]$, with probability at least $1-2^{1-\xi^2/3}$, the absolute value of the blanket noise in $C[t, h_t[j]]$ is at most $\xi\sqrt{\gamma n}$. This follows from the fact that the blanket noise in the entry $C[t, h_t[j]]$ is the sum of $n$ independent $\Ber(\gamma)$ random variables (one contributed by each user). The bound now follows from the multiplicative Chernoff bound.

Finally, by a union bound the overall error is at most $O\left( \gamma n \right)$ with probability at least $1-(kn/s)^{\tau}- 2^{\Theta(\log(s\tau) - \gamma n)}$. % (\ref{eq:cm_unionbound}) follows by a union bound over all $j \in [\B]$. 
\end{proof}

% \begin{lemma}[Accuracy of $P^{\CM}$]
% \label{le:h_CM_acc}
% Let $n$, $\B$, and $\tau$ be positive integers and let $\gamma \in [0,1]$ be a real parameter. Then the estimate $\hat x_j$ produced by $Q^{\CM}$ on input $j \in [B]$ and as an outcome of the shuffled-model protocol $P^{\CM} = (R^{\CM}, S, A^{\CM})$ with input $x_1, \ldots, x_n \in \{0,1\}^\B$ ($\| x_i\|_1 \leq k$) satisfies $\hat x_j \geq \sum_{i=1}^n x_{i,j}$ and
% $$
% \p \left[\hat x_j - \sum_{i=1}^n x_{i,j} \leq O\left( \gamma n \right)\right] \geq 1- (n/s)^{\tau} - 2^{\Theta(\log(s \tau) - \gamma n)}.
% $$
% % In particular, for any $0 \leq \beta \leq 1$, for $s = 2n$, $\tau = \log n$, and $\gamma n = \zeta \log(\B n/\beta)$ for a sufficiently large constant $\zeta$, we have
% % \begin{equation}
% % \label{eq:cm_unionbound}
% % \p \left[ \forall j \in [\B] : \hat x_j - \sum_{i=1}^n x_{i,j} \leq O(\log \B n / \beta) \right] \geq 1 - \beta.
% % \end{equation}
% \end{lemma}

% Note that the protocol based on count-min never underestimates the count of an element.  (However, this will not be the case for the protocol based on the Hadamard response (Section~\ref{subsec:Had}).)
By removing the subtraction of $\gamma n$ on the final line of Algorithm~\ref{alg:hh_CM}, we can guarantee that the estimate returned by the Count Min sketch is never less than the true count of an element. This would lead to, however, an expected error of $O(\gamma n)$ as opposed to $O(\sqrt{\gamma n})$ in Lemma~\ref{le:h_CM_acc}. 
The next lemma shows the efficiency of Algorithm~\ref{alg:hh_CM}.

\begin{restatable}[Efficiency of $P^{\CM}$]{lemma}{CMeff}
\label{le:h_CM_eff}
Let $n, \B, \tau, \tablesize$ be positive integers and $\gamma \in [0,1]$. Then,
\begin{enumerate}
\item\label{it:communication} With probability at least $1- n \cdot 2^{-\Theta(\gamma s \tau)}$, the output of $R^{\CM}(n, \B, \tau, \gamma, \tablesize)$ on input $\MS$ consists of at most $|\MS| + O(\gamma s \tau)$ messages each consisting of $\lceil \log_2(\tau) \rceil + \lceil \log_2(s) \rceil$ bits.
\item\label{it:efficiency_analyzer} The runtime of the analyzer $A^{\CM}(n, B, \tau, s)$ on input $\{ y_1, \ldots, y_m \}$ is $O(\tau s + m)$ and the space of the data structure that it outputs is $O(\tau s \log m)$ bits.
\item\label{it:efficiency_query} The runtime of any query $Q^{\CM}(n, B, \tau, s)$ is $O(\tau)$.
\end{enumerate}
\end{restatable}

\begin{proof}[Proof]
The second and third parts follow immediately from the operation of Algorithm~\ref{alg:hh_CM}. To prove the first part, note that each user sends $|\MS|$ messages corresponding to its inputs along with a number of ``blanket noise'' terms. This number is a random variable drawn from the binomial distribution $\Bin(\tau s, \gamma)$. Moreover, each of these messages is a pair consisting of a repetition index (belonging to $[\tau]$) and a bucket index (belonging to $[s]$). The proof now follows from the multiplicative Chernoff bound along with a union bound over all $n$ users.
\end{proof}

The next lemma establishes the privacy of Algorithm~\ref{alg:hh_CM}.
\begin{restatable}[Privacy of $P^{\CM}$]{lemma}{CMprivacy}
\label{le:h_CM_priv}
Let $n$ and $\B$ be positive integers. Then, for  $\gamma n \geq  \frac{90 k^2 \tau^2 \ln(2\tau k/\delta)}{\ep^2}$, the algorithm $S \circ R^{\CM}(n, \B, \tau, \gamma, s)$ is $(\ep, \delta)$-differentially private. 
\end{restatable}

To prove Lemma~\ref{le:h_CM_priv}, we need some general tools linking sensitivity of vector-valued functions, smoothness of distributions and approximate differential privacy---these are given next in Section~\ref{subsec:approx_DP_insensitive_functions}. The proof of Lemma~\ref{le:h_CM_priv} is deferred to Section~\ref{subsec:privacy_CM_proof}. We are now ready to prove Theorem~\ref{thm:CM_main}.

\begin{proof}[Proof of Theorem~\ref{thm:CM_main}]
  Privacy is an immediate consequence of Lemma~\ref{le:h_CM_priv}. To establish accuracy (i.e., (\ref{eq:qcm_acc})), note first that Lemma~\ref{le:h_CM_acc} guarantees that for any $j \in [\B]$  and any $\xi \in [0,\sqrt{\gamma n}]$, $\left|\hat x_j - \sum_{i=1}^n x_{i,j}\right| \leq \xi \sqrt{\gamma n}$ with probability at least $1 - (kn/s)^\tau - 2^{\log(2s\tau) - \xi^2/3}$. We now choose $\xi = \sqrt{3 \cdot \log \left( \frac{4\B s\tau}{\beta}\right)}$; this ensures that $2^{\log(2s\tau) - \xi^2/3} \leq \beta/(2\B)$. Moreover, we have that $\xi \leq \sqrt{\gamma n}$ by our choice of $\gamma$ in the theorem statement.

  It now follows from a union bound over all $j \in [\B]$ that
  \begin{equation}
\label{eq:qcm_complicated}
\p \left[\forall j \in [\B] ~:~ \left|\hat x_j - \sum_{i=1}^n x_{i,j}\right| \leq O \left( \sqrt{\log\left( \frac{\B n }{\beta} \right) \cdot \left(\log\left( \frac{\B n }{\beta} \right) +  \frac{k^2 \log^2(\B/\beta) \log(\log(\B/\beta) k /\delta)}{\ep^2}\right)} \right) \right] \geq 1 - \beta.
\end{equation}
Here we have used that $k \leq \B$.
% Now (\ref{eq:qcm_acc}) follows from a union bound over all $j \in [\B]$ as well as the fact that for sufficiently large $\zeta$,
% $$
% \gamma n \geq \zeta/2 \cdot \left( k^2 \tau + \log(kn) \right) \geq  \log(2kn \cdot \log(2\B/\beta)) + 1/\zeta_0 \cdot \log(2\B/\beta) = \log(s\tau) + 1/\zeta_0 \cdot \log(2\B/\beta),
% $$
% so that $2^{\zeta_0 (\log(s\tau) - \gamma n)} \leq \beta/(2\B)$. 
\end{proof}

% \subsection{Approximate Differential Privacy for Insensitive Functions}\label{subsec:approx_DP_insensitive_functions}

\paragraph{Improving error and privacy by increasing communication.}
Theorem~\ref{thm:CM_main} bounds the error of Algorithm~\ref{alg:hh_CM} with parameters $s=O(n)$ and $\tau = O(\log B/\beta)$.
For constant $\eta > 0$ it is interesting to consider the parameterization $s = O(n (n/\beta)^\eta)$ and $\tau = O(1/\eta)$.
By Lemma~\ref{le:h_CM_priv} differential privacy can be ensured in this setting with $\gamma n = O(k^2 \log(k/\delta)/\ep^2)$.
The randomizer of Algorithm~\ref{alg:hh_CM} sends a number of blanket messages that is $O(\gamma s)$ in expectation, i.e., $O((n/\beta)^\eta k^2 \log(k/\delta)/\ep^2)$.
An argument mirroring the proof of Lemma~\ref{le:h_CM_acc} shows that the \emph{pointwise} error of an estimate $\hat x_j$ is bounded by $\sqrt{\gamma n \log(1/\beta)} = O(k \sqrt{\log(k/\delta)\log(1/\beta)}/\ep)$ with probability $1-\beta$.
Thus, error as well as communication is independent of the domain size $B$.
To get a bound that is directly comparable to Theorem~\ref{thm:CM_main}, holding for all queries in $[B]$, we may reduce the pointwise error probability $\beta$ by a factor $B$ and apply a union bound, resulting in communication $O((B n/\beta)^\eta k^2 \log(k/\delta)/\ep^2)$ and error $O(k \sqrt{\log(k/\delta)\log(B/\beta)}/\ep)$.
Query time is $\tau = O(1)$.
This strictly improves the results that follow from~\cite{cheu_distributed_2018,DBLP:journals/corr/abs-1906-09116,ghazi2019scalable} (see Table~\ref{table:frequency_estimation_results}).
Very recently, Balcer and Cheu~\cite{balcer2019separating} showed a different trade-off in the case where the number of messages is very large:
$B+1$ messages of size $O(\log B)$ each with error $O(\log(1/\delta)/\ep^2 + \sqrt{\log(1/\delta)\log(n/\beta)}/\ep)$, which is independent of $B$.

\subsection{Useful Tools}\label{subsec:approx_DP_insensitive_functions}

If the blanket noise added to each bucket of the Count Min sketch were distributed as i.i.d.~Gaussian or Laplacian random variables, the proof of Lemma~\ref{le:h_CM_priv} would follow immediately from known results. Due to the discrete and distributed nature of the problem, we are forced to instead use Binomial blanket noise. To prove Lemma~\ref{le:h_CM_priv}, we will need some general tools linking approximate differential privacy to smoothness of distributions (and in particular the Binomial distribution); these tools are essentially known, but due to the lack of a suitable reference we prove all the prerequisite results.

\begin{defn}[Sensitivity]\label{def:sensitivity}
The {\it $\ell_1$-sensitivity} (or {\it sensitivity}, for short) of $f : \MX^n \ra \BZ^m$ is given by:
$$
\Delta(f) = \max_{X \sim X'} \| f(X) - f(X') \|_1.
$$
\end{defn}
It is well-known~\cite{dwork2006calibrating} that the mechanism given by adding independent Laplacian noise with variance $2 \Delta(f)^2 / \ep^2$ to each coordinate of $f(X)$ is $(\ep, 0)$-differentially private.  Laplace noise, however, is unbounded in both the positive and negative directions, and this causes issues in the shuffled model (roughly speaking, it would require each party to send infinitely many messages).  In our setting we will need to ensure that the noise added to each coordinate is bounded, so to achieve differential privacy we will not be able to add Laplacian noise. As a result we will only be able to obtain $(\ep, \delta)$-differential privacy for $\delta > 0$. We specify next the types of noise that we will use instead of Laplacian noise. 
\begin{defn}[Smooth distributions]\label{def:smooth_dist}
 Suppose $\MD$ is a distribution supported on $\BZ$. For $k \in \BN$, $\ep \geq 0$ and $\delta \in [0,1]$, we say that $\MD$ is {\it $(\ep, \delta, k)$-smooth} if for all $-k \leq k' \leq k$,
$$
\p_{Y \sim \MD} \left[ \frac{\p_{Y' \sim \MD}[Y' = Y]}{\p_{Y' \sim \MD}[Y' = Y + k']} \geq e^{|k'|\ep} \right] \leq \delta.
% \quad{\mbox{and}}\quad
% \p_{Y \sim \MD} \left[ \frac{\p_{Y' \sim \MD}[Y' = Y]}{\p_{Y' \sim \MD}[Y' = Y - 1]} \geq e^\ep \right] \leq \delta.
$$
\end{defn}

\begin{defn}[Incremental functions]
\label{def:incremental}
Suppose $k \in \BN$. We define $f : \MX^n \ra \BZ^m$ to be {\it $k$-incremental} if for all neighboring datasets $X \sim X'$, $\| f(X) - f(X') \|_\infty \leq k$. %  and for all $j \in [m]$, $|f(X)_j - f(X')_j| \leq k$.
\end{defn}

% and so will only be able to obtain $(\ep, \delta)$-differential privacy for some $\delta > 0$. % only be able to obtain $(\ep, \delta)$-differential privacy for some $\delta > 0$ -- for this purpose
% The following lemma generalizes the privacy guarantee of the Laplace mechanism to the setting of $(\ep, \delta)$-differential privacy with $\delta > 0$ when the function $f$ is $k$-incremental.
The following lemma formalizes the types of noise we can add to $f(X)$ to obtain such a privacy guarantee. Its proof appears in Appendix~\ref{sec:CM_app}.
\begin{lemma}
\label{lem:lap_gen}
Suppose $f : \MX^n \ra \BZ^m$ is $k$-incremental (Definition~\ref{def:incremental}) and $\Delta(f) = \Delta$. Suppose $\MD$ is a distribution supported on $\BZ$ that is $(\ep, \delta, k)$-smooth. Then the mechanism
$$
X \mapsto f(X) + (Y_1, \ldots, Y_m),
$$
where $Y_1, \ldots, Y_m \sim \MD$, i.i.d., is $(\ep', \delta')$-differentially private, where $\ep' = \ep \cdot \Delta, \delta' = \delta \cdot \Delta$.
\end{lemma}

In order to prove Lemma~\ref{le:h_CM_priv}, we will also use the following statement about the smoothness of the binomial distribution (that we will invoke with a small value of the head probability $\gamma$). Its proof appears in Appendix~\ref{sec:CM_app}.

\begin{lemma}[Smoothness of $\Bin(n, \gamma)$]\label{lem:bin_smooth_small_gamma}
Let $n \in \BN$, $\gamma \in [0,1/2]$, $0 \leq \alpha \leq 1$, and $k \leq \alpha \gamma n / 2$. Then the distribution $\Bin(n, \gamma)$ is $(\ep ,\delta, k)$-smooth with $\epsilon = \ln((1+\alpha)/(1-\alpha))$ and $\delta = e^{-\frac{\alpha^2 \gamma n}{8}} + e^{-\frac{\alpha^2 \gamma n}{8+2\alpha}}$.
\end{lemma}

\subsection{Privacy Proof}\label{subsec:privacy_CM_proof}

We are now ready to prove Lemma~\ref{le:h_CM_priv} using the results on $k$-incremental functions from the previous section, thereby establishing the privacy of Algorithm~\ref{alg:hh_CM}. An alternative approach to establishing privacy of Algorithm \ref{alg:hh_CM} is to first do so for the case $k = 1$ and then apply the advanced composition lemma \cite{DworkRothBook}. However, doing so leads to an error bound that incurs at least an additional $\sqrt{k}$ factor since one has to make $\ep$ smaller by a factor of $\sqrt{k}$. In order to prove Lemma~\ref{le:h_CM_priv}, we could use Theorem 1 of \cite{agarwal2018cpsgd} instead of our Lemma~\ref{lem:bin_smooth_small_gamma} but their result would give worse bounds for $k>1$.

\begin{proof}[Proof of Lemma~\ref{le:h_CM_priv}]
Fix $\ep, \delta$. Notice that $S \circ R^{\CM}(n, \B, \tau, \gamma, s)$ can be obtained as a post-processing of the noisy Count Min data structure $C : [\tau] \times [s] \ra \BN$ in Algorithm~\ref{alg:hh_CM}, so it suffices to show that the algorithm bringing the players' inputs to this Count Min data structure is $(\ep, \delta)$-differentially private. Consider first the Count Min data structure $\tilde C : [\tau] \times [s] \ra \BN$ with no noise, so that $\tilde C[t, \ell]$ measures the number of inputs $x$ inside some user's set $\MS_i$ such that $h_t(x) = \ell$. We next note that the function mapping the users' inputs $(\MS_1, \ldots, \MS_n)$ to $\tilde C$ has sensitivity (in terms of Definition~\ref{def:sensitivity}) at most $k \tau$ and is $k$-incremental (in terms of Definition~\ref{def:incremental}). Moreover, Lemma~\ref{lem:bin_smooth_small_gamma} (with $\alpha = \ep/(3\tau k)$) implies that the binomial distribution $\Bin(n, \gamma)$ is $(\epsilon/(\tau k), \delta/(\tau k), k)$-smooth (in terms of Definition~\ref{def:smooth_dist}) as long as $\delta \geq 2\tau k e^{-\frac{\ep^2 \gamma n}{90 \tau^2 k^2}}$ and $k \leq \ep \gamma n / (6 \tau k)$. In particular, we need
$$
\gamma n \geq  \frac{90 \tau^2 k^2 \ln(2\tau k/\delta)}{\ep^2}. % \max \left\{ 6k^2 \tau/\ep, \frac{90 \ln(2\tau k/\delta)}{\ep^2} \right\}.
$$
By construction in Algorithm~\ref{alg:hh_CM}, $C[t,s] = \tilde C[t,s] + \Bin(n, \gamma)$, where the binomial random variables are independent for each $t,s$. Applying Lemma~\ref{lem:lap_gen}, we get that the Count Min data structure is $(\ep, \delta)$-differentially private (with respect to Definition~\ref{def:dp_shuffled}).
\end{proof}

\section{Multi-Message Protocols for Range Counting Queries}\label{sec:range-queries}
\label{sec:rq}

\begin{table}[t]
    \centering
    \bgroup
    \def\arraystretch{1.5}
    \footnotesize
    \centerline{
    \begin{tabular}{|c|c|c|c|c|}
        \hline
        {\bf Problem} & {\bf \thead{Messages\\ per user}} & {\bf \thead{Message size\\ in bits}} & {\bf \thead{Error}} & {\bf Query time}\\
        \hline
        \hline
        \thead{$d$-dimensional\\ range counting (public) \\ Theorem \ref{thm:cm_rq_final}} & $\frac{\log^{3d+3}(\B)\log \frac{1}{\delta}}{\ep^2}$ & $\log n + \log \log \B$ & $\frac{\log^{2d + 3/2}(\B) \log \frac{1}{\delta}}{\ep}$ & $\log^{d+1} \B$ \\
        \hline
        \thead{$d$-dimensional\\ range counting (private) \\ Theorem \ref{thm:had_rq_final}} & $\frac{\log^{2d}(\B)  \log \frac{1}{\ep\delta}}{\ep^2}$ & $\log(n) \log \B$ & $\frac{\log^{2d + 1/2}(\B)  \log\tfrac{1}{\ep \delta}}{\ep}$  & $\frac{n \log^{3d+2}(\B) \log\tfrac{1}{\ep\delta}}{\ep^2}$ \\
        \hline
    \end{tabular}}
    \egroup
    \caption{Overview of results on differentially private range counting in the shuffled model. The query time stated is the additional time to answer a query, assuming a preprocessing of the output of the shuffler that takes time linear in its length. Note that frequencies and counts are not normalized, i.e., they are integers in~$\{0,\dots,n\}$. For simplicity, constant factors are suppressed, the bounds are stated for error probability $\beta = B^{-O(1)}$, and the following are assumed: dimension $d$ is a constant, $n$ is bounded above by $\B$, and $\delta < 1/\log \B$.
    }
    \label{tab:results_range_queries}
\end{table}

We recall the definition of range queries.  Let $\MX  = [\B]$ and consider a dataset $X = (x_1, \ldots, x_n) \in [\B]^n$.  Notice that a statistical query may be specified by a vector $w \in \BR^\B$, and the answer to this statistical query on the dataset $X$ is given by $\langle w, \hist(X) \rangle$. For all queries $w$ we consider, we will in fact have $w \in \{0,1\}^\B$, and thus $w$ specifies a {\it counting query}. Here $\langle \cdot, \cdot \rangle$ denotes the Euclidean inner product; throughout the paper, we slightly abuse notation and allow an inner product to be taken of a row vector and a column vector. A \emph{1-dimensional range query} $[j, j']$, where $1 \leq j \leq j' \leq \B$, is a counting query such that $w_j = w_{j+1} = \cdots = w_{j'} = 1$, and all other entries of $w$ are 0.  For \emph{$d$-dimensional range queries}, the elements of $[\B]$ will map to points on a $d$-dimensional grid, and a certain subset of vectors $w \in \{0,1\}^\B$ represent the $d$-dimensional range queries. In this section, we use the frequency oracle protocols in Section~\ref{sec:hh} to derive protocols for computing counting queries with per-user communication $\poly\log(\B)$ and additive error $\poly\log(\max\{n,\B\})$. 

In Section~\ref{sec:matrix_randomizer}, we adapt the matrix mechanism of~\cite{li2010optimizing,li2012adaptive} to use the frequency oracle protocols of Section~\ref{sec:hh} as a black-box for computation of counting queries, which include range queries as a special case. In Section~\ref{sec:1d_rq}, we instantiate this technique for the special case of 1-dimensional range queries, and in Section~\ref{sec:multid_rq} we consider the case of multi-dimensional range queries. In Section~\ref{sec:rq_collect} we collect the results from Sections~\ref{sec:matrix_randomizer} through~\ref{sec:multid_rq} to formally state our guarantees on range query computation in the shuffled model, as well as the application to $M$-estimation of the median, as mentioned in the Introduction.

\subsection{Frequency Oracle}

We now describe a basic data primitive that encapsulates the results in Section~\ref{sec:freq_oracle_heavy_hitters} and that we will use extensively in this section.  Fix positive integers $\B$ and $k \leq \B$ as well as positive real numbers $\ha$ and $\hb$. 
For each $v \in [\B]$, let $e_v \in \{0,1\}^\B$ be the unit vector with $(e_v)_j = 1$ if $j = v$, else $(e_v)_j = 0$.
In the {\it $(\ha,\hb,\hk)$-frequency oracle} problem~\cite{hsu_hh,bassily2015local},
each user $i \in [n]$ holds a subset $\MS_i \subset [\B]$ of size at most $k$. Equivalently, user $i$ holds the sum of the unit vectors $e_v$ corresponding to the elements $v$ of $\MS_i$, i.e., the vector $x_i \in \{0,1\}^\B$ such that $(x_i)_j = 1$ if and only if $j \in \MS_i$. Note that $\| x_i \|_1 \leq k$ for all $i$. At times we will restrict ourselves to the case that $k = 1$; in such cases we will often use $x_i$ to denote the single element $x_i \in [\B]$ held by user $i$, and write $e_{x_i} \in \{0,1\}^\B$ for the corresponding unit vector. % For simplicity we focus on the regime where $n \leq \B$, i.e., the universe of possible items is very large. (Our results extend with minor modifications to the general case.)

The goal is to design a (possibly randomized) data structure $\FO$ and a deterministic algorithm $\MA$ ({\it frequency oracle}) that takes as input the data structure $\FO$ and an index $j \in [\B]$, and outputs in time $T$ an estimate that, with high probability, is within an additive $\ha$ from $\sum_{i=1}^n(x_i)_j$. Formally:
\begin{defn}[$(\ha, \hb, \hk)$-frequency oracle]
\label{def:fo}
A protocol with inputs $x_1, \ldots, x_n \in \{0,1\}^\B$ computes an \emph{$(\ha, \hb, \hk)$-frequency oracle} if it outputs a pair $(\FO, \MA)$ such that for all datasets $(x_1, \ldots, x_n)$ with $\| x_i \|_1 \leq k$ for $i \in [n]$,
$$
\p\left[\forall j \in [\B] : \left| \MA(\FO, j) - \sum_{i=1}^n (x_i)_j \right| \leq \ha \right] \geq 1-\beta.
$$
The probability in the above expression is over the randomness in creating the data structure $\FO$. % and in the algorithm $\MA$.
\end{defn}
Note that given such a frequency oracle, one can recover the $(2\ha)$-{\it heavy hitters}, namely those $j$ such that $\sum_{i=1}^n (x_i)_j \geq 2\ha$, in time $O(T \cdot \B)$, by querying $\MA(\FO, 1), \ldots, \MA(\FO, \B)$ (for a more efficient reduction see Appendix~\ref{app:hh_reduction}). 
% $\hat x \in \BR^\B$ that is an additive approximation of $\ha$ to $\sum_{i=1}^n x_i$, i.e., $\left\| \hat x - \sum_{i=1}^n x_i \right\|_\infty \leq \ha$. Such a 
%\footnote{\label{fn:server_decoding_time} In this paper we do not focus on reducing the server decoding time for heavy hitters below $O(TB)$. Nevertheless, we point out that the protocol based on the count-min sketch that we give in Section~\ref{subsec:count-min} can be ensured to have server decoding time $\tilde{O}(n)$ (for recovering all heavy hitters). This is tight up to logarithmic factors and it follows by combining our protocol with the ``prefix tree'' idea of Bassily et al.~\cite{bassily2017practical}. For more details, see the discussion at the end of Section~\ref{subsec:count-min}.}

% It is customary in the literature to use \emph{normalized} frequencies, i.e., divided by $n$, the number of inputs.
% While this does not change the problem, since error bounds can be similarly scaled, it will be convenient to adopt this convention in the proofs of our lower bounds.

\subsection{Reduction to Private Frequency Oracle via the Matrix Mechanism}
\label{sec:matrix_randomizer}

Our protocol for computing range queries is a special case of a more general protocol, which is in turn inspired by the {\it matrix mechanism} of~\cite{li2010optimizing,li2012adaptive}. We begin by introducing this more general protocol and explaining how it allows us to reduce the problem of computing range queries in the shuffled model to that of computing a frequency oracle in the shuffled model.

Finally, for a matrix $\M \in \BR^\B \times \BR^\B$, define the {\it sensitivity} of $\M$ as follows:
\begin{defn}[Matrix sensitivity,~\cite{li2010optimizing}]
\label{def:matrix_sensitivity}
For a matrix $\M$, let the {\it sensitivity} of $\M$, denoted $\Delta_\M$, be the maximum $\ell_1$ norm of a column of $\M$.
\end{defn}
For any column vector $y \in \BR^\B$, $\Delta_\M$ measures the maximum $\ell_1$ change in $\M y$ if a single element of $y$ changes by $1$. 
The matrix mechanism, introduced by Li et al.~\cite{li2010optimizing,li2012adaptive} in the central model of DP, allows one to release answers to a given set of counting queries in a private manner. It is parametrized by an invertible matrix $\M$, and given input $X$, releases the following noisy perturbation of $\hist(X)$: 
\begin{equation}
\label{eq:matrix_mech}
\hist(X) + \Delta_\M \cdot \M^{-1} z,
\end{equation}
where $z \in \BR^\B$ is a random vector whose components are distributed i.i.d.~according to some distribution calibrated to the privacy parameters $\ep, \delta$. The response to a counting query $w \in \BR^\B$ is then given by $\langle w, \hist(X) + \Delta_\M \cdot \M^{-1} z \rangle$. The intuition behind the privacy of (\ref{eq:matrix_mech}) is as follows: (\ref{eq:matrix_mech}) can be obtained as a post-processing of the mechanism $X \mapsto \M( \hist(X)) + \Delta_\M\cdot z$, namely via multiplication by $M^{-1}$. If we choose, for instance, each $z_i$ to be an independent Laplacian of variance $2/\ep$, then the algorithm $X \mapsto \M(\hist(X)) + \Delta_\M \cdot z$ is simply the Laplace mechanism, which is $(\ep, 0)$-differentially private~\cite{dwork2006calibrating}. 

% Given a query $w$, the matrix mechanism~\cite{li2010optimizing,li2012adaptive} takes the dataset $X$ to the output
% \begin{equation}
% \label{eq:matrix_mech-old}
% \left \langle w, y + \Delta_\M \cdot \M^{-1}z \right\rangle = \langle w, y \rangle + \Delta_\M \cdot \langle w, \M^{-1} z \rangle,
% \end{equation}
% where $z \in \BR^\B$ is a random vector whose components are distributed i.i.d.~according to some distribution calibrated to the privacy parameters $\ep, \delta$. Notice that if the entries of $z$ are i.i.d.~Laplacian and $\M = I_\B$, then the above is simply the output of the Laplace mechanism. In~\cite{li2010optimizing,li2012adaptive} it was shown that when the query $w$ is known to belong to a certain set of queries, such as that corresponding to all range queries, other choices of $\M$ (combined with Laplacian noise $z$) can lead to significantly improved privacy and accuracy over the Laplace mechanism. 

% We modify the matrix mechanism in a few respects: first, due to the limitations of the shuffled model, we will not be able to add Laplacian noise $z$, settling instead for a noise distribution that guarantees $(\ep, \delta)$-differential privacy for some $\delta > 0$. Second, 

In our modification of the matrix mechanism, the parties will send data that allows the analyzer to directly compute the ``pre-processed input'' $\M(\hist(X)) + \Delta_\M \cdot z$. Moreover, due to limitations of the shuffled model and to reduce communication, the distribution of the noise $z$ will be different from what has been previously used~\cite{li2010optimizing,li2012adaptive}. For our application, we will require $\M$ to satisfy the following properties: % though it will satisfy that $\|  z\|_\infty \leq \log \log \B$ with high probability. We will require the following properties of
% Notice that the matrix mechanism (\ref{eq:matrix_mech}) can be obtained as a post-processing of the mechanism $X \mapsto \M\hist(X) + \Delta_\M \cdot z$, namely via taking the inner product of $\M\hist(X) + \Delta_\M \cdot z$ with $w\M^{-1}$ (where $w$ is viewed as a row vector here). Now suppose that $\M$ satisfies the following properties: 
\begin{enumerate}[label=\textnormal{\arabic*}]
\item[(1)] For any counting query $w$ corresponding to a $d$-dimensional range query, $w\M^{-1}$ has at most $\poly\log(\B)$ nonzero entries, and all of those nonzero entries are bounded in absolute value by some $c > 0$. (Here $w \in \{0,1\}^\B$ is viewed as a row vector.) \label{it:matrix_accuracy}
% Notice that the output of the matrix mechanism (\ref{eq:matrix_mech}) is a linear function of the histogram $y = \hist(X)$. 
\item[(2)]$\Delta_\M \leq \poly\log(\B)$.\label{it:matrix_sensitivity}
\end{enumerate}
By property (2) above and the fact that all entries of $\M$ are in $\{ 0,1\}$, (approximate) computation of the vector $\M(\hist(X))$ can be viewed as an instance of the frequency oracle problem where user $i \in [n]$ holds the $\leq \poly\log(\B)$ nonzero entries of the vector $\M(\hist(x_i))$. This follows since $\M(\hist(x_i))$ is the $x_i$th column of $\M$, $\Delta_\M \leq \poly\log(\B)$, and $\hist(X) = \sum_{i=1}^n \hist(x_i)$. Moreover, suppose there is some choice of local randomizer and analyzer (such as those in Section~\ref{sec:hh}) that approximately solve the frequency oracle problem, i.e., compute an approximation $\hat y$ of $\M(\hist(X))$ up to an additive error of $\poly \log \B$, in a differentially private manner. % In particular, $\|\hat y -\M \hist(X)\|_\infty \leq \poly \log n$. 
Since $w\M^{-1}$ has at most $\poly \log(\B)$ nonzero entries, each of magnitude at most $c$, it follows that \begin{equation}
\label{eq:analyzer_output}
\langle w \M^{-1},\hat y\rangle
\end{equation}
approximates the counting query $\langle w, \hist(X) \rangle$ up to an additive error of $c \cdot \poly \log (\B)$.

\begin{algorithm}[ht]
\Fn{$R^{\matrix}(n, \B, \M, R^{\FO})$}{
    \KwIn{$x \in [\B]$, parameters $n, \B \in \BN, \M \in \{0,1\}^{\B \times \B}, R^{\FO} : \{0,1\}^\B \ra \MT^*$}
    \KwOut{Multiset $\MS \subset \MT$, where $\MT$ is the output set of $R^{\FO}$}
    Let $\MA_x \gets \{ j \in [\B] : \M_{jx} \neq 0 \}$ \\
    \tcp{$\MA_x$ is the set of nonzero entries of the $x$th column of $\M$}
    \Return{$R^{\FO}(\MA_x)$}
}
\caption{Local randomizer for matrix mechanism}
\label{alg:matrix_rand}
\end{algorithm}
Perhaps surprisingly, for any constant $d \geq 1$, we will be able to find a matrix $\M$ that satisfies properties (1) and (2) above for $d$-dimensional range queries with $c = 1$. This leads to the claimed $\poly \log(\B)$ error for computation of $d$-dimensional range queries, as follows: the local randomizer $R^{\matrix}$ (Algorithm~\ref{alg:matrix_rand}) is parametrized by integers $n, \B \in \BN$, a matrix $\M \in \{ 0, 1\}^{B \times \B}$, and a local randomizer $R^{\FO} : [\B] \ra \MT^*$ that can be used in a shuffled model protocol that computes a frequency oracle. (Here $\MT$ is an arbitrary set, and $R^{\FO}$ computes a sequence of messages in $\MT$.) Given input $x \in [\B]$, $R^{\matrix}$ returns the output of $R^{\FO}$ when given as input the set of nonzero entries of the $x$th column of $\M$. The corresponding analyzer $A^{\matrix}$ (Algorithm~\ref{alg:matrix_analyzer}) is parametrized by integers $n, B \in \BN$, a matrix $\M  \in \{0,1\}^{\B \times \B}$, and an analyzer $A^{\FO}$ for computation of a frequency oracle in the shuffled model. Given a multiset $\MS$ consisting of the shuffled messages output by individual randomizers $R^{\matrix}$, it returns (\ref{eq:analyzer_output}), namely the inner product of $w\M^{-1}$ and the output of $A^{\FO}$ when given $\MS$ as input.
\begin{algorithm}[ht]
\Fn{$A^{\matrix}(n, \B, \M, A^{\FO})$}{
\KwIn{Multiset $\MS \subset [\B]$ consisting of the shuffled reports;\\
Parameters $\MW \subset \{0,1\}^\B$ specifying a set of counting queries, $n, \B \in \BN, \M \in \{0,1\}^{\B \times \B}$, analyzer $A^{\FO}$ for frequency oracle computation}
\KwOut{Map associating each $w \in \MW$ to $f_w \in [0,1]$, specifying an estimate for each counting query $w$}
Let $(\FO, \MA) \gets A^{\FO}(\MS)$ \\
\tcp{Frequency oracle output by $A^{\FO}$ (see Definition~\ref{def:fo})}
% Let $\hat y\gets A^{\FO}(\MS)$\tcp{$\hat y \in \BR^{\B}$}
\Return{Map associating each $w \in \MW$ to $f_w := \sum_{j \in [\B] : (w\M^{-1})_j \neq 0} (w\M^{-1})_j \cdot \MA(\FO, j)$} \\
\tcp{Let $\hat y \in \BR^\B$ be such that $\hat y_j = \MA(\FO, j)$; then this returns the map associating $w \in \MW$ to $\langle w\M^{-1}, \hat y \rangle$.}\label{ln:ma_return}
}
\caption{Analyzer for matrix mechanism}\label{alg:matrix_analyzer}
\end{algorithm}
To complete the construction of a protocol for range query computation in the shuffled model, it remains to find a matrix $\M$ satisfying properties (1) and (2) above. We will do so in Sections~\ref{sec:1d_rq} and~\ref{sec:multid_rq}. First we state here the privacy and accuracy guarantees of the shuffled protocol $P^{\matrix} = (R^{\matrix}, S, A^{\matrix})$.
\begin{theorem}[Privacy of $P^{\matrix}$]
\label{thm:pmatrix_privacy}
Suppose $R^{\FO}$ is a local randomizer for computation of an $(\ha, \hb, \hk)$-frequency oracle with $n$ users and universe size $\B$, which satisfies $(\ep, \delta)$-differential privacy in the shuffled model. Suppose $\M \in \{0,1\}^\B$ satisfies $\Delta_M \leq k$. Then the shuffled protocol $S \circ R^{\matrix}(n, \B, \M, R^{\FO})$ is $(\ep, \delta)$-differentially private.
\end{theorem}

\begin{proof}[Proof]
Let $\MY$ be the message space of the randomizer $R^{\FO}$, and $\MY'$ be the set of multisets consisting of elements of $\MY$. Let $P = S \circ R^{\matrix}(n,\B,\M,R^{\FO})$. Consider neighboring datasets $X = (x_1, \ldots, x_n) \in [\B]^n$ and $X' = (x_1, \ldots, x_{n-1}, x_n') \in [\B]^n$. We wish to show that for any $\MT \subset \MY$,
\begin{equation}
\label{eq:pmatrix_priv}
\p[P(X) \in \MT] \leq e^\ep \cdot \p [P(X') \in \MT] + \delta.
\end{equation}
For $i \in [n]$, let $\MS_i = \{ j \in [\B] : \M_{j,x_i} \neq 0 \}$ and $\MS_n' = \{ j \in [\B] : \M_{j,x_n'} \neq 0\}$. Since $\Delta_M \leq k$, we have $|\MS_i| \leq k$ for $i \in [n]$ and $|\MS_n'| \leq k$. Since the output of $R^{\matrix}$ on input $x_i$ is simply $R^{\FO}(\MS_i)$,
$$
P(X) = S (R^{\FO}(\MS_1), \ldots, R^{\FO}(\MS_n)), \quad P(X') = S(R^{\FO}(\MS_1), \ldots, R^{\FO}(\MS_{n-1}),  R^{\FO}(\MS_n')).
$$
Then (\ref{eq:pmatrix_priv}) follows by  the fact that $(\MS_1, \ldots, \MS_n)$ and $(\MS_1, \ldots, \MS_{n-1}, \MS_n')$ are neighboring datasets for the $(\ha, \hb, \hk)$-frequency problem and $S \circ R^{\FO}$ is $(\ep, \delta)$-differentially private.
\end{proof}

\begin{theorem}[Accuracy \& efficiency of $P^{\matrix}$]
\label{thm:pmatrix_accuracy}
Suppose $R^{\FO}, A^{\FO}$ are the local randomizer and analyzer for computation of an $(\ha, \hb, \hk)$-frequency oracle with $n$ users and universe size $\B$. Suppose also that $\MW \subset \{0,1\}^\B$ is a set of counting queries and $\M \in \{0,1\}^\B$ is such that, for any $w \in \MW$, $\| wM^{-1} \|_1 \leq a$ and $\Delta_\M \leq \hk$. Consider the shuffled model protocol $P^{\matrix} = (R^{\matrix}(n,\B,\M, R^{\FO}), S, A^{\matrix}(n,\B,\M,A^{\FO},\MW))$. For any dataset $X = (x_1, \ldots, x_n)$, let the (random) estimates produced by the protocol $P^{\matrix}$ on input $X$ be denoted by $f_w \in [0,1]$ ($w \in \MW$). Then: % the protocol $P^{\matrix} = (R^{\matrix}, S, A^{\matrix})$ satisfies
\begin{equation}
\label{eq:pmatrix_prob}
\p \left[ \forall w \in \MW:  | f_w - \langle w, \hist(X) \rangle| \leq \ha \cdot a \right] \geq 1 - \beta.
\end{equation}
Moreover, if the set of nonzero entries of $w\M^{-1}$ and their values can be computed in time $T$, and $A^{\FO}$ releases a frequency oracle $(\FO, \MA)$ which takes time $T'$ to query an index $j$, then for any $w \in \MW$, the estimate $f_w$ can be computed in time $O(T + a \cdot T')$ by $A^{\matrix}$.
\end{theorem}

\begin{proof}[Proof]
For $i \in [n]$, let $\MS_i = \{ j \in [\B] : \M_{j,x_i} \neq 0 \}$ be the set of nonzero entries of the $x_i$th column of $\M$. Denote by $(\FO, \MA)$ the frequency oracle comprising the output $A^{\FO}(S(R^{\FO}(\MS_1), \ldots, R^{\FO}(\MS_n)))$. Define $\hat y \in \BR^\B$ by $\hat y_j = \MA(\FO, j)$, for $j \in [\B]$. Then the output of $P^{\matrix}$, namely $$P^{\matrix}(X) = A^{\matrix}(S(R^{\matrix}(x_1), \ldots, R^{\matrix}(x_n))),$$ is given by the map associating each $w \in \MW$ to $\langle w M^{-1}, \hat y \rangle$ %, where $\hat y = A^{\FO}(S(R^{\FO}(\MS_1) | \cdots | R^{\FO}(\MS_n)))$ is the output of the shuffled protocol $P^{\FO} := (R^{\FO}, S, A^{\FO})$ given inputs $(\MS_1, \ldots, \MS_n)$ 
(Algorithms~\ref{alg:matrix_rand} and~\ref{alg:matrix_analyzer}).

Since $(\FO, \MA)$ is an $(\ha, \hb, \hk)$-frequency oracle, we have that
$$
\p \left[ \left\| \hat y - \hist(\MS_1, \ldots, \MS_n) \right\|_\infty \leq \ha\right] \geq 1-\beta.
$$
Notice that the histogram of $\MS_i$ is given by the $x_i$th column of $\M$, which is equal to $\M \hist(x_i)$. Thus $\hist(\MS_1, \ldots, \MS_n) = \M \hist(x_1, \ldots, x_n)$. 
By H{\"o}lder's inequality, it follows that with probability $1-\beta$, for all $w \in \MW$,
$$
\left|\langle w\M^{-1}, \hat y \rangle - \langle w\M^{-1}, \M \hist(x_1, \ldots, x_n) \rangle \right| \leq \ha \cdot \| w\M^{-1} \|_1 \leq \ha \cdot a.
$$
But $\langle w\M^{-1}, \M \hist(x_1, \ldots, x_n) \rangle = w\M^{-1} \M \hist(x_1, \ldots, x_n) = \langle w, \hist(x_1, \ldots, x_n)\rangle$ is the answer to the counting query $w$. This establishes (\ref{eq:pmatrix_prob}).

The final claim involving efficiency follows directly from Line~\ref{ln:ma_return} of Algorithm~\ref{alg:matrix_analyzer}.
\end{proof}

\subsection{Single-Dimensional Range Queries}
\label{sec:1d_rq}
We first present the matrix $\M$ discussed in previous section for the case of $d = 1$, i.e., single-dimensional range queries. In this case, the set $\MX = [\B]$ is simply identified with $\B$ consecutive points on a line, and a range query $[j,j']$ is specified by integers $j,j' \in [\B]$ with $j \leq j'$. We will assume throughout that $\B$ is a power of 2. (This assumption is without loss of generality since we can always pad the input domain to be of size a power of 2, with the loss of a constant factor in our accuracy bounds.) We begin by presenting the basic building block in the construction of $\M$, namely that of a {\it range query tree} $\MT_\B$ with $\B$ leaves and a {\it chosen set} $\MC_\B$ of $\B$ nodes of $\MT_\B$:
\begin{defn}[Range query tree]
\label{def:rqt}
Suppose $\B \in \BN$ is a power of 2, $\ell \in \BN, \gamma \in (0,1)$. Define a complete binary tree $\MT_\B$ of depth $\log \B$, where each node stores a single integer-valued random variable:
\begin{enumerate}[nosep]
\item For a depth $0 \leq t \leq \log \B$ and an index $1 \leq s \leq \B/2^{\log \B-t}$, let $v_{t,s}$ be the $s$th vertex of the tree at depth $t$ (starting from the left). We will denote the value stored at vertex $v_{t,s}$ by $y_{t,s}$.
The values $y_{t,s}$ will always have the property that $y_{t,s} = y_{t+1,2s-1} + y_{t+1, 2s}$; i.e., the value stored at $v_{t,s}$ is the sum of the values stored at the two children of $v_{t,s}$.\label{it:yts_defn}
\item Let $\MC_\B = \{ v_{t,s} : 0 \leq t \leq \log \B, s \equiv 1 \pmod{2}\}$. Let the $\B$ nodes in $\MC_\B$ be ordered in the top-to-bottom, left-to-right order. In particular, $v_{0,1}$ comes first, $v_{1,1}$ is second, $v_{1,3}$ is third, $v_{2,1}$ is fourth, and in general: the $j$th node in this ordering ($1 < j \leq \B$) is $v_{t_j,s_j}$, where $t_j = \lceil \log_2 j \rceil, s_j = 2(j-2^{t_j-1}) - 1$. \label{it:mcb_defn}%
% Then the local randomizer $R_i^{tree}(n,B,\ell,\gamma)$ can be described as follows: consider the binary tree $\MT_\B$ along with the random variables $Y_{t,s}$ at each node $v_{t,s}$. 
% Then for each $j \in [\B]$ with $j \equiv i \pmod{n}$, the randomizer sets $Y_j = Y_{t_j,s_j}$. % where $t_j = \lceil \log_2 j \rceil, s_j = 2(j-2^{t_j - 1}) - 1$. 
\item For $1 \leq j \leq \B$, we will denote $z_j := y_{\log \B, j}$ and $y_j = y_{t_j,s_j}$.
\end{enumerate}
\end{defn}
See Figure~\ref{fig:rqtree} for an illustration of $\MT_4$. The next lemma establishes some basic properties of the set $\MC_\B$:
\begin{figure}[t]
\centering
\subfloat[Range query tree, $\B = 4$]{
\includegraphics[scale=0.75]{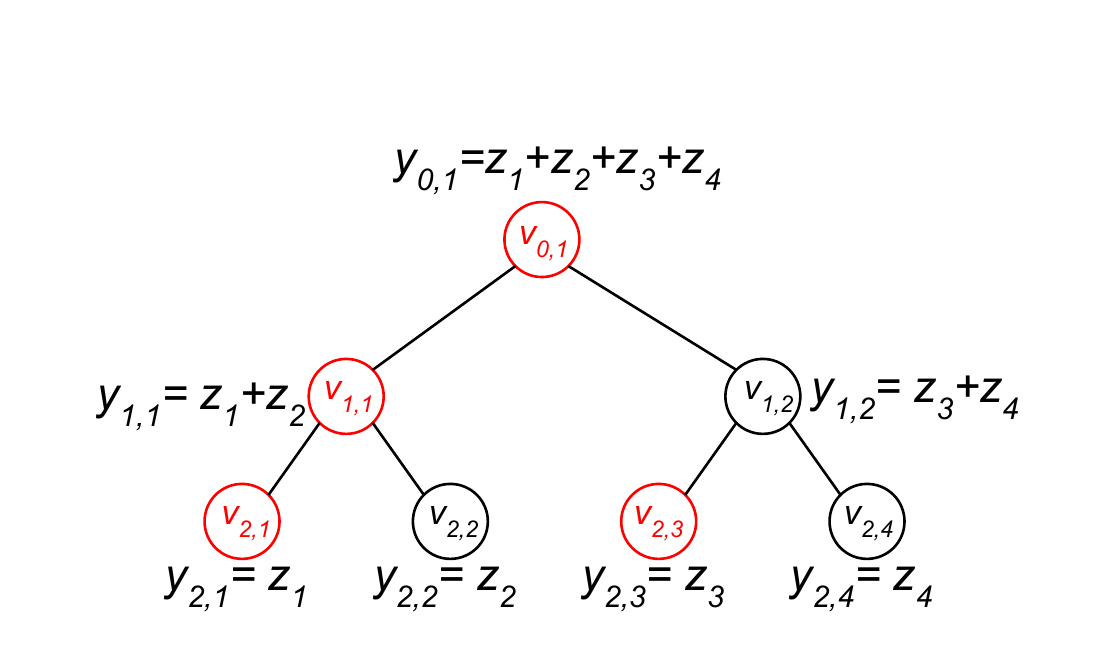}}
\subfloat[Path $P$ constructed to derive (\ref{eq:leaf_sum_contrib})]{
\includegraphics[scale=0.75]{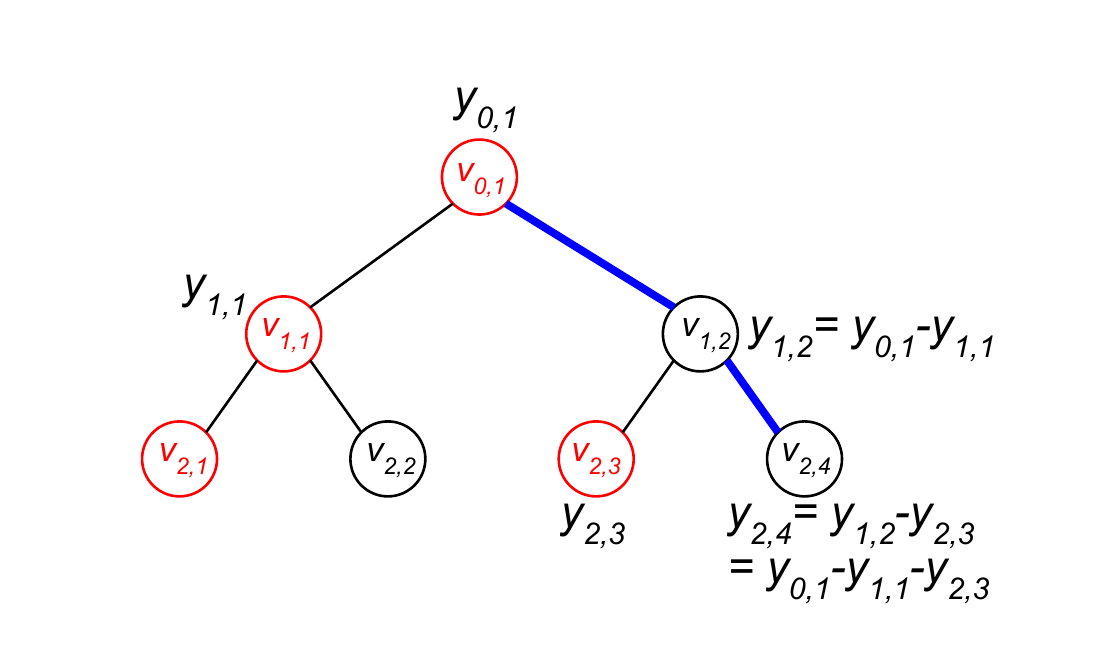}}
\caption{(a) The range query tree $\MT_4$. The nodes in $\MC_4$ are highlighted in red. The labels $y_{t,s}, z_s$ next to nodes show the values stored at the nodes and the relations between them. Notice that in the case $\B= 4$, we have $(t_1, s_1) = (0,1), (t_2, s_2) = (1,1), (t_3,s_3) = (2,1), (t_4, s_4) = (2,3)$. (b) The path $P$ described in (\ref{eq:leaf_sum_contrib}) for $j = 4$ is highlighted in blue. For this case ($\B = j = 4$) we have $z_4 = y_{0,1} - y_{1,2} - y_{2,3}$.} % $z_4 = y_{2,4} = y_{1,2} - y_{2,3} = y_{0,1} - y_{1,2} - y_{2,3}$.}
\label{fig:rqtree}
\end{figure}

\begin{lemma}
\label{lem:s_properties}
Fix $d$ a power of 2. We have the following regarding the set $\MC_\B$ defined in Definition~\ref{def:rqt}:
\begin{enumerate}
\item $\MC_\B$ is the union of the the root and set of nodes of $\MT_\B$ that are the left child of their parent. % (including the root). In particular, for any $v \in \MT_\B$, the left child of $v$ is in $\MC_\B$. % For any $v \in \MC_\B$, the left child of $v$ is also in $\MC_\B$.
\item % $\MC_\B$ is the unique set of nodes in the tree such that 
For any node $u \not \in \MC_\B$, there is some $v \in \MT_\B$ (which is an ancestor of $u$) so that there is a path from $v$ to $u$ that consists entirely of following the right child of intermediate nodes, starting from $v$.
\end{enumerate}
\end{lemma}

\begin{proof}[Proof of Lemma~\ref{lem:s_properties}]
The first part is immediate from the definition of $\MC_\B$. For the second part, given $u$, we walk towards the root, continually going to the parent of the current node. The first time we arrive at a node that is the left child of its parent, we will be at a node in $\MC_\B$; we let this node be $v$.
\end{proof}

Next we make two more definitions that will aid in the analysis:
\begin{defn}
\label{def:2val}
    For an integer $j \in [\B]$, let $v(j)$ denote the number of steps from a node to its parent one must take starting at the leaf $v_{\log \B, j}$ of the tree $\MT_\B$ to get to a node in $\MC_\B$. Equivalently, $v(j)$ is the 2-adic valuation of $j$ (i.e., the base-2 logarithm of the largest power of 2 dividing $j$).
\end{defn}
\begin{defn}
\label{def:countones}
For a positive integer $j$, let $c(j)$ be the number of ones in the binary representation of $j$.
\end{defn}
By property (\ref{it:yts_defn}) of Definition~\ref{def:rqt}, the set of all values $y_{t,s}$, for $0 \leq t \leq \log \B$, $1 \leq s \leq \B/2^t$, is entirely determined by the values $z_s$: in particular, for any $v_{t,s}$, $y_{t,s}$ is the sum of all $z_s$ for which the leaf $v_{\log \B, s}$ is a descendant of $v_{t,s}$. Conversely, given the values of $y_{t,s}$ for which $v_{t,s} \in \MC_\B$ (equivalently, the values $y_{t_j,s_j}$ for $j \in [\B]$), the values $z_j = y_{\log \B, j}$ are determined as follows:%
\begin{equation}
\label{eq:leaf_sum_contrib}
z_j = y_{\log B, j} = y_{\log \B - v(j), j / 2^{v(j)}} - \sum_{t =1}^{v(j) - 1} y_{\log \B - v(j) + t, j/2^{v(j) - t} - 1}.
\end{equation}
Graphically, we follow the path $P$ from $v_{\log \B, j}$ to the root until we hit a node $v_{t, s}$ in $\MC_\B$; then $z_j$ is the difference of $y_{t,s}$ and the sum of the variables stored at the left child of each node in the path $P$. (See Figure~\ref{fig:rqtree} for an example.)

It follows from the argument in the previous paragraph that the linear transformation that sends the vector $(z_1, \ldots, z_\B)$ to the vector $(y_{t_1,s_1}, \ldots, y_{t_\B, s_\B})$ is invertible; let $\M_\B \in \{0,1\}^{\B \times \B}$ be the matrix representing this linear transformation. By (\ref{eq:leaf_sum_contrib}), which describes the linear transformation induced by $\M_\B^{-1}$, we have that $\M_\B^{-1} \in \{-1, 0, 1\}^{\B \times \B}$. 

Since each leaf has $1 + \log \B$ ancestors (including itself), we immediately obtain:
\begin{lemma}
\label{lem:ab_sensitivity}
The sensitivity of $\M_\B$ is given by $\Delta_{\M_\B} = 1 + \log \B$.
\end{lemma}

Next consider any range query $[j,j']$, so that $1 \leq j \leq j' \leq \B$, and let $w \in \BR^\B$ be the row vector representing this range query (see Section~\ref{sec:matrix_randomizer}). In particular all entries of $w$ are 0 apart from $w_j, w_{j+1}, \ldots, w_{j'}$, which are all 1. 
\begin{lemma}
\label{lem:ab_inverse}
For a vector $w$ representing a range query $[j,j']$, the vector $w\M_\B^{-1}$ belongs to $\{-1, 0,1\}^\B$, and it has at most $c(j-1) + c(j') \leq 2\log \B$ nonzero entries. Moreover, the set of these nonzero entries (and their values) can be computed in time $O(\log \B)$. 
\end{lemma}

\begin{proof}[Proof of Lemma~\ref{lem:ab_inverse}]
Since $\M_\B$ is invertible, $w\M_\B^{-1}$ is the unique vector $\nu \in \BR^\B$ such that for any values of $\{y_{t,s}\}_{0 \leq t \leq \B, s \in [\B/ 2^t]}$ satisfying property (\ref{it:yts_defn}) of Definition~\ref{def:rqt}, we have
$$
z_j + z_{j+1} + \cdots + z_{j'} = y_{\log \B, j} + \cdots + y_{\log \B, j'} = \langle \nu, (y_{t_1,s_1}, \ldots, y_{t_\B, s_\B})\rangle.
$$

Next let $v_{\tilde t, \tilde s}$ be the first node in $\MC_\B$ that is reached on the leaf-to-root path starting at $v_{\log \B, j'}$. Recall from Definition~\ref{def:2val} that $\tilde t = \log B - v(j')$. Consider the path on the tree $\MT_\B$ from the root $v_{1,1}$ to the node $v_{\tilde t, \tilde s}$. Suppose the right child is taken at $h-1$ vertices of this path; it is not hard to see that $h = c(j')$ (see Definition~\ref{def:countones}). For $1 \leq k \leq h$, at the $k$th vertex on this path where the right child is taken, set $v_{t_k', s_k'}$ to be the left child of the parent vertex (so that $v_{t_k', s_k'}$ is not on the path). By Lemma~\ref{lem:s_properties}, $v_{t_k',s_k'} \in \MC_\B$. Also set $v_{t_h', s_h'} = v_{\tilde t, \tilde s}$. Then from Definition~\ref{def:rqt} (property (\ref{it:yts_defn})) we have
\begin{equation}
\label{eq:jp_addition}
z_1 + \cdots + z_{j'} = y_{\log B, 1} + \cdots + y_{\log B, j'} = \sum_{k=1}^h y_{t_k', s_k'}.
\end{equation}
The same computation for $j-1$ replacing $j'$ yields, with $\hat h = c(j-1)$,
\begin{equation}
\label{eq:jm1_addition}
z_1 + \cdots + z_{j-1} = y_{\log \B, 1} + \cdots + y_{\log B, j-1} = \sum_{k=1}^{\hat h} y_{\hat  t_k, \hat s_k},
\end{equation}
where the pairs $(\hat t_k, \hat s_k)$ replace the pairs $(t_k', s_k')$. Taking the difference of (\ref{eq:jp_addition}) and (\ref{eq:jm1_addition}) yields
$$
z_{j} + \cdots + z_{j'} =  \sum_{k=1}^h y_{t_k', s_k'} - \sum_{k=1}^{\hat h} y_{\hat  t_k, \hat s_k},
$$
i.e., $z_j + \cdots + z_{j'}$ is a linear combination of at most $c(j-1) + c(j')$ elements of $\{ y_{t,s} : v_{t,s} \in \MC_\B \}$, with coefficients in $\{-1,1\}$. The sets $\{ (t_k', s_k') \}_{1 \leq k \leq h}$ and $\{ (\hat t_k, \hat s_k) \}_{1 \leq k \leq \hat h}$ can be computed in $O(\log \B)$ time by walking on the leaf-to-root path starting at $v_{\log \B, j'}$ and $v_{\log \B, j-1}$, respectively. This establishes Lemma~\ref{lem:ab_inverse}.
\end{proof}

Lemmas~\ref{lem:ab_sensitivity} and~\ref{lem:ab_inverse} establish properties (1) and (2) required of the matrix $\M = \M_\B$ to guarantee $\poly \log(\B)$ accuracy and $\poly \log(\B)$ communication for private computation of 1-dimensional range queries. In the following section we use $\M_\B$ to construct a matrix which satisfies the same properties for $d$-dimensional range queries for any $d \geq 1$.

\subsection{Multi-Dimensional Range Queries}
\label{sec:multid_rq}
Fix any $d \geq 1$, and suppose the universe $\MX$ consists of $\B_0$ buckets in each dimension, i.e., $\MX = [\B_0]^d$. In this case, a range query $[j_1, j_1'] \times [j_2, j_2'] \times \cdots \times [j_d, j_d']$ is specified by integers $j_1, j_2, \dots, j_d, j_1', j_2', \dots, j_d' \in [\B_0]$ with $j_i \leq j_i'$ for all $i=1,2,\dots,d$.

Throughout this section, we will consider the case that $d$ is a constant (and $\B_0$ is large). Moreover suppose that $\B_0$ is a power of 2 (again, this is without loss of generality since we can pad each dimension to be a power of 2 at the cost of a blowup in $|\MX|$ by at most a factor of $2^d$). Write $\B = |\MX| = \B_0^d$. Our goal is to define a matrix $\M_{\B,d}$ which satisfies analogues of Lemmas~\ref{lem:ab_sensitivity} and~\ref{lem:ab_inverse} for $w \in \{0,1\}^\B$ representing multi-dimensional range queries (when $[\B]$ is identified with $[\B_0]^d$).

The idea behind the construction of $\M_{\B,d}$ is to apply the linear transformation $\M_{\B_0}$ in each dimension, operating on a single-dimensional slice of the input vector $(z_{j_1, \ldots, j_d})_{j_1, \ldots, j_d \in [\B_0]}$ (when viewed as a $d$-dimensional tensor) at a time. Alternatively, $\M_{\B,d}$ can be viewed combinatorially through the lens of {\it range trees}~\cite{rangetree}: $\M_{\B,d}$ is a linear transformation that takes the vector $(z_{j_1, \ldots, j_d})$ to a $\B$-dimensional vector whose components are the values stored at the nodes of a range tree defined in a similar manner to the range query tree $\MT_\B$ for the case $d = 1$. However, we opt to proceed linear algebraically: the matrix $\M_{\B,d}$ is defined as follows. Fix a vector $z \in \BR^\B$. We will index the elements of $z$ with $d$-tuples of integers in $[\B_0]$, i.e., we will write $z = (z_{j_1, \ldots, j_d})_{j_1, \ldots, j_d \in [\B_0]}$. For $1 \leq p \leq d$, let $\M^{\pre}_{\B,p}$ be the linear transformation that applies $\M_{\B_0}$ to each vector $(z_{j_1, \ldots, j_{p-1}, 1, j_{p+1}, \ldots, j_d}, \ldots, z_{j_1, \ldots, j_{p-1}, \B_0, j_{p+1}, \ldots, j_d})$, where $j_1, \ldots, j_{p-1}, j_{p+1}, \ldots, j_d \in [\B_0]$. That is, $\M_{\B_0}$ is applied to each slice of the vector $z$, where the slice is being taken along the $p$th dimension. Then let
\begin{equation}
\label{eq:mbd}
\M_{\B,d} := \M^{\pre}_{\B,d} \circ \cdots \circ \M^{\pre}_{\B,1} (z).
\end{equation}
We will also use an alternate characterization of $\M_{\B,d}$, which we develop next. First identify $\BR^\B$ with the $d$-wise tensor product of $\BR^{\B_0}$, in the following (standard) manner: % i.e., $\BR^\B \simeq \BR^{\B_0} \otimes \cdots \otimes \BR^{\B_0}$. 
Let $e_1, \ldots, e_{\B_0} \in \BR^{\B_0}$ be the standard basis vectors in $\BR^{\B_0}$. Then the collection of all $e_{j_1} \otimes \cdots \otimes e_{j_d}$, where $j_1, \ldots, j_d \in [\B_0]$, form a basis for $\BR^{\B_0} \otimes \cdots \otimes \BR^{\B_0}$. Under the identification $\BR^\B \simeq (\BR^{\B_0})^{\otimes d}$, a vector $z = (z_{j_1, \ldots, j_d})_{j_1, \ldots, j_d \in [\B_0]} \in \BR^\B$ is identified with the following linear combination of these basis vectors:
$$
\sum_{j_1, \ldots, j_d \in [\B_0]} z_{j_1, \ldots, j_d} \cdot e_{j_1} \otimes \cdots \otimes e_{j_d}.
$$
Under this identification, the matrix $\M_{\B,d}$ corresponds to the following linear transformation of $(\BR^{\B_0})^{\otimes d}$:
$$
\M_{\B_0} \otimes \cdots \otimes \M_{\B_0} : (\BR^{\B_0})^{\otimes d} \ra (\BR^{\B_0})^{\otimes d}.
$$
In the following lemmas, we will often abuse notation to allow $\M_{\B,d}$ to represent both the above linear transformation as well as the matrix in $\BR^{\B \times \B}$ representing this transformation.
\begin{lemma}
\label{lem:multid_sensitivity}
We have that $\M_{\B,d} \in \{0,1\}^{\B \times \B}$ and the sensitivity of $\M_{\B,d} : \BR^\B \ra \BR^\B$ is bounded by $\Delta_{\M_{\B,d}} \leq (1 + \log \B_0)^d$.
\end{lemma}

\begin{proof}[Proof of Lemma~\ref{lem:multid_sensitivity}]
Notice that the $((j_1, \ldots, j_d), (j_1', \ldots, j_d'))$ entry of $\M_{\B,d}$ is given by the following product:
$$
\prod_{p=1}^d (\M_{\B_0})_{j_p, j_p'}.
$$
Since $\M_{\B_0} \in \{0,1\}^{\B_0 \times \B_0}$, it follows immediately that $\M_{\B,d} \in \{0,1\}^{\B \times \B}$. Moreover, to upper bound the sensitivity of $\M_{\B, d}$ note that for any $(j_1', \ldots, j_d') \in [\B_0]^d$,
$$
\sum_{(j_1, \ldots, j_d) \in [\B_0]^d} \prod_{p=1}^d (\M_{\B_0})_{j_p, j_p'} = \prod_{p=1}^d \left( \sum_{j_p = 1}^{\B_0} (\M_{\B_0})_{j_p,j_p'} \right) \leq (\Delta_{\M_{\B_0}})^d \leq (1 + \log \B_0)^d,
$$
where the last inequality above uses Lemma~\ref{lem:ab_sensitivity}.
% Let $\M_{\B,0}^{\pre} : \BR^\B \ra \BR^\B$ be the identity map. We will show by induction that for $0 \leq p \leq d$, the sensitivity of the map $\M_{\B,p}^{\pre} \circ \cdots \circ M_{\B,1}^{\pre} : \BR^\B \ra \BR^\B$ is bounded above  by $(1 + \log \B)^p$ and the entries of the matrix representing this linear transform are in $\{0,1\}$. The base case $p = 0$ is immediate (as $I_\B$ has sensitivity 1).
% Now suppose the statement of the lemma is true for $p-1$. To complete the inductive step, we will argue separately for each column of $\M_{\B,p}^{\pre}$. In particular, consider a column of $\M_{\B,p}^{\pre}$ indexed by $(j_1, \ldots, j_d) \in [\B_0]^d$. Also consider any vector $z = (z_{j_1',\ldots,j_d'})_{j_1', \ldots, j_d' \in [\B_0]} \in \BR^\B$. Note that each entry of $\M_{\B,p}^{\pre}z$ is a linear combination of the $z_{j_1, \ldots, j_d}$ (given by the rows of $\M_{\B,p}^{\pre}$). The only entries of $\M_{\B,p}^{\pre}z$ to which $z_{j_1, \ldots, j_d}$ has a nonzero contribution are those entries indexed by $(j_1, \ldots, j_{p-1}, j_p', j_{p+1}, \ldots, j_d) \in [\B_0]^d$ where $j_p' \in [\B_0]$ is such that $(M_{\B_0})_{j_p', j_p} = 1$ 
\end{proof}

\begin{lemma}
\label{lem:multid_inverse}
For the vector $w$ representing any range query $[j_1, j_1'] \times \cdots \times [j_d, j_d']$, the vector $w\M_{\B,d}^{-1}$ belongs to $\{-1,0,1\}^\B$ and moreover it has at most $$
\prod_{p=1}^d (c(j_p-1) + c(j_p')) \leq (2 \log \B_0)^d = (2 \log(\B^{1/d}))^d
$$
nonzero entries.
\end{lemma}

\begin{proof}[Proof of Lemma~\ref{lem:multid_inverse}]
The inverse $\M_{\B,d}^{-1}$ of $\M_{\B,d}$ is given by the $d$-wise tensor product $\M_{\B_0}^{-1} \otimes \cdots \otimes \M_{\B_0}^{-1}$. This can be verified by noting that this tensor product and $\M_{\B,d}$ multiply (i.e., compose) to the identity:
\begin{align*}
(\M_{\B_0}^{-1} \otimes \cdots \otimes \M_{\B_0}^{-1}) \cdot \M_{\B,d} & = (\M_{\B_0}^{-1} \otimes \cdots \otimes \M_{\B_0}^{-1}) \cdot (\M_{\B_0}  \otimes \cdots \otimes \M_{\B_0}) \\
&= (\M_{\B_0}^{-1} \cdot \M_{\B_0}) \otimes \cdots \otimes (\M_{\B_0}^{-1} \cdot \M_{\B_0})\\
&= I_{\B_0} \otimes \cdots \otimes I_{\B_0}\\
&= I_{\B}.
\end{align*}
Recall that the (row) vector $w$ representing the range query $[j_1, j_1'] \times \cdots \times [j_d, j_d']$ satisfies, for each $(j_1'', \ldots, j_d'') \in [\B_0]^d$, $w_{j_1'', \ldots, j_d''} = 1$ if and only if $j_p'' \in [j_p, j_p']$ for all $1 \leq p \leq d$, and otherwise $w_{j_1'', \ldots, j_d''} = 0$. Therefore, we may write $w$ as the product of row vectors $w = w_1 \otimes \cdots \otimes w_d$, where for $1 \leq p \leq d$, $w_p$ is the (row) vector representing the range query $[j_p, j_p']$. In particular, for $1 \leq j'' \leq \B_0$, the $j''$th entry of $w_p$ is 1 if and only if $j''\in [j_p, j_p']$. It follows that
\begin{equation}
\label{eq:wminv_prod}
w\M_{\B,d}^{-1} = (w_1 \otimes \cdots \otimes w_d) (\M_{\B_0}^{-1} \otimes \cdots \otimes \M_{\B_0}^{-1}) = w_1\M_{\B_0}^{-1} \otimes \cdots \otimes w_d \M_{\B_0}^{-1}.
\end{equation}
By Lemma~\ref{lem:ab_inverse}, for $1 \leq p \leq d$, the vector $w_p \M_{\B_0}^{-1}$ has entries in $\{-1, 0, 1\}$, at most $c(j_p - 1) + c(j_p')$ of which are nonzero. Since $w\M_{\B,d}^{-1}$ is the tensor product of these vectors and the set $\{-1, 0, 1\}$ is closed under multiplication, it also has entries in $\{-1,0,1\}$, at most $\prod_{p=1}^d (c(j_p - 1) + c(j_p'))$ of which are nonzero.
\end{proof}

The following lemma allows us to bound the running time of the local randomizer (Algorithm~\ref{alg:matrix_rand}) and analyzer (Algorithm~\ref{alg:matrix_analyzer}):
\begin{lemma}
\label{lem:matrix_time_bounds}
Given $\B, d$ with $\B = \B_0^d$, the following can be computed in $O(\log^d \B_0)$ time:
\begin{enumerate}
\item[(1)] Given indices $(j_1, \ldots, j_d) \in [\B_0]^d$, the nonzero indices of $\M_{\B,d}$ for the column indexed by $(j_1, \ldots, j_d)$.\label{it:nonzero_m}
% the entries $(M_{\B,d})_{j_1, \ldots, j_d}$ and $(M_{\B,d}^{-1})_{j_1, \ldots, j_d}$ of $\M_{\B,d}$ and $\M_{\B,d}^{-1}$, respectively, at the position $(j_1, \ldots, j_d)$.  
\item[(2)] Given a vector $w \in \BR^\B$ specifying a range query, the set of nonzero elements of $wM_{\B,d}^{-1}$ and their values (which are in $\{ -1, 1\}$).\label{it:nonzero_minv}
\end{enumerate} 
\end{lemma}

\begin{proof}[Proof of Lemma~\ref{lem:matrix_time_bounds}]
We first deal with the case $d = 1$, i.e., the matrix $\M_{\B,1} = \M_\B$. Given $j, j' \in [\B]$, the $(j',j)$-entry of $\M_\B$ is 1 if and only if the node $v_{t_{j'}, s_{j'}}$ of the tree $\MT_{\B}$ is an ancestor of the leaf $v_{\log \B, j}$. Since $t_j = \lceil \log_2 j \rceil, s_j = 2(j-2^{t_j - 1}) - 1$, whether or not $v_{t_{j'}, s_{j'}}$ is an ancestor of $v_{\log \B, j}$ can be determined in $O(\log \B)$ time, thus establishing (1) for the case $d=1$. Notice that the statement of Lemma~\ref{lem:ab_inverse} immediately gives (2) for the case $d = 1$.

To deal with the case of general $d$, notice that $\M_{\B,d} = (\M_{\B_0})^{\otimes d}$. Therefore, for a given $(j_1, \ldots, j_d)$ the set
\begin{equation}
\label{eq:cartesian_rows}
\{ (j_1', \ldots, j_d') : (\M_{\B,d})_{(j_1', \ldots, j_d'),(j_1, \ldots, j_d)} = 1\}
\end{equation}
of nonzero indices in the $(j_1, \ldots, j_d)$-th column of $\M_{\B,d}$ is equal to the Cartesian product
$$
\bigtimes_{1 \leq p \leq d} \{ j_p' : (\M_{\B_0})_{j_p', j_p} = 1\}.
$$
Since each of the sets $\{ j_p' : (\M_{\B_0})_{j_p', j_p} = 1\}$ can be computed in time $O(\log \B_0)$ (using the case $d = 1$ solved above), and is of size $O(\log \B_0)$, the product of these sets (\ref{eq:cartesian_rows}) can be computed in time $O(\log^d \B_0)$, thus completing the proof of item (1) in the lemma.

The proof of item (2) for general $d$ is similar. For $1 \leq p \leq d$, let $w_p$ be the vector in $\BR^{\B_0}$ corresponding to the 1-dimensional range query $[j_p, j_p']$. Then recall from (\ref{eq:wminv_prod}) we have that $w\M_{\B,d}^{-1} = w_1 \M_{\B_0}^{-1} \otimes \cdots \otimes w_d \M_{\B_0}^{-1}$. By item (2) for $d = 1$, the nonzero entries of each of $w_p \M_{\B_0}^{-1}$ (and their values) can be computed in time $O(\log \B_0)$; since each of these sets has size $O(\log \B_0)$, the set of nonzero entries of $w\M_{\B,d}^{-1}$, which is the Cartesian product of these sets, as well as the values of these entries, can be computed in time $O(\log^d \B_0)$.
\end{proof}

\subsection{Guarantees for Differentially Private Range Queries}
\label{sec:rq_collect}
In this section we state the guarantees of Theorems~\ref{thm:pmatrix_privacy} and~\ref{thm:pmatrix_accuracy} on the privacy and accuracy of the protocol $P^{\matrix} = (R^{\matrix}(n,\B, \M, R^{\FO}), S, A^{\matrix}(n, \B, \M, A^{\FO}))$ for range query computation when $\M = \M_{\B,d}$ and the pair $(R^{\FO}, A^{\FO})$ is chosen to be either  $(R^{\CM}, A^{\CM})$ (Count Min sketch-based approach; Algorithm~\ref{alg:hh_CM}) or $(R^{\Had}, A^{\Had})$ (Hadamard response-based approach; Algorithm~\ref{alg:hh_had}). 

For the Hadamard response-based frequency oracle, we obtain the following:
\begin{theorem}
\label{thm:had_rq_final}
Suppose $\B_0,n,d\in \BN$, $\B = \B_0^d$, and $0 \leq \ep \leq 1$, and $\beta, \delta \geq 0$ with $1/\beta \leq \B^{O(1)}$\footnote{The assumption that $1/\beta$ is polynomial in $\B$ is purely for simplicity and can be removed at the cost of slightly more complicated bounds.}. Consider the shuffled-model protocol $P^{\matrix} = (R^{\matrix}, S, A^{\matrix})$, where:
\begin{itemize}
\item $R^{\matrix} = R^{\matrix}(n,\B,\M_{\B,d},R^{\Had})$ is defined in Algorithm~\ref{alg:matrix_rand};
\item $A^{\matrix} = A^{\matrix}(n,\B,\M_{\B,d},A^{\Had})$ is defined in Algorithm~\ref{alg:matrix_analyzer};
\item and $R^{\Had} = R^{\Had}(n,\B,\log n, \rho, (\log 2\B)^d)$ and $A^{\Had} = A^{\Had}(n,\B,\log n,\rho, (\log 2\B)^d)$ are defined in Algorithm~\ref{alg:hh_had}, and 
\begin{equation}
\label{eq:define_rho}
\rho = \frac{36 (\log 2\B)^{2d} \ln(e(\log 2\B)^d / (\ep \delta))}{\ep^2}.
\end{equation}
\end{itemize}
Then:
\begin{itemize}
\item The protocol $P^{\matrix}$ is $(\ep, \delta)$-differentially private in the shuffled model (Definition~\ref{def:dp_shuffled}).
\item For any dataset $X = (x_1, \ldots, x_n) \in ([\B_0]^d)^n$, with probability $1-\beta$, the frequency estimate of $P^{\matrix}$ for each $d$-dimensional range query has additive error at most $O(\ep^{-1}d^{1/2}(2 \log \B)^{2d + 1/2} \cdot \sqrt{\log((\log \B) / ( \ep \delta))})$.
\item The local randomizers send a total of $O\left(n \cdot \rho\right)$  messages, each of length $O(\log n \log \B)$. The analyzer can either (a) produce a data structure of size $O(\B\log (n\rho))$ bits such that a single range query can be answered in time $O((2\log \B)^d)$, or (b) produce a data structure of size $O(n \rho \log n\log \B)$ such that a single range query can be answered in time $O(n \rho (2 \log \B)^d\log n \log \B)$. 
\end{itemize}
\end{theorem}

\begin{proof}[Proof of Theorem~\ref{thm:had_rq_final}]
Lemma~\ref{lem:multid_sensitivity} guarantees that $\Delta_{\M_{\B,d}} \leq (1 + \log \B)^d = (\log 2\B)^d$. Then by Theorem~\ref{thm:pmatrix_privacy}, to show $(\ep, \delta)$-differential privacy of $P^{\matrix}$ it suffices to show $(\ep, \delta)$-differential privacy of the shuffled-model protocol $P^{\Had} := (R^{\Had}, S, A^{\Had})$. By Theorem \ref{thm:Had_k_main} with $k = (\log 2\B)^d$, this holds with $\rho$ as in (\ref{eq:define_rho}).
% By Lemma~\ref{thm:h_had_priv} with $k = (\log 2\B)^d$, this holds with $\rho$ as in (\ref{eq:define_rho}). (In particular, the parameter $\ep$ in Lemma~\ref{thm:h_had_priv} is set to $\ep / (\log 2\B)^d$, and the parameter $\delta$ in Lemma~\ref{thm:h_had_priv} is set to $\delta \cdot \ep / (e \cdot (\log 2\B)^d$.)

Next we show accuracy of $P^{\matrix}$. Lemma~\ref{lem:multid_inverse} guarantees that for any $w \in \{0,1\}^\B$ representing a range query, $w\M_{\B,d}^{-1}$ has at most $(2 \log \B)^d$ nonzero entries, all of which are either $-1$ or $1$. Moreover, by Theorem \ref{thm:Had_k_main} with $k = (\log 2\B)^d$ and $\rho$ as in (\ref{eq:define_rho}), for any $1 \geq \hb \geq 0$, the shuffled model protocol $P^{\Had}$ provides a
$$
\left( O\left(\log(\B/\beta) + \frac{(\log 2\B)^d\sqrt{\log(\B/\beta) \log((\log2\B)^d/(\ep \delta))}}{\ep} \right)
%\sqrt{3 \ln(2\B/\beta) \cdot \max\{3 \ln(2\B/\beta), 2\rho\}}
,\hb,(\log  2\B)^d\right)
$$
frequency oracle. % (Here we have used that $k \leq \rho$.)
By Theorem~\ref{thm:pmatrix_accuracy} and the assumption that $1/\beta \leq \B^{O(1)}$, it follows that with probability $1-\beta$, the frequency estimates of $P^{\matrix}$ on each $d$-dimensional range query have additive error at most
 \begin{align*}
% & (2 \log \B)^d \cdot \sqrt{3 \ln(2\B/\beta) \cdot \max\{3 \ln(2\B/\beta), 2\rho\}} \\
 & \leq O\left(\frac{(2 \log \B)^{2d + 1/2} \cdot \sqrt{d\left( \log((\log \B)/(\delta\ep))\right)}}{\ep}\right).
 \end{align*}
%  (Recall our treatment of $d$ as a constant, so that, for instance $2^d = O(1)$.)
 This establishes the claim regarding accuracy of $P^{\matrix}$. 
 
 To establish the last item (regarding efficiency), notice that the claims regarding communication (the number of messages and message length) follow from Lemma~\ref{thm:Hadamard_efficiency} with $k = (\log  2\B)^d$. Part (a) of the claim regarding efficiency of the analyzer follows from item 2 of Lemma~\ref{thm:Hadamard_efficiency} and the last sentence in the statement of Theorem~\ref{thm:pmatrix_accuracy}. Part (b) of the claim regarding efficiency of the analyzer follows from item 3 of Lemma~\ref{thm:Hadamard_efficiency} and the last sentence in the statement of Theorem~\ref{thm:pmatrix_accuracy}.
\end{proof}
Similarly, for the Count Min sketch-based frequency oracle, we obtain
\begin{theorem}
\label{thm:cm_rq_final}
There is a sufficiently large constant $\zeta$ such that the following holds. Suppose $\B_0,n,d\in \BN$, $\B = \B_0^d \geq n$\footnote{The assumption $n \leq \B$ is made to simplify the bounds and can be removed.}, and $0 \leq \ep \leq 1$, and $\beta, \delta \geq 0$. Consider the shuffled-model protocol $P^{\matrix} = (R^{\matrix}, S, A^{\matrix})$, where:
\begin{itemize}
\item $R^{\matrix} = R^{\matrix}(n,\B,\M_{\B,d},R^{\CM})$ is defined in Algorithm~\ref{alg:matrix_rand};
\item $A^{\matrix} = A^{\matrix}(n,\B,\M_{\B,d},A^{\CM})$ is defined in Algorithm~\ref{alg:matrix_analyzer};
\item and $R^{\CM} = R^{\CM}(n,\B,\log 2\B/\beta, \gamma, 2kn)$ and $A^{\CM} = A^{\CM}(n,\B,\log 2\B/\beta, 2kn)$ are defined in Algorithm~\ref{alg:hh_CM}, where
$$
\gamma =\frac 1n \cdot  \zeta \cdot \frac{\log^2(\B/\beta) k^2 \log(\log(\B/\beta) k / \delta)}{\ep^2}
%   k^2 \log(B/\beta) / \ep, \frac{\log(\tau k /\beta)}{\ep^2} \right\}
$$
and $k = (\log 2\B^{1/d})^d = (\log 2\B_0)^d$.
\end{itemize}
Then:
\begin{itemize}
\item The protocol $P^{\matrix}$ is $(\ep, \delta)$-differentially private in the shuffled model (Definition~\ref{def:dp_shuffled}).
\item For any dataset $X = (x_1, \ldots, x_n) \in ([\B_0]^d)^n$, with probability $1-\beta$, the frequency estimate of $P^{\matrix}$ for each $d$-dimensional range query has additive error at most
$$
O\left( \frac{(2 \log \B_0)^{2d}}{\ep} \cdot \sqrt{\log^3(\B/\beta) \log \left( (\log(\B/\beta)) (\log 2\B_0)^d/\delta\right)}\right).
%   \cdot \left( \log(n) +   \frac{\log(\B_0/\beta) \cdot (\log 2\B_0)^{2d}}{\ep} + \frac{\log(\log(\B/\beta) \cdot (\log 2\B_0)^d/\delta)}{\ep^2} \right) \right).
$$ 
% $O\left(\frac{(2 \log \B)^{3d} \log(\B n / \beta) (\log d(\log \B) / (\beta\delta))}{\ep^2} \right)$.
\item With probability at least $1 - \beta$, each local randomizer sends a total of at most
  $$
  \tilde m := O \left( \frac{\log^3(\B/\beta) (\log 2\B_0)^{3d}\log((\log (\B/\beta)) (\log 2\B_0)^d/\delta)}{\ep^2} \right)
  $$
  % O \left((\log 2\B_0)^d \cdot \log(\B/\beta) \cdot \left( \log(n) + \frac{\log(\B_0/\beta) \cdot (\log 2\B_0)^{2d}}{\ep} + \frac{\log((\log \B/\beta) (\log 2\B_0)^d/\delta)}{\ep^2}\right)\right)$$% $O \left( \frac{\log^{2d} \B \log(d \log \B/(\beta\delta))}{\ep^2} \right)$ 
messages, each of length $O(\log (\log \B / \beta) + \log ((\log 2\B_0)^d n))$. Moreover, in time $O( n\tilde m)$, the analyzer produces a data structure of size $O(n\log (\B / \beta)(\log 2\B_0)^d\log(n\tilde m))$ bits, such that a single range query can be answered in time $O((2\log  \B_0)^d \cdot \log \B / \beta)$. 
\end{itemize}
\end{theorem}

\begin{proof}[Proof of Theorem~\ref{thm:cm_rq_final}]
Lemma~\ref{lem:multid_sensitivity} guarantees that $\Delta_{\M_{\B,d}} \leq (1 + \log \B_0)^d = (\log 2\B^{1/d})^d$. (Recall our notation that $\B = (\B_0)^d$.) Then by Theorem~\ref{thm:pmatrix_privacy}, to show $(\ep, \delta)$-differential privacy of $P^{\matrix}$ it suffices to show $(\ep, \delta)$-differential privacy of the shuffled-model protocol $P^{\CM} := (R^{\CM}, S, A^{\CM})$. For the parameters above this follows from Theorem~\ref{thm:CM_main}.

Next we show accuracy of $P^{\matrix}$. Lemma~\ref{lem:multid_inverse} guarantees that for any $w \in \{0,1\}^\B$ representing a range query, $w\M_{\B,d}^{-1}$ has at most $(2 \log \B_0)^d$ nonzero entries, all of which are either $-1$ or $1$. Moreover, by Theorem~\ref{thm:CM_main}, the shuffled model protocol $P^{\CM}$ provides an
$
(\ha,\hb,(\log  2\B_0)^d)
$-frequency oracle with
$$
\ha \leq O\left( \frac{(\log 2\B_0)^d}{\ep} \cdot \sqrt{\log^3(\B/\beta) \log \left( (\log(\B/\beta)) (\log 2\B_0)^d/\delta\right)}\right).
% O \left( \log((\log 2\B_0)^d n) + \log(\B / \beta) \cdot \frac{(\log 2\B_0)^{2d}}{\ep} + \frac{\log(\log(\B/\beta) \cdot (\log 2\B_0)^d/\delta)}{\ep^2} \right) .
$$
By Theorem~\ref{thm:pmatrix_accuracy} with $a = (2 \log \B_0)^d$, it follows that with probability $1-\beta$, the frequency estimates of $P^{\matrix}$ on each $d$-dimensional range query have additive error at most 
$$
O\left( \frac{(2 \log \B_0)^{2d}}{\ep} \cdot \sqrt{\log^3(\B/\beta) \log \left( (\log(\B/\beta)) (\log 2\B_0)^d/\delta\right)}\right).
% O\left( (2 \log \B_0)^d \cdot \left( \log(n) + \log(\B_0 / \beta) \cdot \frac{(\log 2\B_0)^{2d}}{\ep} + \frac{\log(\log(\B/\beta) \cdot (\log 2\B_0)^d/(\delta))}{\ep^2} \right) \right).
$$
This establishes the second item. The final item follows from Lemma~\ref{le:h_CM_eff}, part (2) of Lemma~\ref{lem:matrix_time_bounds}, and the final sentence in the statement of Theorem~\ref{thm:pmatrix_accuracy}.
\end{proof}

\section{Conclusion and Open Problems}\label{sec:conclusion}
The shuffled model is a promising new privacy framework motivated by the significant interest on anonymous communication.
In this paper, we studied the fundamental task of frequency estimation in this setup. In the single-message shuffled model, we established nearly tight bounds on the error for frequency estimation and on the number of users required to solve the selection problem. We also obtained communication-efficient multi-message private-coin protocols with exponentially smaller error for frequency estimation, heavy hitters, range counting, and M-estimation of the median and quantiles (and more generally sparse non-adaptive SQ algorithms). We also gave public-coin protocols with, in addition, small query times. Our work raises several interesting open questions and points to fertile future research directions.

Our $\tilde{\Omega}(B)$ lower bound for selection (Theorem~\ref{th:selection_single_message}) holds for single-message protocols even with unbounded communication.  We conjecture that a lower bound on the error of $B^{\Omega(1)}$ should hold even for multi-message protocols (with unbounded communication) in the shuffled model, and we leave this as a very interesting open question. Such a lower bound would imply a first separation between the central and (unbounded communication) multi-message shuffled model.

Another interesting question is to obtain a private-coin protocol for frequency estimation with polylogarithmic error, communication per user, and query time; reducing the query time of our current protocol below $\tilde{O}(n)$ seems challenging.  In general, it would also be interesting to reduce the polylogarithmic factors in our guarantees for range counting as that would make them practically useful.

Another interesting direction for future work is to determine whether our efficient protocols for frequency estimation with much less error than what is possible in the local model could lead to more accurate and efficient shuffled-model protocols for fundamental primitives such as clustering \cite{stemmer_locally_2019} and distribution testing \cite{acharya_test_2018}, for which current locally differentially private protocols use frequency estimation as a black box.

Finally, a promising future direction is to extend our protocols for sparse non-adaptive SQ algorithms to the case of sparse aggregation. Note that the queries made by sparse non-adaptive SQ algorithms correspond to the special case of sparse aggregation where all non-zero queries are equal to $1$. Extending our protocols to the case where the non-zero coordinates can be arbitrary numbers would, e.g., capture sparse stochastic gradient descent (SGD) updates, an important primitive in machine learning. More generally, it would be interesting to study the complexity of various other statistical and learning tasks~\cite{smith_statest,wasserman,bassily_erm, chaudhuri_emprisk,chaudhuri_logistic,chaudhuri_pca} in the shuffled privacy model.

%Finally, coming up with an algorithm that computes the optimal noise probabilities in terms of the $(\epsilon, \delta)$ privacy parameters in the \emph{(multi-message) shuffled model} would be a valuable step towards the practical deployment of the methods studied in this work.

\medskip

% \noindent
% {\bf Acknowledgents.}
\section*{Acknowledgments}
We would like to thank James Bell, Albert Cheu, \'Ulfar Erlingsson, Vitaly Feldman, Adri\`a Gasc\'on, Peter Kairouz, Pasin Manurangsi, Stefano Mazzocchi, Ilya Mironov,  Ananth Raghunathan, Kunal Talwar,  Abhradeep Guha Thakurta, Salil Vadhan, and Vinod Vaikuntanathan as well as the anonymous reviewers of a previous version of this paper, for very helpful comments and suggestions. In particular, we would like to thank an anonymous reviewer who pointed out the connection to nonadaptive SQ algorithms in Corollary~\ref{cor:sq}.

% \newpage

\appendix

\section{Proof of Theorem~\ref{thm:freq_ub}}
\label{apx:freq_ub}
In this section we prove Theorem~\ref{thm:freq_ub}. The proof is a simple consequence of the privacy amplification result of \cite{balle_privacy_2019} and known accuracy bounds for locally-\DP\ protocols. We first recall the privacy amplification result:
\begin{theorem}[Privacy amplification of single-message shuffling, \cite{balle_privacy_2019}, Corollary 5.3.1]
  \label{thm:pamp}
Suppose $R : \MX \ra \MZ$ is an $(\ep_L, 0)$-locally \DP\ randomizer with $\ep_L \leq \frac{\ln(n/\ln(1/\delta))}{2}$ for some $\delta > 0$. Then the shuffled algorithm $(x_1, \ldots, x_n) \mapsto S(R(x_1), \ldots, R(x_n))$ is $(\ep, \delta)$-\DP\ with $\ep = O \left( \ep_L \cdot e^{\ep_L} \cdot \sqrt{\ln(1/\delta)/n} \right)$. %
\end{theorem}
\begin{proof}[Proof of Theorem~\ref{thm:freq_ub}]
  We first treat the case of $n \leq \tilde O(\B^2)$ (i.e., the cases (\ref{eq:const_ub}), (\ref{eq:interm_ub})), where the locally differentially private protocol we use is $\B$-RAPPOR \cite{erlingsson2014rappor,duchi_minimax}. In particular, we consider the protocol $P_{\RAP} = (R_{\RAP}, A_{\RAP})$. For a given privacy parameter $\ep_L \geq 1$, the local randomizer $R_{\RAP} : [\B] \ra \{0,1\}^\B$ is defined as follows: for $v \in [\B]$, $R_{\RAP}(v) = (Z_1, \ldots, Z_\B)$, where each $Z_k$ is an independent bit that equals $(e_v)_k$ with probability $\frac{\exp(\ep_L/2)}{1 + \exp(\ep_L/2)}$ and equals $1 - (e_v)_k$ with probability $\frac{1}{1 + \exp(\ep_L/2)}$. For later use in the proof, we will also define $R_{\RAP}(\varnothing) = (Z_1, \ldots, Z_\B)$, where each $Z_k \sim \Ber \left( \frac{1}{1 + \exp(\ep_L/2)} \right)$.

  The analyzer $A_{\RAPPOR} : (\BB^\B)^n \ra \BB^\B$ is defined as follows: given as input $n$ bit-vectors $(z^1, \ldots, z^n)$, the analyzer outputs the vector $(\hat x_1, \ldots, \hat x_\B) \in [0,1]^n$ of frequency estimates defined by
  \begin{equation}
    \label{eq:arappor}
\hat x_j = \frac 1n \cdot \left( \left(\sum_{i=1}^n z^i_j\right) - \frac{n}{1+\exp(\ep_L/2)}\right) \cdot \left( \frac{\exp(\ep_L/2) + 1}{\exp(\ep_L/2) - 1}\right).
\end{equation}
% (The above expression for $\hat x_j$ is chosen so that $\hat x_j$ is an unbiased estimator of the true frequency $\frac 1n \sum_{i=1}^n (e_{x_i})_j$ when the $z^i$ are the outputs $R(x_i)$ of the local randomizers given inputs $x_i\in [\B]$.)
First we prove the accuracy of the local-model protocol $P_{\RAP} = (R_{\RAP},A_{\RAP})$; it is clear, by symmetry of $A_{\RAP}$ that the same accuracy bounds hold when we insert the shuffler $S$. The expression (\ref{eq:arappor}) satisfies the following property: for any dataset $X = (x_1, \ldots, x_n)$, if $Z^i := R_{\RAP}(x_i)$, and it happens that  $\sum_{i=1}^n Z^i_j$ equals its expected value (over the randomness in the local randomizers $R_{\RAP}$), then $\hat x_j$ is equal to the true frequency $\frac 1n \sum_{i=1}^n (e_{x_i})_j$. It follows that if $\left| \frac 1n\sum_{i=1}^n Z^i_j - \E_{R_{\RAP}} \left[ \frac 1n\sum_{i=1}^n Z^i_j \right] \right| \leq \ha$, then $\left| \hat x_j - \frac 1n \sum_{i=1}^n (e_{x_i})_j \right| \leq \ha \cdot \left( \frac{\exp(\ep_L/2) + 1}{\exp(\ep_L/2) - 1} \right) \leq O(\ha)$, where the final inequality follows from $\ep_L \geq 1$. 

Let $p = \frac{1}{\exp(\ep_L/2) + 1}$. The random variable $\left| \sum_{i=1}^n Z^i_j - \E_{R_{\RAP}} \left[  \sum_{i=1}^n Z^i_j \right] \right|$ is stochastically dominated by the random variable $\left| Y - np \right|$, where $Y \sim \Bin\left( n, p \right)$. It follows that if $Y_1, \ldots, Y_\B \sim \Bin\left( n, p\right)$ i.i.d., then
\begin{align}
  & \E \left[\max_{j \in [\B]}\left| \sum_{i=1}^n Z^i_j - \E_{R_{\RAP}} \left[ \sum_{i=1}^n Z^i_j \right] \right|\right] \nonumber\\
  & \leq \E \left[\max_{j \in [\B]} \left| Y_j - np \right|\right] \nonumber\\
  \label{eq:maxjy}
  & \leq \E \left[\left( \max_{j \in [\B]} \{ Y_j\} - np \right)\right] + E[|Y_1 - np|].
\end{align}
By Jensen's inequality $\E[|Y_1 - np|] \leq \sqrt{np}$. By Exercise 2.19 of \cite{boucheron2012concentration}, we have that
$$
\E \left[ \max_{j \in [\B]} \{ Y_j \} \right] \leq np \exp\left( 1 + W \left( \frac{\ln(\B) - np}{enp} \right) \right),
$$
where $W(\cdot)$ is the Lambert $W$ function.\footnote{The lambert $W$ function $W : [-1/e, \infty) \ra \BR$ is defined implicitly by $W(x) \exp(W(x)) = x$ and $W(x) \geq -1$ for all $x \geq -1/e$. It is an increasing function.} We consider two cases regarding the value of $pn$:

{\bf Case 1.} $epn < \ln \B$. In this case we use the fact that $W(x) \leq \ln(e \cdot x)$ for all $x \geq 1$. Then since $W$ is increasing,
\begin{align*}
  np \exp\left( 1 + W \left( \frac{\ln(\B) - np}{enp} \right) \right) & \leq np \exp \left( 1 + W \left( \ln(\B) / (enp) \right) \right) \\
                                                                      & \leq np \exp(1 + \ln(\ln(\B)/(np)))\\
                                                                      & \leq np \exp(\ln(e\ln(\B)/(np)))\\
  & = np \cdot \frac{e \ln (\B)}{np} = e \ln \B.
\end{align*}
Thus in this case (\ref{eq:maxjy}) is bounded above by $O(\ln \B)$.

{\bf Case 2.} $epn \geq \ln \B$. In this case we use the fact that $W\left( \frac{-1}{e} + x \right) \leq -1 + 3\sqrt{x}$ for all $x \geq 0$. In particular, it follows from this fact that
\begin{align*}
  np \exp\left( 1 + W \left( \frac{\ln(\B)}{enp} - \frac{1}{e} \right) \right) & \leq np \exp\left( 1 -1 + 3\sqrt{\frac{\ln(\B)}{enp}} \right) \\
                                                                               & \leq np \cdot O \left( \sqrt{\frac{\ln \B}{enp}}\right)\\
  & = O \left( \sqrt{np \ln(\B) } \right),
\end{align*}
where the second inequality uses the fact that $\ln(\B) / np = O(1)$. 
Thus in this case (\ref{eq:maxjy}) is bounded above by $O(\sqrt{np \ln \B})$. 

Next we analyze privacy of $R_{\RAP}$ in the $n$-user shuffled model. It is clear that $R_{\RAP}$ is $(\ep_L, 0)$-\DP. In fact, $R_{\RAP}$ satisfies the following stronger property: for any $v \in [\B]$, and any vector $z \in \BB^\B$, we have that $e^{-\ep_L/2} \p[R_{\RAP}(\varnothing) = z] \leq \p[R_{\RAP}(v) = z] \leq e^{\ep_L/2} \p[R_{\RAP}(\varnothing) = z]$. Now write $M(x_1, \ldots, x_n) = S(R_{\RAP}(x_1), \ldots, R_{\RAP}(x_n))$. It is not difficult to see by inspecting the proof of Theorem~\ref{thm:pamp} that the following holds, as long as $\ep, \delta$ are chosen so that $\frac{\ep_L}{2} \leq \frac{\ln(n/\ln(1/\delta))}{2}$ and $\ep = O \left( \ep_L e^{\ep_L/2} \sqrt{\ln(1/\delta)/n} \right)$: For any subset $\MS \subset (\BB^\B)^n$, and any dataset $(x_1, \ldots, x_n) \in [\B]^n$,
\begin{align*}
& \p\left[ M(x_1,\ldots, x_{n-1}, x_n) \in \MS \right] \leq e^{\ep} \p\left[ M(x_1, \ldots, x_{n-1}, \nn) \in \MS \right ] + \delta\\
& \p\left[ M(x_1,\ldots, x_{n-1}, \nn) \in \MS \right] \leq e^{\ep} \p\left[ M(x_1, \ldots, x_{n-1}, x_n) \in \MS \right ] + \delta.
\end{align*}
It follows that $(x_1, \ldots, x_n) \mapsto M(x_1, \ldots,x_n)$ is $(2\ep, \delta(1+e^{\ep}))$-\DP\ (i.e., in the $n$-user shuffled model). Thus, by choosing $\ep_L = \ln(n/\ln(1/\delta)) - 2\ln\ln n + 2 \ln(\ep) + O(1)$, we obtain the accuracy bounds in (\ref{eq:const_ub}) and (\ref{eq:interm_ub}); in particular, the accuracy bound in Case 1 corresponds to $n \leq \frac{\ep^2 \log^2 \B}{\log^3 \log \B}$ and the accuracy bound in Case 2 corresponds to $n \geq \frac{\ep^2 \log^2 \B}{\log^3 \log \B}$. 

Finally we treat the case $n >  \Omega(\B^2)$ (i.e., the case (\ref{eq:small_ub})). In this case we will use the local randomizer of $\B$-randomized response \cite{warner1965randomized}. In particular, the local randomizer $R_{\RR} : [\B] \ra [\B]$ is defined as follows: for $u,v \in [\B]$,
$$
\p[R_{\RR}(v) = u] = \begin{cases}
  \frac{\exp(\ep_L)}{\exp(\ep_L) + \B - 1} &: u = v \\
  \frac{1}{\exp(\ep_L) + \B - 1} &: u \neq v.
\end{cases}
$$
The analyzer $A_{\RR} : [\B]^n \ra [\B]^n$, when given outputs of local randomizers $(z_1, \ldots, z_n) \in [\B]^n$, produces frequency estimates $A(z_1, \ldots, z_n) = (\hat x_1, \ldots, \hat x_\B)$, given by
\begin{equation}
  \label{eq:arr}
\hat x_j = \frac{1}{n} \left( \left( \sum_{i=1}^n (e_{z_i})_j \right) - \frac{n}{\exp(\ep_L) + \B - 1}\right) \cdot \left( \frac{\exp(\ep_L) + \B - 1}{\exp(\ep_L) -1} \right).
\end{equation}
First we analyze the accuracy of $A_{\RR}$. %
The analysis is quite similar to that of $A_{\RAP}$. In particular, first note that (\ref{eq:arr}) satisfies the following property: for any dataset $X = (x_1, \ldots, x_n)$, if $Z^i := R_{\RR}(x_i)$, and it happens that  $\sum_{i=1}^n (e_{Z^i})_j$ equals its expected value (over the randomness in the local randomizers $R_{\RR}$), then $\hat x_j$ is equal to the true frequency $\frac 1n \sum_{i=1}^n (e_{x_i})_j$. It follows that if $\left| \frac 1n\sum_{i=1}^n (e_{Z^i})_j - \E_{R_{\RR}} \left[ \frac 1n\sum_{i=1}^n (e_{Z^i})_j \right] \right| \leq \ha$, then $\left| \hat x_j - \frac 1n \sum_{i=1}^n (e_{x_i})_j \right| \leq \ha \cdot \left( \frac{\exp(\ep_L) + \B - 1}{\exp(\ep_L) - 1} \right) \leq O \left(\ha \cdot \left(\frac{\exp(\ep_L) + \B}{\exp(\ep_L)}\right)\right)$, where the last inequality follows from $\ep_L \geq 1$ (as we will see below).

Let $q = \frac{\B-1}{\exp(\ep_L) + \B - 1}$. The random variable $\left| \sum_{i=1}^n (e_{Z^i})_j - \E_{R_{\RR}} \left[  \sum_{i=1}^n (e_{Z^i})_j \right] \right|$ is stochastically dominated by the random variable $\left| Y - n(1-q) \right|$, where $Y \sim \Bin\left( n, 1-q \right)$. It follows by binomial concentration that if $Y_1, \ldots, Y_\B \sim \Bin\left( n, 1-q\right)$ i.i.d. and $qn \geq \Omega(\ln \B)$, then
$$
\E_{R_{\RR}} \left[ \max_{j \in [\B]} \left|\sum_{i=1}^n (e_{Z^i})_j  - \E_{R_{\RR}}\left[ \sum_{i=1}^n (e_{Z^i})_j \right]\right| \right] \leq O(\sqrt{qn \ln(\B)}). % NOTE: getting this bound is somewhat tricky, no time to spell it out (probably not worth it anyway)
$$
Thus the expected error of $P_{\RR} = (R_{\RR}, A_{\RR})$ is bounded above by
\begin{equation}
  \label{eq:rr_1m_error}
\E_{(\hat x_1, \ldots, \hat x_\B) \sim P_{\RR}}\left[ \max_{j \in [\B]} \left| \hat x_j -\frac 1n \sum_{i=1}^n (e_{x_i})_jk \right| \right]\leq O\left( \sqrt{q \ln(\B)/n} \cdot \left( \frac{\exp(\ep_L) + \B}{\exp(\ep_L)} \right) \right).
\end{equation}

Next we analyze the privacy of $R_{\RR}$ in the $n$-user shuffled model. To do so, note that $R_{\RR}$ is clearly $(\ep_L, 0)$-(locally) \DP. Thus, by \cite[Theorem 3.1]{balle_privacy_2019}\footnote{In particular, the parameter $\gamma$ in Theorem 3.1 of \cite{balle_privacy_2019} is set to $\frac{\B}{\exp(\ep_L) + \B-1}$.}, if we take $\ep_L = {\ln(n/\ln(1/\delta))} + 2\ln(\ep) + O(1)$, then by the assumption $\ep \geq \omega(\ln^2(n) / \min\{\sqrt \B, \sqrt n \})$, the shuffled-model protocol $(R_{\RR}, S, A_{\RR})$ is $(\ep, \delta)$-\DP\ in the $n$-user shuffled model. % \footnote{Alternatively, the $(\ep, \delta)$-differential privacy of this protocol in the shuffled model can be established from Theorem~\ref{thm:pamp} using a similar (though more involved) trick as was used for $P_{\RAP}$.} 
As long as $\sqrt{n} > \B$ (i.e., $n > \B^2$, so that $\frac{n\ep^2}{\ln 1/\delta} \geq \Omega(\B)$), the error in (\ref{eq:rr_1m_error}) is bounded above by $O(\sqrt{q \ln(\B)/n}) = O\left(\frac{1}{n\ep} \cdot \sqrt{\B \ln(n) \ln(\B)}\right)$.
\end{proof}

\section{Low-Communication Simulation of Sparse Non-Adaptive SQ Algorithms}\label{app:counting-queries}
% \label{sec:scq}
% \noah{Am thinking of making the following changes:
%   \begin{itemize}
%   \item Call counting queries statistical queries (which actually is somewhat standard in the literature), and statistical queries ``real-valued statistical queries''.
%   \item Move last para of this section to beginning.
%   \item Change terminology from ``$k$-sparse family of counting queries'' to ``non-adaptive $k$-sparse SQ algorithm''. (in which case the main corollary of this section will say that the total variation distance between the output of the SQ algorithm and the shuffled algorithm is $\leq \delta$).
%   \item Assume iid data, and state the main corollary here for a non-adaptive $k$-sparse SQ algorithm with $\B$ queries, in the regime that $\ep \leq \tau$ ($\tau$ is the SQ accuracy). The assumption $\ep \leq \tau$ means the non-private sample complexity is ``small'' (i.e., $1/(\ep \tau) > 1/(\tau^2)$),  so we can just take a sample and then have all $n$ users execute all statistical queries using CM or Hadamard.
%   \item This improves upon local DP, which would get $n \approx \frac{k}{\tau^2 \ep^2}$, and is the same as central DP up to log factors.
%     \item Perhaps put an informatl theorem statement in intro, perhaps as corollary to Theorem 1.3 -- my concern is the SQ reviewer will shoot the paper down unless we have something like that.
%     \end{itemize}}
We now discuss an equivalent formulation of our results in terms of non-adaptive statistical query algorithms. A {\it statistical query} on a set $\MX$ is specified by a binary-valued predicate function $q : \MX \ra \{0,1\}$, and, for a distribution $\MD$ on $\MX$, takes the value $q(\MD) :=  \E_{x \sim \MD}[q(x)]$. Special cases of statistical queries include \emph{frequency queries}, specified by $q(x) = \One[x = y]$ for some $y \in \MX$, and \emph{range queries}, given by $q(x) = \One[x \in R]$, where $R$ is a rectangle in $\MX$. For $\tau \in [0,1]$, a statistical query oracle $\SQ_{\MD,\tau}$ of {\it tolerance $\tau$}, takes as input a statistical query $q$ and outputs a value $\SQ_{\MD,\tau}(q) \in [q(\MD) - \tau, q(\MD) + \tau]$.  A {\it statistical query (SQ) algorithm of tolerance $\tau$}, $\MQ$, may access the distribution $\MD$ through a number of queries $q$ to an oracle $\SQ_{\MD,\tau}$. $\MQ$ is called {\it non-adaptive} if the distribution of its queries is fixed \emph{a priori}, i.e., does not depend on the results of any of these queries. It was observed in the work of Blum et al.~\cite{blum2005practical} that any statistical query (SQ) algorithm can be simulated by a differentially private protocol (in the central model). The same was shown for locally differentially private protocols by Kasiviswanathan et al.~\cite{kasiviswanathan2008what}, albeit with worse parameters. % \footnote{In particular, for a set $\MQ$ of statistical queries, the simple protocol of~\cite{blum2005practical} achieves an error of $\tilde O(\sqrt{|\MQ|}/n\ep)$ whereas that of~\cite{kasiviswanathan2008what} gets error of $\tilde O(\sqrt{|\MQ|/n\ep^2})$. The results of~\cite{blum2008learning,hardt2010multiplicative,nikolov2013geometry} and many others, combined with the work of Bshouty and Feldman~\cite[Lemma 5]{bshouty2002using} give better bounds in the central model for many choices of $\MQ, n, \ep, \delta$.} In fact, % in the setting that the database entries $x_1, \ldots, x_n$ are drawn i.i.d.~from some distribution $\MD$,
In fact, it is known~\cite{kasiviswanathan2008what} that (non-adaptive) SQ algorithms are equivalent to (noninteractive) locally-differentially private algorithms, up to a polynomial factor in the tolerance
$\tau$. We refer the reader to \cite{kasiviswanathan2008what} for further background on SQ algorithms. % Much work in differential privacy in the central model (e.g.,~\cite{vadhan2017complexity,blum2008learning,hardt2009geometry,hardt2010multiplicative,nikolov2013geometry}) has focused on understanding the optimal error for a family $\MQ$ of statistical queries, namely: for which $\ha$ is there a differentially-private mechanism that, given as input $(x_1, \ldots, x_n)$, outputs, for each $q \in \MQ$, an $\ha$-additive approximation to $q(x_1, \ldots, x_n)$?

A straightforward corollary of the techniques used to show Theorem~\ref{thm:hist_intro} is that one can efficiently and privately simulate {\it sparse} non-adaptive statistical query algorithms in the shuffled model. In particular, for $k \in \BN$, we say that a non-adaptive SQ algorithm $\MQ$ is {\it $k$-sparse} if, for each $x \in \MX$, with probability 1, there are at most $k$ distinct statistical queries $q$ that $\MQ$ makes satisfying $q(x) = 1$. Sparsity is a more stringent condition than having low sensitivity: in particular, if a $k$-sparse algorithm makes $\B$ queries $q_1, \ldots, q_\B$, then, letting $\hat \MD$ be the empirical distribution over a dataset $(x_1, \ldots, x_n)$, the mapping $(x_1, \ldots, x_n) \mapsto (q_1(\hat \MD), \ldots, q_\B(\hat \MD))$ has $\ell_1$ sensitivity $2k/n$. % the collection of queries made by a $k$-sparse algorithm has sensitivity at most $k$, since the addition or removal of a % For instance, if $\MQ$ makes only frequency queries, then it is 1-sparse; on the other hand, if $\MQ$ makes at most $\B$ queries, it is trivially $\B$-sparse. 
We show in Corollary~\ref{cor:cq_sparse} below that for any universe $\MX$, if $\MQ$ is a $k$-sparse SQ algorithm making at most $\B$ queries, then one can privately simulate $\MQ$ in the shuffled model, as long as the tolerance $\tau$ is roughly a multiplicative factor of $k/n$ times the corresponding error in Theorem~\ref{thm:hist_intro} (and communication blown up by a factor of $k^{2}$). 

Next we state and prove the formal version of Corollary \ref{cor:sq}:
\begin{corollary}
  \label{cor:cq_sparse}
  Fix any set $\MX$, and let $\MQ$ be a $k$-sparse non-adaptive SQ algorithm on $\MX$ making at most $\B$ queries of tolerance $\tau$ for a distribution $\MD$.
  Suppose that $\ep, \delta \in (0,1)$,
  \begin{equation}
    \label{eq:sq-private-tol}
n  \geq \Omega \left( \frac{\log(\B/\beta)}{\tau^2} + \frac{k \log(\B/\beta) \sqrt{ \log \left( \frac{k}{\delta \ep} \right)}}{\ep\tau} \right).
\end{equation}
% Then there exist:
  % \begin{itemize}
  
Then there is a private-coin $(\epsilon, \delta)$-differentially private algorithm $P$ in the shuffled that receives as input $n$ iid samples $x_1, \ldots, x_n \sim \MD$, and produces output that agrees with that of $\MQ$ with probability at least $1-\beta$. Moreover, each user sends in expectation $O\left(\frac{k^2 \log(k/(\delta\ep))}{\ep^2} \right)$ messages consisting of $O(\log n \log \B)$ bits each.
\end{corollary}
The sample complexity bound (\ref{eq:sq-private-tol}) improves upon an analogous result for locally differentially private simulation of $\MQ$, for which $n \geq \tilde \Omega\left(\frac{k}{\tau^2 \ep^2} \right)$ samples suffice \cite[Theorem 5.7]{edmondspower2020,kasiviswanathan2008what}. Moreover, for small $k$, it is close to what one gets in the central model, namely that $n \geq \tilde \Omega\left( \frac{1}{\tau^2} + \frac{\sqrt{k}}{\tau \ep} \right)$ \cite{blum2005practical} samples suffice. These observations follow from the fact that the $\ell_2$ sensitivity of the collection of queries made by $\MQ$ is bounded above by $\sqrt{k}$.
\begin{proof}[Proof of Corollary \ref{cor:cq_sparse}]
  Let us fix a set of $\B$ queries $q_1, \ldots, q_\B$ made by $\MQ$. We will show that with probability $1-\beta$ over the sample $X := (x_1, \ldots, x_n)$, the algorithm $P$ can output real numbers $P_1(X), \ldots, P_\B(X)$ so that for $1 \leq j \leq \B$, $\left| q_j(\MD) - P_j(X)  \right| \leq \tau$. Thus, conditioned on $\MQ$ making the queries $q_1, \ldots, q_\B$, $P$ simulates an SQ oracle of tolerance $\tau$ with probability $1-\beta$. The claimed result regarding the accuracy of $P$ follows by taking expectation over $q_1, \ldots, q_\B$.

  As long as $n \geq C \cdot \frac{\log \B}{\tau^2}$ for a sufficiently large constant $C$, the Chernoff-Hoeffding bound guarantees that with probability $1-\beta/2$, for $1 \leq j \leq \B$, we have $\left| q_j(\MD) - \frac{1}{n} \sum_{i=1}^n q_j(x_i) \right| \leq \tau/2$. Define a new universe $\MX' := \{ (q_1(x), q_2(x), \ldots, q_\B(x)) : x \in \MX \}$. Since $\MQ$ is $k$-sparse, we have that $\MX' \subseteq \{ v \in \{0,1\}^\B : \| v \|_1 \leq k\}$. Now let $P$ be the protocol $P^{\Had}$ with the parameters from Theorem \ref{thm:Had_k_main}, and $P(X)_j$, $1 \leq j \leq \B$, be the values $\hat x_j/n$ from Theorem \ref{thm:Had_k_main}. Then Theorem \ref{thm:Had_k_main} gives that
  $$
  \p\left[ \forall j \in [\B] : \left| \frac{1}{n} \sum_{i=1}^n q_j(x_i) - P(X)_j \right| \leq O \left( \frac{\log(\B) k \sqrt{\log(k/(\ep\delta))}}{n} \right)\right] \geq 1-\beta/2.
  $$
  By the choice of $n$ in (\ref{eq:sq-private-tol}), with probability at least $1-\beta/2$ over $P^{\Had}$, we have $\left| P(X)_j - \frac 1n \sum_{i=1}^n q_j(x_i) \right| \leq \tau/2$. Thus, with probability at least $1-\beta$ over $P^{\Had}$ and the sample $X$, we have $\left| P(X)_j - q_j(\MD) \right| \leq \tau$ for all $j \in [\B]$, as desired.
\end{proof}

Instead of the Hadamard response-based protocol, we could use the (public coin) count-min sketch based protocol $P^{\CM}$ of Theorem \ref{thm:CM_main} in the above corollary. This would give the inferior sample complexity bound of % ; A public-coin $(\epsilon, \delta)$-differentially private algorithm in the shuffled model that, on input $(x_1, \ldots, x_n) \in \MX^n$, outputs a vector $v \in \BR^\B$ such that with probability at least $1-\beta$, for all $j \in [\B]$,
  $$
n \geq  \Omega \left(  \frac{k \log^{3/2}(\B)\sqrt{\log(k(\log \B)/\delta)}}{\ep\tau }  \right),
$$% with error $O(\frac{\log\B}{\ep} + \frac{\log (1/\delta)}{\ep^2})$ and with $O(( \frac{\log \B}{\ep} + \frac{\log(1/\delta)}{\ep^2}) \log \B)$ messages per user with each message consisting of $O(\log{n} + \log\log{B})$ bits.
an would involve each user sending in expectation $O \left( \frac{(\log^3 \B)k^3 \log(k (\log \B)/\delta)}{\ep^2} \right)$ messages consisting of $O(\log k + \log n + \log \log \B)$ bits each. (Recall that the advantage of the count-min sketch based protocol was efficient computation of a given statistical query in the data-structural setting when $\B$ is prohibitively large to compute all of them.)

Notice that the error and communication bounds in Corollary~\ref{cor:cq_sparse} degrade polynomially in $k$; thus, for families $\MQ$ which do not have any particular sparsity structure and for which $|\MQ| \geq n$, the bounds of Corollary~\ref{cor:cq_sparse} are vacuous. In the central model of differential privacy, much effort has gone into determining the optimal sample complexity of simulating in a differentially private manner an arbitrary (not necessarily sparse) non-adaptive statistical query algorithm. For instance, Nikolov et al.~\cite{nikolov2013geometry} demonstrated a mechanism that achieves sample complexity nearly equal to an efficiently computable lower bound for any differentially private mechanism releasing a fixed set of statistical queries.\footnote{This optimality is with respect to mean squared error over the set of queries.} 
The earlier work of Hardt and Rothblum~\cite{hardt2010multiplicative} showed that there is an $(\ep, \delta)$-differentially private mechanism that with high probability simulates a non-adaptive SQ algorithm $\MQ$ making $\B$ queries as long as  $n \geq \Omega\left(\frac{\log(\B)\sqrt{\log |\MX|\log(1/\delta)}}{\ep \tau^2} \right)$. % We leave the question of generalizing these results to the shuffled model as an interesting direction for future work.

We leave the question of generalizing Corollary~\ref{cor:cq_sparse} to the case of non-sparse $\MQ$ in a way analogous to \cite{nikolov2013geometry,hardt2010multiplicative} as an interesting question for future work. In particular, we would hope to maintain polylogarithmic-in-$\B$ growth of the tolerance and communication, while settling for a number of samples $n$ that grows as $\frac{1}{\ep \tau^\alpha}$ for some $\alpha > 1$.% , as in, for instance,~\cite{blum2008learning,hardt2010multiplicative,nikolov2013geometry}, to future work.

% \noah{Remove most of this...} Counting queries themselves are particular cases of statistical queries, which are specified by a {\it real-valued} function $R : \MX \ra [0,1]$. Given a dataset $(x_1, \ldots, x_n) \in \MX^n$, the statistical query (SQ) specified by $R$ takes the value $q_R(x_1, \ldots, x_n) := \frac 1n \sum_{i=1}^n R(x_i)$.\footnote{Depending on the context, the normalizing factor $1/n$ may sometimes be dropped.} An {\it SQ algorithm} may access the dataset through a number of {\it statistical query oracles} $\MO_{R,\tau}$ (with $R : \MX \ra [0,1]$ and $\tau \geq 0$), where a call to $\MO_{R,\tau}(x_1, \ldots, x_n)$ returns a value $v$ such that $|v - q_R(x_1, \ldots, x_n)| \leq \tau$. %  We focus on {\it nonadaptive} SQ algorithms, namely those that make all of their queries before seeing any of the answers.
% It is natural to wonder what can be said about the private computation of arbitrary classes of statistical queries in the shuffled model (beyond the algorithms from the local model that immediately carry over):
 % One might hope to apply the result of Bshouty and Feldman~\cite[Lemma 5]{bshouty2002using} stating that any statistical query with error $\tau$ can be simulated with $O(\log 1/\tau)$ counting queries to generalize Corollary~\ref{cor:cq_sparse} to some subset of ``sparse statistical queries''. Unfortunately, the reduction in~\cite{bshouty2002using} does not preserve any meaningful notion of sparsity, so we also leave this question for future work.

\section{Proofs of Auxiliary Lemmas from Section~\ref{sec:freq_oracle_heavy_hitters}}\label{sec:CM_app}

In this section, we prove Lemmas~\ref{lem:lap_gen} and~\ref{lem:bin_smooth_small_gamma}.

\begin{replemma}{lem:lap_gen}
Suppose $f : \MX^n \ra \BZ^m$ is $k$-incremental (Definition~\ref{def:incremental}) and $\Delta(f) = \Delta$. Suppose $\MD$ is a distribution supported on $\BZ$ that is $(\ep, \delta, k)$-smooth. Then the mechanism
$$
X \mapsto f(X) + (Y_1, \ldots, Y_m),
$$
where $Y_1, \ldots, Y_m \sim \MD$, i.i.d., is $(\ep', \delta')$-differentially private, where $\ep' = \ep \cdot \Delta, \delta' = \delta \cdot \Delta$.
\end{replemma}
% \lemlapgen*
\begin{proof}[Proof of Lemma~\ref{lem:lap_gen}]
    Consider neighboring datasets $X = (x_1, \ldots, x_{n-1}, x_n)$ and $X' = (x_1, \ldots, x_{n-1}, x_n')$. We will show 
    \begin{equation}
    \label{eq:zzp_1}
    \p_{y_1, \ldots, y_m \sim \MD} \left[ \frac{\p_{Y_1, \ldots, Y_m \sim \MD}[f(X) + (Y_1, \ldots, Y_m) = f(X) + (y_1, \ldots, y_m)]}{\p_{Y_1, \ldots, Y_m \sim \MD}[f(X') + (Y_1, \ldots, Y_m) = f(X) + (y_1, \ldots, y_m)]}\geq e^{\ep'} \right] \leq \delta'.
    \end{equation}
    To see that (\ref{eq:zzp_1}) suffices to prove the Lemma~\ref{lem:lap_gen}, fix any subset $\MS \subset \BZ^m$, and write $P(X) = f(X) + (Y_1, \ldots, Y_m)$ to denote the randomized protocol. Let $\MT$ denote the set of $f(X) + (y_1, \ldots, y_m) \in \BZ^m$ such that the event in (\ref{eq:zzp_1}) does not hold; then we have $\p[P(X) \not \in \MT] \leq \delta'$. It follows that 
    \begin{align*}
    \p\left[P(X) \in \MS \right] & \leq \delta' + \sum_{w \in \MT \cap \MS} \p[P(X) = w]\\
    & = \delta' + \sum_{w \in \MT \cap \MS} \p_{Y_1, \ldots, Y_m \sim \MD}[f(X) + (Y_1, \ldots, Y_m) = w] \\
    & \leq \delta' + \sum_{w \in \MT \cap \MS} e^{\ep'}\p_{Y_1, \ldots, Y_m \sim \MD}[f(X') + (Y_1, \ldots, Y_m) = w ]\\
    & = \delta' + \sum_{w \in \MT \cap \MS} e^{\ep'} \cdot \p[P(X') = w] \\
    & \leq \delta' + e^{\ep'} \p[P(X') \in \MS].
    \end{align*}
It then suffices to show (\ref{eq:zzp_1}).  For $j \in [m]$, let $k_j = f(X)_j - f(X')_j$. Since the sensitivity of $f$ is $\Delta$, we have $\sum_{j=1}^m |k_j| \leq \Delta$. It follows that (\ref{eq:zzp_1}) is equivalent to

% Let $\MJ_1$ be the set of $j \in [m]$ such that $f(X)_j = f(X')_j + 1$ and let $\MJ_{-1}$ be the set of $j \in [m]$ such that $f(X)_j = f(X')_j - 1$. Since $f$ is incremental, for all $j \in [m] \backslash (\MJ_1 \cup \MJ_{-1})$, we have $f(X)_j = f(X')_j$. Moreover, since the sensitivity of $f$ is $\Delta$, we have that $|\MJ_1 \cup \MJ_{-1}| \leq \Delta$. 

% It follows that (\ref{eq:zzp_1}) is equivalent to
\begin{equation}
\label{eq:zzp_2}
\p_{y_1, \ldots, y_m \sim \MD} \left[ \prod_{j=1}^m
\frac{\p_{Y_j \sim \MD}[Y_j = y_j]}{\p_{Y_j \sim \MD}[Y_j = y_j + k_j]} % \cdot \prod_{j \in \MJ_{-1}} \frac{\p_{Y_j \sim \MD}[Y_j= y_j]}{\p_{Y_j \sim \MD}[Y_j = y_j - 1]} 
\geq e^{\ep'} \right] \leq \delta'.
\end{equation}
For (\ref{eq:zzp_2}) to hold it in turn suffices, by a union bound and the fact that at most $\Delta$ of the $k_j$ are nonzero, that for each $j$ with $k_j \neq 0$,
\begin{equation}
\label{eq:use_smooth_1}
\p_{y \sim \MD} \left[ \frac{\p_{Y \sim \MD}[Y = y]}{\p_{Y \sim \MD}[Y = y + k_j]} \geq e^{|k_j|\ep'/\Delta} \right] \leq \delta' / \Delta.
\end{equation}
% and
% \begin{equation}
% \label{eq:use_smooth_2}
% \p_{y \sim \MD} \left[ \frac{\p_{Y \sim \MD}[Y = y]}{\p_{Y \sim \MD}[Y = y - 1]} \geq e^{\ep'/\Delta} \right] \leq \delta' / \Delta.
% \end{equation}
But (\ref{eq:use_smooth_1}) follows for $\ep'/\Delta = \ep, \delta'/\Delta = \delta$ since $\MD$ is $(\ep, \delta, k)$-smooth. This completes the proof.
%\noah{Can improve result to $\ep' \approx \sqrt{\Delta} \cdot \ep$ by using Azuma's inequality + slight modification to smoothness definition (same proof as composition of $(\ep,\delta$)-DP.)}
\end{proof}

\begin{replemma}{lem:bin_smooth_small_gamma}
Let $n \in \BN$, $\gamma \in [0,1/2]$, $0 \leq \alpha \leq 1$, and $k \leq \alpha \gamma n / 2$. Then the distribution $\Bin(n, \gamma)$ is $(\ep ,\delta, k)$-smooth with $\epsilon = \ln((1+\alpha)/(1-\alpha))$ and $\delta = e^{-\frac{\alpha^2 \gamma n}{8}} + e^{-\frac{\alpha^2 \gamma n}{8+2\alpha}}$.
\end{replemma}

\begin{proof}[Proof of Lemma~\ref{lem:bin_smooth_small_gamma}]
Recall that for $Y \sim \Bin(n, \gamma)$ and $0 \leq y \leq n$, we have $\p[Y = y] = \gamma^y (1-\gamma)^{n - y} \binom{n }{ y}$. Thus, we have that, for any $k \geq k' \geq -k$,
\begin{equation}\label{eq:ratio_bin_probs}
    \frac{\p_{Y \sim \Bin(n, \gamma)}[Y = y]}{\p_{Y \sim \Bin(n, \gamma)}[Y = y + k']} =  \frac{(1-\gamma)^{k'}}{\gamma^{k'}} \cdot \frac{(y+k')! (n-y-k')!}{y! (n-y)!}.
    %  \frac{(y+1) \cdots (y+k')}{(n-y) \cdots (n-y-k'+1)}.
\end{equation}
We define the interval $\mathcal{E} := [(1-\alpha)\gamma n + k', (1+\alpha)\gamma n - k']$ where $\alpha$ is any positive constant smaller than $1$. As long as $k' \leq \alpha \gamma n/2$, $\mathcal{E}$ contains the interval $\ME' := [(1-\alpha/2) \gamma n, (1 + \alpha/2) \gamma n]$. By the multiplicative Chernoff Bound, we have that
\begin{equation}\label{eq:tail_prob_bin}
    \p_{y \sim \Bin(n, \gamma)}[y \notin \mathcal{E}] \le e^{-\frac{\alpha^2 \gamma n}{8}} + e^{-\frac{\alpha^2 \gamma n}{8+2\alpha}}.
\end{equation}
Note that for any $y \in \mathcal{E}$, if $k' \geq 0$, it is the case that
\begin{equation}\label{eq:bound_non_tail_region}
    \frac{(1-\gamma)^{k'}}{\gamma^{k'}} \frac{(y+k')! (n-y-k')!}{y!(n-y)!} =
    \frac{(1-\gamma)^{k'}}{\gamma^{k'}} \cdot \frac{(y+1) \cdots (y+k')}{(n-y) \cdots (n-y-k'+1)}
    \le (1+\alpha)^{k'}.
\end{equation}
% where we have used the assumptions that $n$ is sufficiently larger, $\gamma = \omega(1/n)$ and $\gamma = o(1)$. 
For $y \in \ME$ and if $k' \leq 0$, it is the case that
\begin{equation}
\label{eq:bound_non_tail_region_neg}
    \frac{(1-\gamma)^{k'}}{\gamma^{k'}} \frac{(y+k')! (n-y-k')!}{y!(n-y)!} =
    \frac{\gamma^{|k'|}}{(1-\gamma)^{|k'|}} \cdot \frac{(n-y+1) \cdots (n-y+|k'|)}{y (y-1) \cdots (y-|k'|)} \leq \left( \frac{(1-\gamma + \gamma \alpha)}{(1-\gamma)(1-\alpha)} \right)^{|k'|} \leq \left(\frac{1+\alpha}{1-\alpha}\right)^{|k'|},
\end{equation}
where the last inequality above uses $\gamma \leq 1/2$. 
We now proceed to show smoothness by conditioning on the event $y \in \mathcal{E}$ as follows:
\begin{align}
& \p_{y \sim \Bin(n, \gamma)} \left[ \frac{\p_{Y \sim \Bin(n, \gamma)}[Y = y]}{\p_{Y \sim Bin(n, \gamma)}[Y = y + k']} \geq e^{|k'|\ep} \right]\nonumber \\ 
& \le \p_{y \sim \Bin(n, \gamma)} \left[ \frac{\p_{Y \sim \Bin(n, \gamma)}[Y = y]}{\p_{Y \sim \Bin(n, \gamma)}[Y = y + k']} \geq e^{|k'|\ep} ~ | ~ y \in \mathcal{E} \right] + \p[y \not \in {\mathcal{E}}]\nonumber \\ 
& \le \p_{y \sim \Bin(n, \gamma)} \left[ \frac{\p_{Y \sim \Bin(n, \gamma)}[Y = y]}{\p_{Y \sim \Bin(n, \gamma)}[Y = y + k']} \geq e^{|k'|\ep} ~ | ~y \in \mathcal{E} \right] + e^{-\frac{\alpha^2 \gamma n}{8}} + e^{-\frac{\alpha^2 \gamma n}{8+2\alpha}}\label{eq:bounding_tail_event}\\
% & \le \p_{y \sim \Bin(n, \gamma)} \left[ \frac{(1-\gamma)}{\gamma} \frac{(y+1)}{(n-y)} \geq e^\ep ~ | ~ \mathcal{E} \right] + e^{-\frac{\alpha^2 \gamma n}{2}} + e^{-\frac{\alpha^2 \gamma n}{2+\alpha}}\label{eq:plugging_mass_ratio} \\ 
& = e^{-\frac{\alpha^2 \gamma n}{8}} + e^{-\frac{\alpha^2 \gamma n}{8+2\alpha}}\label{eq:zeroing_prob},
\end{align}
where (\ref{eq:bounding_tail_event}) follows from (\ref{eq:tail_prob_bin}) and (\ref{eq:zeroing_prob}) follows from (\ref{eq:ratio_bin_probs}), (\ref{eq:bound_non_tail_region}), and (\ref{eq:bound_non_tail_region_neg}) as well as our choice of $\ep$.
% A similar argument implies that
% \begin{align*}
%     \p_{y \sim \Bin(n, \gamma)} \left[ \frac{\p_{Y \sim \Bin(n, \gamma)}[Y = y]}{\p_{Y \sim \Bin(n, \gamma)}[Y = y + 1]} \geq e^\ep \right] \le e^{-\frac{\alpha^2 \gamma n}{2}} + e^{-\frac{\alpha^2 \gamma n}{2+\alpha}},
% \end{align*}
% while using instead of (\ref{eq:ratio_bin_probs}) we use the equality
% \begin{equation*}
%     \frac{\p_{Y \sim \Bin(n, \gamma)}[Y = y]}{\p_{Y \sim \Bin(n, \gamma)}[Y = y - 1]} = \frac{\gamma}{(1-\gamma)} \frac{(n-y+1)}{y},
% \end{equation*}
% and instead of (\ref{eq:bound_non_tail_region}) we use the following inequality which holds for all $y \in \mathcal{E}$:
% \begin{equation*}
%     \frac{\gamma}{(1-\gamma)} \frac{(n-y+1)}{y} \le \frac{1}{1-\alpha}.\qedhere
% \end{equation*}
\end{proof}

\section{Heavy Hitters}\label{app:hh_reduction}

Let $\tau$ denote the heavy hitter threshold, and assume that $\tau$ is large enough so that with high probability the maximum frequency estimation error of Theorem~\ref{thm:CM_main} is at most $\tau/2$. We wish to return a set of $O(n/\tau)$ elements that include all heavy hitters (it may also return other elements, so in that sense this is an approximate answer). 
One option is to use Theorem~\ref{thm:CM_main} directly: 
Iterate over all elements in $[B]$, compute an estimate of each count, and output the elements whose estimate is larger than $\tau / 2$. This gives a runtime of $\tilde{O}(B)$.

Algorithm~\ref{alg:hh_CM} can be combined with the prefix tree idea of Bassily et al.~\cite{bassily2017practical} to reduce the server decoding time (for recovering all heavy hitters and their counts up to additive polylogarithmic factors) from $\tilde{O}(B)$ to $\tilde{O}(n/\tau)$.
For completeness we sketch the reduction here.
The combined algorithm would use $\lceil \log_2(B) \rceil$ differentially private frequency estimation data structures obtained from Algorithm~\ref{alg:hh_CM}.
To make the whole data structure differentially private we decrease the privacy parameters, such that each data structure is $(\epsilon/\lceil \log_2(B) \rceil, \delta/\lceil \log_2(B) \rceil)$-differentially private.
(In turn, this increases the error and the bound on how small $\tau$ can be by a polylogarithmic factor in $B$.)

For each element $x\in [B]$ we would consider the prefixes of the binary representation of $x$, inserting the length-$i$ prefix in the $i$th frequency estimation data structure.
The decoding procedure iteratively identifies the prefixes of length 1, 2, 3, \dots with a true count of at least $\tau$.
With high probability this is a subset of the prefixes that have an estimated count of at least $\tau/2$, and there are $O(n/\tau)$ such prefixes at each level.
When a superset of these ``heavy'' prefixes have been determined at level $i$, we only need to estimate the frequencies of the two length-$(i+1)$ extensions of each considered prefix.
This reduces the server decoding time to $\tilde{O}(n/\tau)$ with high probability (while maintaining a polylogarithmic bound on the number of bits of communication per user).

\section{M-Estimation of Median and Quantiles}\label{sec:basic_stats}
We now discuss how to obtain results for M-estimation of the median and quantiles using our result for range counting (from Section~\ref{sec:range-queries}). For simplicity, let $x_1, x_2, \dots, x_n \in [0,1]$ be data points held by $n$ users.

We recall that there is no differentially private algorithm for estimating the value of the true median with error $o(1)$, i.e., computing $\tilde x$ which is within additive error $o(1)$ of the true median. This is because the true median can be highly sensitive to a single data point, which precludes the possibility of outputting a close approximation of the median without revealing much information about a single user. For instance, consider the case in which $n = 2k+1$ and $x_1 = x_2 = \cdots = x_k = 0$ while $x_{k+1} = x_{k+2}  = \cdots = x_n = 1$. It is clear that the median of this set is 0, but changing a single value $x_k$ to 1 would change the median of the set to 1. 

To get around the above limitations, we consider a different notion known as \emph{M-estimation} of the median, defined as follows: Consider the function
\[
 M(y) = \frac{1}{2} \sum_{i=1}^n |x_i - y|.
\]
Note that the median is a value $\tilde x$ that minimizes the quantity $M(\tilde x)$. The problem of M-estimation seeks to compute a value of $y$ which approximates this quantity, and the error is considered to be the additive error between $M(y)$ and $M(\tilde x)$. Our results imply a differentially private multi-message protocol for M-estimation that obtains both error $\poly\log{n}$ and communication per user $\poly\log{n}$ bits.

\begin{theorem}[Multi-message protocol for M-estimating the median]\label{thm:M_est_median}
Suppose $x_1, x_2, \dots, x_n \in [0,1]$. Then there is a differentially private multi-message protocol in the shuffled model that M-estimates the median of $x_1, x_2, \dots, x_n$ with communication $\poly\log{n}$ bits per user and additive error $\poly\log{n}$, i.e., outputs $y \in [0,1]$ such that
\[
M(y) \leq \min_{\tilde x} M(\tilde x) + \poly\log{n}.
\]
\end{theorem}
\begin{proof}
 We reduce the problem of M-estimation to range counting. First, we divide the interval $[0,1]$ into $B = n$ subintervals $I_1, I_2, \dots, I_B$, where $I_j = [(j-1)/B, j/B]$. Each user will associate his element $x_i$ with an index $z_i \in [B]$ corresponding to an interval $I_{z_i}$ which contains $x_j$. Note that if $j$ is the smallest element of $[B]$ such that $\left|[1, j] \cap \{z_1, z_2, \dots, z_n\}\right| \geq n/2$, then $I_j$ contains a minimizer of $M(y)$. Thus, we wish to determine this value of $j$.
 
 Thus, we obtain a protocol as follows: We use our protocol for range counting queries in the shuffled model (see Section~\ref{sec:range-queries}) to compute the queries $[1,j]$ for $j=1,2,\dots, B$ for the dataset $z_1, z_2, \dots, z_B$ and compute the first $j$ for which the query $[1,j]$ yields a count of $\geq n/2$ (or $j=B$ if no query yields such a count). Then the analyzer outputs $j/B$ as the estimate for the median.
 
 We now determine the error of thee aforementioned protocol. Note that by the guarantees of the range counting protocol, the error for the range counts is $\poly \log B$. This results in a corresponding $\poly \log B$ additive error due to range counting queries for the estimation of $M(\tilde x)$. Moreover, note that there is an additional error resulting from the discretization. Since each interval is of length $1/B$, the error resulting from discretization is $n/B$. Hence, the total error is $n/B + \poly\log B = \poly\log n$ for our choice of $B$.
\end{proof}

\begin{remark}
It should be noted that virtually the same argument as above yields a differentially private protocol for M-estimation of \emph{quantiles}. Given a set of points $x_1, x_2, \dots, x_n \in [0,1]$, we say that $y$ is a $k^\text{th}$ $q$-quantile of the dataset if
\[
  \left| [0,y) \cap \{x_1, x_2, \dots, x_n\} \right| \leq \frac{k}{q}
\]
and
\[
  \left| [0,y] \cap \{x_1, x_2, \dots, x_n\} \right| \geq \frac{k}{q}.
\]

In particular, the median is a special case, namely, the (only) $q$-quantile for $q=2$. The above argument applies, except that the function $M$ to be minimized (which is minimized by $k^\text{th}$ $q$-quantiles) is given by
\[
  M(y) = \sum_{i=1}^n \left( \left(1-\frac{k}{q}\right) (y - x_i)_+ + \frac{k}{q} (y-x_i)_- \right), 
\]
and again, the task is to find a $y$ such that
\[
M(y) \leq \min_{\tilde x} M(\tilde x) + \poly\log{n}.
\]
\end{remark}
Moreover, in the reduction to range counting queries, one instead determines the smallest value $j$ such that $\left| [1,j] \cap \{z_1,z_2, \dots, z_n\}\right| \geq kn/q$ and the rest of the analysis follows verbatim.

\bibliographystyle{alpha}
\bibliography{privacy.bib}

\newcommand{\etalchar}[1]{$^{#1}$}
\begin{thebibliography}{BBGN19b}

\bibitem[ACFT19]{acharya_test_2018}
Jayadev Acharya, Cl\'{e}ment Canonne, Cody Freitag, and Himanshu Tyagi.
\newblock Test without trust: Optimal locally private distribution testing.
\newblock In {\em AISTATS}, pages 2067--2076, 2019.

\bibitem[{App}17]{dp2017learning}
{Apple Differential Privacy Team}.
\newblock Learning with privacy at scale.
\newblock {\em Apple Machine Learning Journal}, 2017.

\bibitem[AS19]{acharya2019communication}
Jayadev Acharya and Ziteng Sun.
\newblock Communication complexity in locally private distribution estimation
  and heavy hitters.
\newblock In {\em ICML}, pages 97:51--60, 2019.

\bibitem[ASY{\etalchar{+}}18]{agarwal2018cpsgd}
Naman Agarwal, Ananda~Theertha Suresh, Felix Xinnan~X Yu, Sanjiv Kumar, and
  Brendan McMahan.
\newblock cpsgd: Communication-efficient and differentially-private distributed
  sgd.
\newblock In {\em Advances in Neural Information Processing Systems}, pages
  7564--7575, 2018.

\bibitem[ASZ19]{acharya2019hadamard}
Jayadev Acharya, Ziteng Sun, and Huanyu Zhang.
\newblock Hadamard response: Estimating distributions privately, efficiently,
  and with little communication.
\newblock In {\em AISTATS}, pages 1120--1129, 2019.

\bibitem[BBGN19a]{DBLP:journals/corr/abs-1906-09116}
Borja Balle, James Bell, Adri{\`{a}} Gasc{\'{o}}n, and Kobbi Nissim.
\newblock Differentially private summation with multi-message shuffling.
\newblock {\em CoRR}, abs/1906.09116, 2019.

\bibitem[BBGN19b]{balle_privacy_2019constantIKOS}
Borja Balle, James Bell, Adri{\`{a}} Gasc{\'{o}}n, and Kobbi Nissim.
\newblock Improved summation from shuffling.
\newblock {\em arXiv: 1909.11225}, 2019.

\bibitem[BBGN19c]{balle_privacy_2019}
Borja Balle, James Bell, Adri{\`{a}} Gasc{\'{o}}n, and Kobbi Nissim.
\newblock The privacy blanket of the shuffle model.
\newblock In {\em CRYPTO}, pages 638--667, 2019.

\bibitem[BBGN20]{balle_private_2020}
Borja Balle, James Bell, Adri{\`{a}} Gasc{\'{o}}n, and Kobbi Nissim.
\newblock Private summation in the multi-message shuffle model.
\newblock {\em arXiv:2002.00817}, 2020.

\bibitem[BC19]{balcer2019separating}
Victor Balcer and Albert Cheu.
\newblock Separating local \& shuffled differential privacy via histograms,
  2019.

\bibitem[BDNM05]{blum2005practical}
Avrim Blum, Cynthia Dwork, Kobbi Nissim, and Frank McSherry.
\newblock Practical privacy: the {SuLQ} framework.
\newblock In {\em PODS}, pages 128--138, 2005.

\bibitem[BEM{\etalchar{+}}17]{bittau17}
Andrea Bittau, {\'{U}}lfar Erlingsson, Petros Maniatis, Ilya Mironov, Ananth
  Raghunathan, David Lie, Mitch Rudominer, Ushasree Kode, Julien Tinn{\'{e}}s,
  and Bernhard Seefeld.
\newblock Prochlo: Strong privacy for analytics in the crowd.
\newblock In {\em SOSP}, pages 441--459, 2017.

\bibitem[Ben79]{rangetree}
Jon~Louis Bentley.
\newblock Decomposable searching problems.
\newblock {\em IPL}, 8(5):244--251, 1979.

\bibitem[BLM12]{boucheron2012concentration}
Stephane Boucheron, Gabor Lugosi, and Pascal Massart.
\newblock {\em Concentration Inequalities: a nonasmpytotic theory of
  independence}.
\newblock Clarendon Press, Oxford, 2012.

\bibitem[BLR08]{blum2008learning}
Avrim Blum, Katrina Ligett, and Aaron Roth.
\newblock A learning theory approach to non-interactive database privacy.
\newblock In {\em STOC}, pages 609--618, 2008.

\bibitem[BNS13]{beimelNissimStemmer2013}
Amos Beimel, Kobbi Nissim, and Uri Stemmer.
\newblock Private learning and sanitization: Pure vs. approximate differential
  privacy.
\newblock In {\em {APPROX}-{RANDOM}}, pages 363--378, 2013.

\bibitem[BNS18]{bun2018heavy}
Mark Bun, Jelani Nelson, and Uri Stemmer.
\newblock Heavy hitters and the structure of local privacy.
\newblock In {\em PODS}, pages 435--447, 2018.

\bibitem[BNST17]{bassily2017practical}
Raef Bassily, Kobbi Nissim, Uri Stemmer, and Abhradeep~Guha Thakurta.
\newblock Practical locally private heavy hitters.
\newblock In {\em NIPS}, pages 2288--2296, 2017.

\bibitem[BNSV15]{bun2015differentially}
Mark Bun, Kobbi Nissim, Uri Stemmer, and Salil Vadhan.
\newblock Differentially private release and learning of threshold functions.
\newblock In {\em FOCS}, pages 634--649, 2015.

\bibitem[BS15]{bassily2015local}
Raef Bassily and Adam Smith.
\newblock Local, private, efficient protocols for succinct histograms.
\newblock In {\em STOC}, pages 127--135, 2015.

\bibitem[BST14]{bassily_erm}
Raef Bassily, Adam~D. Smith, and Abhradeep Thakurta.
\newblock Private empirical risk minimization: Efficient algorithms and tight
  error bounds.
\newblock In {\em FOCS}, pages 464--473, 2014.

\bibitem[CCFC02]{charikar2002finding}
Moses Charikar, Kevin Chen, and Martin Farach-Colton.
\newblock Finding frequent items in data streams.
\newblock In {\em ICALP}, pages 693--703, 2002.

\bibitem[CDN15]{cramer2015secure}
Ronald Cramer, Ivan~Bjerre Damg{\aa}rd, and Jesper~Buus Nielsen.
\newblock {\em Secure Multiparty Computation}.
\newblock Cambridge University Press, 2015.

\bibitem[CH08]{cormode2008finding}
Graham Cormode and Marios Hadjieleftheriou.
\newblock Finding frequent items in data streams.
\newblock {\em VLDB}, 1(2):1530--1541, 2008.

\bibitem[CKS18]{cormode2018marginal}
Graham Cormode, Tejas Kulkarni, and Divesh Srivastava.
\newblock Marginal release under local differential privacy.
\newblock In {\em SIGMOD}, pages 131--146, 2018.

\bibitem[CKS19]{cormode2018answering}
Graham Cormode, Tejas Kulkarni, and Divesh Srivastava.
\newblock Answering range queries under local differential privacy.
\newblock In {\em Proceedings of International Conference on Management of Data
  (SIGMOD)}, page 1832–1834, 2019.

\bibitem[CM05a]{cormode2005improved}
Graham Cormode and Shan Muthukrishnan.
\newblock An improved data stream summary: {T}he {C}ount-{M}in sketch and its
  applications.
\newblock {\em Journal of Algorithms}, 55(1):58--75, 2005.

\bibitem[CM05b]{cormode2005s}
Graham Cormode and Shan Muthukrishnan.
\newblock What's hot and what's not: tracking most frequent items dynamically.
\newblock {\em TODS}, 30(1):249--278, 2005.

\bibitem[CM08]{chaudhuri_logistic}
Kamalika Chaudhuri and Claire Monteleoni.
\newblock Privacy-preserving logistic regression.
\newblock In {\em NIPS}, pages 289--296, 2008.

\bibitem[CMS11]{chaudhuri_emprisk}
Kamalika Chaudhuri, Claire Monteleoni, and Anand~D. Sarwate.
\newblock Differentially private empirical risk minimization.
\newblock {\em JMLR}, 12:1069--1109, 2011.

\bibitem[Cor11]{cormode2011sketch}
Graham Cormode.
\newblock Sketch techniques for approximate query processing.
\newblock {\em Foundations and Trends in Databases. NOW publishers}, 2011.

\bibitem[CPS{\etalchar{+}}12]{cormode_spatialdecomp}
Graham Cormode, Cecilia Procopiuc, Divesh Srivastava, Entong Shen, and Ting Yu.
\newblock Differentially private spatial decompositions.
\newblock In {\em ICDE}, pages 20--31, 2012.

\bibitem[CSS11]{ChanSS11}
T.{-}H.~Hubert Chan, Elaine Shi, and Dawn Song.
\newblock Private and continual release of statistics.
\newblock {\em {ACM} Trans. Inf. Syst. Secur.}, 14(3):26:1--26:24, 2011.

\bibitem[CSS12]{chan_optimal_2012}
T-H.~Hubert Chan, Elaine Shi, and Dawn Song.
\newblock Optimal lower bound for differentially private multi-part
  aggregation.
\newblock In {\em European Symposium on Algorithms}, 2012.

\bibitem[CSS13]{chaudhuri_pca}
Kamalika Chaudhuri, Anand~D. Sarwate, and Kaushik Sinha.
\newblock A near-optimal algorithm for differentially-private principal
  components.
\newblock {\em JMLR}, 14(1):2905--2943, 2013.

\bibitem[CSU{\etalchar{+}}19]{cheu_distributed_2018}
Albert Cheu, Adam~D. Smith, Jonathan Ullman, David Zeber, and Maxim Zhilyaev.
\newblock Distributed differential privacy via mixnets.
\newblock In {\em EUROCRYPT}, pages 375--403, 2019.

\bibitem[CT91]{coverthomas}
Thomas~A. Cover and Joy~M. Thomas.
\newblock {\em Elements of Information Theory}.
\newblock Wiley, 1991.

\bibitem[CY20]{CormodeYi20}
Graham Cormode and Ke~Yi.
\newblock {\em Small Summaries for Big Data}.
\newblock Cambridge University Press, 2020.

\bibitem[DJW13]{duchi2013local}
John~C Duchi, Michael~I Jordan, and Martin~J Wainwright.
\newblock Local privacy and statistical minimax rates.
\newblock In {\em FOCS}, pages 429--438, 2013.

\bibitem[DJW18]{duchi_minimax}
John~C. Duchi, Michael~I. Jordan, and Martin~J. Wainwright.
\newblock Minimax optimal procedures for locally private estimation.
\newblock {\em JASA}, 113(521):182--201, 2018.

\bibitem[DKM{\etalchar{+}}06]{dwork2006our}
Cynthia Dwork, Krishnaram Kenthapadi, Frank McSherry, Ilya Mironov, and Moni
  Naor.
\newblock Our data, ourselves: Privacy via distributed noise generation.
\newblock In {\em EUROCRYPT}, pages 486--503, 2006.

\bibitem[DKY17]{ding2017collecting}
Bolin Ding, Janardhan Kulkarni, and Sergey Yekhanin.
\newblock Collecting telemetry data privately.
\newblock In {\em NIPS}, pages 3571--3580, 2017.

\bibitem[DL09]{dwork2009differential}
Cynthia Dwork and Jing Lei.
\newblock Differential privacy and robust statistics.
\newblock In {\em STOC}, pages 371--380, 2009.

\bibitem[DMNS06]{dwork2006calibrating}
Cynthia Dwork, Frank McSherry, Kobbi Nissim, and Adam Smith.
\newblock Calibrating noise to sensitivity in private data analysis.
\newblock In {\em TCC}, pages 265--284, 2006.

\bibitem[DNPR10]{DworkNPR10}
Cynthia Dwork, Moni Naor, Toniann Pitassi, and Guy~N. Rothblum.
\newblock Differential privacy under continual observation.
\newblock In {\em STOC}, pages 715--724, 2010.

\bibitem[DNR{\etalchar{+}}09]{dwork2009complexity}
Cynthia Dwork, Moni Naor, Omer Reingold, Guy~N Rothblum, and Salil Vadhan.
\newblock On the complexity of differentially private data release: {E}fficient
  algorithms and hardness results.
\newblock In {\em STOC}, pages 381--390, 2009.

\bibitem[DNRR15]{dwork2015pure}
Cynthia Dwork, Moni Naor, Omer Reingold, and Guy~N Rothblum.
\newblock Pure differential privacy for rectangle queries via private
  partitions.
\newblock In {\em ASIACRYPT}, pages 735--751, 2015.

\bibitem[DR14a]{DworkRothBook}
Cynthia Dwork and Aaron Roth.
\newblock {\em The Algorithmic Foundations of Differential Privacy}.
\newblock Now Publishers Inc., 2014.

\bibitem[DR{\etalchar{+}}14b]{dwork2014algorithmic}
Cynthia Dwork, Aaron Roth, et~al.
\newblock The algorithmic foundations of differential privacy.
\newblock {\em Foundations and Trends{\textregistered} in Theoretical Computer
  Science}, 9(3--4):211--407, 2014.

\bibitem[Dwo06]{Dwork2006DP}
Cynthia Dwork.
\newblock Differential privacy.
\newblock In {\em ICALP}, pages 1--12, 2006.

\bibitem[EFM{\etalchar{+}}19]{erlingsson2019amplification}
{\'U}lfar Erlingsson, Vitaly Feldman, Ilya Mironov, Ananth Raghunathan, Kunal
  Talwar, and Abhradeep Thakurta.
\newblock Amplification by shuffling: From local to central differential
  privacy via anonymity.
\newblock In {\em SODA}, pages 2468--2479, 2019.

\bibitem[EFM{\etalchar{+}}20]{erlingsson2020encode}
{\'U}lfar Erlingsson, Vitaly Feldman, Ilya Mironov, Ananth Raghunathan, Shuang
  Song, Kunal Talwar, and Abhradeep Thakurta.
\newblock Encode, shuffle, analyze privacy revisited: Formalizations and
  empirical evaluation.
\newblock {\em arXiv preprint arXiv:2001.03618}, 2020.

\bibitem[ENU20]{edmondspower2020}
Alexander Edmonds, Aleksander Nikolov, and Jonathan Ullman.
\newblock The power of factorization methods in local and central differential
  privacy.
\newblock In {\em Symposium on the Theory of Computing}, 2020.

\bibitem[EPK14]{erlingsson2014rappor}
{\'U}lfar Erlingsson, Vasyl Pihur, and Aleksandra Korolova.
\newblock {RAPPOR}: Randomized aggregatable privacy-preserving ordinal
  response.
\newblock In {\em CCS}, pages 1054--1067, 2014.

\bibitem[EV03]{estan2003new}
Cristian Estan and George Varghese.
\newblock New directions in traffic measurement and accounting: Focusing on the
  elephants, ignoring the mice.
\newblock {\em TOCS}, 21(3):270--313, 2003.

\bibitem[GGI{\etalchar{+}}02]{gilbert2002fast}
Anna~C Gilbert, Sudipto Guha, Piotr Indyk, Yannis Kotidis, Sivaramakrishnan
  Muthukrishnan, and Martin~J Strauss.
\newblock Fast, small-space algorithms for approximate histogram maintenance.
\newblock In {\em STOC}, pages 389--398, 2002.

\bibitem[GGK{\etalchar{+}}20]{ghazi_pure_2020}
Badih Ghazi, Noah Golowich, Ravi Kumar, Pasin Manurangsi, Rasmus Pagh, and
  Ameya Velingker.
\newblock Pure differentially private summation from anonymous messages.
\newblock In {\em Information Theoretic Cryptography (ITC)}, 2020.

\bibitem[GK{\etalchar{+}}01]{greenwald2001space}
Michael Greenwald, Sanjeev Khanna, et~al.
\newblock Space-efficient online computation of quantile summaries.
\newblock {\em ACM SIGMOD Record}, 30(2):58--66, 2001.

\bibitem[GMPV19]{ghazi2019private}
Badih Ghazi, Pasin Manurangsi, Rasmus Pagh, and Ameya Velingker.
\newblock Private aggregation from fewer anonymous messages.
\newblock {\em arXiv:1909.11073}, 2019.

\bibitem[GPV19]{ghazi2019scalable}
Badih Ghazi, Rasmus Pagh, and Ameya Velingker.
\newblock Scalable and differentially private distributed aggregation in the
  shuffled model.
\newblock {\em arXiv:1906.08320}, 2019.

\bibitem[Gre16]{greenberg2016apple}
Andy Greenberg.
\newblock {Apple's} ``differential privacy'' is about collecting your data --
  but not your data.
\newblock {\em Wired, June}, 13, 2016.

\bibitem[HKR12]{hsu_hh}
Justin Hsu, Sanjeev Khanna, and Aaron Roth.
\newblock Distributed private heavy hitters.
\newblock In {\em ICALP}, pages 461--472, 2012.

\bibitem[HLM12]{hardt_release}
Moritz Hardt, Katrina Ligett, and Frank McSherry.
\newblock A simple and practical algorithm for differentially private data
  release.
\newblock In {\em NIPS}, pages 2339--2347, 2012.

\bibitem[HR10]{hardt2010multiplicative}
Moritz Hardt and Guy~N. Rothblum.
\newblock A multiplicative weights mechanism for privacy-preserving data
  analysis.
\newblock In {\em FOCS}, pages 61--70, 2010.

\bibitem[HRMS10]{hay_hist}
Michael Hay, Vibhor Rastogi, Gerome Miklau, and Dan Suciu.
\newblock Boosting the accuracy of differentially private histograms through
  consistency.
\newblock {\em VLDB}, 3(1-2):1021--1032, 2010.

\bibitem[IKOS06]{ishai2006cryptography}
Yuval Ishai, Eyal Kushilevitz, Rafail Ostrovsky, and Amit Sahai.
\newblock Cryptography from anonymity.
\newblock In {\em FOCS}, pages 239--248, 2006.

\bibitem[KBR16]{kairouz2016discrete}
Peter Kairouz, Keith Bonawitz, and Daniel Ramage.
\newblock Discrete distribution estimation under local privacy.
\newblock In {\em ICML}, pages 2436--2444, 2016.

\bibitem[KLL16]{karnin2016optimal}
Zohar Karnin, Kevin Lang, and Edo Liberty.
\newblock Optimal quantile approximation in streams.
\newblock In {\em FOCS}, pages 71--78, 2016.

\bibitem[KLN{\etalchar{+}}08]{kasiviswanathan2008what}
Shiva~Prasad Kasiviswanathan, Homin~K. Lee, Kobbi Nissim, Sofya Rashkodnikova,
  and Adam Smith.
\newblock What can we learn privately?
\newblock In {\em FOCS}, pages 531--540, 2008.

\bibitem[KMSZ08]{kilian2008fast}
Joe Kilian, Andr{\'e} Madeira, Martin~J Strauss, and Xuan Zheng.
\newblock Fast private norm estimation and heavy hitters.
\newblock In {\em TCC}, pages 176--193, 2008.

\bibitem[Lei11]{lei2011differentially}
Jing Lei.
\newblock Differentially private $m$-estimators.
\newblock In {\em NIPS}, pages 361--369, 2011.

\bibitem[LHR{\etalchar{+}}10]{li2010optimizing}
Chao Li, Michael Hay, Vibhor Rastogi, Gerome Milau, and Andrew McGregor.
\newblock Optimizing linear counting queries under differential privacy.
\newblock In {\em PODS}, pages 123--134, 2010.

\bibitem[LLV07]{li2007t}
Ninghui Li, Tiancheng Li, and Suresh Venkatasubramanian.
\newblock $t$-closeness: Privacy beyond $k$-anonymity and $l$-diversity.
\newblock In {\em ICDE}, pages 106--115, 2007.

\bibitem[LM12]{li2012adaptive}
Chao Li and Gerome Miklau.
\newblock An adaptive mechanism for accurate query answering under differential
  privacy.
\newblock In {\em VLDB}, volume 5(6), pages 514--525, 2012.

\bibitem[MG82]{misra1982finding}
Jayadev Misra and David Gries.
\newblock Finding repeated elements.
\newblock {\em Science of Computer Programming}, 2(2):143--152, 1982.

\bibitem[MN12]{MuthukrishnanN12}
S.~Muthukrishnan and Aleksandar Nikolov.
\newblock Optimal private halfspace counting via discrepancy.
\newblock In {\em STOC}, pages 1285--1292, 2012.

\bibitem[MP80]{munro1980selection}
J~Ian Munro and Mike~S Paterson.
\newblock Selection and sorting with limited storage.
\newblock {\em TCS}, 12(3):315--323, 1980.

\bibitem[MRL98]{manku1998approximate}
Gurmeet~Singh Manku, Sridhar Rajagopalan, and Bruce~G Lindsay.
\newblock Approximate medians and other quantiles in one pass and with limited
  memory.
\newblock {\em ACM SIGMOD Record}, 27(2):426--435, 1998.

\bibitem[MT07]{mcsherry2007mechanism}
Frank McSherry and Kunal Talwar.
\newblock Mechanism design via differential privacy.
\newblock In {\em FOCS}, pages 94--103, 2007.

\bibitem[NTZ13]{nikolov2013geometry}
Aleksandar Nikolov, Kunal Talwar, and Li~Zhang.
\newblock On the geometry of differential privacy: the sparse and approximate
  cases.
\newblock In {\em STOC}, pages 351--360, 2013.

\bibitem[NXY{\etalchar{+}}16]{nguyen2016collecting}
Thong Nguyen, Xiaokui Xiao, Yin Yang, Sui~Cheung Hui, Hyejin Shin, and Junbum
  Shin.
\newblock Collecting and analyzing data from smart device users with local
  differential privacy.
\newblock In {\em arXiv:1606.05053}, 2016.

\bibitem[O'D14]{odonnell2014analysis}
Ryan O'Donnell.
\newblock {\em Analysis of Boolean functions}.
\newblock Cambridge University Press, 2014.

\bibitem[QYL13]{qardaji}
Wahbeh Qardaji, Weining Yang, and Ninghui Li.
\newblock Understanding hierarchical methods for differentially private
  histograms.
\newblock {\em VLDB}, 6(14):1954--1965, 2013.

\bibitem[Roo06]{roos}
Bero Roos.
\newblock Binomial approximation to the {P}oisson binomial distribution: The
  {K}rawtchouk {E}xpansion.
\newblock {\em Theory of Probability and its Applications}, 45(2):258--272,
  2006.

\bibitem[Sha14]{CNET2014Google}
Stephen Shankland.
\newblock How {Google} tricks itself to protect {Chrome} user privacy.
\newblock {\em CNET, October}, 2014.

\bibitem[Smi11]{smith_statest}
Adam~D. Smith.
\newblock Privacy-preserving statistical estimation with optimal convergence
  rates.
\newblock In {\em STOC}, pages 813--822, 2011.

\bibitem[Ste20]{stemmer_locally_2019}
Uri Stemmer.
\newblock Locally private k-means clustering.
\newblock In {\em Proceedings of the 2020 Symposium on Discrete Algorithms},
  2020.

\bibitem[SU16]{steinke2013between}
Thomas Steinke and Jonathan Ullman.
\newblock Between pure and approximate differential privacy.
\newblock {\em Journal of Privacy and Confidentiality}, 7(2):3--22, 2016.

\bibitem[SU17]{steinke2017tight}
Thomas Steinke and Jonathan Ullman.
\newblock Tight lower bounds for differentially private selection.
\newblock In {\em FOCS}, pages 552--563, 2017.

\bibitem[Ull18]{ullman2018tight}
Jonathan Ullman.
\newblock Tight lower bounds for locally differentially private selection.
\newblock In {\em arXiv:1802.02638}, 2018.

\bibitem[Vad17]{vadhan2017complexity}
Salil Vadhan.
\newblock The complexity of differential privacy.
\newblock In {\em Tutorials on the Foundations of Cryptography}, pages
  347--450. Springer, 2017.

\bibitem[War65]{warner1965randomized}
Stanley~L Warner.
\newblock Randomized response: A survey technique for eliminating evasive
  answer bias.
\newblock {\em JASA}, 60(309):63--69, 1965.

\bibitem[WBLJ17]{wang_freqest}
Tianhao Wang, Jeremiah Blocki, Ninghui Li, and Somesh Jha.
\newblock Locally differentially private protocols for frequency estimation.
\newblock In {\em USENIX Security}, pages 729--745, 2017.

\bibitem[WXD{\etalchar{+}}19]{wang2019practical}
Tianhao Wang, Min Xu, Bolin Ding, Jingren Zhou, Ninghui Li, and Somesh Jha.
\newblock Practical and robust privacy amplification with multi-party
  differential privacy.
\newblock {\em arXiv:1908.11515}, 2019.

\bibitem[WZ10]{wasserman}
Larry Wasserman and Shuheng Zhou.
\newblock A statistical framework for differential privacy.
\newblock {\em JASA}, 105(489):375--389, 2010.

\bibitem[XWG10]{xiao2010differential}
Xiaokui Xiao, Guozhang Wang, and Johannes Gehrke.
\newblock Differential privacy via wavelet transforms.
\newblock {\em TKDE}, 23(8):1200--1214, 2010.

\bibitem[YB17]{ye2017optimal}
Min Ye and Alexander Barg.
\newblock Optimal schemes for discrete distribution estimation under local
  differential privacy.
\newblock In {\em ISIT}, pages 759--763, 2017.

\bibitem[YZ13]{yi2013optimal}
Ke~Yi and Qin Zhang.
\newblock Optimal tracking of distributed heavy hitters and quantiles.
\newblock {\em Algorithmica}, 65(1):206--223, 2013.

\end{thebibliography}

\end{document}